\documentclass[11pt,a4paper]{article} 

\usepackage[margin=27mm]{geometry}

\pdfoutput=1

\usepackage{amssymb,amsmath,amsfonts,amsthm, amscd, mathrsfs, helvet, mathtools}
\usepackage{bm}
\usepackage{setspace}
\usepackage{caption}
\usepackage{graphicx,verbatim}
\usepackage{enumitem}
\usepackage[usenames, dvipsnames]{color}
\usepackage[mathcal]{euscript}
\usepackage[utf8]{inputenc}
\usepackage[T1]{fontenc}

\usepackage{hyperref}
\hypersetup{
    colorlinks=true, 
    linktoc=all,     
    linkcolor=myPurple,  
    citecolor=myGreen,
    menucolor=black,
    linktocpage=true,
}

\usepackage{tikz}    
\usetikzlibrary{arrows.meta}
\usetikzlibrary{cd}
\usetikzlibrary{decorations.pathreplacing,positioning}
\usetikzlibrary{decorations.markings,shapes.geometric,shapes.misc}

\usepackage{pgfplots}
\usepackage{tikz-3dplot}
\tdplotsetmaincoords{60}{115}
\pgfplotsset{compat=newest}

\theoremstyle{plain}
\newtheorem{theorem}{Theorem}[section]
\newtheorem*{theorem*}{Theorem}
\newtheorem{proposition}[theorem]{Proposition}
\newtheorem*{proposition*}{Proposition}
\newtheorem{definition}[theorem]{Definition}
\newtheorem{lemma}[theorem]{Lemma}
\newtheorem{corollary}[theorem]{Corollary}
\newtheorem{def-lem}[theorem]{Definition/Lemma}
\newtheorem{def-prop}[theorem]{Definition/Proposition}

\theoremstyle{remark}

\newtheorem{remarkx}[theorem]{Remark}

\newenvironment{remark}
  {\pushQED{\qed}\remarkx}
  {\popQED\endremarkx}

\newcounter{assumption}

\newenvironment{assumption}
{
  \noindent
  \refstepcounter{assumption}%
  \textbf{Assumption \theassumption\ }%
  \begin{itshape}%
}%
{
  \end{itshape}%
}

\definecolor{myPurple}{rgb}{.4,0.2,0.8}

\definecolor{brightRed}{rgb}{1,0,0}

\definecolor{myBlue}{rgb}{0,0,1}

\definecolor{myGreen}{rgb}{0,0.7,0}

\newcommand{\nord}[1]{\mathopen{:}#1\mathclose{:}}

\DeclareMathOperator{\id}{id}
\DeclareMathOperator{\boplus}{\bar\oplus}

\DeclareMathOperator{\End}{End}

\DeclareMathOperator{\Hom}{Hom}
\DeclareMathOperator{\Map}{Map}

\DeclareMathOperator{\wgt}{wt}
\DeclareMathOperator{\wgtb}{\overline{wt}}

\DeclareMathOperator{\ext}{ext}
\DeclareMathOperator{\Ind}{Ind}
\DeclareMathOperator{\im}{im}

\DeclareMathOperator{\supp}{supp}

\DeclareMathOperator{\Sym}{Sym}
\DeclareMathOperator{\bSym}{\overline{Sym}}
\DeclareMathOperator{\Conf}{Conf}

\newcommand{\longhookrightarrow}{\lhook\joinrel\relbar\joinrel\rightarrow}

\def\tp{{\scalebox{.7}{$\scriptstyle +$}}}
\def\tm{{\scalebox{.7}{$\scriptstyle -$}}}
\def\tmo{{\scalebox{.7}{$\scriptstyle -1$}}}

\def\cent{{\bf 1}}
\def\kent{{\bf 1}}

\def\t{\mathsf t}

\def\a{\mathsf{a}}
\def\fl{\mathfrak{a}}

\def\b{\mathsf{b}}

\def\conf{\Omega}
\def\d{\mathsf{d}}

\def\f{\mathfrak{f}}
\def\g{\mathfrak{g}}

\def\x{\mathsf{x}}

\def\R{\mathcal R}

\def\hg{\widehat{\fl}}
\def\hbg{\widehat{\bar{\fl}}}

\def\G{\mathcal{G}}

\def\vac{|0\rangle}

\def\mcF{\mathcal F}
\def\U{\mathcal U}

\def\CE{{\rm CE}}
\def\bCE{\overline{\rm CE}}
\def\chain{\bullet}
\def\Ber{\mathfrak B}

\def\Dz{{\widehat{\Sigma}}}

\def\Alg{\mathsf{Alg}}
\def\Top{\mathsf{Top}}

\def\Vec{\mathsf{Vec}}
\def\L{\mathcal{L}}

\def\uLocLie{\mathsf{uLocLie}}
\def\dgLie{\mathsf{dgLie}}
\def\udgLie{\mathsf{udgLie}}
\def\udgVec{\mathsf{udgVec}}
\def\uLie{\mathsf{uLie}}
\def\dgVec{\mathsf{dgVec}}
\def\PFac{\mathsf{PFac}}

\def\dg{\text{\scriptsize \rm DG}}

\newcommand{\SimTo}{\xrightarrow{\raisebox{-0.3 em}{\smash{\ensuremath{\text{\tiny$\cong$}}}}}}
\newcommand{\qSimTo}{\xrightarrow{\raisebox{-0.3 em}{\smash{\ensuremath{\text{\tiny$\simeq$}}}}}}

\newcommand{\ms}[1]{{\mathsf #1}}

\newcommand{\wt}[1]{\widetilde #1}

\newcommand{\ul}[1]{{\underline{#1}}}

\def\A{\mathcal{A}}

\def\CC{\mathbb{C}}
\def\CP{\mathbb{C}P^1} 
\def\RR{\mathbb{R}}

\def\VV{\mathbb{V}}
\def\hVV{\mathbb{F}}
\def\ZZ{\mathbb{Z}}

\def\O{\mathcal{O}}

\def\R{\mathcal{R}}

\def\L{\mathcal{L}}

\def\hV{\mathcal F}
\def\X{\mathsf{X}}
\def\Y{\mathsf{Y}}

\def\ii{\mathsf{i}}

\title{%
Full universal enveloping vertex algebras from factorisation
}

\author{%
Beno\^{\i}t Vicedo\vspace{4mm}\\
{\small Department of Mathematics, University of York,}\\
{\small Heslington, York YO10 5GH, United Kingdom}\vspace{2mm}\\
{\small \begin{tabular}{ll}
\texttt{benoit.vicedo@gmail.com}\\
\vspace{2mm}
\end{tabular}
}
}

\numberwithin{equation}{section}

\begin{document} 

\maketitle

\begin{abstract}
\noindent We give a systematic construction of the symmetries, or observables in the vacuum sector, of a full conformal field theory on an arbitrary real two-dimensional conformal manifold $\Sigma$. Specifically, we construct a prefactorisation algebra on $\Sigma$ which locally encodes the full (non-chiral) version $\hVV^{\fl,\alpha} = \VV^{\fl,\alpha} \otimes \bar \VV^{\fl,\alpha}$ of a universal enveloping vertex algebra $\VV^{\fl,\alpha}$, where $\fl$ is a finite-dimensional vector space labelling the set of fields and $\alpha$ is a $2$-cocycle controlling the central extension of their Lie brackets. Our construction provides a unified treatment of the three canonical examples of (full) universal enveloping vertex algebras \--- Kac-Moody, Virasoro and $\beta\gamma$ system \--- using the notion of unital local Lie algebra. By using the coordinate-invariant nature of prefactorisation algebras we derive an analogue of Huang's change of variable formula for full vertex algebras. We give a careful treatment of (both Euclidean and Lorentzian) reality conditions in this formalism which allows us, in the Kac-Moody and Virasoro cases, to construct a Hermitian sesquilinear form on these full vertex algebras by using the factorisation product to the global observables on $S^2$. We also give an explicit derivation of Borcherds type identities and a construction of the operator formalism for $\hVV^{\fl,\alpha}$.
\end{abstract}

\setcounter{tocdepth}{2}
\hypersetup{linkcolor=myPurple}
\tableofcontents


\section{Introduction} \label{sec: intro}

Two-dimensional conformal field theory \cite{Belavin:1984vu} has been extensively studied by physicists and mathematicians alike over the past four decades. Its original applications in physics ranged from the study of critical phenomena in two-dimensional statistical systems to string theory, see for instance \cite{CFTbook}, but it has also found far-reaching applications in mathematics, from the monstrous moonshine \cite{Borcherds:1983sq} to the geometric Langlands correspondence \cite{FrLanglands, Frbook}.

An important distinction to be made is between \emph{chiral} and \emph{full} two-dimensional conformal field theories. Indeed, in both the Euclidean and Lorentzian settings, the infinite dimensional symmetry algebra of a two-dimensional conformal field theory involves a direct sum of two copies of the Virasoro algebra, see e.g. \cite{Schottenloher}.
The collection of fields which transform trivially with respect to either of these two copies generates the so-called \emph{chiral} and \emph{anti-chiral} sectors of the conformal field theory, respectively.
By contrast, the \emph{full} conformal field theory deals with all of the fields, transforming under both copies of the Virasoro algebra.

\medskip
The main purpose of this paper is to initiate the construction and study of full conformal field theories within the modern framework of prefactorisation algebras \cite{CGBook1, CGBook2} recently introduced by Costello and Gwilliam. We focus on describing the full/non-chiral versions of a broad class of vertex algebras, known as universal enveloping vertex algebras \cite{Primc}, which are defined as Verma modules over infinite-dimensional Lie algebras, including centrally extended loop algebras, the Virasoro algebra and the infinite-dimensional Weyl algebra. These full vertex algebras encode the vacuum sectors, or symmetry algebras, of many rational and logarithmic conformal field theories, including the Wess-Zumino-Witten model \cite{WessZumino, WittenWZ1, WittenWZ2, Novikov} and its more recent generalisations associated to non-semisimple Lie algebras \cite{Babichenko:2012uq, Quella:2020uhk}, minimal models and the $\beta\gamma$ system; see \cite{CFTbook} for an extensive review.

The prefactorisation algebra perspective that we follow has several important advantages. It provides a coordinate-invariant description of (full) vertex algebras; this will, in particular, allow us to derive a generalisation of Huang's change of variable formula \cite{Huang} for full vertex algebras which leads to a simple formulation of the operator formalism in full conformal field theories. It also provides an elegant geometric formulation of (full) vertex algebras allowing us to work on an arbitrary two-dimensional conformal manifold; for instance, by paying close attention to reality conditions we will construct a canonical Hermitian sesquilinear form on the full vertex algebra using the factorisation product to the global observables on $S^2$.

\medskip

In the remainder of the introduction we begin in \S\ref{sec: approaches to CFT} by providing a general overview of different mathematical approaches to conformal field theory and their interconnections. In \S\ref{sec: orientation double} we explain the main idea for realising the notion of full vertex algebras in the framework of prefactorisation algebras, based on an observation from \cite{FuchsSchweigert}. In \S\ref{sec: class of VOAs} we describe the class of vertex algebras whose non-chiral versions we will formulate using prefactorisation algebras, and in \S\ref{sec: paper outline} we give an outline the content of the paper to help guide the reader.

\subsection{Approaches to conformal field theory} \label{sec: approaches to CFT}

There are various mathematical formulations of both chiral and full two-dimensional conformal field theory,
each of which is
based on a different axiomatic framework for quantum field theory. These include the G\r{a}rding--Wightman axioms \cite{Wightman, PCTbook} or the Osterwalder--Schrader axioms \cite{OS1, OS2} related by analytic continuation, the Haag--Kastler axioms \cite{HaagKastler} or the closely related Brunetti--Fredenhagen--Verch axioms \cite{BFV} in the locally covariant case, and the Atiyah--Segal--Witten axioms \cite{SegalCFT, AtiyahTQFT, WittenTQFT} for functorial quantum field theory.

\subsubsection{(Full) vertex operator algebras} \label{sec: full VOA intro}

A vertex (operator) algebra is a vector space $V$ equipped with a state-field correspondence, also known as a field map or a
vertex operator map, $Y : V \otimes V \to V(\!(\zeta)\!)$, valued in formal Laurent series in the formal variable $\zeta$ with coefficients in $V$, satisfying certain axioms \cite{Borcherds:1983sq, FHL, Huang, KacVA, FBbook, LepowskyLi}.
The vertex operator map already encodes an `algebraic' version of the Wightman fields \cite{PCTbook} for a chiral two-dimensional conformal field theory, see \cite[\S1]{KacVA}. If, moreover, $V$ is \emph{unitary} in the sense of \cite{DongLin}, then under certain extra mild assumptions one can use the vertex operator map to associate \cite{CCHW, Raymond:2022riz} to each state $a \in V$ an operator-valued distribution on $S^1$, i.e. a Wightman field on $S^1$, acting on a dense open subset of the Hilbert space completion $\mathcal H_V$ of $V$ with respect to its positive-definite sesquilinear form. In the very recent work \cite{Carpi:2024vpp} a rigorous connection is also established in the broader non-unitary context and without the need for any additional technical assumptions.

Vertex operator algebras only encode \emph{chiral} two-dimensional conformal field theories, but a number of closely related extensions describing full two-dimensional conformal field theories have also been developed: a non-chiral version of vertex algebras was first introduced in \cite{KapustinOrlov} and further studied in \cite{Rosellen} under the name \emph{OPE-algebras}, the notion of a \emph{full field algebra} was defined in \cite{Huang:2005gz, Kong:2006wa}, that of a \emph{full vertex algebra} in \cite{Moriwaki:2020dlj, Moriwaki:2020cxf} and that of a \emph{non-chiral vertex operator algebra} in \cite{Singh:2023mom}. All of these non-chiral versions of vertex (operator) algebras are closely related and we refer the reader to the original articles for the various comparisons.
A notion of unitarity was also recently introduced in \cite{Adamo:2024etu} to relate the algebraic approach of  \cite{Moriwaki:2020dlj, Moriwaki:2020cxf} with an analytic one based on the Osterwalder--Schrader axioms.

\subsubsection{Conformal nets}

A conformal net \cite{FrohlichGabbiani, FredenhagenJorss, KawaigashiLongo1},
by constrast, encodes the local observables of a chiral two-dimensional conformal field theory as a conformal Haag--Kastler net \cite{HaagKastler} on $S^1$, i.e. as a functor from the category of open intervals in $S^1$ to that of von Neumann subalgebras of $B(\mathcal H)$ for a fixed Hilbert space $\mathcal H$.
There is also a `coordinate-free' formulation of conformal nets \cite{BDH} built instead on the axiomatic
framework of Brunetti--Fredenhagen--Verch \cite{BFV}, see also \cite{FewsterVerch, BDRY, Rejzner, BSW}, which rather than working on a fixed spacetime ($S^1$ in this case) treats all contractible compact $1$-dimensional manifolds on an equal footing.

Conformal nets on $S^1$ similarly only encode \emph{chiral} two-dimensional conformal field theories, whereas full two-dimensional conformal field theories are encoded instead in terms of conformal Haag-Kastler nets on two-dimensional Minkowski space $\RR^{1,1}$ or the Einstein cylinder $\mathcal E$.
The relationship between a full two-dimensional conformal field theory and its chiral and anti-chiral sectors was studied in the conformal nets setting in \cite{Rehren1, KawaigashiLongo2, BKL},
see also \cite{Rehren2} for a review.
This was also recently studied from a purely categorical perspective in \cite{BGS} within a general operadic reformulation \cite{BSW} of the Brunetti--Fredenhagen--Verch framework.

\subsubsection{Segal's functorial approach}

In the chiral setting, unitary vertex operator algebras and conformal nets are intimately related \cite{Capri:2015fga, Gui}, see also \cite{Carpi:2023onx} in the supersymmetric case.
A more geometric perspective on this relationship was also obtained in \cite{Tener1, Tener2} based on ideas from \cite{Three-tierCFT}, using Segal’s functorial definition of a chiral conformal field theory \cite{SegalCFT} to interpolate between the two.
Although a similar comparison in the non-chiral case does not yet exist, note that Segal's axioms \cite{SegalCFT} can encode both chiral and full two-dimensional conformal field theories.

\subsubsection{Factorisation algebras}

In the present work we will, instead, be working in the modern formulation of quantum field theory based on factorisation algebras due to Costello--Gwilliam \cite{CGBook1, CGBook2}.
It is known to be closely related, at least on a heuristic level \cite[\S1.4]{CGBook1}, to all of the frameworks mentioned above. The precise relationship between factorisation algebras in the Lorentzian setting and the Brunetti--Fredenhagen--Verch framework was clarified in the recent series of works \cite{Gwilliam:2017ses, Benini:2019ujs, Benini:2022mlb, Gwilliam:2022vja}.
In the setting of chiral two-dimensional conformal field theories, conformal nets were also shown in \cite{Henriques:2016ajg} to be a particular instance of factorisation algebras on $S^1$.

The notion of (pre)factorisation algebra \`a la Costello--Gwilliam encapsulates, in a simple and general axiomatic framework, the algebra of observables in general quantum field theories \cite{CGBook1, CGBook2}. Its development was originally inspired by the theory of factorisation algebras \`a la Beilinson--Drinfel'd \cite{BDbook}, see also \cite{Francis:2012} and the book \cite{FBbook} for a nice review, which provides a sheaf-theoretic coordinate-free formulation of vertex (operator) algebras, and hence of chiral two-dimensional conformal field theory.
Indeed, the axioms of a vertex algebra also admit a very natural and elegant reformulation in terms of \emph{holomorphic} prefactorisation algebras \`a la Costello--Gwilliam on $\CC$, see \cite[\S5]{CGBook1} and \cite{Brugmann1}. There is also a close connection with the theory of geometric vertex (operator) algebra first introduced in \cite{Huang}, see \cite{Brugmann2}.
Note that the relation between the two different notions of factorisation algebras of Beilinson--Drinfel'd and Costello--Gwilliam remains to be elucidated, see however \cite{HennionKapranov} for the comparison in the \emph{locally constant} setting corresponding to topological field theories \cite{LurieHA}. From now on, we shall always consider the notion of (pre)factorisation algebra in the sense of Costello--Gwilliam and we will refer to these simply as (pre)factorisation algebras.

Even though the origins of factorisation algebras are deeply rooted in the theory of vertex algebras, hence in chiral two-dimensional conformal field theory, their scope of applicability is much broader, extending to general quantum field theories in arbitrary space-time dimensions. This is very reminiscent of the axiomatic framework for quantum field theory in curved space-time of Hollands--Wald \cite{Hollands:2009bke, Hollands:2014eia, Hollands:2023txn}, see also \cite{HollandsOlbermann, HollandsHollands}, in which the \emph{operator product expansion} takes centre stage, exactly as in a vertex (operator) algebra $V$ where it is given by the state-field correspondence $Y : V \otimes V \to V(\!(\zeta)\!)$. Indeed, factorisation algebras can encode operator product expansions of local observables in a general quantum field theory \cite[\S10]{CGBook2}.

\subsection{The orientation double construction} \label{sec: orientation double}

Let $\Sigma$ be a real two-dimensional conformal manifold. The decomposition of a full conformal field theory on $\Sigma$ into its chiral and anti-chiral sectors is often motivated from a dynamical point of view by the fact that the space of solutions of the classical field theory factorises into left and right moving solutions in the Lorentzian setting, or holomorphic and anti-holomorphic solutions in the Euclidean setting. We will instead adopt the more geometric perspective, as first advocated in \cite{FuchsSchweigert}, which relates the chiral and anti-chiral sectors of a full conformal field theory to the two possible choices of orientations one can locally make on $\Sigma$. For concreteness we will work with a Euclidean conformal manifold $\Sigma$, in which case the choice an orientation corresponds to a choice of complex structure. However, we will also deal with the Lorentzian case in \S\ref{sec: Fourier modes} by working in the Hamiltonian formalism and restricting fields to a Cauchy surface.

The orientation bundle $\text{Or}(\Sigma) \rightarrow \Sigma$ is the principal $\ZZ_2$-bundle whose fibre over any $p \in \Sigma$ consists of two points corresponding to the two possible orientations at $p$. The \emph{double} of $\Sigma$, denoted $\pi : \Dz \to \Sigma$, is defined as the quotient of $\text{Or}(\Sigma)$ by the relation which identifies the pair of points above any $p \in \partial \Sigma$.
By construction, $\Dz$ is a 2-dimensional conformal manifold without boundary which is naturally oriented and thus carries a canonical complex structure. Moreover, $\Dz$ comes naturally equipped with an orientation reversing involution $\tau_E : \Dz \to \Dz$, such that $\pi \tau_E = \pi$, induced by the involution $\text{Or}(\Sigma) \to \text{Or}(\Sigma)$ exchanging the pair of points in each $\ZZ_2$ fibre. In particular, the fixed point set of $\tau_E$ corresponds to the boundary $\partial \Sigma \hookrightarrow \Dz$.

Following \cite{FuchsSchweigert}, we can describe a full conformal field theory on $\Sigma$ using a \emph{chiral} conformal field theory on the oriented cover $\Dz$. For instance, the double of the $2$-sphere $\Sigma = S^2$ is given by the disjoint union $\Dz = \Sigma_+ \sqcup \Sigma_-$ of two copies of $\Sigma$ equipped with opposite orientations and the chiral and anti-chiral sectors of the full conformal field theory then correspond to $\Sigma_+$ and $\Sigma_-$, respectively. We will later specialise to this particular case. The subscript `$E$' on the orientation reversing involution $\tau_E : \Dz \to \Dz$ refers to the fact that it then describes reality conditions in $E$uclidean signature. In this case we shall also consider a different orientation reversing involution $\tau_L : \Dz \to \Dz$ which acts individually on each of the two spheres $\Sigma_\pm$, by reflection through the equatorial plane, rather than exchanging them as $\tau_E$ does:
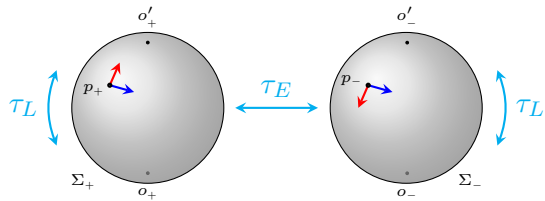
\begin{equation*}
\begin{tikzpicture}
  \shade[ball color = lightgray, opacity = 0.5] (0,0,0) circle (1);
  \tdplotsetrotatedcoords{0}{0}{120};
 
  \draw (0,0) circle (1) node[black, below left=7.5mm and 5.5mm]{\tiny $\Sigma_\tp$};
  \draw[tdplot_rotated_coords, fill] (0,0,1) node[above=.7mm]{\tiny $o'_\tp$} circle (0.5pt);
  \draw[tdplot_rotated_coords, fill, black!60] (0,0,-1) node[black, below=.7mm]{\tiny $o_\tp$} circle (0.5pt);

  \draw[thick, blue, -stealth] (-.5,.3) -- (-.18, .21);
  \draw[thick, red, -stealth] (-.5,.3) -- (-.37,.6);
  \draw[thick] (-.5,.3) node[below left=-1.5mm and -.7mm]{\tiny $p_\tp$} circle (0.02);

\def\x{3.4}
  \shade[ball color = lightgray, opacity = 0.5] (\x,0,0) circle (1);
  \tdplotsetrotatedcoords{0}{0}{120};

  \draw (\x,0) circle (1) node[black, below right=7.5mm and 5.5mm]{\tiny $\Sigma_\tm$};
  \draw[tdplot_rotated_coords, fill] (\x,-.3,1) node[above=.7mm]{\tiny $o'_\tm$} circle (0.5pt);
  \draw[tdplot_rotated_coords, fill, black!60] (\x,-.3,-1) node[black, below=.7mm]{\tiny $o_\tm$} circle (0.5pt);

  \draw[thick, blue, -stealth] (\x-.5,.3) -- (\x-.18, .21);
  \draw[thick, red, -stealth] (\x-.5,.3) -- (\x-.63,0);
  \draw[thick] (\x-.5,.3) node[above left=-1.5mm and -.7mm]{\tiny $p_\tm$} circle (0.02);

  \draw [thick, cyan, stealth-stealth] (1.15,0) -- node[above]{$\tau_E$} (0.35*\x + 1.05,0);
  \draw [thick, cyan, stealth-stealth] (-1.2,-.5) to [out=110,in=-110] node[left]{$\tau_L$} (-1.2,.5);
  \draw [thick, cyan, stealth-stealth] (\x+1.2,-.5) to [out=70,in=-70] node[right]{$\tau_L$} (\x+1.2,.5);
\end{tikzpicture}
\end{equation*}
We use the subscript `$L$' here since this will correspond to reality conditions in the $L$orentzian setting (also in the Euclidean setting after Wick rotation from the Lorentzian setting), which send the chiral and anti-chiral fields to themselves rather than exchanging them.

\subsection{Universal enveloping vertex algebras} \label{sec: class of VOAs}

A general class of vertex algebras can be constructed using a vertex-algebraic analogue of the universal enveloping algebra construction for Lie algebras, starting from the simpler notion of a vertex Lie algebra, also known as a Lie conformal algebra, introduced in \cite{Primc}; see also \cite{KacVA, FBbook}. Specifically, a \emph{vertex Lie algebra} $\mathscr L$ is determined by a vector space of fields $\a(z) = \sum_{n \in \ZZ} \a_{(n)} z^{-n-1}$ labelled by $\a \in \fl$ for some finite-dimensional vector space $\fl$, which together with all of their derivatives $\partial_z^k \a(z)$ for $k \in \ZZ_{\geq 1}$ is closed under taking the singular part of their operator product expansions. Then the associated \emph{universal enveloping vertex algebra} \cite{Primc} is the vertex algebra $\VV(\mathscr L)$ that is freely generated by the fields of the vertex Lie algebra $\mathscr L$ under taking derivatives and normal ordered products, modulo the relations encoded in $\mathscr L$. More precisely, a general element of the vector space $\VV(\mathscr L)$ is given by a linear combination of states of the form $\a^r_{(-m_r)} \ldots \a^1_{(-m_1)} \vac$ where $\a^i \in \fl$ and $m_i \in \ZZ_{\geq 1}$ for $i \in \{ 1,\ldots, r\}$, with associated fields
\begin{equation} \label{Y map intro}
Y \big( \a^r_{(-m_r)} \ldots \a^1_{(-m_1)} \vac, z \big) = \nord{\frac{1}{(m_r -1)!} \partial_z^{m_r - 1} \a^r(z) \ldots \frac{1}{(m_1 -1)!} \partial_z^{m_1 - 1} \a^1(z)} .
\end{equation}
The fact that this formula endows $\VV(\mathscr L)$ with the structure of a vertex algebra is an immediate application of the recontruction theorem, see for instance \cite[Theorem 4.5]{KacVA} or \cite[\S2.3.11]{FBbook}. The proof of the latter, however, relies on various technical results including Dong's lemma. We will see that the formula \eqref{Y map intro} (and in particular its counterpart for full vertex algebras) has a very simple and elegant geometric origin from the prefactorisation algebra perspective on (full) vertex algebras. Therefore, rather than invoke the reconstruction theorem (for which there does not appear to be a full vertex algebra analogue in the literature) to construct (full) vertex algebras from holomorphic prefactorisation algebras, as is done for instance in \cite[\S5]{CGBook1} in the chiral setting, we will show directly that the formula \eqref{Y map intro} emerges very naturally from the structure of the prefactorisation algebra itself.

\subsubsection{Main examples} \label{sec: vertex Lie examples}

We will consider the following three canonical examples of vertex Lie algebras:

\paragraph{Kac-Moody:} For any finite-dimensional complex Lie algebra $\g$ equipped with a symmetric invariant bilinear form $\kappa = \langle \cdot, \cdot \rangle : \g \otimes \g \to \CC$, the Kac-Moody vertex Lie algebra $\mathscr L$ is generated by fields $\X(z)$ for each $\X \in \g$ with operator product expansions given, for any $\X, \Y \in \g$, by
\begin{equation*}
\X(z) \Y(w) \, \sim \, \frac{[\X, \Y](w)}{z-w} + \frac{\langle \X, \Y \rangle}{(z-w)^2}.
\end{equation*}
When $\g$ is reductive the associated universal enveloping vertex algebra $\VV(\mathscr L)$ is known as the affine (Kac-Moody) vertex algebra $\VV_\kappa(\g)$, see for instance \cite{LepowskyLi, KacVA, FBbook}.
This includes the Heisenberg vertex algebra as a special case for the trivial Lie algebra $\CC$ with invariant bilinear form given by multiplication.
Note, however, that the Lie algebra $\g$ need not be reductive. For instance, one may take $\g$ to be a Takiff algebra \cite{Takiff}, i.e. $\g = \mathcal T^n \f \coloneqq \f[t]/ t^{n+1} \f[t]$ for some reductive Lie algebra $\f$ and $n \in \ZZ_{\geq 0}$. This can be equipped with a non-degenerate symmetric invariant bilinear form \cite{Quella:2020uhk} and the associated universal enveloping vertex algebra $\VV(\mathscr L)$ gives rise to examples of (logarithmic) chiral conformal field theories first studied in \cite{Babichenko:2012uq}, see also \cite{Rasmussen:2017eus, Rasmussen:2019zfu}, and related to generalised WZW models in \cite{Quella:2020uhk}.

\paragraph{Virasoro:} The Virasoro vertex Lie algebra $\mathscr L$ is generated by one field, often called $T(z)$, whose operator product expansion with itself depends on the central charge $c \in \RR$ and reads
\begin{equation*}
T(z) T(w) \, \sim \, \frac{\partial_w T(w)}{z-w} + \frac{2 T(w)}{(z-w)^2} + \frac{\frac 12 c}{(z-w)^4}.
\end{equation*}
The associated universal enveloping vertex algebra $\VV(\mathscr L)$ is the Virasoro vertex algebra $\text{Vir}_c$.

\paragraph{$\bm \beta \bm \gamma$ system:} The Weyl vertex Lie algebra $\mathscr L$ is generated by two free bosonic ghost fields $\beta(z)$ and $\gamma(z)$ with operator product expansion
\begin{equation*}
\beta(z) \gamma(w) \, \sim \, \frac{1}{z-w}.
\end{equation*}
The associated universal enveloping vertex algebra $\VV(\mathscr L)$ is the $\beta\gamma$ vertex algebra $\VV_{\beta\gamma}$, named after the $\beta\gamma$ bosonic ghost system in the physics literature.

\subsubsection{Analytic Langlands correspondence}

One important motivation for studying the affine vertex algebra $\VV_\kappa(\g)$ comes from the fact that at the critical level $\kappa = \kappa_c$ it plays a central role in the study of the geometric Langlands correspondence, as formulated by Beilinson and Drinfel'd in their seminal work \cite{BDgeometric} on the quantisation of Hitchin's integrable system. For an extensive review on the subject we refer the reader to \cite{FrLanglands, Frbook} and references therein.
Much more recently, an analytic version of the geometric Langlands correspondence was formulated and studied by Etingof, Frenkel and Kazhdan in the series of works \cite{Frenkel:2018dic, Etingof:2019pni, Etingof:2021eub, Etingof:2021eeu, Etingof:2023drx}, implementing and extending earlier ideas of Teschner \cite{Teschner:2017djr}.
The gauge theoretic interpretation of the geometric Langlands correspondence as electric-magnetic duality in twisted $\mathcal N=4$ super Yang-Mills theory in four dimensions \cite{KapustinWitten} has also been recently extended in \cite{Gaiotto:2021tsq} to the analytic version. It is expected, see for instance \cite{Teschner:2017djr, Frenkel:2018dic}, that the role of the vertex algebra $\VV_{\kappa_c}(\g)$ in the geometric Langlands correspondence should be replaced by that of the full affine vertex algebra $\hVV_{\kappa_c}(\g) = \VV_{\kappa_c}(\g) \otimes \overline{\VV}_{\kappa_c}(\g)$ at critical level $\kappa = \kappa_c$ in the analytic version.

\subsection{Outline of the paper} \label{sec: paper outline}

In \S\ref{sec: prefac alg} we construct a local Lie algebra $\L^{\Dz}_\alpha$ over the complex curve $\Dz$, which is an analogue of the notion of vertex Lie algebra in the present prefactorisation algebra setting. Since we are interested in describing a full vertex algebra on $\Sigma$, to forget about the orientation introduced by passing from $\Sigma$ to the orientation double $\Dz$, see \S\ref{sec: orientation double}, the main object of interest will be the pushforward $\L^\Sigma_\alpha$ of the local Lie algebra along the projection $\pi : \Dz \to \Sigma$.
We focus on the three examples of vertex Lie algebras described in \S\ref{sec: vertex Lie examples}. Importantly, in all of these the identity operator appears as one term in the operator product expansions. Correspondingly, the local Lie algebra $\L^\Sigma_\alpha$ will be centrally extended and to keep track of its crucial central extension $\kent$ we will work with the notion of \emph{unital} local Lie algebras. The counterpart of the universal enveloping vertex algebra construction is the prefactorisation envelope $\U\L^\Sigma_\alpha$ which we adapt from the case of local Lie algebras in \cite[Definition 3.6.1]{CGBook1} to the unital case. The technical details are relegated to Appendix \ref{sec: twisted PF env}. This allows us to deal with (full) vertex algebras over $\CC$ rather than over the base ring $\CC[s \kent]$, as in \cite[Definition 5.5.2]{CGBook1}. Our construction thus unifies three key examples of holomorphic prefactorisation algebras in a single framework: the Kac-Moody case \cite[\S5.5]{CGBook1}, the Virasoro case \cite{Williams-Vir} and the $\beta\gamma$ system case \cite[\S5.4]{CGBook1}.

In \S\ref{sec: vertex alg} we consider the vector space $\hVV^{\fl,\alpha} = \VV^{\fl,\alpha} \otimes \bar \VV^{\fl,\alpha}$ where $\VV^{\fl,\alpha}$ denotes the vector space underlying the universal enveloping vertex algebra $\VV(\mathscr L)$ and $\bar\VV^{\fl,\alpha}$ its anti-chiral analogue. We realise this vector space in the limit $\hV^{\fl,\alpha}_p$ of the prefactorisation algebra $\U\L^\Sigma_\alpha$ over any point $p \in \Sigma^\circ$ in the interior $\Sigma^\circ$ of $\Sigma$. In Proposition \ref{prop: vertex operator general} we show that the state-field correspondence at a point $p \in \Sigma^\circ$ is determined by the factorisation product $\hV^{\fl,\alpha}_q \to \U\L^\Sigma_\alpha(Y)$ for some annulus $q \in Y \subset \Sigma^\circ$ around $p$.
In particular, the formula for the state field correspondence of an arbitrary state in $\hVV^{\fl, \alpha}$, i.e. the analogue of \eqref{Y map intro}, has a very elegant geometric origin from the prefactorisation algebra perspective.
We derive a number of properties satisfied by the modes of the vertex operators of arbitrary states in $\hVV^{\fl, \alpha}$, including `Borcherds type' identities when one of the states involved is (anti-)chiral, i.e. belongs to $\VV^{\fl,\alpha}$ or $\bar \VV^{\fl,\alpha}$. In \S\ref{sec: chiral states} we introduce the (anti-)conformal states of $\hVV^{\fl,\alpha}$ and use these to derive a non-chiral version of Huang's ``change of variable'' formula for the state-field correspondence of $\hVV^{\fl,\alpha}$, see Theorem \ref{thm: change of coordinate} and Corollary \ref{cor: Huang's formula}. In \S\ref{sec: invariant bilinear} we specialse to $\Sigma = S^2$ and use the factorisation product to $\U\L^\Sigma_\alpha(S^2)$ to define an invariant bilinear form on $\hVV^{\fl,\alpha}$. We then use this to prove in Theorem \ref{thm: F is a full VOA} that $\hVV^{\fl,\alpha}$ satisfies the axioms of a full vertex algebra in the sense of \cite{Moriwaki:2020dlj, Moriwaki:2020cxf}. Finally, in \S\ref{sec: unitarity} we discuss reality conditions on $\hVV^{\fl,\alpha}$ and define a Hermitian sesquilinear form on $\hVV^{\fl, \alpha}$.

In \S\ref{sec: Fourier modes} we make use of the framework developed in \S\ref{sec: vertex alg} to describe the operator formalism for $\hVV^{\fl,\alpha}$, cf. \cite[\S6]{CFTbook}. In particular, we use the change of variable formula applied to the conformal transformation from the plane to the cylinder in order to describe Fourier series decompositions of quantum operators on $S^1$ associated to generic states in $\hVV^{\fl,\alpha}$. In \S\ref{sec: Borcherds Fourier} we also derive `Borcherds type' identities for the Fourier modes of states in $\hVV^{\fl,\alpha}$, again in the case when one of the states involved is (anti-)chiral. Finally, in \S\ref{sec: reality conditions} we discuss reality conditions, define a natural notion of Hermitian conjugate for quantum operators on $S^1$ and show that it corresponds to the adjoint with respect to the Hermitian sesquilinear from on $\hVV^{\fl, \alpha}$.

\paragraph{Acknowledgements.}
I would like to thank Simen Bruinsma and Alex Schenkel for valuable discussions in relation to various aspects of this work, and also Simen Bruinsma for detailed feedback on a first draft of the paper. I am also grateful to the anonymous referees for their valuable comments and suggestions, which greatly helped to broaden the scope of this paper. I gratefully acknowledge the support of the Leverhulme Trust through a Leverhulme Research Project Grant (RPG-2021-154).

\section{Prefactorisation algebra \texorpdfstring{$\U\L^\Sigma_\alpha$}{Ug}} \label{sec: prefac alg}

Throughout this section we let $\Sigma$ be an arbitrary connected $2$-dimensional conformal manifold but which could be non-orientable and possibly with boundary.

Let $\Vec_\CC$ denote the symmetric monoidal category of complex vector spaces with the tensor product over $\CC$ serving as the monoidal product. We shall also need the symmetric monoidal category $\dgVec_\CC$ of \dg{} vector spaces over $\CC$, also equipped with the tensor product over $\CC$ as the monoidal product. Let $\dgLie_\CC$ denote the symmetric monoidal category of \dg{} Lie algebras over $\CC$ with the direct sum of \dg{} Lie algebras serving as monoidal product.

\subsection{Unital local Lie algebra \texorpdfstring{$\L^\Sigma_\alpha$}{gk}} \label{sec: twisted dgla}

Let $\Omega^\chain_c$ denote the cosheaf of compactly supported sections of the de Rham complex on the complex manifold $\Dz$.
Let $\Omega^{0, \chain}_c$ be the cosheaf of compactly supported sections of the Dolbeault complex on $\Dz$, equipped with the Dolbeault differential $\bar\partial$, namely the anti-holomorphic part of the de Rham differential $\d = \partial + \bar \partial$. These are cosheaves of commutative \dg{} algebras.

\subsubsection{General setup} \label{sec: general setup}

Let $L$ be a holomorphic vector bundle over $\Dz$. We will make a number of assumptions about $L$, motivated in part by the list of examples considered in \S\ref{sec: main examples} below, which will allow us to define a unital local Lie algebra $\L^{\Sigma}_\alpha$ on $\Sigma$ that will form the basis for the rest of the paper.

\medskip
\begin{assumption}
The compactly supported sections of the $L$-valued Dolbeault complex $U \mapsto \Omega^{0, \chain}_c(U, L)$ define a precosheaf of \dg{} Lie algebras on $\Dz$, whose Lie bracket we denote by
\begin{equation} \label{Omega L bracket}
[\cdot, \cdot] : \Omega_c^{0, \chain}(U, L) \otimes \Omega_c^{0, \chain}(U, L) \longrightarrow \Omega_c^{0, \chain}(U, L).
\end{equation}
\end{assumption}

\medskip
\begin{assumption}
The holomorphic vector bundle $L$ is locally holomorphically trivial over any coordinate chart on $\Dz$, that is for any open subset $U \subset \Dz$ equipped with a local holomorphic coordinate $\xi : U \to \CC$ we have $L|_U \cong U \times \fl_\xi$ for some finite-dimensional vector space $\fl_\xi$ with the Dolbeault operator $\bar\partial_L$ of $L$ given locally by $\bar\partial$.
\end{assumption}

\medskip

Although the vector space $\fl_\xi$ will typically depend on the choice of coordinate $\xi$, hence the use of the subscript $\xi$ in the notation, there is a canonical isomomorphism $\fl_\eta \SimTo \fl_\xi$ between the vectors spaces associated with different local holomorphic coordinates $\xi, \eta : U \to \CC$ determined by the the change of coordinate from $\eta$ to $\xi$. We then have the isomorphism of \dg{} vector spaces
\begin{equation} \label{Omega L iso}
\Omega^{0, \chain}_c(U, L) \cong \fl_\xi \otimes \Omega^{0, \chain}_c(U).
\end{equation}

It will be useful to also introduce a canonical copy $\fl$ of the finite-dimensional vector spaces $\fl_\xi$ with a corresponding isomorphism $(\cdot)_\xi : \fl \SimTo \fl_\xi$. However, it is important to note at this stage that the induced map
$(\cdot)_{\eta \to \xi} \coloneqq (\cdot)_\xi \circ ( (\cdot)_\eta )^{-1} : \fl_\eta \SimTo \fl_\xi$
will generally be different from the canonical change of coordinate isomomorphism $\fl_\eta \SimTo \fl_\xi$ described above.

\medskip
\begin{assumption}
For each $U \in \Dz$, the \dg{} Lie algebra $\Omega_c^{0,\chain}(U, L)$ is equipped with a $2$-cocycle
\begin{equation} \label{cocycle alpha}
\alpha : \Omega_c^{0, \chain}(U, L) \otimes \Omega_c^{0, \chain}(U, L) \longrightarrow \CC.
\end{equation}
\end{assumption}
\medskip

We use this $2$-cocycle to introduce the \emph{twisted} precosheaf of \dg{} Lie algebras $\L^\Dz_\alpha$ as a central extension of the precosheaf of \dg{} Lie algebras $U \mapsto \Omega^{0, \chain}_c(U, L)$ following \cite{CGBook1}. Specifically, for any open $U \subset \Dz$ we define the \dg{} vector space
\begin{equation*}
\L^\Dz_\alpha(U) \coloneqq \Omega_c^{0, \chain}(U, L) \oplus \CC \kent,
\end{equation*}
where $\kent$ is a formal variable of cohomological degree $1$ so that $\CC \kent$ is a $1$-dimensional complex \dg{} vector space concentrated in degree $1$, isomorphic to the inverse suspension $\CC[-1]$ of $\CC$.

We have a morphism of \dg{} vector spaces $\ext_{U, V} : \L^\Dz_\alpha(U) \to \L^\Dz_\alpha(V)$ defined as the extension by zero on the first summand and as the identity on $\CC \kent$.
We define the $\CC$-linear operation
\begin{subequations} \label{Lie bracket gXkappa}
\begin{equation} \label{Lie bracket gXkappa a}
[\cdot, \cdot]_\alpha : \L^\Dz_\alpha(U) \otimes \L^\Dz_\alpha(U) \longrightarrow \L^\Dz_\alpha(U)
\end{equation}
by declaring $\CC \kent$ to be central and for every $a, b \in \Omega^{0,\chain}_c(U, L)$ setting
\begin{align} \label{Lie bracket gXkappa b}
[a, b]_\alpha &\coloneqq [a, b] + \alpha(a, b) \, \kent.
\end{align}
\end{subequations}

It is straightforward to show that the above structure makes $\L^\Dz_\alpha$ into a precosheaf of \dg{} Lie algebras on $\Dz$. In fact, since the central extension will play a key role in what follows, it will be important to keep track of it by introducing the notion of a unital \dg{} Lie algebra\footnote{I thank Alex Schenkel for this suggestion which makes some of the later constructions more elegant.}.
A unital \dg{} Lie algebra can be described as a \dg{} Lie algebra $L$ equipped with a map $\eta \in \Map_{\dgVec_\CC}(\CC, L)$ of degree $1$ whose image is both central and closed in $L$.
In other words, a \dg{} Lie algebra $L$ is unital if it has a central cocycle in degree $1$, namely $\eta(1)$. We will refer to the degree $1$ map $\eta : \CC \to L$ as the \emph{unit} of the unital \dg{} Lie algebra. See \S\ref{sec: udgLie} for more details.
Finally, a unital local Lie algebra is a precosheaf of unital \dg{} Lie algebras with the property that sections over disjoint opens commute. Let $\uLocLie_\CC(D)$ denote the category of unital local Lie algebras on a manifold $D$, see \S\ref{sec: uLocLie} for more details. It is straightforward to check by direct computation that $\L^\Dz_\alpha$ is a unital local Lie algebra on $\Dz$.

In order to get rid of the dependence on the orientation which was introduced by going from $\Sigma$ to its oriented cover $\Dz$, from now on we will only consider the pushforward
\begin{equation} \label{pushforward}
\L^\Sigma_\alpha \coloneqq \pi_\ast \L^\Dz_\alpha \in \uLocLie_\CC(\Sigma)
\end{equation}
of the unital local Lie algebra $\L^\Dz_\alpha$ on $\Dz$ introduced above, along the projection $\pi : \Dz \to \Sigma$. In other words, the sections of \eqref{pushforward} over any open $U \subset \Sigma$ are given by
\begin{equation}
\L^\Sigma_\alpha(U) = \L^\Dz_\alpha\big( \pi^{-1}(U) \big).
\end{equation}
Note that on a general $2$-dimensional conformal manifold $\Sigma$, with possibly non-empty boundary $\partial \Sigma$ and interior region $\Sigma^\circ \coloneqq \Sigma \setminus \partial \Sigma$, there are two types of connected open subsets $U \subset \Sigma$ one may consider:
\begin{itemize}
  \item[$(i)$] $U \subset \Sigma^\circ$ in which case $\pi^{-1}(U) \subset \Dz$ consists of two connected components,
  \item[$(ii)$] $U \subset \Sigma$ with $U \cap \partial \Sigma \neq \emptyset$ for which $\pi^{-1}(U) \subset \Dz$ has a single connected component.
\end{itemize}
Throughout this paper we will only consider bulk fields in the twisted prefactorisation envelope of $\L^\Sigma_\alpha$, to be introduced in \S\ref{sec: twisted PFA envelope}, namely observables living on open subsets $U \subset \Sigma^\circ$ of type $(i)$. We leave the study of boundary fields associated with open subsets $U \subset \Sigma$ of type $(ii)$, and their interplay with bulk fields, for future work.

\subsubsection{Main examples} \label{sec: main examples}

The three primary examples of the above general setup which we will consider throughout this paper correspond to the three canonical examples of vertex Lie algebras discussed in \S\ref{sec: vertex Lie examples}. Upon taking their twisted prefactorisation envelopes in \S\ref{sec: twisted PFE} below, these will lead to the full versions of the Kac-Moody, Virasoro and $\beta\gamma$ vertex algebras, respectively. We now describe the data underlying the unital local Lie algebra $\L^\Sigma_\alpha$ in each of these cases.

\paragraph{Kac-Moody:} Let $\g$ be any finite-dimensional complex Lie algebra and consider the trivial vector bundle
\begin{equation} \label{vector bundle KM}
L = \Dz \times \g.
\end{equation}
This clearly satisfies the assumption of local triviality over any open $U \subset \Dz$ equipped with a holomorphic coordinate $\xi : U \to \CC$, where $\fl_\xi = \g$ is actually independent of the coordinate $\xi$ in this case. The Lie bracket \eqref{Omega L bracket} on $\Omega^{0,\chain}_c(U,L) \cong \fl_\xi \otimes \Omega^{0,\chain}_c(U)$ is given by
\begin{subequations} \label{bracket and cocycle KM}
\begin{equation} \label{Lie bracket KM}
[\X \otimes \omega, \Y \otimes \eta] \coloneqq [\X, \Y] \otimes \omega \wedge \eta
\end{equation}
for $\X, \Y \in \g$ and $\omega, \eta \in \Omega^{0, \chain}_c(U)$.

We also define $\fl = \g$ and the isomorphism $(\cdot)_\xi : \fl \SimTo \fl_\xi$ is the identity.

We fix a symmetric invariant bilinear form $\kappa : \g \otimes \g \to \CC$ on $\g$. Using this we can define a $2$-cocycle \eqref{cocycle alpha} on $\Omega_c^{0, \chain}(U, L)$ by the following formula, see \cite[\S5.5.1]{CGBook1},
\begin{align} \label{cocycle KM}
\alpha(\X \otimes \omega, \Y \otimes \eta) \coloneqq - \frac{\kappa(\X, \Y)}{2 \pi \ii} \int_U \partial \omega \wedge \eta
\end{align}
\end{subequations}
for any $\X, \Y \in \g$ and $\omega, \eta \in \Omega^{0, \chain}_c(U)$.
The integral over $U$ in \eqref{cocycle KM} is trivially zero on degree grounds if $|\omega| + |\eta| \neq 1$. Indeed, if $\omega \in \Omega^{0, r}_c(U)$ and $\eta \in \Omega^{0, s}_c(U)$ then $\partial \omega \wedge \eta \in \Omega^{1, r+s}_c(U)$, which can be integrated over $U$ only if $r+s = 1$.

The unital local Lie algebra $\L^\Dz_\alpha$ in this case corresponds to the affine Kac-Moody example considered in \cite[\S5.5]{CGBook1}.

\paragraph{Virasoro:} Consider the holomorphic tangent bundle
\begin{equation} \label{Virasoro L def}
L = T^{1,0} \Dz.
\end{equation}
Over any open $U \subset \Dz$ equipped with a holomorphic coordinate $\xi : U \to \CC$ we have the local trivialisation $L|_U \cong U \times \fl_\xi$ with $\fl_\xi = \text{span}_\CC \{ \partial_\xi \}$. The induced isomorphism \eqref{Omega L iso} then allows us to represent an element of $\Omega^{0,\chain}_c(U,L)$ in the form $\partial_\xi \otimes u \in \fl_\xi \otimes \Omega^{0,\chain}_c(U)$ with $u \in \Omega^{0,\chain}_c(U)$ and the Lie bracket \eqref{Omega L bracket} on $\Omega^{0,\chain}_c(U,L)$ explicitly reads
\begin{subequations} \label{bracket and cocycle Virasoro}
\begin{equation} \label{Lie bracket Virasoro}
[\partial_\xi \otimes u, \partial_\xi \otimes v] = \partial_\xi \otimes (u \wedge \partial_\xi v - v \wedge \partial_\xi u)
\end{equation}
for any $u, v \in \Omega^{0, \chain}_c(U)$. One can view $\partial_\xi \otimes u$ as a compactly supported Dolbeault form valued vector field on $U$, which would be more conventionally denoted $u \partial_\xi$. The Lie bracket \eqref{Lie bracket Virasoro} is then simply the commutator of vector fields combined with the wedge product on forms.

We let $\fl = \text{span}_\CC \{ \Omega \}$ and the isomorphism $(\cdot)_\xi : \fl \SimTo \fl_\xi$ is given by $\Omega \mapsto - \partial_\xi$.

Following \cite{Williams-Vir}, a formula for the $2$-cocycle \eqref{cocycle alpha} on $\Omega^{0,\chain}_c(U, L)$ is obtained by generalising the formula for the Gelfand-Fuchs $2$-cocycle on the Lie algebra of vector fields on $S^1$. First, we suppose that $U \subset \Dz$ is covered by a local coordinate chart $\xi : U \to \CC$. A pair of elements in $\Omega_c^{0, \chain}(U, L)$ can then be represented as $\partial_\xi \otimes u$ and $\partial_\xi \otimes v$, whose coefficients $u, v \in \Omega^{0,\chain}_c(U)$ transform as components of $(1,0)$ vector fields on $U$. We then define the $2$-cocycle \cite[\S2.2.2]{Williams-Vir}
\begin{align} \label{cocycle Virasoro}
\alpha_\xi( \partial_\xi \otimes u, \partial_\xi \otimes v ) &\coloneqq \frac{c}{24 \pi \ii} \int_U \partial (\partial_\xi u) \wedge \partial_\xi v,
\end{align}
\end{subequations}
where $c \in \RR$ is fixed. As in the previous example, the integral over $U$ in \eqref{cocycle Virasoro} is trivially zero on degree grounds if $|u| + |v| \neq 1$.
The subscript $\xi$ on the left hand side of \eqref{cocycle Virasoro} is used to indicate that this $2$-cocycle explicitly depends on the chosen coordinate $\xi$. However, it depends on the coordinate $\xi$ only up to a $2$-coboundary. Indeed, given any other local coordinate $\eta : U \to \CC$ on the open $U \subset \Dz$ and introducing the $1$-cochain
\begin{equation} \label{coboundary def}
\beta_{\eta,\xi} : \Omega^{0,\chain}_c(U, L) \longrightarrow \CC, \qquad
\partial_\xi \otimes u \longmapsto \frac{c}{24 \pi \ii} \int_U (S \eta)(\xi) \d \xi \wedge u,
\end{equation}
where $(S \eta)(\xi) = \frac{\partial_\xi^3 \eta}{\partial_\xi \eta} - \frac 32 \big( \frac{\partial_\xi^2 \eta}{\partial_\xi \eta} \big)^2$ denotes the Schwarzian derivative, it is straightforward to show that
\begin{equation*}
\alpha_\eta( \partial_\eta \otimes \tilde u, \partial_\eta \otimes \tilde v ) = \alpha_\xi( \partial_\xi \otimes u, \partial_\xi \otimes v ) + \beta_{\eta, \xi}\big( [\partial_\xi \otimes u, \partial_\xi \otimes v] \big)
\end{equation*}
for any pair of vector fields $\partial_\eta \otimes \tilde u = \partial_\xi \otimes u$ and $\partial_\eta \otimes \tilde v = \partial_\xi \otimes v$ in $\Omega^{0,\chain}_c(U, L)$. In other words, $\alpha_\eta = \alpha_\xi + \delta \beta_{\eta, \xi}$ so that we have an isomorphism of unital local Lie algebras
\begin{equation} \label{Phi iso change alpha}
\Phi_{\eta \to \xi} : \L^\Dz_{\alpha_\eta}(U) \overset{\cong}\longrightarrow \L^\Dz_{\alpha_\xi}(U) , \qquad
a + x \kent \longmapsto a + \big( x - \beta_{\eta,\xi}(a) \big) \kent.
\end{equation}
The $2$-cocycle \eqref{cocycle Virasoro} thus leads to a well-defined and coordinate independent central extension \eqref{Lie bracket gXkappa}. We refer to \cite[\S5.2]{Williams-Vir} for further details.

The resulting unital local Lie algebra $\L^{\Dz}_\alpha$ coincides with the local Lie algebra underlying the Virasoro factorisation algebra in \cite{Williams-Vir}.

\paragraph{$\bm \beta \bm \gamma$ system:} Consider the direct sum of the trivial $1$-dimensional line bundle on $\Dz$ with the holomorphic cotangent bundle, namely
\begin{equation} \label{beta gamma L def}
L = (\Dz \times \CC) \oplus T^{\ast 1,0} \Dz.
\end{equation}
Over any open $U \subset \Dz$ equipped with a holomorphic coordinate $\xi : U \to \CC$ we have the local trivialisation $L|_U \cong U \times \fl_\xi$ where $\fl_\xi = \text{span}_\CC \{ 1, \d \xi \}$. The Lie bracket \eqref{Omega L bracket} on $\Omega^{0,\chain}_c(U,L)$ is taken to be trivial.
The isomorphism \eqref{Omega L iso} here implies $\Omega_c^{0, \chain}(U, L) \cong \Omega_c^{0, \chain}(U) \oplus \Omega_c^{1, \chain}(U)$.

We let $\fl = \text{span}_\CC \{ \beta, \gamma \}$ and the isomorphism $(\cdot)_\xi : \fl \SimTo \fl_\xi$ is given by $\beta \mapsto 1$, $\gamma \mapsto \d \xi$.

A $2$-cocycle \eqref{cocycle alpha} on $\Omega^{0,\chain}_c(U, L)$ is given by
\begin{align} \label{cocycle beta gamma}
\alpha(\omega, \eta) \coloneqq \frac{1}{2 \pi \ii} \int_U \omega \wedge \eta
\end{align}
for any $\omega, \eta \in \Omega^{0, \chain}_c(U) \oplus \Omega^{1, \chain}_c(U)$, where once again the integral is understood to vanish unless $\omega \wedge \eta \in \Omega^{1,1}_c(U)$ so that if $\omega \in \Omega^{0, r}_c(U)$, say, then we should have $\eta \in \Omega^{1, s}_c(U)$ with $r + s = 1$.

This corresponds to the $\beta\gamma$ system example discussed in \cite[\S5.4]{CGBook1}.

\subsection{Cohomology of \texorpdfstring{$\L^\Sigma_\alpha$}{gk}} \label{sec: coh of gD}

In this section we compute the cohomology of the \dg{} vector space $\L^\Sigma_\alpha(U)$ for connected open subsets $U \subset \Sigma^\circ$ in the interior $\Sigma^\circ$ of $\Sigma$, equipped with a local holomorphic coordinate, see the last part of Theorem \ref{thm: cohomology gSigma} below. Although this result is well known, see e.g. \cite[\S9]{LueckingRubel} or \cite[Theorem 5.4.7]{CGBook1}, we shall give the details of the proof in the present context since the full statement of Theorem \ref{thm: cohomology gSigma} will be crucial for us in \S\ref{sec: vertex alg}.

Given any $U \subset \CC$ or $U \subset \Dz$ we denote by $\overline{U}$ its closure.
The following result is standard.

\begin{lemma} \label{lem: Cauchy-Pompeiu}
Let $U \subset \CC$ be a bounded open subset with $C^1$ boundary.
\begin{itemize}
  \item[$(a)$] For any $f \in \Omega^0(\overline{U})$ we have the \emph{Cauchy-Pompeiu formula}
\begin{equation*}
\frac{1}{2 \pi \ii} \int_{\partial U} \frac{f \d\mu}{\mu-\lambda} + \frac{1}{2 \pi \ii} \int_U \frac{\d\mu \wedge \bar \partial f}{\mu-\lambda} = \left\{ 
\begin{array}{ll}
f(\lambda) & \textup{if}\; \lambda \in U,\\
0 & \textup{if}\; \lambda \in \CC \setminus \overline{U}.
\end{array}
\right.
\end{equation*}
  \item[$(b)$] For any $\omega \in \Omega^{0,1}_{c}(U)$, the $\bar \partial$-problem $\bar \partial f = \omega$ for $f \in \Omega^{0,0}(\CC)$ on $\CC$ is solved by
\begin{equation*}
\pushQED{\qed}
f(\lambda) = \omega^\downarrow(\lambda) \coloneqq \frac{1}{2 \pi \ii} \int_U \frac{\d\mu \wedge \omega}{\mu-\lambda}. \qedhere
\popQED
\end{equation*}
\end{itemize}
\end{lemma}

We denote the complement of any subset $S \subset \CP \coloneqq \CC \cup \{ \infty \}$ of the Riemann sphere by $S^{\rm c} \coloneqq \CP \setminus S$. For any $S \subset \CC$, let $\O_\infty(S^{\rm c})$ denote the algebra of germs of holomorphic functions on $S^{\rm c}$ vanishing at $\infty$.
Explicitly, an element of $\O_\infty(S^{\rm c})$ is the equivalence class $[f]$ of a holomorphic function $f \in \O_{\CP}(\Delta_f)$ defined on an open neighbourhood $\Delta_f \supset S^{\rm c} \ni \infty$ and such that $f(\infty) = 0$, where two such functions $f$ and $g$ are considered equivalent if $f|_V = g|_V$ for some open neighbourhood $V \subset \Delta_f \cap \Delta_g$ of $S^{\rm c}$.

\begin{lemma} \label{lem: omega check germ}
Let $W \subset \CC$ be a bounded open subset. For any $\omega \in \Omega^{0,1}_c(W)$ we have:
\begin{itemize}
  \item[$(i)$] $\omega^\downarrow \in \Omega^{0,0}(\CC)$ defines a germ $[\omega^\downarrow] \in \O_\infty( W^{\rm c} )$,
  \item[$(ii)$] $[\omega^\downarrow] = 0$ in $\O_\infty( W^{\rm c} )$ if and only if $\omega$ is $\bar\partial$-exact.
\end{itemize}
\begin{proof}
Firstly, since $\omega$ has compact support the integral entering the definition of $\omega^\downarrow(\lambda)$ is well defined, in particular the integrand is locally integrable near the singularity at $\lambda$.
And by Lemma \ref{lem: Cauchy-Pompeiu}$(b)$, this function provides a solution to the $\bar\partial$-problem for $\omega$.
We then also deduce from the support property of $\omega$ that $\omega^\downarrow$ is holomorphic on an open neighbourhood $\Delta_{\omega^\downarrow}$ of $W^{\rm c}$. Moreover, since for every $\mu \in \supp \, \omega \subset W$ we have $|\mu - \lambda| \geq \text{dist}(\lambda, \supp \, \omega)$, it follows that
\begin{equation*}
\big| \omega^\downarrow(\lambda) \big| \leq \frac{C \sup_{\mu \in U} | \omega(\mu) |}{\text{dist}(\lambda, \supp \, \omega)}
\end{equation*}
for some $C > 0$, and hence $\omega^\downarrow(\lambda) \to 0$ as $\lambda \to \infty$. The statement $(i)$ now follows.

If $[ \omega^\downarrow ] = 0$ then $\omega^\downarrow$ vanishes on an open neighbourhood of $W^{\rm c}$ and so it has compact support in the bounded open $W \subset \CC$, i.e. $\omega^\downarrow \in \Omega^{0,0}_c(W)$. The `only if' part of $(ii)$ now follows as $\omega^\downarrow$ solves the $\bar\partial$-problem for $\omega$.

Conversely, if $\omega = \bar\partial \eta$ for some $\eta \in \Omega^{0,0}_c(W)$ then it is immediate from the Cauchy-Pompeiu formula that $\omega^\downarrow = \eta$. The `if' part of $(ii)$ now follows since $\omega^\downarrow \in \Omega^{0,0}_c(W)$.
\end{proof}
\end{lemma}

Given any bounded open subset $W \subset \CC$, let us fix a basis of the space of germs $\O_\infty(W^c)$ and for each basis element $[f] \in \O_\infty(W^c)$ we pick a representative $f \in \O_{\CP}(\Delta_f)$ defined on an open neighbourhood $\Delta_f$ of $W^c$ and we choose a smooth bump function $\rho_f \in \Omega_c^{0,0}(W)$ which is equal to $1$ on some open neighbourhood of $(\Delta_f)^c \subset W$.
Since $f|_{\Delta_f \cap W} \in \Omega^{0,0}(\Delta_f \cap W)$ is defined on $\supp \bar\partial \rho_f \subset \Delta_f \cap W$ we can set
\begin{equation} \label{up arrow def}
\lceil f \rceil \coloneqq - f \, \bar\partial \rho_f \in \Omega^{0,1}_c(W)
\end{equation}
for every basis element $[f] \in \O_\infty(W^c)$ and then extend by linearity to all of $\O_\infty(W^c)$.

\begin{lemma} \label{lem: up arrow indep choices}
Let $W \subset \CC$ be a bounded open subset. For any germ $[f] \in \O_\infty(W^c)$,
the $1$-form $\lceil f \rceil \in \Omega^{0,1}_c(W)$ is independent of the choices made, up to $\bar\partial$-exact terms.
In particular, for any representative $f \in \O_{\CP}(\Delta_f)$ of $[f]$ and any smooth bump function $\rho \in \Omega^{0,0}_c(W)$ equal to $1$ on some open neighbourhood of $(\Delta_f)^c \subset W$ we have that $\lceil f \rceil + f \, \bar\partial \rho$ is $\bar\partial$-exact.
\begin{proof}
Let $[f] \in \O_\infty(W^c)$. In terms of the fixed basis of $\O_\infty(W^c)$, we can write it as a finite linear combination $[f] = \sum_i \alpha_i [f_i]$ of basis elements $[f_i] \in \O_\infty(W^c)$. Then $f|_V = \sum_i \alpha_i f_i|_V$ for some neighbourhood $V \subset \Delta_f \cap \bigcap_i \Delta_{f_i}$ of $W^c$. Let $\rho' \in \Omega^{0,0}_c(W)$ be a smooth bump function equal to $1$ on some neighbourhood of $V^c \subset W$. By definition $\lceil f \rceil = - \sum_i \alpha_i f_i \, \bar\partial \rho_{f_i}$ so that
\begin{equation*}
\lceil f \rceil + f\, \bar\partial \rho = - \sum_i \alpha_i f_i\, \bar\partial (\rho_{f_i} - \rho') - \bigg( \sum_i \alpha_i f_i - f \bigg) \bar\partial \rho' + f \, \bar\partial (\rho - \rho').
\end{equation*}
The second term on the right hand side is equal to zero since $\supp \bar\partial \rho' \subset V \cap W$ and $\sum_i \alpha_i f_i - f$ vanishes on $V$. On the other hand, the first term on the right hand side is $\bar\partial$-exact because $\rho_{f_i} - \rho'$ vanishes on $(\Delta_{f_i})^c \subset V^c$ so that $f_i$ is holomorphic on the support of $\rho_{f_i} - \rho'$. Likewise, the last term is also $\bar\partial$-exact, as required.
\end{proof}
\end{lemma}

\begin{lemma} \label{lem: germ representative}
Let $W \subset \CC$ be a bounded open subset.
For $[f] \in \O_\infty(W^{\rm c})$ we have $\big[ \lceil f \rceil^\downarrow \big] = [f]$.
Moreover, for $\omega \in \Omega^{0,1}_c(W)$ we have $\omega - \lceil \omega^\downarrow \rceil = \bar\partial \big( ( \omega - \lceil \omega^\downarrow \rceil )^\downarrow \big)$.
\begin{proof}
For any $f \in \O_{\CP}(\Delta_f)$ tending to zero at infinity, we have
\begin{equation*}
- ( f \, \bar\partial \rho_f)^\downarrow = \big( f \bar\partial \big(1 - \rho_f \big) \big)^\downarrow = \big( \bar\partial \big( (1 - \rho_f) f \big) \big)^\downarrow =  ( 1 - \rho_f ) f,
\end{equation*}
where in the second step we used the fact that $\bar\partial f = 0$ on $\supp (1- \rho_f) \subset \Delta_f$. In order to see the last step, consider the Cauchy-Pompeiu formula applied to the function $(1-\rho_f) f \in \Omega^0(\overline{D_R})$ with $D_R \supset W$ an open disc of radius $R > 0$, namely
\begin{equation*}
\frac{1}{2 \pi \ii} \int_{\partial D_R} \frac{f \d\mu}{\mu-\lambda} + \frac{1}{2 \pi \ii} \int_W \frac{\d\mu \wedge \bar \partial \big( (1-\rho_f) f \big)}{\mu-\lambda} = \big( 1-\rho_f(\lambda) \big) f(\lambda) \quad \textup{if}\; \lambda \in D_R,
\end{equation*}
where in the first integral on the left hand side we used the fact that $(1-\rho_f) f = f$ on $\partial D_R$ and the domain of the second integral was restricted from $D_R$ to $W$ using the fact that the integrand has support in $W$. The first integral is bounded by $\text{sup}_{\mu \in \partial D_R} |f(\mu)|$ which tends to zero as $R \to \infty$. And in this limit, the right hand side of the Cauchy-Pompieu formula tends to $(1 - \rho_f(\lambda)) f(\lambda)$ for all $\lambda \in \CC$.
Finally, we note that $(1- \rho_f) f = f$ on the open neighbourhood $(\supp \rho_f)^{\rm c}$ of $W^{\rm c}$ from which the first claim now follows.

For the `moreover' part, it follows using Lemma \ref{lem: omega check germ}$(i)$ and the above that $\omega^\downarrow$ and $\lceil \omega^\downarrow \rceil^\downarrow$ define the same germ in $\O_\infty(W^c)$. The desired result now follows by using Lemma \ref{lem: omega check germ}$(ii)$, see in particular its proof.
\end{proof}
\end{lemma}

Let $U \subset \Sigma^\circ$ be a connected open subset of the interior $\Sigma^\circ \subset \Sigma$, i.e. not intersecting the boundary $\partial \Sigma$ of $\Sigma$, with a bounded holomorphic coordinate on each connected component of $\pi^{-1}(U) \subset \Dz$. That is, $\pi^{-1}(U) = \sqcup_{i=1}^m W_i$ with $m \in \{ 1, 2\}$, where each $W_i \subset \Dz$ is a connected open subset equipped with a holomorphic coordinate $\xi_i : W_i \to \CC$ such that $\xi_i(W_i) \subset \CC$ is bounded. We will denote this set of holomorphic coordinates collectively as $\bm\xi : \pi^{-1}(U) \to \CC$. Consider the \dg{} vector space concentrated in degree $0$ given by
\begin{equation*}
\L_\O^{\bm\xi} ( \pi^{-1}(U) ) \coloneqq \bigoplus_{i=1}^m \fl \otimes \O_\infty\big( \xi_i(W_i)^{\rm c} \big) \oplus \CC.
\end{equation*}
We define a morphism in the category of pointed vector spaces $\udgVec_\CC$, see \S\ref{sec: udgLie}, as
\begin{subequations} \label{iph def}
\begin{align}
i_U : \L_\O^{\bm\xi} ( \pi^{-1}(U) )[-1] &\longrightarrow \L^\Sigma_\alpha(U)\\
s^{-1} \big(\a_i \otimes [f_i] \big)_{i=1}^m &\longmapsto \sum_{i=1}^m (\a_i)_{\xi_i} \otimes \xi_i^\ast \lceil f_i \rceil \notag\\
s^{-1} a &\longmapsto a \cent \notag
\end{align}
for any $\a_i \in \fl$, $a \in \CC$ and $[f_i] \in \O_\infty\big( \xi_i(W_i)^{\rm c} \big)$.
We also define another morphism in $\udgVec_\CC$, going in the other direction, as
\begin{align}
p_U : \L^\Sigma_\alpha(U) &\longrightarrow \L_\O^{\bm\xi} ( \pi^{-1}(U) )[-1]\\
\sum_{i=1}^m (\a_i)_{\xi_i} \otimes \omega_i &\longmapsto s^{-1} \Big( \a_i \otimes \big[ \big( (\xi_i^{-1} )^\ast \omega_i \big)^\downarrow \big] \Big)_{i=1}^m \notag\\
\sum_{i=1}^m (\b_i)_{\xi_i} \otimes \eta_i &\longmapsto 0 \notag\\
a \cent &\longmapsto s^{-1} a \notag
\end{align}
for any $\a_i, \b_i \in \fl$, $a \in \CC$, $\eta_i \in \Omega^{0,0}_c(W_i)$ and $\omega_i \in \Omega^{0,1}_c(W_i)$.
Finally, we define a degree $-1$ linear map
\begin{align}
h_U : \L^\Sigma_\alpha(U) &\longrightarrow \L^\Sigma_\alpha(U)\\
\sum_{i=1}^m (\a_i)_{\xi_i} \otimes \omega_i &\longmapsto \sum_{i=1}^m (\a_i)_{\xi_i} \otimes \xi_i^\ast \Big( (\xi_i^{-1})^\ast \omega_i - \big\lceil \big( (\xi_i^{-1})^\ast \omega_i \big)^\downarrow \big\rceil \Big)^\downarrow \notag\\
\sum_{i=1}^m (\b_i)_{\xi_i} \otimes \eta_i &\longmapsto 0 \notag\\
a \cent &\longmapsto 0 \notag
\end{align}
\end{subequations}
for any $\a_i, \b_i \in \fl$, $a \in \CC$, $\eta_i \in \Omega^{0,0}_c(W_i)$ and $\omega_i \in \Omega^{0,1}_c(W_i)$.

\begin{theorem} \label{thm: cohomology gSigma}
For any connected open subset $U \subset \Sigma^\circ$ equipped with holomorphic coordinates $\bm\xi : \pi^{-1}(U) \to \CC$ as above, we have a strong deformation retract
\begin{equation*}
\begin{tikzcd}
\L_\O^{\bm\xi}( \pi^{-1}(U) )[-1] \arrow[r, "i_U", shift left=1mm] & \L^\Sigma_\alpha(U) \arrow[l, "p_U", shift left=1mm] \arrow["h_U"', out=-15,in=15,distance=10mm]
\end{tikzcd}
\end{equation*}
in $\udgVec_\CC$. In other words, the collection of linear maps in \eqref{iph def} define a \emph{deformation retract} in the sense that $p_U i_U = \id$ and $i_U p_U - \id = - (\bar\partial h_U + h_U \bar\partial)$, and this deformation retract is \emph{strong} in the sense that $p_U h_U = h_U i_U = h_U h_U = 0$.

In particular, $H^0\big( \L^\Sigma_\alpha(U) \big) = 0$ and $H^1\big( \L^\Sigma_\alpha(U) \big) \cong \L_\O^{\bm\xi}( \pi^{-1}(U) )$.
\begin{proof}
Let $\a_i, \b_i \in \fl$, $a \in \CC$, $\eta_i \in \Omega^{0,0}_c(W_i)$, $\omega_i \in \Omega^{0,1}_c(W_i)$ and $[f_i] \in \O_\infty\big( \xi(W_i)^{\rm c} \big)$. Then
\begin{equation*}
p_U i_U\Big( s^{-1} \big( (\a_i \otimes [f_i])_{i=1}^m, a \big) \Big) = s^{-1} \bigg( \Big( \a_i \otimes \big[ \lceil f_i \rceil^\downarrow \big] \Big)_{i=1}^m, a \bigg),
\end{equation*}
so using Lemma \ref{lem: germ representative} we deduce that $p_U i_U = \id$.
Next, we have
\begin{align*}
&(i_U p_U - \id)\bigg( \sum_{i=1}^m \big( (\a_i)_{\xi_i} \otimes \omega_i + (\b_i)_{\xi_i} \otimes \eta_i \big) + a \cent \bigg)\\
&\qquad\qquad = - \sum_{i=1}^m (\a_i)_{\xi_i} \otimes \Big( \omega_i - \xi_i^\ast \big\lceil \big( (\xi_i^{-1} )^\ast \omega_i \big)^\downarrow \big\rceil \Big) - \sum_{i=1}^m (\b_i)_{\xi_i} \otimes \eta_i\\
&\qquad\qquad = - \bar\partial \bigg( h_U \bigg( \sum_{i=1}^m \big( (\a_i)_{\xi_i} \otimes \omega_i + (\b_i)_{\xi_i} \otimes \eta_i \big) + a \cent \bigg) \bigg) - \sum_{i=1}^m h_U\big( (\b_i)_{\xi_i} \otimes \bar\partial \eta_i \big),
\end{align*}
where the first sum in last step has been rewritten using the `moreover' part of Lemma \ref{lem: germ representative} and the definition of $h_U$. To rewrite the second sum, we note that by Lemma \ref{lem: omega check germ}$(ii)$, see in particular its proof, we have $( (\xi_i^{-1})^\ast \bar\partial \eta_i)^\downarrow = (\xi_i^{-1})^\ast \eta_i \in \Omega^{0,0}_c(\CC)$ which defines the trivial germ in $\O_\infty\big( \xi_i(W_i)^{\rm c} \big)$ and hence $\lceil ( (\xi_i^{-1})^\ast \bar\partial \eta_i )^\downarrow \rceil = 0$. Therefore by definition of $h_U$ we have
\begin{equation*}
\sum_{i=1}^m h_U\big( (\b_i)_{\xi_i} \otimes \bar\partial \eta_i \big) = h_U\bigg( \sum_{i=1}^m (\b_i)_{\xi_i} \otimes \bar\partial \eta_i \bigg) = \sum_{i=1}^m (\b_i)_{\xi_i} \otimes \eta_i.
\end{equation*}
We have shown $i_U p_U - \id = - (\bar\partial h_U + h_U \bar\partial)$ so that we have a deformation retract.

It remains to show that this deformation retract is strong. But it is immediate on degree grounds that $p_U h_U = 0$ and $h_U h_U = 0$. On the other hand, we have
\begin{align*}
h_U i_U\Big( s^{-1} \big( (\a_i \otimes [f_i])_{i=1}^m, a \big) \Big)
&= \sum_{i=1}^m (\a_i)_{\xi_i} \otimes \xi_i^\ast \Big( \lceil f_i \rceil - \big\lceil \lceil f_i \rceil^\downarrow \big\rceil \Big)^\downarrow.
\end{align*}
Now by Lemma \ref{lem: germ representative} we know that $[f_i] = \big[ \lceil f_i \rceil^\downarrow \big]$ and hence the right hand side vanishes, so we have a strong deformation retract.
\end{proof}
\end{theorem}

\subsection{Twisted prefactorisation envelope of \texorpdfstring{$\L^\Sigma_\alpha$}{gk}} \label{sec: twisted PFE}

We now introduce the prefactorisation algebra $\U\L^\Sigma_\alpha$, referred to as the twisted prefactorisation envelope of the unital local Lie algebra $\L^\Sigma_\alpha$, by adapting the general construction in \cite{CGBook1}.
To formulate the notion of prefactorisation algebra most succinctly it is convenient to use the framework of \emph{multicategories}, which are also known as \emph{coloured operads}. These consist of classes of objects and arrows, just like an ordinary category, but where the arrows describe $n$-ary operations with an arbitrary number $n \in \ZZ_{\geq 0}$ of inputs but still just a single output. We refer the reader to \cite[\S 2]{Leinster} or \cite[\S A.3.3]{CGBook1} for an introduction to multicategories.

\subsubsection{Prefactorisation algebras}

Let $\Top(\Sigma)$ denote the category whose objects consist of $\Sigma$ itself and all open subsets which are homeomorphic to either $\CC$ or $\CC \setminus \{ 0 \}$, and whose morphisms are given by inclusions, i.e. for all $U, V \in \Top(\Sigma)$ the set of morphisms $\Hom_{\Top(\Sigma)}(U, V)$ contains a single morphism $U \to V$ if $U \subset V$ and is empty otherwise. Note that we do not restricting attention to open subsets homeomorphic to $\CC$, as is done for instance in the locally constant setting of \cite{AyalaFrancis, LurieHA} which describes topological quantum field theories. The reason for also allowing open subsets homeomorphic to $\CC \setminus \{ 0 \}$ will become clear later in \S\ref{sec: VA products} below, see in particular Proposition \ref{prop: vertex operator general} where the state-field correspondence will be constructed as a factorisation product from a disc to an annulus. Likewise, the reason for including $\Sigma$ itself as an open subset in $\Top(\Sigma)$ will become clear in \S\ref{sec: invariant bilinear} when we introduce the notion of invariant bilinear form. Note that our choice of category $\Top(\Sigma)$ also differs from the one used in \cite{CGBook1, CGBook2} whose objects consist of all open subsets of $\Sigma$. However, unlike \cite{CGBook1, CGBook2}, where the focus is on general perturbative quantum field theories, here we restrict to conformal field theories, for which our choice of category $\Top(\Sigma)$ is particularly well-adapted.

We define an associated multicategory $\Top(\Sigma)^\sqcup$ as follows. It has the same set of objects as $\Top(\Sigma)$ and for any finite collection of open subsets $U_i$ with $i \in \{1, \ldots, n\}$ and $V$ in $\Top(\Sigma)^\sqcup$, with $n \in \ZZ_{\geq 0}$, its set of $n$-operations is
\begin{equation*}
P^{\Top(\Sigma)^\sqcup}_n\big( (U_i)_{i=1}^n, V \big) \coloneqq
\left\{ \begin{array}{ll}
\Hom_{\Top(\Sigma)}(\sqcup_{i=1}^n U_i, V) &\textup{if} \; U_i \cap U_j = \emptyset \; \textup{for all} \; i \neq j,\\
\emptyset & \textup{otherwise}.
\end{array}
\right.
\end{equation*}
In other words, it contains the single element $\sqcup_{i=1}^n U_i \to V$ if all the $U_i$ are pairwise disjoint and contained within $V$, and is empty otherwise. By convention, for any open $V \subset \Sigma$ there is a unique $0$-operation $\bm\emptyset \to V$ from the empty collection of disjoint open subsets to $V$. Since the disjoint union is symmetric we have $\sqcup_{i=1}^n U_i = \sqcup_{i=1}^n U_{\sigma(i)}$ for any $\sigma \in S_n$ and the corresponding action of the symmetric group on the $n$-operations
\begin{equation*}
\sigma^\ast : P^{\Top(\Sigma)^\sqcup}_n\big( (U_i)_{i=1}^n, V \big) \overset{\cong}\longrightarrow P^{\Top(\Sigma)^\sqcup}_n\big( (U_{\sigma(i)})_{i=1}^n, V \big)
\end{equation*}
is simply the identity.

Given any symmetric monoidal category $\mathsf C$ with monoidal product denoted $\otimes$, we denote the associated multicategory by $\mathsf C^\otimes$. A $\mathsf C^\otimes$-valued \emph{prefactorisation algebra} $\mcF$ on $\Sigma$, is an object in the category $\Alg_{\Top(\Sigma)^\sqcup}(\mathsf C^\otimes)$ of $\Top(\Sigma)^\sqcup$-algebras in $\mathsf C^\otimes$, i.e. it is a multifunctor
\begin{equation*}
\mcF : \Top(\Sigma)^\sqcup \longrightarrow \mathsf C^\otimes.
\end{equation*}
More explicitly, this is an assignment of:
\begin{itemize}
  \item an object $\mcF(U) \in \mathsf C$ to each open subset $U \subset \Sigma$ in $\Top(\Sigma)$, and
  \item a morphism $m^\mcF_{(U_i), V} : \bigotimes_{i=1}^n \mcF(U_i) \to \mcF(V)$ in $\mathsf C$, called a \emph{factorisation product}, to each inclusion of $n \in \ZZ_{\geq 0}$ disjoint open sets $\sqcup_{i=1}^n U_i \subset V$ in $\Top(\Sigma)$, such that the diagram
\begin{equation} \label{PFA commutativity}
\begin{tikzcd}
{\displaystyle \bigotimes_{i=1}^n \bigotimes_{j=1}^{m_i} \mcF(U_{ij})} \arrow[r, "\bigotimes_{i=1}^n m^\mcF_{(U_{ij}), V_i}"] \arrow[rd, "m^\mcF_{(U_{ij}), W}"', shorten <= -1mm] &[10ex] {\displaystyle \bigotimes_{i=1}^n \mcF(V_i)} \arrow[d, "m^\mcF_{(V_i), W}"]\\
& \mcF(W)
\end{tikzcd}
\end{equation}
in $\mathsf C$ commutes, for any inclusion of disjoint open subsets $\sqcup_{j=1}^{m_i} U_{ij} \subset V_i$ for $i =1,\ldots, n$ and $\sqcup_{i=1}^n V_i \subset W$ in $\Top(\Sigma)^\sqcup$. In particular, the unique $0$-operation $\bm\emptyset \to V$ is assigned a morphism $u \to \mcF(V)$ from the identity object $u$ of $\mathsf C$, i.e. $\mcF(V)$ is \emph{pointed}.
\end{itemize}
When $\mcF$ is clear from the context we denote the factorisation products simply as $m_{(U_i), V}$. We also use the abbreviation $m_{U, V} \coloneqq m_{(U), V}$ for any inclusion of open subsets $U \subset V$ in $\Top(\Sigma)$. Note that a prefactorisation algebra $\mcF$ is, in particular, a precosheaf on $\Sigma$ restricted to $\Top(\Sigma)$ with extension morphisms $m_{U, V}$ for every inclusion of open subsets $U \subset V$ in $\Top(\Sigma)$.

A \emph{morphism of prefactorisation algebras} $\phi : \mcF \to \G$ is a natural transformation
\begin{equation*}
\begin{tikzcd}
\Top(\Sigma)^\sqcup \arrow[r, bend left=40, "\mcF", ""{name=U, below}]
\arrow[r, bend right=35, "\mathcal G"{below}, ""{name=D}]
& \mathsf C^\otimes
\arrow[from=U, to=D, "\phi"]
\end{tikzcd}
\end{equation*}
of multifunctors $\mcF, \G : \Top(\Sigma)^\sqcup \to \mathsf C^\otimes$. Let $\PFac(\Sigma, \mathsf C^\otimes) \coloneqq \Alg_{\Top(\Sigma)^\sqcup}(\mathsf C^\otimes)$ denote the category of prefactorisation algebras on $\Sigma$ valued in the multicategory $\mathsf C^\otimes$. Given any morphisms of prefactorisation algebras $\phi : \mcF \to \G$ and $\psi : \G \to \mathcal H$ with $\mcF, \G, \mathcal H \in \PFac(\Sigma, \mathsf C^\otimes)$ we denote by $\psi \circ \phi : \mcF \to \mathcal H$ their composition, i.e. the vertical composition of these natural transformations of multifunctors. The horizontal composition of natural transformations will be denoted by concatenation.

\subsubsection{Twisted prefactorisation envelope of \texorpdfstring{$\L^\Sigma_\alpha$}{gk}} \label{sec: twisted PFA envelope}

Any local Lie algebra on $\Sigma$ defines, in a canonical way, a prefactorisation algebra on $\Sigma$ valued in $\dgLie_\CC^{\oplus}$. Similarly, any unital local Lie algebra on $\Sigma$ defines a prefactorisation algebra on $\Sigma$ valued in $\udgLie_\CC^{\boplus}$, see Proposition \ref{prop: constructing prefac 1} for details. In \cite[Definition 3.6.4]{CGBook1}, the notion of twisted prefactorisation envelope is defined by applying the $C$hevalley-$E$ilenberg functor $\CE_\chain$ for Lie algebra homology to such a unital \dg{} Lie algebra (or more generally to a \dg{} Lie algebra with a $1$-dimensional central extension in any degree $-k \in \ZZ$). This produces a prefactorisation algebra over the base ring $\CC[s \kent]$ and ultimately in the holomorphic setting of \cite[\S 5.5]{CGBook1} to vertex algebra structures over the base ring $\CC[s \kent]$. In order to obtain vertex algebras over $\CC$, we introduce a variant $\bCE_\chain$ of the homological $C$hevalley-$E$ilenberg functor in \S\ref{sec: CE for udgLie}, in which we additionally quotient by the ideal generated by $s \kent - 1$, see Proposition \ref{prop: bCE functor}. Applying this functor to the unital local Lie algebra $\L^\Sigma_\alpha$ viewed as a prefactorisation algebra on $\Sigma$ valued in $\udgLie_\CC^{\boplus}$, and taking $0^{\rm th}$ cohomology to obtain a prefactorisation algebra valued in vector spaces, leads to
\begin{equation*}
\U \L^\Sigma_\alpha \coloneqq H^0 \, \bCE_\chain \, \L^\Sigma_\alpha \in \PFac(\Sigma, \Vec_\CC^{\otimes}).
\end{equation*}
By a slight abuse of terminology, in this paper we will still refer to this construction as taking the \emph{twisted prefactorisation envelope} of the unital local Lie algebra $\L^\Sigma_\alpha$. We refer the reader to \S\ref{sec: twisted PF env} for the full details of this construction, leading to the definition \eqref{PFA envelope}.

Since $\L^\Sigma_\alpha(U)$ is concentrated in degrees $0$ and $1$ for any open subset $U$ in $\Top(\Sigma)$, it follows that for every $i > 0$ we have $\CE_i\big(\L^\Sigma_\alpha(U)\big) = 0$ and hence also $\bCE_i\big(\L^\Sigma_\alpha(U)\big) = 0$. In particular, any $\mathcal A \in \bCE_0(\L^\Sigma_\alpha(U))$ is closed and we will denote by $[\mathcal A]_U \in \U\L^\Sigma_\alpha(U)$ its $0^{\rm th}$ cohomology class.

The factorisation products $m_{(U_i), V}$ of $\U \L^\Sigma_\alpha$ can be described explicitly as follows. For any inclusion $U \subset V$ of open subsets in $\Top(\Sigma)$, the factorisation product $m_{U, V}$ is induced by the extension morphism $\L^\Sigma_\alpha(U) \to \L^\Sigma_\alpha(V)$ of the precosheaf $\L^\Sigma_\alpha$. For any inclusion $U \sqcup V \subset W$ in $\Top(\Sigma)^\sqcup$, the factorisation product $m_{(U,V), W}$ is given by the composition
\begin{align} \label{Fac prod explicit}
&\U\L^\Sigma_\alpha(U) \otimes \U\L^\Sigma_\alpha(V) \overset{\cong}\longrightarrow H^0\Big( \bCE_\chain\big( \L^\Sigma_\alpha(U) \big) \otimes \bCE_\chain\big( \L^\Sigma_\alpha(V) \big) \Big) \notag\\
&\qquad\qquad \overset{\cong}\longrightarrow H^0\Big( \bCE_\chain\big( \L^\Sigma_\alpha(U) \boplus \L^\Sigma_\alpha(V) \big)\Big) \longrightarrow \U\L^\Sigma_\alpha(W),
\end{align}
where we have first applied the canonical isomorphism given by $[\mathcal A]_U \otimes [\mathcal B]_V \mapsto [\mathcal A \otimes \mathcal B]$ for any $\mathcal A \in \bCE_0(\L^\Sigma_\alpha(U))$ and $\mathcal B \in \bCE_0(\L^\Sigma_\alpha(V))$ where we denote by $[\mathcal A \otimes \mathcal B]$ the $0^{\rm th}$ cohomology class of $\mathcal A \otimes \mathcal B \in \bCE_0(\L^\Sigma_\alpha(U)) \otimes \bCE_0(\L^\Sigma_\alpha(V))$.
The second isomorphism in \eqref{Fac prod explicit} is induced by the $\Sym$-product $\mathcal A \otimes \mathcal B \to \mathcal A \, \mathcal B$. The final morphism in \eqref{Fac prod explicit} is induced by the factorisation product $\L^\Sigma_\alpha(U) \boplus \L^\Sigma_\alpha(V) \to \L^\Sigma_\alpha(W)$ of the unital local Lie algebra $\L^\Sigma_\alpha$ regarded as a prefactorisation algebra valued in $\udgLie_\CC^{\boplus}$, which is given explicitly by extending by zero elements of $\L^\Sigma_\alpha(U)$ and $\L^\Sigma_\alpha(V)$ to $\L^\Sigma_\alpha(W)$ then taking their sum in $\L^\Sigma_\alpha(W)$ and identifying their central extensions.
In summary, the above factorisation product is given by
\begin{center}
\raisebox{2mm}{
\begin{minipage}[b]{0.6\linewidth}
\begin{align*}
m_{(U,V), W} : \U\L^\Sigma_\alpha(U) \otimes \U\L^\Sigma_\alpha(V) &\longrightarrow \U\L^\Sigma_\alpha(W), \qquad\\
[\mathcal A]_U \otimes [\mathcal B]_V &\longmapsto [\mathcal A \, \mathcal B]_W.
\end{align*}
\end{minipage}
\quad }
\begin{minipage}[b]{0.2\linewidth}
\begin{tikzpicture}
\def\R{1}
  \fill[blue!90,
        opacity      = 0.3]
        (.05*\R,1.2*\R) circle (.8*\R);
  \draw[thick, blue] (.05*\R,1.2*\R) circle (.8*\R) node[above right=5.8mm and 3.5mm, blue]{\tiny $W$};
  \fill[red!90,
       opacity      = 0.4] (-.2*\R,1.4*\R) ellipse (.23cm and .26cm);
  \draw[thick, red] (-.2*\R,1.4*\R) ellipse (.23cm and .26cm) node[above right=1mm and 1mm, red]{\tiny $U$};
  \fill[green!90!black,
       opacity      = 0.4] (.3*\R,.9*\R) ellipse (.29cm and .32cm);
  \draw[thick, green!90!black] (.3*\R,.9*\R) ellipse (.29cm and .32cm) node[above right=1.2mm and 1mm, green!90!black]{\tiny $V$};
\end{tikzpicture}
\end{minipage}
\end{center}
The factorisation product $m_{(U_i), V} : \bigotimes_{i=1}^n \U\L^\Sigma_\alpha(U_i) \to \U\L^\Sigma_\alpha(V)$ for any inclusion $\sqcup_{i=1}^n U_i \subset V$ in $\Top(\Sigma)^\sqcup$ with $n \in \ZZ_{\geq 3}$ is obtained recursively from the above factorisation product.

\subsection{Reality conditions} \label{sec: g Sigma reality}

To describe reality conditions on the prefactorisation algebra $\U\L^\Sigma_\alpha$, which will be needed later in \S\ref{sec: unitarity}, we will show in Proposition \ref{prop: equiv PFac} below that $\U\L^\Sigma_\alpha$ is equivariant under a certain action of the group $\ZZ_2 = \langle \t \rangle$ in which $\t$ acts by an anti-linear isomorphism of prefactorisation algebras. We begin by constructing a $\ZZ_2$-action on the unital local Lie algebra $\L^\Sigma_\alpha$ after making some further assumptions on the holomorphic vector bundle $L$ from \S\ref{sec: general setup}. We will show in \S\ref{sec: reality examples} that all these assumptions hold in each of the main examples from \S\ref{sec: main examples}.

In this section we take $\tau : \Dz \to \Dz$ to be any orientation reversing involution of $\Dz$ such that there is an involution  $\tau : \Sigma \to \Sigma$ of $\Sigma$ with the property that $\pi \tau = \tau \pi$. The main examples are the Euclidean involution $\tau_E : \Dz \to \Dz$, introduced in \S\ref{sec: orientation double}, together with the identity involution $\id : \Sigma \to \Sigma$ or the Lorentzian involution $\tau_L : \Dz \to \Dz$ also introduced in \S\ref{sec: orientation double}, in the case $\Sigma = S^2$ so that $\Dz = \Sigma_+ \sqcup \Sigma_-$, together with the orientation reversing involution $\tau : \Sigma \to \Sigma$ defined in the same way as $\tau_L$ on each copy $\Sigma_\pm$.

\subsubsection{\texorpdfstring{$\ZZ_2$}{Z2}-equivariance of \texorpdfstring{$\U\L^\Sigma_\alpha$}{UL}} \label{sec: equivariance UL}

Recall from \S\ref{sec: general setup} that we are assuming the anti-holomorphic vector bundle $L$ on $\Dz$ to be such that $\Omega^{0,\chain}_c(U,L)$ is a \dg{} Lie algebra for every open $U \subset \Dz$ with Lie bracket denoted by \eqref{Omega L bracket}. We will make two further assumptions on $L$ below.

\medskip
\begin{assumption}
The orientation reversing involution $\tau : \Dz \to \Dz$ lifts to an anti-linear involutive holomorphic vector bundle automorphism
\begin{equation} \label{anti-linear auto L}
\begin{tikzcd}
L \arrow[r, "\widetilde\tau"] \arrow[d] & L \arrow[d]\\
\Dz \arrow[r, "\tau"'] & \Dz
\end{tikzcd}
\end{equation}
over $\tau : \Dz \to \Dz$ such that the induced isomorphism
\begin{equation} \label{tau involution}
\hat\tau : \Omega^{0,\chain}_c(U,L) \overset{\cong}\longrightarrow \Omega^{0,\chain}_c\big( \tau(U), L \big), \qquad
\sigma \otimes \omega \longmapsto \widetilde\tau \sigma \otimes \overline{\tau^\ast \omega}
\end{equation}
for any open subset $U \subset \Dz$, with $\sigma \in \Gamma_c(U,L)$ a smooth compactly supported section of $L$ and $\omega \in \Omega^{0,\chain}_c(U)$, is an anti-linear morphism of \dg{} Lie algebras.
\end{assumption}

\medskip
Note that since $\tau$ reverses the orientation we have $\tau^\ast I = - I \tau^\ast$ where $I : \Omega^1_c \to \Omega^1_c$ denotes the complex structure of $\Dz$, and hence we obtain an anti-linear isomorphism of cosheaves $\Omega^{0, \chain}_c \SimTo \tau^\ast \Omega^{0, \chain}_c$, $\omega \mapsto \overline{\tau^\ast \omega}$. The fact that \eqref{tau involution} commutes with the differential $\bar\partial$ follows because $\widetilde{\tau}$ is a holomorphic vector bundle automorphism so that $\bar\partial_L \widetilde{\tau} = \widetilde{\tau} \bar\partial_L$ and for every $\omega \in \Omega^{0, \chain}_c(U)$ with $U \subset \Dz$ we have
\begin{equation} \label{tau bardel commute}
\overline{\tau^\ast \bar\partial \omega} = \overline{\tfrac 12 \tau^\ast \d\omega + \tfrac 12 \ii \tau^\ast I (\d\omega)}
= \tfrac 12 \d\overline{\tau^\ast \omega} + \tfrac 12 \ii I \big(\d \overline{\tau^\ast \omega} \big)
= \bar\partial (\overline{\tau^\ast \omega}),
\end{equation}
where in the third equality we have used the fact that $\tau^\ast I = - I \tau^\ast$.

Recall from \S\ref{sec: general setup} that we are also assuming $L$ to be locally trivial over coordinate charts on $\Dz$, i.e. for every open subset $U \subset \Dz$ equipped with a holomorphic coordinate $\xi : U \to \CC$ we have $L|_U \cong U \times \fl_\xi$. Similarly, on the image open $\tau(U) \subset \Dz$ we have the holomorphic coordinate $\hat\tau \xi \coloneqq \overline{\xi \circ \tau} : \tau(U) \to \CC$ with the induced local trivialisation $L|_{\tau(U)} \cong \tau(U) \times \fl_{\hat\tau\xi}$. Since the anti-linear vector bundle automorphism $\widetilde\tau : L \to L$ acts fibrewise, it induces an anti-linear map $\tau : \fl_\xi \to \fl_{\hat\tau\xi}$ and under the isomorphism \eqref{Omega L iso} we can represent the isomorphism \eqref{tau involution} locally by
\begin{equation} \label{reality Omega fl}
\hat\tau : \fl_\xi \otimes \Omega^{0, \chain}_c(U) \overset{\cong}\longrightarrow \fl_{\hat\tau\xi} \otimes \Omega^{0, \chain}_c\big( \tau(U) \big), \qquad
\a \otimes \omega \longmapsto \tau \a \otimes \overline{\tau^\ast \omega}.
\end{equation}
Using the isomorphisms $(\cdot)_\xi : \fl \SimTo \fl_\xi$ and $(\cdot)_{\hat\tau\xi} : \fl \SimTo \fl_{\hat\tau\xi}$ with the canonical copy $\fl$ introduced in \S\ref{sec: general setup}, we obtain an induced anti-linear involution $\tau : \fl \to \fl$.

Recall finally from \S\ref{sec: general setup} that we are also assuming $L$ to be such that the \dg{} Lie algebra $\Omega^{0,\chain}_c(U,L)$ is equipped with a $2$-cocycle \eqref{cocycle alpha}. 

\medskip
\begin{assumption}
The $2$-cocycle \eqref{cocycle alpha} satisfies the reality condition
\begin{equation} \label{reality cocycle}
\overline{\alpha(a,b)} = \alpha(\hat\tau a, \hat\tau b)
\end{equation}
for all $a, b \in \Omega^{0,\chain}_c(U,L)$.
\end{assumption}

\medskip
It then follows that the anti-linear morphism of \dg{} Lie algebras \eqref{tau involution} extends to an anti-linear isomorphism of unital local Lie algebras
\begin{equation*}
\hat\tau : \L^\Dz_\alpha \overset{\cong}\longrightarrow \tau^\ast \L^\Dz_\alpha,
\end{equation*}
defined by sending $\kent \mapsto \kent$, which satisfies $(\tau^\ast \hat\tau) \circ \hat\tau = \id$. Note here that the pullback $\tau^\ast \L^\Sigma_\alpha$ of the precosheaf of unital $\dg{}$ Lie algebras $\L^\Sigma_\alpha$ also satisfies the local Lie algebra condition \eqref{LocLie condition}, so that $\tau^\ast \L^\Sigma_\alpha \in \uLocLie_\CC(\Sigma)$. Since $\pi \tau = \tau \pi$ we have an induced anti-linear isomorphism
\begin{equation} \label{equiv LocLie}
\hat\tau : \L^\Sigma_\alpha \overset{\cong}\longrightarrow \tau^\ast \L^\Sigma_\alpha
\end{equation}
also satisfying $(\tau^\ast \hat\tau) \circ \hat\tau = \id$. We will say that $\L^\Sigma_\alpha \in \uLocLie_\CC(\Sigma)$ is $\ZZ_2$-equivariant.

Recall from \S\ref{sec: twisted PFE} that we denote the composition of morphisms in $\PFac(\Sigma, \mathsf C^\otimes)$ by $\circ$.

\begin{proposition} \label{prop: equiv PFac}
We have an anti-linear isomorphism of prefactorisation algebras
\begin{equation*}
\hat \tau : \U\L^\Sigma_\alpha \overset{\cong}\longrightarrow \tau^\ast \U\L^\Sigma_\alpha
\end{equation*}
satisfying $(\tau^\ast \hat\tau) \circ \hat\tau = \id_{\U\L^\Sigma_\alpha}$, i.e. the prefactorisation algebra $\U\L^\Sigma_\alpha$ is $\ZZ_2$-equivariant.
\begin{proof}
Firstly, the anti-linear isomorphism \eqref{equiv LocLie} of unital local Lie algebras on $\Sigma$ induces, by Proposition \ref{prop: constructing prefac 1}, an anti-linear isomorphism of $\PFac(\Sigma, \udgLie_\CC^{\boplus})$, i.e. a natural isomorphism of multifunctors $\Top(\Sigma)^\sqcup \to \udgLie_\CC^{\boplus}$. Abusing notation slightly, we denote it in the same way, as $\hat \tau : \L^\Sigma_\alpha \SimTo \tau^\ast \L^\Sigma_\alpha$, but where now we view $\L^\Sigma_\alpha$ as a prefactorisation algebra on $\Sigma$ valued in $\udgLie_\CC^{\boplus}$ and $\tau^\ast \L^\Sigma_\alpha = \L^\Sigma_\alpha \, \tau$ as its pre-composition with (the identity natural transformation of) the induced multifunctor $\tau : \Top(\Sigma)^\sqcup \to \Top(\Sigma)^\sqcup$. The property $(\tau^\ast \hat\tau) \circ \hat\tau = \id$ of the morphism \eqref{equiv LocLie} of unital local Lie algebras on $\Sigma$ then turns into the property $(\tau^\ast \hat\tau) \circ \hat\tau = \id_{\L^\Sigma_\alpha}$ for the natural isomorphism $\hat \tau : \L^\Sigma_\alpha \SimTo \tau^\ast \L^\Sigma_\alpha$ of multifunctors $\Top(\Sigma)^\sqcup \to \udgLie^{\bar\oplus}_\CC$.

The horizontal post-composition of this natural isomorphism with (the identity natural transformation of)
the composite functor $H^0 \, \bCE_\chain$ yields a natural isomorphism
\begin{equation} \label{CE circ ralpha}
H^0 \, \bCE_\chain \, \hat \tau : \U\L^\Sigma_\alpha \overset{\cong}\longrightarrow \U\L^\Sigma_\alpha \circ \tau
\end{equation}
of multifunctors $\Top(\Sigma)^\sqcup \to \Vec_\CC^{\otimes}$. Then by the interchange law for the vertical and horizontal composition of natural transformations we have
\begin{equation} \label{interchange law CE r}
\big( H^0 \, \bCE_\chain \, (\tau^\ast \hat \tau) \big) \circ (H^0 \, \bCE_\chain \, \hat \tau)
= H^0 \, \bCE_\chain \, \big( (\tau^\ast \hat \tau) \circ \hat \tau \big) = H^0 \, \bCE_\chain \, \id_{\L^\Sigma_\alpha} = \id_{\U\L^\Sigma_\alpha}.
\end{equation}
If by a slight abuse of notation we denote the morphism in \eqref{CE circ ralpha} also by $\hat \tau$, so that the first morphism on the left hand side of \eqref{interchange law CE r} is the pre-composition of $\hat\tau$ with the identity natural transformation of the multifunctor $\tau$, i.e. $\tau^\ast \hat\tau = \hat\tau \, \tau$, then the relation \eqref{interchange law CE r} is then equivalent to the desired identity $(\tau^\ast \hat\tau) \circ \hat\tau = \id_{\U\L^\Sigma_\alpha}$.
\end{proof}
\end{proposition}

\subsubsection{Examples} \label{sec: reality examples}

We now show how to construct, for each of the three main examples of holomorphic vector bundles $L$ over $\Dz$ from \S\ref{sec: main examples}, the vector bundle automorphism $\widetilde\tau : L \to L$ in \eqref{anti-linear auto L} satisfying the further conditions assumed in the general discussion of \S\ref{sec: equivariance UL}.

\paragraph{Kac-Moody:} We fix an anti-linear involution $\tau : \g \SimTo \g$. This lifts the orientation reversing involution $\tau : \Dz \to \Dz$ to an anti-linear involutive automorphism of the trivial bundle \eqref{vector bundle KM}. The induced isomorphism \eqref{reality Omega fl} is clearly a morphism of \dg{} Lie algebra since
\begin{align*}
\hat \tau\big( [\X \otimes \omega, \Y \otimes \eta] \big) &= \tau [\X, \Y] \otimes \overline{\tau^\ast (\omega \wedge \eta)} = [\tau \X, \tau \Y] \otimes \overline{\tau^\ast \omega} \wedge \overline{\tau^\ast \eta}\\
&= \big[ \tau \X \otimes \overline{\tau^\ast \omega}, \tau \Y \otimes \overline{\tau^\ast \eta} \big] = \big[ \hat \tau (\X \otimes \omega), \hat \tau (\Y \otimes \eta) \big].
\end{align*}
In the second equality we used the fact that $\tau : \g \SimTo \g$ is an (anti-linear) automorphism of $\g$.

In order to show that the $2$-cocycle in \eqref{cocycle KM} satisfies \eqref{reality cocycle}, we assume that the bilinear form $\kappa : \g \otimes \g \to \CC$ is such that
\begin{equation} \label{tau on kappa}
\kappa(\tau \X, \tau \Y) = \overline{\kappa(\X, \Y)}
\end{equation}
for every $\X, \Y \in \g$.
Note that, $\overline{\tau^\ast \partial \omega} = \partial (\overline{\tau^\ast \omega})$ for every $\omega \in \Omega^{0, \chain}_c(U)$ by the same computation as in \eqref{tau bardel commute} and so for any $\omega, \eta \in \Omega^{0, \chain}_c(U)$ we then have
\begin{align} \label{conjugate central extension}
- \frac{1}{2 \pi \ii} \overline{\int_{\Dz} \partial \omega \wedge \eta}
= - \frac{1}{2 \pi \ii} \int_{\Dz} \tau^\ast \big( \overline{\tau^\ast \partial \omega} \wedge \overline{\tau^\ast \eta} \big)
= \frac{1}{2 \pi \ii} \int_{\Dz} \partial \big( \overline{\tau^\ast \omega} \big) \wedge \overline{\tau^\ast \eta}
\end{align}
where in the second equality we have also used the fact that $\tau : \Dz \to \Dz$ is orientation reversing. Note here that we have taken the integration over all of $\Dz$ rather than just $U$ by implicitly using the extension morphism $\Omega^{0, \chain}_c(U) \to \Omega^{0, \chain}_c(\Dz)$. It now follows from combining \eqref{tau on kappa} and \eqref{conjugate central extension} that the $2$-cocycle in \eqref{cocycle KM} satisfies \eqref{reality cocycle}.

We observe for later that the bilinear form $\langle \cdot, \cdot \rangle : \g \otimes \g \to \CC$, $\langle \X, \Y \rangle \coloneqq - \kappa(\tau \X, \Y)$ is a Hermitian sesquilinear form on $\g$, i.e. it is anti-linear in the first argument, linear in the second and $\langle \X, \Y \rangle = \overline{\langle \Y, \X \rangle}$. It is non-degenerate if $\kappa$ is.
In particular, $\langle \X, \X \rangle \in \RR$ for all $\X \in \g$.

\paragraph{Virasoro:} Since $\tau : \Dz \to \Dz$ is orientation reversing, its differential defines a vector bundle isomorphism $\d \tau : T^{1,0} \Dz \SimTo T^{0,1} \Dz$. Postcomposing this with complex conjugation we obtain the desired anti-linear involutive vector bundle automorphism $\widetilde \tau \coloneq \overline{\d \tau (\cdot)} : L \to L$ in \eqref{anti-linear auto L} of the holomorphic tangent bundle $L = T^{1,0} \Dz$ from \eqref{Virasoro L def}.

Recalling that $\fl_\xi = \text{span}_\CC \{ \partial_\xi \}$, and likewise $\fl_{\hat\tau\xi} = \text{span}_\CC \{ \partial_{\hat\tau\xi} \}$, the induced anti-linear map $\tau : \fl_\xi \to \fl_{\hat\tau\xi}$ along the fibres is given explicitly by $\partial_\xi \mapsto \partial_{\hat\tau \xi}$.
The isomorphism \eqref{reality Omega fl} then sends $\partial_\xi \otimes u$ for any $u \in \Omega^{0,\chain}_c(U)$ to $\partial_{\hat\tau\xi} \otimes \overline{\tau^\ast u}$. To see that this is a morphism of \dg{} Lie algebras, with respect to the Lie bracket \eqref{Lie bracket Virasoro}, note that for $u, v \in \Omega^{0, \chain}_c(U)$ we have
\begin{align*}
\hat\tau \big( [\partial_\xi \otimes u, \partial_\xi \otimes v] \big) &= \hat\tau \big( \partial_\xi \otimes (u \partial_\xi v - v \partial_\xi u) \big) = \partial_{\hat\tau\xi} \otimes \big( \overline{\tau^\ast u} \, \partial_{\hat\tau\xi} \overline{\tau^\ast v} - \overline{\tau^\ast v} \, \partial_{\hat\tau\xi} \overline{\tau^\ast u} \big)\\
&= \big[ \partial_{\hat\tau\xi} \otimes \overline{\tau^\ast u}, \partial_{\hat\tau \xi} \otimes \overline{\tau^\ast v} \big]
= \big[ \hat\tau(\partial_\xi \otimes u), \hat\tau(\partial_\xi \otimes v) \big].
\end{align*}
The anti-linear involution $\tau : \fl \to \fl$ induced on the canonical copy $\fl = \text{span}_\CC \{ \Omega \}$ is simply given by complex conjugation $x\, \Omega \mapsto \bar x \,\Omega$.

To see that the $2$-cocycle \eqref{cocycle Virasoro} satisfies the condition \eqref{reality cocycle}, we use the identity \eqref{conjugate central extension} applied to $\omega = \partial_\xi u$ and $\eta = \partial_\xi v$, and noting that $\overline{\tau^\ast \partial_\xi u} = \partial_{\hat\tau\xi} \overline{\tau^\ast u}$ and $\overline{\tau^\ast \partial_\xi v} = \partial_{\hat\tau\xi} \overline{\tau^\ast v}$ we find
\begin{align*}
\overline{\alpha_\xi( \partial_\xi \otimes u, \partial_\xi \otimes v )} &= \alpha_{\hat\tau\xi}\big( \partial_{\hat\tau\xi} \otimes \overline{\tau^\ast u}, \partial_{\hat\tau\xi} \otimes \overline{\tau^\ast v} \big) = \alpha_{\hat\tau\xi}\big( \hat\tau (\partial_\xi \otimes u), \hat\tau(\partial_\xi \otimes v) \big).
\end{align*}
The fact that $c \in \RR$ was used in the first equality.
We then recall from \S\ref{sec: main examples} that the $2$-cocycle \eqref{cocycle Virasoro} depends on the coordinate $\xi$ only up to a $2$-coboundary, explicitly $\alpha_{\hat\tau\xi} = \alpha_\xi + \delta \beta_{\hat\tau\xi, \xi}$. The $2$-cocycle property \eqref{reality cocycle} is therefore satisfied in this example only up to a $2$-coboundary.

\paragraph{$\bm \beta \bm \gamma$ system:} The pullback by the orientation reversing map $\tau : \Dz \to \Dz$ defines a vector bundle isomorphism $\tau^\ast : T^{\ast 1,0} \Dz \SimTo T^{\ast 0,1} \Dz$ so that its postcomposition with complex conjugation yields an anti-linear involutive vector bundle automorphism $\widetilde \tau \coloneq \overline{\tau^\ast (\cdot)} : T^{\ast 1,0} \Dz \to T^{\ast 1,0} \Dz$ over $\tau : \Dz \to \Dz$ of the holomorphic cotangent bundle $T^{\ast 1,0} \Dz$. Notice that this vector bundle automorphism naturally covers $\tau^{-1} : \Dz \to \Dz$, however since $\tau$ is an involution we have $\tau^{-1} = \tau$ so that it also defines a vector bundle automorphism over $\tau : \Dz \to \Dz$. Combining this with the canonical anti-linear isomorphism of the trivial vector bundle $\Dz \times \CC \to \Dz \times \CC$ given by $\tau$ on the base $\Dz$ and complex conjugation on the fibre $\CC$, we obtain the desired anti-linear involutive vector bundle automorphism $\widetilde \tau : L \to L$ of the holomorphic vector bundle \eqref{beta gamma L def}.

Recalling that in the present case we have $\fl_\xi = \text{span}_\CC \{ 1, \d \xi \}$ and $\fl_{\hat\tau\xi} = \text{span}_\CC \{ 1, \d (\hat\tau\xi) \}$, the induced anti-linear involution $\tau : \fl_\xi \to \fl_{\hat\tau\xi}$ is given on the basis elements by $1 \mapsto 1$ and $\d \xi \mapsto \d (\hat\tau \xi)$. The induced anti-linear involution $\tau : \fl \to \fl$ on $\fl = \text{span}_\CC \{ \beta, \gamma \}$ is again just given by complex conjugation $x\,\beta + y \, \gamma \mapsto \bar x \, \beta + \bar y \, \gamma$.

Since the Lie bracket is trivial the only thing to check is that the $2$-cocycle \eqref{cocycle beta gamma} satisfies the condition \eqref{reality cocycle}. But under the canonical isomorphism $\Omega^{0,\chain}_c(U,L) \cong \Omega^{0,\chain}_c(U) \oplus \Omega^{1,\chain}_c(U)$, the isomorphism \eqref{reality Omega fl} acts simply as $\omega \mapsto \overline{\tau^\ast \omega}$ on any $\omega \in \Omega^{0,\chain}_c(U) \oplus \Omega^{1,\chain}_c(U)$. It then follows, exactly as in \eqref{conjugate central extension}, that
\begin{align*}
\overline{\alpha(\omega, \eta)} = - \frac{1}{2 \pi \ii} \overline{\int_{\Dz} \omega \wedge \eta}
= \frac{1}{2 \pi \ii} \int_{\Dz} \overline{\tau^\ast \omega} \wedge \overline{\tau^\ast \eta} = \alpha(\hat\tau \omega, \hat\tau\eta)
\end{align*}
for any $\omega, \eta \in \Omega^{0,\chain}_c(U) \oplus \Omega^{1,\chain}_c(U)$, where as before we implicitly use the extension morphism $\Omega^{0, \chain}_c(U) \to \Omega^{0, \chain}_c(\Dz)$, as required.

\section{Full vertex algebra \texorpdfstring{$\hV^{\fl,\alpha}_p$}{Fg}} \label{sec: vertex alg}

As in \S\ref{sec: prefac alg}, we will keep working with an arbitrary connected $2$-dimensional conformal manifold $\Sigma$ which could be non-orientable and also with boundary. However, throughout this paper we will not treat boundary conditions and therefore ignore points on the boundary $\partial \Sigma$ by only working locally arounds points in the interior $\Sigma^\circ \subset \Sigma$. Later, towards the end of \S\ref{sec: chiral states}, we will specialise to the case of the $2$-sphere $\Sigma = S^2$.

\subsection{Underlying vector space} \label{sec: Vec structure}

As mentioned in \S\ref{sec: class of VOAs}, we are interested in full vertex algebras whose underlying vector spaces are built as induced modules of certain infinite-dimensional Lie algebras. In what follows we will focus on the three main examples of infinite-dimensional Lie algebras given by centrally extended loop algebras which includes affine Kac-Moody algebras, the Virasoro algebra and the $\beta\gamma$ system. These will closely correspond to the three examples of unital local Lie algebra $\L^\Sigma_\alpha$ given in \S\ref{sec: main examples}. It will be convenient to describe all these examples uniformly as follows.

Recall the finite-dimensional vector space $\fl$ introduced in \S\ref{sec: general setup}, and consider the associated infinite-dimensional vector spaces
\begin{equation} \label{inf dim algebras}
\hg \coloneqq \fl \otimes \CC[t, t^{-1}] \oplus \CC {\ms k}, \qquad \hbg \coloneqq \fl \otimes \CC[\bar t, \bar t^{-1}] \oplus \CC \bar{\ms k}
\end{equation}
where $t$ and $\bar t$ are independent formal variables referred to as loop parameters. These vector spaces come equipped with Lie algebra structures $[\cdot, \cdot] : \hg \otimes \hg \to \hg$ and $[\cdot, \cdot] : \hbg \otimes \hbg \to \hbg$, respectively, with respect to which $\ms k$ and $\bar{\ms k}$ are central. We will describe these explicitly below in each of the three main examples in terms of the loop generators
\begin{equation} \label{loop generators}
\a_{(n)} \coloneqq \a \otimes t^n \in \hg, \qquad \bar\a_{(n)} \coloneqq \a \otimes \bar t^n \in \hbg,
\end{equation}
for $\a \in \fl$ and $n \in \ZZ$. In all cases, the Lie algebras \eqref{inf dim algebras} admit a decomposition into a direct sum $\hg = \hg_+ \oplus \hg_-$ and $\hbg = \hbg_+ \oplus \hbg_-$ of Lie subalgebras
\begin{alignat*}{2}
\hg_+ &\coloneqq \fl \otimes \CC[t] \oplus \CC {\ms k}, &\qquad \hg_- &\coloneqq \fl \otimes t^{-1} \CC[t^{-1}]\\
\hbg_+ &\coloneqq \fl \otimes \CC[\bar t] \oplus \CC \bar{\ms k}, &\qquad \hbg_- &\coloneqq \fl \otimes \bar t^{-1} \CC[\bar t^{-1}].
\end{alignat*}

Let $\CC \vac$ be the $1$-dimensional module over the direct sum Lie algebra $\hg_+ \oplus \hbg_+$ on which $\fl \otimes \CC[t]$ and $\fl \otimes \CC[\bar t]$ act trivially and the central elements ${\ms k}$ and $\bar{\ms k}$ both act by multiplication by $1$. Define the \emph{full affine vertex algebra} $\hVV^{\fl,\alpha}$ as the induced module over the direct sum Lie algebra $\hg \oplus \hbg$, namely
\begin{equation*}
\hVV^{\fl,\alpha} \coloneqq \Ind_{\hg_+ \oplus \hbg_+}^{\hg \oplus \hbg} \CC \vac.
\end{equation*}
Note that we have included the $2$-cocycle $\alpha$ from \eqref{cocycle alpha} in the notation. This is a slight abuse of notation since $\alpha$ is not directly used in the definition but it will be closely related to the central extension of the Lie algebra structures on \eqref{inf dim algebras} in each of the examples described below. We therefore use the subscript $\alpha$ in the notation to emphasise this choice of central extension.
In each of the three cases we will have an isomorphism of vector spaces $\hVV^{\fl,\alpha} \cong U\big( \hg_- \oplus \hbg_- \big) \vac$ so that a general element of $\hVV^{\fl,\alpha}$ is given by a linear combination of expressions of the form
\begin{equation} \label{gen state Vkx}
\a^r_{(-m_r)} \ldots \a^1_{(-m_1)} \bar \b^{\bar r}_{(-n_{\bar r})} \ldots \bar \b^1_{(-n_1)} \vac
\end{equation}
for any $\a^i \in \fl$, $m_i \in \ZZ_{\geq 1}$ with $i \in \{ 1, \ldots, r \}$ and $\b^j \in \fl$, $n_j \in \ZZ_{\geq 1}$ with $j \in \{ 1, \ldots, \bar r \}$ where $r, \bar r \in \ZZ_{\geq 0}$. The full affine vertex algebra $\hVV^{\fl,\alpha}$ is canonically isomorphic to the tensor product of two copies of the usual affine vertex algebra, namely we have an isomorphism of complex vector spaces $\hVV^{\fl,\alpha} \cong \VV^{\fl,\alpha} \otimes \bar{\VV}^{\fl,\alpha}$ where
\begin{equation*}
\VV^{\fl,\alpha} \coloneqq \Ind_{\hg_+}^{\hg} \CC \vac, \qquad
\bar{\VV}^{\fl,\alpha} \coloneqq \Ind_{\hbg_+}^{\hbg} \CC \vac
\end{equation*}
are the affine vertex algebra and its `anti-chiral' copy.

We now describe the Lie algebra structures on \eqref{inf dim algebras} in all three cases of interest, where the finite-dimensional vector space $\fl$ was defined in each case in \S\ref{sec: general setup}.

\paragraph{Kac-Moody:}
We can endow \eqref{inf dim algebras} for $\fl = \g$ with Lie brackets described in terms of the loop generators \eqref{loop generators} by
\begin{subequations} \label{KM algebra vertex modes def}
\begin{align}
\label{KM algebra vertex modes def a} \big[ \X_{(m)}, \Y_{(n)} \big] &= [\X, \Y]_{(m+n)} + m\, \kappa(\X, \Y) \delta_{m+n, 0} \, {\ms k}, \\
\label{KM algebra vertex modes def b} \big[ \bar \X_{(m)}, \bar \Y_{(n)} \big] &= \overline{[\X, \Y]}_{(m+n)} + m\, \kappa(\X, \Y) \delta_{m+n, 0} \, \bar{\ms k}.
\end{align}
\end{subequations}
for every $\X, \Y \in \g$ and $m,n \in \ZZ$. In this case $\hg$ and $\hbg$ are two copies of the centrally extended loop algebra associated with $\g$. Since the case when $\g$ is a simple Lie algebra corresponds to a pair of untwisted affine Kac-Moody algebras, by a slight abuse of terminology we will keep referring to the general case with $\g$ arbitrary as the `Kac-Moody' case.

\paragraph{Virasoro:}
We endow \eqref{inf dim algebras} for $\fl = \text{span}_\CC \{ \Omega \}$ with the Lie brackets given in terms of the loop generators \eqref{loop generators} of the fixed basis element $\Omega$ by
\begin{subequations} \label{Virasoro algebra vertex modes def}
\begin{align}
\label{Virasoro algebra vertex modes def a} \big[ \Omega_{(m)}, \Omega_{(n)} \big] &= (m-n) \Omega_{(m+n-1)} + \frac{m(m-1)(m-2)}{12} c \, \delta_{m+n, 2} \, {\ms k}, \\
\label{Virasoro algebra vertex modes def b} \big[ \bar\Omega_{(m)}, \bar\Omega_{(n)} \big] &= (m-n) \bar\Omega_{(m+n-1)} + \frac{m(m-1)(m-2)}{12} c\, \delta_{m+n, 2} \, \bar{\ms k}.
\end{align}
\end{subequations}
In this case $\hg$ and $\hbg$ are two copies of the Virasoro algebra, where the usual generators $L_{(n)}$ and $\bar L_{(n)}$ satisfying the more familiar looking Virasoro relations are given by a simple shift in the indices, namely $L_{(n)} = \Omega_{(n+1)}$ and $\bar L_{(n)} = \bar \Omega_{(n+1)}$.

\paragraph{$\bm \beta \bm \gamma$ system:}
We endow \eqref{inf dim algebras} for $\fl = \text{span}_\CC \{ \beta, \gamma \}$ with the Lie brackets given in terms of the loop generators \eqref{loop generators} of the fixed basis elements $\beta$ and $\gamma$ by
\begin{subequations} \label{beta gamma algebra vertex modes def}
\begin{align}
\label{beta gamma algebra vertex modes def a} \big[ \beta_{(m)}, \gamma_{(n)} \big] &= \delta_{m+n, -1} \, {\ms k}, \\
\label{beta gamma algebra vertex modes def b} \big[ \bar\beta_{(m)}, \bar\gamma_{(n)} \big] &= \delta_{m+n, -1} \, \bar{\ms k}.
\end{align}
\end{subequations}
In this case $\hg$ and $\hbg$ are two copies of the infinite-dimensional Weyl algebra, also know as the $\beta\gamma$ system, whose generators are often denoted by $a_{(n)}$ and $a^\ast_{(n)}$ with the more standard Lie algebra relations obtained from \eqref{beta gamma algebra vertex modes def} by a simple shift of indices 
\begin{equation} \label{a a star def}
a_{(n)} = \beta_{(n)}, \quad a^\ast_{(n)} = \gamma_{(n-1)}, \quad \bar a_{(n)} = \bar \beta_{(n)}, \quad \bar a^\ast_{(n)} = \bar \gamma_{(n-1)}.
\end{equation}

\subsubsection{Geometric realisation in \texorpdfstring{$\U\L^\Sigma_\alpha(U)$}{UL}}

Let $U \subset \Sigma^\circ$ be a connected open subset of the interior $\Sigma^\circ \subset \Sigma$. Its preimage under the projection $\pi : \Dz \to \Sigma$ is a disjoint union of two copies of $U$ which we denote $\pi^{-1}(U) = U_+ \sqcup U_-$. Suppose that $U_+$ is equipped with a holomorphic coordinate $\xi : U_+ \to \CC$ such that $\xi(U_+) \subset \CC$ is bounded. Since $U_-$ comes equipped with the opposite complex structure to $U_+$, the same map $\xi$ viewed as a function on $U_-$ defines an anti-holomorphic coordinate $\xi : U_- \to \CC$ on $U_- \subset \pi^{-1}(U)$. In other words, the complex conjugate map $\bar \xi : U_- \to \CC$ defines a holomorphic coordinate on $U_-$. We then refer to $\xi : U \to \CC$ as a \emph{local (holomorphic) coordinate} on $U \subset \Sigma^\circ$.

Given any $p \in U \subset \Sigma^\circ$, we let $p_\pm \in U_\pm$ denote its preimages under $\pi : \Dz \to \Sigma$. We refer to $\xi(p) \coloneqq \xi(p_+)$, i.e. the value of $\xi : U_+ \to \CC$ at $p_+ \in U_+$, as its \emph{holomorphic coordinate}. Similarly, we refer to $\bar\xi(p) \coloneqq \bar\xi(p_-)$, i.e. the value of $\bar\xi : U_- \to \CC$ at $p_- \in U_-$, as its \emph{anti-holomorphic coordinate}. We also define \emph{shifted local (anti-)holomorphic coordinates} as

\begin{equation} \label{shifted coords def}
\xi_p \coloneqq \xi - \xi(p) : U_+ \longrightarrow \CC, \qquad
\bar \xi_p \coloneqq \bar \xi - \bar \xi(p) : U_- \longrightarrow \CC
\end{equation}
so that any $q \in U \subset \Sigma$ has holomorphic coordinate $\xi_p(q) = \xi(q) - \xi(p)$ and anti-holomorphic coordinate $\bar \xi_p(q) = \bar \xi(q) - \bar \xi(p)$.

In view of giving a geometric description of the vector space $\hVV^{\fl,\alpha}$ using the prefactorisation algebra $\U\L^\Sigma_\alpha$ it is useful to first consider the level $1$ subspace $(\hg_- \oplus \hbg_-) \vac \subset \hVV^{\fl,\alpha}$.
We have an injective linear map
\begin{align} \label{level 1 injection}
(\hg_- \oplus \hbg_-) \vac &\longhookrightarrow \fl \otimes \O_\infty\big( \xi_p(U_+)^{\rm c} \big) \oplus \fl \otimes \O_\infty \big( \bar \xi_p(U_-)^{\rm c} \big),\\
\a_{(-m)} \vac + \bar \b_{(-n)} \vac &\longmapsto \big( \a \otimes [\lambda^{-m}], \b \otimes [\bar\lambda^{-n}] \big) \notag
\end{align}
for $\a, \b \in \fl$ and $m, n \in \ZZ_{\geq 1}$, where $[\lambda^{-m}] \in \O_\infty\big( \xi_p(U_+)^c \big)$ denotes the germ of the holomorphic function $\xi_p(U_+)^c \subset \CP \setminus \{ 0 \} \to \CC$, $\lambda \mapsto \lambda^{-m}$ and similarly $[\bar\lambda^{-n}] \in \O_\infty\big( \xi_p(U_-)^c \big)$ denotes the germ of the holomorphic function $\bar \xi_p(U_-)^c \subset \CP \setminus \{ 0 \} \to \CC$, $\bar\lambda \mapsto \bar\lambda^{-n}$.

Recall the $(0,1)$-form $\lceil f \rceil \in \Omega^{0,1}_c(W)$ on a bounded open subset $W \subset \CC$ associated with a given germ $[f] \in \O_\infty(W^c)$ defined in \S\ref{sec: coh of gD}. In the case of the germ $[\lambda^{-m}] \in \O_\infty\big( \xi_p(U_+)^c \big)$ we can pick the representative $f : \Delta_f \coloneqq \CP \setminus \{ 0 \} \to \CC$, $\lambda \mapsto \lambda^{-m}$. By Lemma \ref{lem: up arrow indep choices}, the $(0,1)$-form $\lceil \lambda^{-m} \rceil \in \Omega^{0,1}_c\big( \xi_p(U_+) \big)$ is then given, up to $\bar\partial$-exact terms, by $- f\, \bar\partial \rho$ where $\rho \in \Omega^{0,0}_c\big(\xi_p(U_+)\big)$ is a smooth bump function equal to $1$ on some neighbourhood of $\{ 0\} = (\Delta_f)^c \subset \xi_p(U_+)$. The pullback of $f|_{\xi_p(U_+) \setminus \{0\}}$ along $\xi_p : U_+ \to \xi_p(U_+)$ is $\xi_p^\ast f : U_+ \setminus \{ p_+\} \to \CC$, $q \mapsto \xi_p(q)^{-m}$, i.e. the function $\xi_p^{-m} : U_+ \setminus \{ p_+ \} \to \CC$. Then $\xi_p^\ast \lceil \lambda^{-m} \rceil \in \Omega^{0,1}_c(U_+)$, which for brevity we denote by $\lceil \xi_p^{-m} \rceil$, is given by Lemma \ref{lem: up arrow indep choices} up to $\bar\partial$-exact terms by $\xi_p^{-m} \bar\partial (\xi_p^\ast \rho)$ where $\xi_p^\ast \rho \in \Omega^{0,0}_c(U_+)$ is equal to $1$ in a neighbourhood of $p_+ \in U_+$. Similarly, we let $\lceil \bar\xi_p^{-n} \rceil$ stand for $\bar\xi_p^\ast \lceil \bar\lambda^{-n} \rceil \in \Omega^{0,1}_c(U_-)$ which is given again by Lemma \ref{lem: up arrow indep choices} up to $\bar\partial$-exact terms by $\bar\xi_p^{-n} \bar\partial (\bar\xi_p^\ast \rho')$ for some smooth bump function $\rho' \in \Omega^{0,0}_c\big( \bar\xi_p(U_-) \big)$ equal to $1$ in some neighbourhood of $0 \in \bar\xi_p(U_-)$.

Combining \eqref{level 1 injection} with the isomorphism from Theorem \ref{thm: cohomology gSigma} defined using the local holomorphic coordinate on $\pi^{-1}(U) = U_+ \sqcup U_-$ given by \eqref{shifted coords def} we obtain an injective linear map
\begin{align} \label{hg- to H1}
\big( \hg_- \oplus \hbg_- \big) \vac &\longhookrightarrow H^1\big( \L^\Sigma_\alpha(U) \big) = H^0\big( \L^\Sigma_\alpha(U)[1] \big),\\
\a_{(-m)} \vac + \bar \b_{(-n)} \vac &\longmapsto \big[ \a_{\xi_p} \otimes \lceil \xi_p^{-m} \rceil + \b_{\bar\xi_p} \otimes \lceil \bar \xi_p^{-n} \rceil \big]_U. \notag
\end{align}
Note, in particular, that although $\lceil \xi_p^{-m} \rceil$ and $\lceil \bar\xi_p^{-n} \rceil$ were described above only up to $\bar\partial$-exact terms, this ambiguity drops out from taking the cohomology class $[\cdot]_U$ in \eqref{hg- to H1}.

If we regard $H^0\big( \L^\Sigma_\alpha(U)[1] \big)$ as a subspace of $\U\L^\Sigma_\alpha(U)$ then \eqref{hg- to H1} gives an embedding of the subspace $(\hg_- \oplus \hbg_-) \vac \subset \hVV^{\fl,\alpha}$ into $\U\L^\Sigma_\alpha(U)$ for any neighbourhood $U \subset \Sigma^\circ$ of $p$ equipped with the local coordinate $\xi$.
A general element of $\hVV^{\fl,\alpha}$ involves products of elements from $\hg_-$ and $\hbg_-$ in some order, and we can realise such an ordering geometrically in $\U\L^\Sigma_\alpha(U)$ through the choice of supports of the smooth bump functions entering the definitions of $\lceil \xi_p^{-m} \rceil \in \Omega^{0,1}_c(U_+)$ and $\lceil \bar\xi_p^{-n} \rceil \in \Omega^{0,1}_c(U_-)$. To describe this ordering explicitly we introduce the important notion of \emph{nested} open subsets that will be used repeatedly throughout the rest of the paper.

\begin{definition}
Given two open subsets $V, W \subset U_\pm$, we say that $V$ is \emph{nested} in $W$ if $V$ is relatively compact in $W$, i.e. if $\overline{V} \subset W$ and $\overline{V}$ is compact, and we write this as $V \Subset W$.
\end{definition}
Let $\Omega^{0,0}_c(W)^1_V \subset \Omega^{0,0}_c(W)$ denote the subset consisting of elements equal to $1$ on $V$.
Given any nested neighbourhoods $p_+ \in V \Subset W \subset U_+$ we pick a $\rho^V_W \in \Omega^{0,0}_c(W)^1_V$ and define
\begin{equation} \label{up arrow def U V}
\lceil \xi^n_p \rceil^V_W \coloneqq - \xi^n_p \bar\partial \rho^V_W \in \Omega^{0,1}_c(W)
\end{equation}
for all $n \in \ZZ$. It follows from Lemma \ref{lem: up arrow indep choices} that $\lceil \xi^n_p \rceil^V_W$ differs by $\bar\partial$-exact terms from $\lceil \xi^n_p \rceil$ defined above \eqref{hg- to H1}. We similarly define $\lceil \bar\xi^n_p \rceil^V_W$ for nested neighbourhoods $p_- \in V \Subset W \subset U_-$.

\begin{lemma} \label{lem: residue integral}
Let $U \subset \Sigma^\circ$ be a neighbourhood of $p \in \Sigma^\circ$ equipped with a holomorphic coordinate $\xi : U \to \CC$.
For any $n \in \ZZ$ we have
$\frac{1}{2 \pi \ii} \int_{U_+} \lceil \xi_p^n \rceil \wedge \d\xi_p = \delta_{n,-1}$ and
$\frac{1}{2 \pi \ii} \int_{U_-} \lceil \bar\xi_p^n \rceil \wedge \d\bar\xi_p = \delta_{n,-1}$.
\begin{proof}
Consider the integral over $U_+$. Let $p_+ \in V \Subset W \subset U_+$ be nested neighbourhoods and pick a smooth bump function $\rho^V_W \in \Omega^{0,0}_c(W)^1_V$. Since $\lceil \xi_p^n \rceil - \lceil \xi_p^n \rceil^V_W$ is $\bar\partial$-exact by Lemma \ref{lem: up arrow indep choices}, we can replace $\lceil \xi_p^n \rceil$ in the integrand by $\lceil \xi_p^n \rceil^V_W$. The given integral over $U_+$ thus evaluates to
\begin{center}
\vspace{-3mm}
\begin{minipage}[b]{0.75\linewidth}
\begin{align*}
&\frac{1}{2 \pi \ii} \int_{U_+} \lceil \xi_p^n \rceil \wedge \d\xi_p = \frac{1}{2 \pi \ii} \int_{U_+} \xi_p^n \d\xi_p \wedge \bar \partial \rho^V_W
= \frac{1}{2 \pi \ii} \int_{W \setminus V} \xi_p^n \d \xi_p \wedge \bar \partial \rho^V_W\\
&\qquad\qquad = - \frac{1}{2 \pi \ii} \int_{W \setminus V} \d \big( \rho^V_W \xi_p^n \d \xi_p \big)
= \frac{1}{2 \pi \ii} \int_{\partial V} \xi_p^n \d \xi_p = \delta_{n,-1},
\end{align*}
\end{minipage}
\begin{minipage}[b]{0.15\linewidth}
\centering
\raisebox{2mm}{
\begin{tikzpicture}
\def\R{1}
  \fill[blue!90,
        opacity      = 0.3,
        even odd rule]
        (.05*\R,1.2*\R) circle[radius=.8*\R] circle[radius=.3*\R];
  \draw[thick, blue] (.05*\R,1.2*\R) circle (.8*\R);
  \draw[thick, blue] (.05*\R,1.2*\R) circle (.3*\R);
  \filldraw[thick] (.05*\R,1.2*\R) node[below=-.7mm]{\tiny $p$} circle (0.02*\R) node[above right=1.9mm and -2.5mm, blue]{\tiny $W \!\setminus\! V$};
\end{tikzpicture}}
\end{minipage}
\end{center}
where the second equality follows from the support property of $\bar \partial \rho^V_W$. The third equality uses the fact that $\xi_p^n \d\xi_p$ is closed on $W \setminus V$ and in the fourth equality we used Stokes's theorem together with the support properties of $\rho^V_W$. The last integral is evaluated using the residue theorem, noting that $\partial V$ is oriented counterclockwise.

The evaluation of the integral over $U_-$ is completely analogous.
\end{proof}
\end{lemma}

We are now in a position to describe the geometric realisation of the vector space $\hVV^{\fl,\alpha}$ in $\U\L^\Sigma_\alpha(U)$ around a point $p \in U$ in any open subset $U \subset \Sigma^\circ$ of $\Top(\Sigma)$ equipped with a local coordinate $\xi : U \to \CC$ such that $\xi(U) \subset \CC$ is bounded. Specifically, we define a linear map
\begin{subequations} \label{cal Y map defined}
\begin{equation} \label{cal Y map defined a}
(\cdot)^{\xi_p}_U : \hVV^{\fl,\alpha} \longrightarrow \U\L^\Sigma_\alpha(U), \qquad
A \longmapsto A^{\xi_p}_U
\end{equation}
as follows. It sends the state of the form \eqref{gen state Vkx} to
\begin{equation} \label{cal Y map defined b}
\Bigg[ \prod_{i=1}^r s\Big( \a^i_{\xi_p} \otimes \lceil \xi_p^{-m_i} \rceil^{U_{i-1}}_{U_i} \Big) \prod_{j=1}^{\bar r} s\Big( \b^j_{\bar\xi_p} \otimes \lceil \bar \xi_p^{-n_j} \rceil^{V_{j-1}}_{V_j} \Big) \Bigg]_U,
\end{equation}
\end{subequations}
where $p_+ \in U_0 \Subset \ldots \Subset U_r \subset U_+$ and $p_- \in V_0 \Subset \ldots \Subset V_{\bar r} \subset U_-$ are choices of nested sequences of open subsets. In particular, the vacuum $\vac$ is sent to the class $[1]_U$.

\begin{proposition} \label{prop: Y map def}
The linear map \eqref{cal Y map defined} is well defined. In particular, we have
\begin{subequations} \label{commutator check}
\begin{align}
\label{commutator check a} \ldots \big( \a_{(m)} \b_{(n)} - \b_{(n)} \a_{(m)} \big) \ldots \vac^{\xi_p}_U &= \ldots \big[ \a_{(m)}, \b_{(n)} \big] \ldots \vac^{\xi_p}_U,\\
\label{commutator check b} \ldots \big( \a_{(m)} \bar \b_{(n)} - \bar \b_{(n)} \a_{(m)} \big) \ldots \vac^{\xi_p}_U &= 0,\\
\label{commutator check c} \ldots \big( \bar \a_{(m)} \bar \b_{(n)} - \bar \b_{(n)} \bar \a_{(m)} \big) \ldots \vac^{\xi_p}_U &= \ldots \big[ \bar \a_{(m)}, \bar \b_{(n)} \big] \ldots \vac^{\xi_p}_U
\end{align}
for all $\a, \b \in \fl$ and $m, n \in \ZZ$, where on the right hand sides of \eqref{commutator check a} and \eqref{commutator check c} we use the Lie brackets on \eqref{inf dim algebras} and in each equation the ellipses on either side of each term represent the same sum of product of elements from $\hg \oplus \hbg$. We also have, for any $k \in \ZZ_{\geq 0}$,
\begin{equation}
\label{commutator check d} \ldots \a_{(k)} \vac^{\xi_p}_U = 0, \qquad
\ldots \bar \a_{(k)} \vac^{\xi_p}_U = 0.
\end{equation}
\end{subequations}
\begin{proof}
Focusing on the chiral part of the state, first note that the cohomology class in \eqref{cal Y map defined b} is independent of the choice of smooth bump functions $\rho^{U_{i-1}}_{U_i} \in \Omega^{0,0}_c(U_i)^1_{U_{i-1}}$ and of the nested sequences of open subsets $p_+ \in U_0 \Subset \ldots \Subset U_r \subset U_+$. Indeed, Lemma \ref{lem: up arrow indep choices} allows one to modify the smooth bump functions $\rho^{U_{i-1}}_{U_i}$ one by one while maintaining disjoint supports 
\begin{equation} \label{disjoint supports}
\supp \big( \bar\partial \rho^{U_{i-1}}_{U_i} \big) \cap \supp \big( \bar\partial \rho^{U_{j-1}}_{U_j} \big) = \emptyset,
\end{equation}
for all $i \neq j$, until we arrive at a desired set of new smooth bump functions $\rho^{U'_{i-1}}_{U'_i} \in \Omega^{0,0}_c(U'_i)^1_{U'_{i-1}}$ associated with another choice of nested sequence of open subsets $p_+ \in U'_0 \Subset \ldots \Subset U'_r \subset U_+$ in $\Top(\Sigma)$. Keeping the supports disjoint as in \eqref{disjoint supports} ensures that there is never any contribution from the differential $\d_{[\cdot,\cdot]}$, as defined in \S\ref{sec: udgLie}, when working up to $\d_\CE$-exact terms, so that we can effectively work up to $\bar\partial$-exact terms. The same applies to the anti-chiral part.

Next, we must check that the cohomology class $A^{\xi_p}_U \in \U\L^\Sigma_\alpha(U)$ as given in \eqref{cal Y map defined b} does not dependent on the way in which the input state $A \in \hVV^{\fl,\alpha}$ is written. It is enough to check the identities \eqref{commutator check}. We will focus on \eqref{commutator check a} and the first identity in \eqref{commutator check d}, the proof of the other identities being completely analogous.

Choosing a nested sequence of open subsets $p_+ \in U_0 \Subset U_1 \Subset U_2 \Subset U_3 \subset U_+$ we can write the left hand side of \eqref{commutator check a} as
\begin{align} \label{commutator check 2}
&\Big[ \ldots s\big( \a_{\xi_p} \otimes \xi_p^m \bar\partial \rho^{U_1}_{U_2} \big) s\big( \b_{\xi_p} \otimes \xi_p^n \bar\partial \rho^{U_0}_{U_1} \big) \ldots \Big]_U - \Big[ \ldots s\big( \b_{\xi_p} \otimes \xi_p^n \bar\partial \rho^{U_2}_{U_3} \big) s\big( \a_{\xi_p} \otimes \xi_p^m \bar\partial \rho^{U_1}_{U_2} \big) \ldots \Big]_U \notag\\
&\qquad\qquad\qquad = \Big[ \ldots s\Big( \b_{\xi_p} \otimes \bar\partial \Big( \xi_p^n \big( \rho^{U_0}_{U_1} - \rho^{U_2}_{U_3} \big) \Big) \Big) s\big( \a_{\xi_p} \otimes \xi_p^m \bar\partial \rho^{U_1}_{U_2} \big) \ldots \Big]_U,
\end{align}
where in the first term on the left hand side
we chose to use the nested sequence of open subsets $p_+ \in U_0 \Subset U_1 \Subset U_2 \subset U_+$ in $\Top(\Sigma)$ and in the second term we used $p_+ \in U_1 \Subset U_2 \Subset U_3 \subset U_+$ instead.
In each of the three terms in \eqref{commutator check 2}, the ellipses on either side represent the same sum of products of degree $0$ elements from $\L^\Sigma_\alpha(U)[1]$, whose support is disjoint from $U_3 \setminus U_0$.
In the second line of \eqref{commutator check 2} we have used the fact that $\xi_p^n$ is holomorphic on the support of $\rho^{U_0}_{U_1} - \rho^{U_2}_{U_3}$ since the latter vanishes in a neighbourhood of $p$.

By definition of the differential $\d_\CE$ in the Chevalley-Eilenberg complex, see \S\ref{sec: udgLie}, we find
\begin{align*}
&\d_\CE\Big( \ldots s\big( \b_{\xi_p} \otimes \xi_p^n \big( \rho^{U_0}_{U_1} - \rho^{U_2}_{U_3} \big) \big) s\big( \a_{\xi_p} \otimes \xi_p^m \bar\partial \rho^{U_1}_{U_2} \big) \ldots \Big)\\
&\qquad = \ldots  \bigg( \!\! - s\Big( \b_{\xi_p} \otimes \bar\partial\Big( \xi_p^n \big( \rho^{U_0}_{U_1} - \rho^{U_2}_{U_3} \big) \Big) \Big) s\big( \a_{\xi_p} \otimes \xi_p^m \bar\partial \rho^{U_1}_{U_2} \big)\\
&\qquad\qquad\qquad\qquad - s \Big[ \b_{\xi_p} \otimes \xi_p^n \big( \rho^{U_0}_{U_1} - \rho^{U_2}_{U_3} \big), \a_{\xi_p} \otimes \xi_p^m \bar\partial \rho^{U_1}_{U_2} \Big]_\alpha \bigg) \ldots
\end{align*}
where the minus sign in the last line comes from the definition of the differential $\d_{[\cdot, \cdot]_\alpha}$ in \S\ref{sec: udgLie} and using the fact that the first argument in the Lie bracket has degree $0$ in $\L^\Sigma_\alpha(U)$. Note also that the action of $\d_\CE$ on the ellipses terms vanishes since these consist of sums of products of degree $0$ elements from $\L^\Sigma_\alpha(U)[1]$, so they are killed by $\bar\partial$, and whose support is disjoint from $U_3 \setminus U_0$, so they also do not contribute to the action of $\d_{[\cdot,\cdot]_\alpha}$.
It now follows that the cohomology class on the right hand side of \eqref{commutator check 2} can be rewritten as
\begin{align} \label{commutator check 3}
&\Big[ \ldots s\big( \a_{\xi_p} \otimes \xi_p^m \bar\partial \rho^{U_1}_{U_2} \big) s\big( \b_{\xi_p} \otimes \xi_p^n \bar\partial \rho^{U_0}_{U_1} \big) \ldots \Big]_U - \Big[ \ldots s\big( \b_{\xi_p} \otimes \xi_p^n \bar\partial \rho^{U_2}_{U_3} \big) s\big( \a_{\xi_p} \otimes \xi_p^m \bar\partial \rho^{U_1}_{U_2} \big) \ldots \Big]_U \notag\\
&\qquad\qquad\qquad = - \Big[ \ldots s \Big[ \b_{\xi_p} \otimes \xi_p^n \big( \rho^{U_0}_{U_1} - \rho^{U_2}_{U_3} \big), \a_{\xi_p} \otimes \xi_p^m \bar\partial \rho^{U_1}_{U_2} \Big]_\alpha \ldots \Big]_U.
\end{align}

In order to complete the proof, it remains to explicitly evaluate the Lie bracket in $\L^\Sigma_\alpha(U)$ on the second line of the right hand side, using the formulae in \S\ref{sec: main examples}, and relate it to the Lie bracket in $\hg$ on the right hand side of \eqref{commutator check a}. We proceed by separately considering the three main examples of \S\ref{sec: main examples}, namely the Kac-Moody, Virasoro and $\beta\gamma$ system cases.

\paragraph{Kac-Moody:} Recall from \S\ref{sec: general setup} that the map $(\cdot)_{\xi_p} : \fl \SimTo \fl_{\xi_p}$ in this case is the identity so that with $\a = \X \in \g$ and $\b = \Y \in \g$ we find
\begin{align*}
&- s \Big[ \Y \otimes \xi_p^n \big( \rho^{U_0}_{U_1} - \rho^{U_2}_{U_3} \big), \X \otimes \xi_p^m \bar\partial \rho^{U_1}_{U_2} \Big]_\alpha\\
&\qquad\qquad\qquad = - s\big( [\X, \Y] \otimes \xi_p^{m+n} \bar\partial \rho^{U_1}_{U_2} \big) - \frac{\kappa(\X, \Y)}{2 \pi \ii} \int_{U_+} \partial( \xi_p^n) \wedge \xi_p^m \bar\partial \rho^{U_1}_{U_2}\\
&\qquad\qquad\qquad = s\big( [\X, \Y] \otimes \lceil \xi_p^{m+n} \rceil^{U_1}_{U_2} \big) + m \, \kappa(\X, \Y) \delta_{m+n,0},
\end{align*}
where in the first step we used the explicit form of the Lie bracket on $\L^\Sigma_\alpha(U)$ given by \eqref{bracket and cocycle KM} and the fact that
\begin{equation} \label{rho support identity}
(\rho^{U_0}_{U_1} - \rho^{U_2}_{U_3}) \bar\partial \rho^{U_1}_{U_2} = - \bar\partial \rho^{U_1}_{U_2}.
\end{equation}
In the last term we also used the fact that $\partial \rho^{U_0}_{U_1}$ and $\partial \rho^{U_2}_{U_3}$ have disjoint supports with $\bar \partial \rho^{U_1}_{U_2}$. The integral can be evaluated using Lemma \ref{lem: residue integral} to give $2 \pi \ii \, n \, \delta_{m+n,0}$. For the first term on the last line we recall the minus sign in the definition \eqref{up arrow def U V}. Using \eqref{KM algebra vertex modes def a} and the definition of the map \eqref{cal Y map defined}, the resulting expression on the right hand side of \eqref{commutator check 3} is therefore precisely the right hand side of \eqref{commutator check a} in the case at hand.

\paragraph{Virasoro:} Recall from \S\ref{sec: main examples} that in this case the $2$-cocycle \eqref{cocycle Virasoro} explicitly depends on the coordinate used. Since we are working in the local coordinate $\xi_p$ we therefore use the $2$-cocycle $\alpha_{\xi_p}$ relative to this coordinate, i.e. we work in the unital local Lie algebra $\L^\Sigma_{\alpha_{\xi_p}}(U)$.

Here $(\cdot)_{\xi_p} : \fl \SimTo \fl_{\xi_p}$ is given by $\Omega \mapsto -\partial_{\xi_p}$ so that with $\a = \b = \Omega \in \fl$ we find
\begin{align*}
&- s \Big[ \partial_{\xi_p} \otimes \xi_p^n \big( \rho^{U_0}_{U_1} - \rho^{U_2}_{U_3} \big), \partial_{\xi_p} \otimes \xi_p^m \bar\partial \rho^{U_1}_{U_2} \Big]_{\alpha_{\xi_p}}\\
&\;= s\Big( \partial_{\xi_p} \otimes (m-n) \xi_p^{m+n-1} \bar\partial \rho^{U_1}_{U_2} \Big) - \bar\partial \Big( s \big( \partial_{\xi_p} \otimes \xi_p^{m+n} \partial_{\xi_p} \rho^{U_1}_{U_2} \big) \Big)\\
&\qquad\qquad\qquad\qquad\qquad\qquad\qquad\qquad\qquad + \frac{c}{24 \pi \ii} \int_{U_+} \partial( n \xi_p^{n-1}) \wedge m \xi_p^{m-1} \bar\partial \rho^{U_1}_{U_2}\\
&\;= (m-n) s\big( \! - \partial_{\xi_p} \otimes \lceil \xi_p^{m+n-1} \rceil^{U_1}_{U_2} \big) + \frac{m (m-1)(m-2)}{12} c \, \delta_{m+n,2} - \bar\partial \Big( s \big( \partial_{\xi_p} \otimes \xi_p^{m+n} \partial_{\xi_p} \rho^{U_1}_{U_2} \big) \Big),
\end{align*}
where in the first step we used again the identity \eqref{rho support identity}. Here, in the term coming from the Lie bracket \eqref{Lie bracket Virasoro} we note that the term involving $\partial_{\xi_p} (\rho^{U_0}_{U_1} - \rho^{U_2}_{U_3})$ vanishes using the fact that it disjoint support with $\bar \partial \rho^{U_1}_{U_2}$. The evaluation of the term coming from the $2$-cocycle \eqref{cocycle Virasoro} is similar to the above Kac-Moody case but for the terms involving $\partial_{\xi_p} (\rho^{U_0}_{U_1} - \rho^{U_2}_{U_3})$ we used the additional fact that this has disjoint support with $\bar \partial \rho^{U_1}_{U_2}$ and for the term involving $\bar \partial (\partial_{\xi_p} \rho^{U_1}_{U_2})$ we note that this vanishes using Stokes's theorem, cf. the proof of Lemma \ref{lem: residue integral}. In the second step the integral is then evaluated using Lemma \ref{lem: residue integral}. The exact term on the last line drops out in the cohomology class \eqref{commutator check 3} since $\partial_{\xi_p} \rho^{U_1}_{U_2}$ has disjoint support with the remaining terms in the ellipses. Comparing the above with \eqref{Virasoro algebra vertex modes def a} using the definition of the map \eqref{cal Y map defined}, the resulting expression on the right hand side of \eqref{commutator check 3} also agrees with the right hand side of \eqref{commutator check a} in the present case.

\paragraph{$\beta\gamma$ system:} Recall that here the map $(\cdot)_{\xi_p} : \fl \SimTo \fl_{\xi_p}$ sends $\beta \mapsto 1$ and $\gamma \mapsto \d \xi_p$ so that with $\a = \beta$ and $\b = \gamma$ in $\fl$ we find
\begin{align*}
- s \Big[ \d \xi_p \otimes \xi_p^n \big( \rho^{U_0}_{U_1} - \rho^{U_2}_{U_3} \big), 1 \otimes \xi_p^m \bar\partial \rho^{U_1}_{U_2} \Big]_\alpha &=  \frac{1}{2 \pi \ii} \int_{U_+} \xi_p^{m+n} \d \xi_p \wedge \bar\partial \rho^{U_1}_{U_2} = \delta_{m+n,-1}.
\end{align*}
where in the first step we used the $2$-cocycle \eqref{cocycle beta gamma} and again the identity \eqref{rho support identity}.
Comparing this with \eqref{beta gamma algebra vertex modes def a} using the definition of the map \eqref{cal Y map defined}, the resulting expression on the right hand side of \eqref{commutator check 3} again agrees with the right hand side of \eqref{commutator check a}, as required.

\medskip

Finally, in all three cases the first identity in \eqref{commutator check d} similarly follows using the fact that $\lceil \xi^k_p \rceil^V_W$ is $\bar\partial$-exact when $k \in \ZZ_{\geq 0}$.
\end{proof}
\end{proposition}

\begin{proposition} \label{prop: Y map injective}
The linear map \eqref{cal Y map defined} is injective and for every inclusion of open subsets $p \in U \subset V \subset \Sigma^\circ$ in $\Top(\Sigma)$ covered by the coordinate $\xi$, we have the commutative diagram
\begin{equation} \label{m circ Y comm diag}
\begin{tikzcd}[column sep=2mm]
& \hVV^{\fl,\alpha} \arrow[ld, "(\cdot)^{\xi_p}_U"', hook'] \arrow[rd, "(\cdot)^{\xi_p}_V", hook]\\
\U\L^\Sigma_\alpha(U) \arrow[rr, "m_{U, V}"'] & & \U\L^\Sigma_\alpha(V)
\end{tikzcd}
\end{equation}
\begin{proof}
To begin with, we note that the Poincar\'e-Birkhoff-Witt theorem yields an isomorphism $\text{PBW} : \hVV^{\fl, \alpha} \SimTo \Sym(\hg_- \oplus \hbg_-) \vac$ which is given by the inverse of the total symmetrisation map $\Sym(\hg_- \oplus \hbg_-) \vac \SimTo U(\hg_- \oplus \hbg_-) \vac$, or more explicitly by writing a given element in $\hVV^{\fl, \alpha}$ as a linear combination of totally symmetrised monomials acting on $\vac$ which then maps canonically to $\Sym(\hg_- \oplus \hbg_-) \vac$.
Applying the functor \eqref{bSym functor def} to the strong deformation retract in $\udgVec_\CC$ from Theorem \ref{thm: cohomology gSigma}, where here ${\bm\xi} = (\xi, \bar\xi)$, we obtain a strong deformation retract
\begin{equation} \label{SDR Sym}
\begin{tikzcd}
\Sym\Big( \fl \otimes \O_\infty\big( \xi(U_+)^{\rm c} \big) \oplus \fl \otimes \O_\infty\big( \bar \xi(U_-)^{\rm c} \big) \Big) \arrow[r, "I_U", shift left=1mm] & \bSym_\chain \big( \L^\Sigma_\kappa(U) \big) \arrow[l, "P_U", shift left=1mm] \arrow["H_U"', out=-15,in=15,distance=10mm]
\end{tikzcd}
\end{equation}
in $\dgVec_\CC$, where the left hand side is a \dg{} vector space concentrated in degree $0$. Here $P_U$ and $I_U$ are morphisms of $\dgVec_\CC$ whose components in $\Sym$-degree $n \in \ZZ_{\geq 0}$ are given explicitly by $P_U^n \coloneqq p_U[1]^{\otimes n}$ and $I_U^n \coloneqq i_U[1]^{\otimes n}$, respectively. An explicit expression for the homotopy $H_U$ can be found, for instance, in \cite[(4.11b)]{Bruinsma:2023kip}.
In order to perturb the differential on the right hand side of \eqref{SDR Sym} by $\d_{[\cdot, \cdot]}$ we apply the homological perturbation lemma \cite{Crainic:2004bxw, LodayVallette}, noting that this perturbation is small since it lowers by $1$ the $\Sym$-degree which is bounded below by $0$. We obtain the perturbed strong deformation retract
\begin{equation*}
\begin{tikzcd}
\Sym\Big( \fl \otimes \O_\infty\big( \xi(U_+)^{\rm c} \big) \oplus \fl \otimes \O_\infty\big( \bar \xi(U_-)^{\rm c} \big) \Big) \arrow[r, "I_U", shift left=1mm] & \bCE_\chain \big( \L^\Sigma_\kappa(U) \big) \arrow[l, "\widehat{P}_U", shift left=1mm] \arrow["\widehat{H}_U"', out=-15,in=15,distance=10mm]
\end{tikzcd}
\end{equation*}
where $\widehat{P}_U \coloneqq \sum_{j \geq 0} P_U (\d_{[\cdot, \cdot]} H_U)^j$ and the expression for the homotopy $\widehat{H}_U$ will not be needed. Notice that neither $I_U$ nor the differential on the left hand side are perturbed, for the same reason as in the proof of \cite[Proposition 4.4]{Bruinsma:2023kip}. Since $\widehat{P}_U$ is a quasi-isomorphism, applying the $0^{\rm th}$-cohomology functor we obtain an isomorphism
\begin{equation} \label{Ug(U) isomorphism}
H^0 \widehat{P}_U : \U\L^\Sigma_\kappa(U) \overset{\cong}\longrightarrow \Sym\Big( \fl \otimes \O_\infty\big( \xi(U_+)^{\rm c} \big) \oplus \fl \otimes \O_\infty\big( \bar \xi(U_-)^{\rm c} \big) \Big).
\end{equation}
We will show that we have a commutative diagram
\begin{equation} \label{diagram for injective}
\begin{tikzcd}
\hVV^{\fl,\alpha} \arrow[r, "(\cdot)^{\xi_p}_U"] \arrow[d, "\cong", "\text{PBW}"'] & \U\L^\Sigma_\kappa(U) \arrow[d, "\cong"', "H^0 \widehat{P}_U"]\\
\Sym(\hg_- \oplus \hbg_-) \vac \arrow[r, "\iota"', hook] & \Sym\Big( \fl \otimes \O_\infty\big( \xi(U_+)^{\rm c} \big) \oplus \fl \otimes \O_\infty\big( \bar \xi(U_-)^{\rm c} \big) \Big)
\end{tikzcd}
\end{equation}
in $\Vec_\CC$, where the bottom morphism is obtained by applying the $\Sym$-functor to the injection \eqref{level 1 injection}. Since this bottom morphism is injective, see Remark \ref{rem: Sym injections} below, it will then follow that the linear map \eqref{cal Y map defined} is also injective. By the Poincar\'e-Birkhoff-Witt theorem any state in $\hVV^{\fl, \alpha}$ can be written as a linear combination of totally symmetrised monomials acting on $\vac$, so it suffices to consider a single such monomial, say a totally symmetrised version of \eqref{gen state Vkx}, i.e.
\begin{equation} \label{symmetrised state}
\frac{1}{r! \bar r!} \sum_{\sigma \in S_r} \sum_{\bar \sigma \in S_{\bar r}} \a^{\sigma(r)}_{(-m_{\sigma(r)})} \ldots \a^{\sigma(1)}_{(-m_{\sigma(1)})} \bar \b^{\bar\sigma(\bar r)}_{(-n_{\bar\sigma(\bar r)})} \ldots \bar \b^{\bar\sigma(1)}_{(-n_{\bar\sigma(1)})} \vac \in \hVV^{\fl, \alpha}.
\end{equation}
We can make a further important simplification, similar to \cite[\S4.2]{Bruinsma:2023kip}. Since we are focusing on totally symmetrised monomials, it is sufficient to consider ones built out of a single pair of elements in $\hg_-$ and $ \hbg_-$. Namely, we may focus on states in $\hVV^{\fl, \alpha}$ of the form
\begin{equation} \label{special state}
A \coloneqq \mathfrak{A}^r \mathfrak{B}^{\bar r} \vac, \qquad
\mathfrak{A} = \sum_{i=1}^r t_i \a^i_{(-m_i)} \in \hg_-, \quad
\mathfrak{B} = \sum_{j=1}^{\bar r} \bar t_j \bar \b^j_{(-n_j)} \in \hbg_-,
\end{equation}
where $t_i, \bar t_j \in \CC$, $\a^i, \b^j \in \fl$ and $m_i, n_j \in \ZZ_{\geq 1}$ for each $i \in \{1, \ldots, r\}$ and $j \in \{ 1, \ldots, \bar r\}$ with $r, \bar r \in \ZZ_{\geq 0}$. 
Indeed, the more general state \eqref{symmetrised state} is recovered by polarization, i.e. it reads
\begin{equation*}
\frac{1}{r! \bar r!} \frac{\partial^r}{\partial t_1 \cdots \partial t_r} \bigg|_{t_1, \ldots, t_r = 0} \frac{\partial^{\bar r}}{\partial \bar t_1 \cdots \partial \bar t_{\bar r}} \bigg|_{\bar t_1, \ldots, \bar t_{\bar r} = 0} \bigg( \sum_{i=1}^r t_i \a^i_{(-m_i)} \bigg)^r \bigg( \sum_{j=1}^{\bar r} \bar t_j \bar \b^j_{(-n_j)} \bigg)^{\bar r} \vac.
\end{equation*}
Applying the isomorphism on the left hand side of \eqref{diagram for injective} to the special state \eqref{special state} then gives $\text{PBW}(A) = \mathfrak{A}^r \mathfrak{B}^{\bar r} \vac \in \Sym(\hg_- \oplus \hbg_-) \vac$. Subsequently applying the injection at the bottom of the diagram \eqref{diagram for injective} gives an element in the symmetric algebra on the bottom right corner, given by realising $\mathfrak A$ and $\mathfrak B$ in terms of germs of holomorphic functions on $\xi(U_+)^c$ and $\bar \xi(U_-)^c$, respectively, i.e. in terms of the notation after \eqref{level 1 injection} we have
\begin{equation} \label{bottom left of diagram}
\iota \big( \text{PBW}(A) \big) = \bigg( \sum_{i=1}^r t_i \a^i \otimes \lambda^{-m_i} \bigg)^r \bigg( \sum_{j=1}^{\bar r} \bar t_j \bar \b^j \otimes \mu^{-n_j} \bigg)^{\bar r}.
\end{equation}
Going instead along the top right of the diagram \eqref{diagram for injective}, applying the first map \eqref{cal Y map defined} we find
\begin{equation*}
A^{\xi_p}_U = \Bigg[ s\bigg( \sum_{i=1}^r t_i \a^i_{\xi_p} \otimes \lceil \xi_p^{-m_i} \rceil^{U_{i-1}}_{U_i} \bigg)^r s\bigg( \sum_{i=1}^{\bar r} \bar t_j \b^j_{\bar\xi_p} \otimes \lceil \bar \xi_p^{-n_j} \rceil^{V_{j-1}}_{V_j} \bigg)^{\bar r} \Bigg]_U,
\end{equation*}
Finally, one checks that upon applying the isomorphism $H^0 \widehat{P}_U$ to the above, we obtain exactly the same result as in \eqref{bottom left of diagram}. For this one needs to consider the explicit form of $H_U$ given in \cite[(4.11b)]{Bruinsma:2023kip}. One can make a further simplification by choosing the smooth bump functions $\rho_f \in \Omega^{0,0}_c(U_+)$ in \eqref{up arrow def}, entering the definition of the strong deformation retract in Theorem \ref{thm: cohomology gSigma}, for every $f \in \O_\infty(\CP \setminus \{0\})$ to coincide and to be equal to $1$ on $U_r$, and likewise for the smooth bump functions $\rho_g \in \Omega^{0,0}_c(U_-)$. We then observe that only the $j=0$ term in the sum defining $\widehat{P}_U$ contributes, as a consequence of the very special form of the state $A$ in \eqref{special state}.

The commutativity of the diagram \eqref{m circ Y comm diag} is immediate from the explicit description of the morphisms $(\cdot)^{\xi_p}_U$, $(\cdot)^{\xi_p}_V$ and of the factorisation products $m_{U, V}$ described in \S\ref{sec: twisted PFE}.
\end{proof}
\end{proposition}

\begin{remark} \label{rem: Sym injections}
In general, the symmetric algebra functor $\Sym : \mathsf{Mod}_R \to \mathsf{g}^{\geq 0}\mathsf{CAlg}_R$ from the category $\mathsf{Mod}_R$ of modules over a commutative ring $R$ to the category $\mathsf{g}^{\geq 0}\mathsf{CAlg}_R$ of positively graded commutative algebras over $R$ is right exact and hence preserves surjections. However, it is not left exact and so it does \emph{not}, in general, preserve injections. In the present setting we are working over $R = \CC$ and in this case the $\Sym$ functor \emph{does} preserve injections. Indeed,  every injective linear map $i : V \hookrightarrow W$ between complex vector spaces has a left inverse, i.e. a linear map $r : W \to V$ such that $r i = \id_V$ and so applying the functor $\Sym$ we deduce that $\Sym(i) : \Sym V \to \Sym W$ is also injective.
\end{remark}

\subsubsection{Local copies of \texorpdfstring{$\hVV^{\fl,\alpha}$}{F}}

Let $\Top(\Sigma)_p$ be the category of neighbourhoods $U \subset \Sigma^\circ$ of $p \in \Sigma^\circ$, 
with inclusions $U \subseteq V$ as morphisms $U \to V$.
It is a subcategory of the multicategory $\Top(\Sigma)^{\sqcup}$ from \S\ref{sec: twisted PFE}, so that we may restrict the multifunctor $\U\L^\Sigma_\alpha : \Top(\Sigma)^{\sqcup} \to \Vec_\CC^{\otimes}$ to this subcategory to form a diagram $(\U\L^\Sigma_\alpha)_p : \Top(\Sigma)_p \to \Vec_\CC$ of shape $\Top(\Sigma)_p$.
Since the category $\Vec_\CC$ is complete and $\Top(\Sigma)_p$ is small, we can form the limit $\lim \, (\U\L^\Sigma_\alpha)_p$ in $\Vec_\CC$, which for any open subset $U$ in $\Top(\Sigma)_p$ comes with a canonical linear map $m_{p, U} : \lim \, (\U\L^\Sigma_\alpha)_p \to \U\L^\Sigma_\alpha(U)$. We also have such a map for any open subset $U \subset \Sigma$ in $\Top(\Sigma)$ containing $p$ by post-composing the latter with a factorisation product.

By the universal property of limits there is a unique injective linear map
\begin{equation} \label{Y at p}
(\cdot)^\xi_p : \hVV^{\fl,\alpha} \longhookrightarrow \lim \, (\U\L^\Sigma_\alpha)_p, \qquad
A \longmapsto A^\xi_p
\end{equation}
for $p \in \Sigma^\circ$ such that the following diagram
\begin{equation*}
\begin{tikzcd}[column sep=2mm]
& \hVV^{\fl,\alpha} \arrow[ldd, "(\cdot)^{\xi_p}_U"', hook', bend right=35] \arrow[rdd, "(\cdot)^{\xi_p}_V", hook, bend left=35] \arrow[d, "\exists! (\cdot)^\xi_p", hook, dashed] &\\
& \lim \, (\U\L^\Sigma_\alpha)_p \arrow[ld, "m_{p,U}"']
\arrow[rd, "m_{p, V}"] &\\
\U\L^\Sigma_\alpha(U) \arrow[rr, "m_{U, V}"'] & & \U\L^\Sigma_\alpha(V)
\end{tikzcd}
\end{equation*}
is commutative, for every inclusion of open neighbourhoods $p \in U \subset V \subset \Sigma^\circ$.

\begin{proposition} \label{prop: indep loc coord}
The image of the linear map \eqref{Y at p} does not depend on the local coordinate $\xi$ used on a neighbourhood of $p \in \Sigma^\circ$. We denote this image as $\hV^{\fl,\alpha}_p \in \Vec_\CC$.
\begin{proof}
Let $\xi$ and $\eta$ be two local coordinates on a neighbourhood $U$ of $p$. It suffices to show that for any $A \in \hVV^{\fl,\alpha}$, the vector $A^{\xi_p}_U$ also lies in the image of the linear map \eqref{Y at p} defined relative to the local coordinate $\eta$. And by linearity it is enough to consider states of the form \eqref{gen state Vkx}. Its image $A^{\xi_p}_U$ is given by \eqref{cal Y map defined b}. Writing $\xi = \varrho^{\eta \to \xi} \circ \eta$ for some holomorphic function $\varrho^{\eta \to \xi} : \eta(U) \to \xi(U)$, so that also $\bar \xi = \bar \varrho^{\eta \to \xi} \circ \bar \eta$, we have expansions (see \eqref{hat varrho def} later for the exact expressions)
\begin{align*}
\xi_p^{-m_i} = \sum_{k \geq 0} \alpha_{k, m_i} \, \eta_p^{-m_i + k}, \qquad
\bar\xi_p^{-n_j} = \sum_{\ell \geq 0} \overline{\alpha_{\ell, n_j}} \, \bar \eta_p^{-n_j + \ell}
\end{align*}
for all $i \in \{1, \ldots, r \}$, $j \in \{1, \ldots, \bar r \}$. Here the coefficients $\alpha_{k,n}$ are complex numbers for every $k \in \ZZ_{\geq 0}$, $n \in \ZZ_{\geq 1}$ whose exact expressions will not be needed for the argument below, but for instance we have $\alpha_{0, n} = (\varrho^{\eta \to \xi})'(\eta(p))^{-n}$.
Substituting these expansions into \eqref{cal Y map defined b} we obtain the formal expression
\begin{align} \label{A rewrite as B}
A^{\xi_p}_U &= \bigg( \sum_{k_r \geq 0} \alpha_{k_r, m_r} \, \a^r_{(-m_r+k_r)} \bigg) \ldots \bigg( \sum_{k_1 \geq 0} \alpha_{k_1, m_1} \, \a^1_{(-m_1+k_1)} \bigg) \notag\\
&\qquad\qquad \times \bigg( \sum_{\ell_{\bar r} \geq 0} \overline{\alpha_{\ell_{\bar r}, n_{\bar r}}} \, \bar \b^{\bar r}_{(-n_{\bar r}+\ell_{\bar r})} \bigg) \ldots \bigg( \sum_{\ell_1 \geq 0} \overline{\alpha_{\ell_1, n_1}} \, \bar \b^1_{(-n_1 + \ell_1)} \bigg) \vac^{\eta_p}_U,
\end{align}
It remains to observe that the right hand side is of the form $B^{\eta_p}_U$ for some well defined state $B \in \hVV^{\fl,\alpha}$ since only finitely many terms contribute from each of the infinite sums. Indeed, this follows from repeatedly applying the identities \eqref{commutator check} from Proposition \ref{prop: Y map def}.
\end{proof}
\end{proposition}

For any point $p \in \Sigma^\circ$, we refer to $\hV^{\fl,\alpha}_p$ as the `\emph{local copy of $\hVV^{\fl,\alpha}$ attached to $p$}', and for any $A \in \hVV^{\fl,\alpha}$ we will refer to $A^\xi_p \in \hV^{\fl,\alpha}_p$ as the `\emph{state $A$ prepared at $p$ in the local coordinate $\xi$}'. Explicitly, the state \eqref{gen state Vkx} in $\hVV^{\fl,\alpha}$ prepared at $p$ in the local coordinate $\xi$ will be denoted
\begin{equation*}
\a^r_{(-m_r)} \ldots \a^1_{(-m_1)} \bar \b^{\bar r}_{(-n_{\bar r})} \ldots \bar \b^1_{(-n_1)} \vac^\xi_p.
\end{equation*}
For any neighbourhood $U \subset \Sigma^\circ$ of $p$ with local coordinate $\xi : U \to \CC$, we refer to $A^{\xi_p}_U \in \U\L^\Sigma_\alpha(U)$ as the `\emph{state $A$ prepared at $p$ in the local chart $(U, \xi)$}'.

Let $\xi$ and $\eta$ be two local holomorphic coordinates in the neighbourhood of a point $p \in \Sigma^\circ$.
Given a state $A \in \hVV^{\fl,\alpha}$, it can be prepared at $p$ in either the local coordinate $\xi$ to give the element $A^\xi_p \in \hV^{\fl,\alpha}_p$, or in the local coordinate $\eta$ to give the element $A^\eta_p \in \hV^{\fl,\alpha}_p$. These will, in general, describe different elements of $\hV^{\fl,\alpha}_p$ so we obtain a linear isomorphism
\begin{equation} \label{change of preparation}
(\cdot)^{\eta \to \xi}_p \coloneqq (\cdot)_p^\xi \circ \big( (\cdot)_p^\eta \big)^{-1} : \hV^{\fl,\alpha}_p \longrightarrow \hV^{\fl,\alpha}_p, \qquad A^\eta_p \longmapsto A^\xi_p
\end{equation}
which takes as input any element of $\hV^{\fl,\alpha}_p = \im \, (\cdot)_p^\eta$, identifies the state $A \in \hVV^{\fl,\alpha}$ corresponding to it in the local coordinate $\eta$ near $p$, which is unique by the injectivity of \eqref{Y at p}, and then prepares the same state $A$ at the same point $p$ but in the different local coordinate $\xi$.
We note here that the isomorphism \eqref{change of preparation} plays a similar role for the local copy $\hV^{\fl,\alpha}_p$ of $\hVV^{\fl,\alpha}$ as the map $(\cdot)_{\eta \to \xi} : \fl_\eta \SimTo \fl_\xi$ defined in \S\ref{sec: general setup} does for the holomorphic vector bundle $L$.

\subsubsection{Translation operators}

In the vertex algebra setting it is customary to introduce the \emph{infinitesimal translation operator} $D$ acting as an endomorphism of the vertex algebra. In the present non-chiral setting we have two such endomorphisms
\begin{subequations} \label{D bar D def}
\begin{equation}
D, \bar D : \hVV^{\fl,\alpha} \longrightarrow \hVV^{\fl,\alpha}
\end{equation}
defined as follows. For any $A \in \hVV^{\fl,\alpha}$ of the form \eqref{gen state Vkx} we set
\begin{align}
 \label{D bar D def a}
D A &\coloneqq \sum_{i = 1}^r m_i \, \a^r_{(-m_r)} \ldots \a^i_{(-m_i-1)} \ldots \a^1_{(-m_1)} \bar\b^{\bar r}_{(-n_{\bar r})} \ldots \bar\b^1_{(-n_1)} \vac,\\
 \label{D bar D def b}
\bar D A &\coloneqq \sum_{j = 1}^{\bar r} n_j \, \a^r_{(-m_r)} \ldots \a^1_{(-m_1)} \bar\b^{\bar r}_{(-n_{\bar r})} \ldots \bar\b^j_{(-n_j-1)} \ldots \bar\b^1_{(-n_1)} \vac.
\end{align}
\end{subequations}
In particular, $D \vac = \bar D \vac = 0$. We call a state $A \in \hVV^{\fl,\alpha}$ \emph{chiral} if $\bar D A = 0$ and \emph{anti-chiral} if $D A = 0$. We let $\VV^{\fl,\alpha}$ denote the subspace of chiral states since it is isomorphic to the affine vertex algebra and, correspondingly, we let $\bar{\VV}^{\fl,\alpha}$ denote the subspace of anti-chiral states.

\begin{proposition} \label{prop: Y deriv}
For any $A \in \hVV^{\fl,\alpha}$ and any local coordinate $\xi$ in the neighbourhood of a point $p \in \Sigma^\circ$, we have the relations
\begin{equation*}
(D A)^\xi_p = \partial_{\xi(p)} A^\xi_p, \qquad
(\bar D A)^\xi_p = \partial_{\bar \xi(p)} A^\xi_p
\end{equation*}
in $\hV^{\fl,\alpha}_p$. In particular, if $A$ is (anti-)chiral then $A^\xi_p$ depends (anti-)holomorphically on $p$.
\begin{proof}
It is enough to prove the relations
\begin{equation*}
(D A)^{\xi_p}_U = \partial_{\xi(p)} A^{\xi_p}_U, \qquad
(\bar D A)^{\xi_p}_U = \partial_{\bar \xi(p)} A^{\xi_p}_U
\end{equation*}
in $\U\L^\Sigma_\alpha(U)$ for an open neighbourhood $U \subset \Sigma^\circ$ of $p$.
This is immediate from the definition of the operators $D, \bar D$ in \eqref{D bar D def} and of the map $(\cdot)^{\xi_p}_U$ in \eqref{cal Y map defined}, noting that for any $m \in \ZZ$ we have $\partial_{\xi(p)} \xi_p^{-m} = m \, \xi_p^{-m-1}$ and $\partial_{\bar \xi(p)} \bar \xi_p^{-m} = m \, \bar \xi_p^{-m-1}$.
\end{proof}
\end{proposition}

\subsection{Vertex modes and vertex operators} \label{sec: VA products}

For any $R \in \RR_{>0}$ we use the notation $D_R \coloneqq \{ \zeta \in \CC \,|\, |\zeta| < R \} \subset \CC$ for the open disc of radius $R$ around the origin in $\CC$. Similarly, for any $R^-, R^+ \in \RR_{> 0}$ with $R^- < R^+$ we denote by $A_{R^-,R^+} \coloneqq \{ \zeta \in \CC \,|\, R^- < |\zeta| < R^+ \} \subset \CC$ the open annulus with radii $R^-$ and $R^+$ around the origin. Note that $A_{R^-, R^+} = D_{R^+} \setminus \overline{D_{R^-}}$.

Given a local coordinate $\xi : U \to \CC$ on a neighbourhood $U \subset \Sigma^\circ$ of $p \in \Sigma^\circ$, we say that $V \subset U$ is a \emph{disc shaped open subset} around $p$ in the local coordinate $\xi : U \to \CC$ if $V = \xi_p^{-1}(D_R)$ for some $R > 0$. Likewise, we say that $Y \subset U$ is an \emph{annulus shaped open subset} around $p$ in the local coordinate $\xi : U \to \CC$ if $Y = \xi_p^{-1}(A_{R^-, R^+})$ for some $0 < R^- < R^+$. We can then define disc shaped open subsets $V^\pm \coloneqq \xi_p^{-1}(D_{R^\pm})$ so that $Y = V^+ \setminus \overline{V^-}$.

By definition of the category $\Top(\Sigma)$, see the start of \S\ref{sec: twisted PFE}, disc shaped open subsets and annulus shaped open subsets around a given point $p$ are both objects in $\Top(\Sigma)$ homeomorphic to $\CC$ and $\CC \setminus \{0\}$, respectively.

\subsubsection{Vertex modes}

To every $\a \in \fl$, $n, \bar n \in \ZZ$ and any local coordinate $\xi$ in the neighbourhood of $p \in \Sigma^\circ$, we associate a pair of endomorphisms of the vector space $\hV^{\fl,\alpha}_p$, called the \emph{chiral vertex $n^{\rm th}$-mode} and \emph{anti-chiral vertex $\bar n^{\rm th}$-mode} of $\a$ at $p$, defined as follows. Given an annulus shaped open subset $Y \subset \Sigma^\circ$ around $p$ in the local coordinate $\xi$, we let $V^\pm \subset \Sigma^\circ$ denote the disc shaped open subsets around $p$ such that $Y = V^+ \setminus \overline{V^-}$. The preimages of $V^\pm$ under the projection $\pi : \Dz \to \Sigma$ is a disjoint union of disc shaped open subsets around $p_+$ and $p_-$, respectively, which we denote as $\pi^{-1}(V^\pm) = V^\pm_+ \sqcup V^\pm_-$, as in \S\ref{sec: Vec structure}. We define the elements
\begin{equation} \label{X bar X vertex modes}
\a^{\; \xi_p}_{(n)} \coloneqq \Big[ s\Big( \a_{\xi_p} \otimes \lceil \xi_p^n \rceil^{V^\tm_\tp}_{V^\tp_\tp}\Big) \Big]_Y, \qquad
\bar \a^{\; \xi_p}_{(\bar n)} \coloneqq \Big[ s\Big( \a_{\bar\xi_p} \otimes \lceil {\bar \xi}_p^{\bar n} \rceil^{V^\tm_\tm}_{V^\tp_\tm}\Big) \Big]_Y
\end{equation}
in $\U\L^\Sigma_\alpha(Y)$. We do not include the subset $Y$ in the notation for the elements \eqref{X bar X vertex modes} of $\U\L^\Sigma_\alpha(Y)$
since the endomorphisms of $\hV^{\fl,\alpha}_p$ they define will be independent of $Y$.
Explicitly, given any state $B \in \hVV^{\fl,\alpha}$ and any neighbourhood $U \subset \Sigma^\circ$ of $p$ we may form the factorisation products
\begin{equation*}
m_{(Y, p), U} \big( \a^{\; \xi_p}_{(n)} \otimes B^\xi_p \big) = ( \a_{(n)} B )^{\xi_p}_U, \qquad
m_{(Y, p), U} \big( \bar \a^{\; \xi_p}_{(\bar n)} \otimes B^\xi_p \big) = ( \bar \a_{(\bar n)} B )^{\xi_p}_U
\end{equation*}
in $\U\L^\Sigma_\alpha(U)$ for every $\a \in \fl$, where $Y \subset U$ is any annulus shaped open subset around $p$ and the map $m_{(Y,p), U} : \U\L^\Sigma_\alpha(Y) \otimes \lim (\U\L^\Sigma_\alpha)_p \to \U\L^\Sigma_\alpha(U)$ is defined as $m_{(Y,p), U} \coloneqq m_{(Y, V), U} \circ (\id \otimes m_{p, V})$ for any open neighbourhood $V \ni p$ disjoint from $Y$. The above equalities both follow from the definitions \eqref{cal Y map defined} and \eqref{X bar X vertex modes}, see also Proposition \ref{prop: Y map def}.
In particular, we obtain linear maps
\begin{subequations} \label{endo Xn bar Xn}
\begin{align}
\a^{\; \xi_p}_{(n)} : \hV^{\fl,\alpha}_p &\longrightarrow \hV^{\fl,\alpha}_p, \qquad
B^\xi_p \longmapsto \a^{\; \xi_p}_{(n)} B^\xi_p = (\a_{(n)} B)^\xi_p,\\
\bar \a^{\; \xi_p}_{(\bar n)} : \hV^{\fl,\alpha}_p &\longrightarrow \hV^{\fl,\alpha}_p, \qquad
B^\xi_p \longmapsto \bar \a^{\; \xi_p}_{(\bar n)} B^\xi_p = (\bar \a_{(\bar n)} B)^\xi_p
\end{align}
\end{subequations}
for any $\a \in \fl$, $n, \bar n \in \ZZ$ and any local coordinate $\xi$ in a neighbourhood of $p$.

\medskip

The infinite sequences of chiral vertex $n^{\rm th}$-modes in \eqref{X bar X vertex modes} can be described collectively, for $n \geq 0$ and $n < 0$ respectively, as the coefficients in the expansion of single elements of $\U\L^\Sigma_\alpha(Y)$ as follows. The same will be true of the anti-chiral vertex $\bar n^{\rm th}$-modes, see below, but focusing on the chiral vertex $n^{\rm th}$-modes $\a^{\; \xi_p}_{(n)}$ of $\a \in \fl$ we consider
\begin{equation} \label{generator Xn modes}
\Big[ s\Big( \a_{\xi_q} \otimes \lceil \xi_q^{-1} \rceil^{V^\tm_\tp}_{V^\tp_\tp} \Big) \Big]_Y \in \U\L^\Sigma_\alpha(Y)
\end{equation}
for some other point $q \in U$ in the local coordinate patch covered by $\xi : U \to \CC$, where the notation used is the same as in \eqref{X bar X vertex modes}. Since $\xi_q^{-1}$ is singular at $q$, the latter should not lie in the closure of the annulus shaped subset $Y = V^+ \setminus \overline{V^-}$, so there are two cases to consider:
\begin{itemize}[leftmargin=12mm]
  \item[$(\rm v_{<0})$] If $q \in V^-$ then $|\xi_p(t)| > |\xi_p(q)|$ for any $t \in Y$ so that the function $\xi_q^{-1} = (\xi_p - \xi_p(q))^{-1}$ can be expanded in small $\xi_p(q)$ and hence \eqref{generator Xn modes} has the following expansion
\begin{equation*}
\phantom{-} \Big[ s\Big( \a_{\xi_q} \otimes \lceil \xi_q^{-1} \rceil^{V^\tm_\tp}_{V^\tp_\tp}\Big) \Big]_Y = \sum_{n < 0} \a^{\; \xi_p}_{(n)} \xi_p(q)^{-n-1} \eqqcolon \a[\xi_p(q)]_+.
\qquad\quad
\raisebox{-6mm}{\begin{tikzpicture}
\def\R{1}
  \fill[blue!90,
        opacity      = 0.3]
        (.05*\R,1.2*\R) circle (.6*\R);
  \draw[thick, blue] (.05*\R,1.2*\R) circle (.6*\R);
  \fill[red!90,
        opacity      = 0.3]
        (.05*\R,1.2*\R) circle (.5*\R);
  \draw[thick, red] (.05*\R,1.2*\R) circle (.5*\R);
  \filldraw[thick] (.05*\R,1.2*\R) node[below=-.7mm]{\tiny $p$} circle (0.02*\R) node[above right=.5mm and -1.5mm, red]{\tiny $V^\tm$} node[above right=4.2mm and 1.7mm, blue]{\tiny $V^\tp$};
  \filldraw[thick] (-.3*\R,1.24*\R) node[below=-.7mm]{\tiny $q$} circle (0.02*\R);
\end{tikzpicture}
}
\end{equation*}
   \item[$(\rm v_{\geq 0})$] If $q \not \in \overline{V^+}$ then instead $|\xi_p(t)| < |\xi_p(q)|$ for any $t \in Y$ so that $\xi_q^{-1} = (\xi_p - \xi_p(q))^{-1}$ can be expanded in small $\xi_p$, and hence \eqref{generator Xn modes} now has the following expansion
\begin{equation*}
- \Big[ s\Big( \a_{\xi_q} \otimes \lceil \xi_q^{-1} \rceil^{V^\tm_\tp}_{V^\tp_\tp}\Big) \Big]_Y = \sum_{n \geq 0} \a^{\; \xi_p}_{(n)} \xi_p(q)^{-n-1} \eqqcolon \a[\xi_p(q)]_-.
\qquad
\raisebox{-6mm}{\begin{tikzpicture}
\def\R{1}
  \fill[blue!90,
        opacity      = 0.3]
        (.05*\R,1.2*\R) circle (.6*\R);
  \draw[thick, blue] (.05*\R,1.2*\R) circle (.6*\R);
  \fill[red!90,
        opacity      = 0.3]
        (.05*\R,1.2*\R) circle (.5*\R);
  \draw[thick, red] (.05*\R,1.2*\R) circle (.5*\R);
  \filldraw[thick] (.05*\R,1.2*\R) node[below=-.7mm]{\tiny $p$} circle (0.02*\R) node[above right=.5mm and -1.5mm, red]{\tiny $V^\tm$} node[above right=4.2mm and 1.7mm, blue]{\tiny $V^\tp$};
  \filldraw[thick] (-.7*\R,1.24*\R) node[below=-.7mm]{\tiny $q$} circle (0.02*\R);
\end{tikzpicture}
}
\end{equation*}
\end{itemize}

\begin{remark} \label{rem: a xi q no expand}
When $\a_{\xi_q}$ explicitly depends on the coordinate $\xi_q$, as in the Virasoro or $\beta\gamma$ system cases, we should also re-express this in terms of the local coordinate $\xi_p$ near $p$ in both of the above expansions in $(\rm v_{<0})$ and $(\rm v_{\geq 0})$. However, since the local coordinates $\xi_q$ and $\xi_p$ are related simply by a constant shift $\xi_q = \xi_p - \xi_p(q)$, in all three examples we consider it is the case that $\a_{\xi_q} = \a_{\xi_p}$ for every $\a \in \fl$. Indeed, in the Virasoro case we have $\Omega_{\xi_q} = - \partial_{\xi_q} = - \partial_{\xi_p} = \Omega_{\xi_p}$ and likewise in the $\beta\gamma$ system case we have $\gamma_{\xi_q} = \d \xi_q = \d \xi_p = \gamma_{\xi_p}$.
\end{remark}

It is convenient to gather the negative modes $\a^{\; \xi_p}_{(n)}$ for $n < 0$ and positive modes $\a^{\; \xi_p}_{(n)}$ for $n \geq 0$ together into a chiral vertex operator, the single doubly infinite series
\begin{equation} \label{X as X+ X-}
\a[\xi_p(q)] \coloneqq \a[\xi_p(q)]_+ + \a[\xi_p(q)]_-,
\end{equation}
with coefficients in $\End \hV^{\fl,\alpha}_p$. Strictly speaking, the nested sequence of open neighbourhoods $p_+ \in V^-_+ \Subset V^+_+ \subset U_+$ used in defining the modes $\a^{\; \xi_p}_{(n)}$ should be different here depending on whether $n < 0$ or $n \geq 0$ so that the point $q \in U$, which is fixed in \eqref{X as X+ X-}, satisfies $q \in V^-$ in the first term and $q \not \in \overline{V^+}$ in the second.
See the proof of Proposition \ref{prop: vertex operator general} for details.

We can similarly assemble all the anti-chiral elements $\bar \a^{\; \xi_p}_{(\bar n)} \in \U\L^\Sigma_\alpha(Y)$ in \eqref{X bar X vertex modes} into series $\bar\a[\bar\xi_p(q)]_\pm$ for $\bar n < 0$ and $\bar n \geq 0$, respectively, and formally set $\bar \a[\bar\xi_p(q)] \coloneqq \bar\a[\bar\xi_p(q)]_+ + \bar\a[\bar\xi_p(q)]_-$.

\subsubsection{Homogeneous vertex modes} \label{sec: hom vertex modes}

Since the vertex modes \eqref{X bar X vertex modes} living in $\U\L^\Sigma_\alpha(Y)$ are supported on an annulus shaped open subset $Y \subset \Sigma^\circ$ around $p$ in the local coordinate $\xi$, i.e. the shifted local coordinate $\xi_p$ provides a bijection $\xi_p : Y \SimTo A_{R^-, R^+}$ to an open annulus with radii $0 < R^- < R^+$, we have an action of the circle $S^1 = \{ \lambda \in \CC \,|\, |\lambda| = 1 \}$ on $Y$ induced by the obvious $S^1$-action on $A_{R^-, R^+} \subset \CC$ given by multiplication by $\lambda \in S^1$. In other words, in the local coordinate $\xi_p : Y \to \CC$ this action is given by $\xi_p \mapsto \lambda \xi_p$. We will say that $\a \in \fl$ has \emph{conformal dimension} $\Delta_\a$ if the behaviour of the associated chiral vertex $n^{\rm th}$-modes in \eqref{X bar X vertex modes} under this $S^1$-action on $Y$ takes the form
\begin{equation*}
\a^{\; \lambda \xi_p}_{(n)} = \lambda^{n-\Delta_\a+1} \a^{\; \xi_p}_{(n)}
\end{equation*}
for every $n \in \ZZ$. Note that since $|\lambda|=1$ we have $\bar\xi_p \mapsto \lambda^{-1} \bar\xi_p$ so that the anti-chiral vertex $n^{\rm th}$-modes associated with $\a \in \fl$ transform for any $n \in \ZZ$ as
\begin{equation*}
\bar\a^{\; \lambda \xi_p}_{(n)} = \lambda^{-n+\Delta_\b-1} \bar\a^{\; \xi_p}_{(n)}.
\end{equation*}

For any $\a \in \a$ of definite conformal dimension $\Delta_\a$, it is convenient to define its associated \emph{homogeneous chiral vertex $n^{\rm th}$-mode} and \emph{homogeneous anti-chiral vertex $n^{\rm th}$-mode} at $p$ in the local coordinate $\xi$, respectively, as
\begin{equation} \label{hom modes a def}
\a^{\xi_p}_{[n]} \coloneqq \a^{\xi_p}_{(n+\Delta_\a - 1)}, \qquad
\bar\a^{\xi_p}_{[n]} \coloneqq \bar\a^{\xi_p}_{(n+\Delta_\b - 1)}.
\end{equation}
The reason that these (anti-)chiral vertex $n^{\rm th}$-modes are called `homogeneous' comes from the fact that their behaviour under the $S^1$-action is just a rescaling by $\lambda^{\pm n}$, which is independent of $\a \in \fl$.
In terms of homogeneous chiral vertex $n^{\rm th}$-modes, the chiral vertex operator \eqref{X as X+ X-} and its anti-chiral counterpart take the form
\begin{equation} \label{vertex operator hom rewrite}
\a[\xi_p(q)] = \sum_{n \in \ZZ} \a^{\; \xi_p}_{[n]} \xi_p(q)^{-n-\Delta_\a}, \qquad
\bar\a[\bar\xi_p(q)] = \sum_{n \in \ZZ} \bar\a^{\; \xi_p}_{[n]} \bar\xi_p(q)^{-n-\Delta_\b}.
\end{equation}
We also define the homogeneous loop generators of the Lie algebra $\hg \oplus \hbg$ as, cf. \eqref{loop generators},
\begin{equation*}
\a_{[n]} \coloneqq \a \otimes t^{n+\Delta_\a - 1} \in \hg, \qquad \bar\a_{[n]} \coloneqq \a \otimes \bar t^{n+\Delta_\a - 1} \in \hbg.
\end{equation*}

Given any monomial state $A \in \hVV^{\fl,\alpha}$ as in \eqref{gen state Vkx} such that $\a^i$ and $\b^j$ have definite conformal dimensions $\Delta_{\a^i}$ and $\Delta_{\b^j}$ for each $i=1,\ldots, r$ and $j = 1, \ldots, \bar r$, we will say that $A$ has \emph{chiral conformal dimension} $\Delta_A \coloneqq \sum_{i=1}^r (\Delta_{\a^i} + m_i - 1)$ and that it has \emph{anti-chiral conformal dimension} $\bar \Delta_A \coloneqq \sum_{j=1}^{\bar r} (\Delta_{\b^j} + n_j - 1)$. In all cases we shall consider, every possible rewriting of the state $A$ is given by a linear combination of monomial states with the same chiral conformal dimension $\Delta_A$ and anti-chiral conformal dimension $\bar\Delta_A$ so we have a well defined $\ZZ^2$-grading on $\hVV^{\fl,\alpha}$.
Rewriting the relations \eqref{hom modes a def} as
\begin{equation*}
\a^{\xi_p}_{(-n)} = \a^{\xi_p}_{[-n-\Delta_\a+1]}, \qquad \bar\a^{\xi_p}_{(-n)} = \bar\a^{\xi_p}_{[-n-\Delta_\a+1]}
\end{equation*}
we see that the $\ZZ^2$-grading on $\hVV^{\fl,\alpha}$ is measured by the negative of the sums of the chiral and anti-chiral mode numbers of a homogeneous state $A$, respectively, when written in terms of homogeneous chiral and anti-chiral vertex modes.

We say that $A^\xi_p \in \hV^{\fl,\alpha}_p$ is \emph{homogeneous in the coordinate $\xi$} and call $\wgt (A^\xi_p) \coloneqq \Delta_A$ the \emph{chiral weight} of $A^\xi_p$ in the coordinate $\xi$ and $\wgtb(A^\xi_p) \coloneqq \bar\Delta_A$ its \emph{anti-chiral weight} in the coordinate $\xi$. Every choice of local coordinate $\xi$ in a neighbourhood of the point $p \in \Sigma^\circ$ induces a $\ZZ^2$-grading on the vector space $\hV^{\fl,\alpha}_p$. These $\ZZ^2$-gradings of $\hV^{\fl,\alpha}_p$ depend on the coordinate used because if $A^\xi_p \in \hV^{\fl,\alpha}_p$ is homogeneous in the coordinate $\xi$ then it will generally not be homogeneous in another local coordinate $\eta$ near $p$. More precisely, if $A^\xi_p = B^\eta_p$ for some $B \in \hVV^{\fl,\alpha}$ as in the proof of Proposition \ref{prop: indep loc coord}, then $B$ will generally not be of definite bi-grade.

The conformal dimensions of basis elements in $\fl$ and their corresponding homogeneous vertex operators in the three main examples from \S\ref{sec: main examples} are as follows.

\paragraph{Kac-Moody:} Recall that $\fl = \g$. Since $\X_{\xi_p} = \X$ is independent of the local coordinate for any $\X \in \g$, it follows that $\Delta_\X = 1$. In particular, homogeneous vertex $n^{\rm th}$-modes coincide with the vertex $n^{\rm th}$-modes, namely $\X_{[n]} = \X_{(n)}$ for all $n \in \ZZ$ and the Lie algebra relations of $\hg \oplus \hbg$ in \eqref{KM algebra vertex modes def} take the exact same form
\begin{subequations} \label{KM algebra vertex modes def hom}
\begin{align}
\label{KM algebra vertex modes def hom a} \big[ \X_{[m]}, \Y_{[n]} \big] &= [\X, \Y]_{[m+n]} + m\, \kappa(\X, \Y) \delta_{m+n, 0} \, {\ms k}, \\
\label{KM algebra vertex modes def hom b} \big[ \bar \X_{[m]}, \bar \Y_{[n]} \big] &= \overline{[\X, \Y]}_{[m+n]} + m\, \kappa(\X, \Y) \delta_{m+n, 0} \, \bar{\ms k}.
\end{align}
\end{subequations}
for every $\X, \Y \in \g$ and $m,n \in \ZZ$. It is evident from these relations that the $\ZZ^2$-grading on $\hVV^{\fl, \alpha}$ is well defined since each term has the same total homogeneous mode number.

\paragraph{Virasoro:} We have $\fl = \text{span}_\CC \{ \Omega \}$ and because $\Omega_{\xi_p} = - \partial_{\xi_p}$ we deduce that $\Delta_\Omega = 2$. The homogeneous chiral and anti-chiral vertex $n^{\rm th}$-modes of $\Omega$ are then given $\Omega_{[n]} = \Omega_{(n+1)}$ and $\bar\Omega_{[n]} = \bar\Omega_{(n+1)}$. The defining relations \eqref{Virasoro algebra vertex modes def} of $\hg \oplus \hbg$ then take the form
\begin{subequations} \label{Virasoro algebra vertex modes def hom}
\begin{align}
\label{Virasoro algebra vertex modes def hom a} \big[ \Omega_{[m]}, \Omega_{[n]} \big] &= (m-n) \Omega_{[m+n]} + \frac{m^3-m}{12} c \, \delta_{m+n, 0} \, {\ms k}, \\
\label{Virasoro algebra vertex modes def hom b} \big[ \bar\Omega_{[m]}, \bar\Omega_{[n]} \big] &= (m-n) \bar\Omega_{[m+n]} + \frac{m^3-m}{12} c\, \delta_{m+n, 0} \, \bar{\ms k},
\end{align}
\end{subequations}
which are the usual Virasoro algebra relations for $L_{(n)} = \Omega_{[n]}$ and $\bar L_{(n)} = \bar\Omega_{[n]}$. In particular, it is again clear from these relations that the $\ZZ^2$-grading on $\hVV^{\fl, \alpha}$ is well defined.

\paragraph{$\beta\gamma$ system:} Here $\fl = \text{span}_\CC \{ \beta, \gamma \}$ with $\beta_{\xi_p} = 1$ and $\gamma_{\xi_p} = \d \xi_p$ so that $\Delta_\beta = 1$ and $\Delta_\gamma = 0$. The homogeneous (anti-)chiral vertex $n^{\rm th}$-modes are then $\beta_{[n]} = \beta_{(n)}$, $\bar\beta_{[n]} = \bar\beta_{(n)}$, $\gamma_{[n]} = \gamma_{(n-1)}$ and $\bar\gamma_{[n]} = \bar\gamma_{(n-1)}$, cf. \eqref{a a star def}. The defining relations \eqref{beta gamma algebra vertex modes def} of $\hg \oplus \hbg$ then read
\begin{subequations} \label{beta gamma algebra vertex modes hom def}
\begin{align}
\label{beta gamma algebra vertex modes def hom a} \big[ \beta_{[m]}, \gamma_{[n]} \big] &= \delta_{m+n, 0} \, {\ms k}, \\
\label{beta gamma algebra vertex modes def hom b} \big[ \bar\beta_{[m]}, \bar\gamma_{[n]} \big] &= \delta_{m+n, 0} \, \bar{\ms k},
\end{align}
\end{subequations}
which are the usual infinite-dimensional Weyl algebra relations. Again, it is immediate from these relations that the $\ZZ^2$-grading on $\hVV^{\fl, \alpha}$ is well defined.

\subsubsection{Vertex operators}

The vertex operator associated with a state $A \in \hVV^{\fl,\alpha}$ can be described geometrically in terms of the prefactorisation algebra $\U\L^\Sigma_\alpha$ by preparing this state at a point $q \in U \subset \Sigma^\circ$ in some local coordinate $\xi : U \to \CC$, to obtain the element $A^\xi_q \in \hV^{\fl,\alpha}_q$, and then applying the factorisation product $m_{q, Y} : \hV^{\fl,\alpha}_q \to \U\L^\Sigma_\alpha(Y)$ to an annulus shaped open subset $Y \subset U$ with $q \in Y$ and encircling another point $p \in U$. The resulting element of $\U\L^\Sigma_\alpha(Y)$ then naturally acts on $\hV^{\fl,\alpha}_p$ via the factorisation products $m_{(Y,p), U}$, cf. the action of vertex modes on $\hV^{\fl,\alpha}_p$ described in \eqref{endo Xn bar Xn}, and encodes the vertex operator of the state $A \in \hVV^{\fl,\alpha}$, see the next proposition.

For any $\lambda \in \CC$ we use the notation $\underline{\lambda}$ as a shorthand for the pair $(\lambda, \bar \lambda)$. Similarly, we use the notation $\underline{\zeta}$ for a pair of formal variables $(\zeta, \bar \zeta)$. Recall that the normal ordered product of $r \in \ZZ_{\geq 2}$ operators $\O_i(\zeta)$ for $i \in \{ 1, \ldots, r \}$ with a decomposition $\O_i(\zeta) = \O_i(\zeta)_+ + \O_i(\zeta)_-$ into creation and annihilation operators $\O_i(\zeta)_\pm$ is defined recursively by
\begin{equation*}
\nord{\O_r(\zeta) \ldots \O_1(\zeta)} \,\coloneqq\, \O_r(\zeta)_+ \, \nord{\O_{r-1}(\zeta) \ldots \O_1(\zeta)} + \nord{\O_{r-1}(\zeta) \ldots \O_1(\zeta)} \, \O_r(\zeta)_-.
\end{equation*}

\begin{proposition} \label{prop: vertex operator general}
Let $q \in Y \subset \Sigma^\circ$ be an annulus shaped open subset around a point $p \in \Sigma^\circ$ in a local coordinate $\xi$. For any $A^\xi_q \in \hV^{\fl,\alpha}_q$, the element $m_{q,Y}(A^\xi_{q}) \in \U\L^\Sigma_\alpha(Y)$ expands as
\begin{equation*}
\qquad m_{q,Y}(A^\xi_{q}) = \mathcal Y\big( A^\xi_q, \underline{\xi_p(q)} \big),
\qquad\qquad
\raisebox{-9mm}{\begin{tikzpicture}
\def\R{1}
  \fill[blue!90,
        opacity      = 0.3,
        even odd rule]
        (.05*\R,1.2*\R) circle[radius=.8*\R] circle[radius=.3*\R];
  \draw[thick, blue] (.05*\R,1.2*\R) circle (.8*\R);
  \draw[thick, blue] (.05*\R,1.2*\R) circle (.3*\R);
  \filldraw[thick] (.05*\R,1.2*\R) node[below=-.7mm]{\tiny $p$} circle (0.02*\R) node[above right=5.5mm and 4.2mm, blue]{\tiny $Y$};
  \filldraw[thick] (-.5*\R,1.24*\R) node[below=-.7mm]{\tiny $q$}
							circle (0.02*\R);
\end{tikzpicture}
}
\end{equation*}
where the right hand side is the \emph{vertex operator} defined as the usual normal ordered product
\begin{align*}
\mathcal Y\big( A^\xi_q, \underline{\xi_p(q)} \big) &\coloneqq \nord{\frac{1}{(m_r -1)!} \partial_{\xi_p(q)}^{m_r - 1} \a^r[\xi_p(q)] \ldots \frac{1}{(m_1 -1)!} \partial_{\xi_p(q)}^{m_1 - 1} \a^1[\xi_p(q)]}\\
&\qquad \times \nord{\frac{1}{(n_{\bar r} -1)!} \partial_{\bar \xi_p(q)}^{n_{\bar r} - 1} \bar \b^{\bar r}[\bar \xi_p(q)] \ldots \frac{1}{(n_1 -1)!} \partial_{\bar \xi_p(q)}^{n_1 - 1} \bar \b^1[\bar \xi_p(q)]}
\end{align*}
for any monomial state $A \in \hVV^{\fl,\alpha}$ as in \eqref{gen state Vkx}, and extended by linearity to all of $\hVV^{\fl,\alpha}$.
\begin{proof}
By linearity it is sufficient to consider a state $A \in \hVV^{\fl,\alpha}$ as in \eqref{gen state Vkx}. In fact, we will focus on proving the statement for a chiral state $A = \a^r_{(-m_r)} \ldots \a^1_{(-m_1)} \vac$ since the treatment of the anti-chiral part is completely analogous.

Let $q_+ \in Y_0 \Subset \ldots \Subset Y_r \subset Y_+$ be a nested sequence of annuli shaped open subsets around the point $p_+$.
We have
\begin{equation*}
m_{q,Y}(A^\xi_{q}) = \Bigg[ \prod_{i=1}^r s\Big( \a^i_{\xi_q} \otimes \lceil \xi_q^{-m_i} \rceil^{Y_{i-1}}_{Y_i} \Big) \Bigg]_Y.
\end{equation*}
Each open subset $Y_j$ for $j \in \{ 0, \ldots, r \}$ can be written as a difference $Y_j = V^+_j \setminus \overline{V^-_j}$ for disc shaped open subsets $V^\pm_j$ around the point $p_+$. We can thus write each $\rho^{Y_{i-1}}_{Y_i} \in \Omega^{0,0}_c(Y_i)^1_{Y_{i-1}}$ for $i \in \{ 1, \ldots, r \}$ as a difference of two smooth bump functions
\begin{equation*}
\rho^{Y_{i-1}}_{Y_i} = \rho^{V^\tp_{i-1}}_{V^\tp_i} - \rho^{V^\tm_i}_{V^\tm_{i-1}}, \quad \textup{with} \quad
\rho^{V^\tp_{i-1}}_{V^\tp_i} \in \Omega^{0,0}_c(V^+_i)^1_{V^\tp_{i-1}}, \quad
\rho^{V^\tm_i}_{V^\tm_{i-1}} \in \Omega^{0,0}_c(V^-_{i-1})^1_{V^\tm_i}
\end{equation*}
where $p_+ \in V^-_i \Subset V^-_{i-1} \Subset V^+_{i-1} \Subset V^+_i$ is a nested sequence of disc shaped open subsets around $p_+$ in the local coordinate $\xi$ such that $q_+ \in V^+_j$ and $q_+ \not\in V^-_j$ for all $j \in \{ 0, \ldots, r \}$, which we can depict schematically as
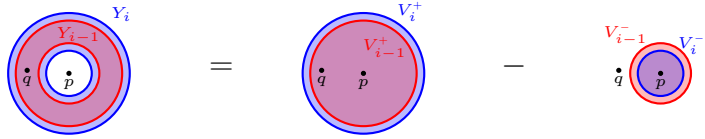
\begin{center}
\begin{tikzpicture}
\def\R{1}
  \fill[blue!90,
        opacity      = 0.3,
        even odd rule]
        (.05*\R,1.2*\R) circle[radius=.8*\R] circle[radius=.3*\R];
  \draw[thick, blue] (.05*\R,1.2*\R) circle (.8*\R);
  \draw[thick, blue] (.05*\R,1.2*\R) circle (.3*\R);
  \fill[red!90,
        opacity      = 0.3,
        even odd rule]
        (.05*\R,1.2*\R) circle[radius=.7*\R] circle[radius=.4*\R];
  \draw[thick, red] (.05*\R,1.2*\R) circle (.7*\R);
  \draw[thick, red] (.05*\R,1.2*\R) circle (.4*\R);
  \filldraw[thick] (.05*\R,1.2*\R) node[below=-.7mm]{\tiny $p$} circle (0.02*\R) node[above right=2.6mm and -3mm, red]{\tiny $Y_{i-1}$} node[above right=5.5mm and 4.2mm, blue]{\tiny $Y_i$};
  \filldraw[thick] (-.5*\R,1.24*\R) node[below=-.7mm]{\tiny $q$} circle (0.02*\R);
\end{tikzpicture}
\qquad \raisebox{8mm}{\large $=$} \qquad
\begin{tikzpicture}
\def\R{1}
  \fill[blue!90,
        opacity      = 0.3]
        (.05*\R,1.2*\R) circle (.8*\R);
  \draw[thick, blue] (.05*\R,1.2*\R) circle (.8*\R);
  \fill[red!90,
        opacity      = 0.3]
        (.05*\R,1.2*\R) circle (.7*\R);
  \draw[thick, red] (.05*\R,1.2*\R) circle (.7*\R);
  \filldraw[thick] (.05*\R,1.2*\R) node[below=-.7mm]{\tiny $p$} circle (0.02*\R) node[above right=.5mm and -1.5mm, red]{\tiny $V^\tp_{i-1}$} node[above right=5.5mm and 3mm, blue]{\tiny $V^\tp_i$};
  \filldraw[thick] (-.5*\R,1.24*\R) node[below=-.7mm]{\tiny $q$} circle (0.02*\R);
\end{tikzpicture}
\qquad \raisebox{8mm}{\large $-$} \qquad
\raisebox{4mm}{\begin{tikzpicture}
\def\R{1}
  \fill[red!90,
        opacity      = 0.3]
        (.05*\R,1.2*\R) circle (.4*\R);
  \draw[thick, red] (.05*\R,1.2*\R) circle (.4*\R);
  \fill[blue!90,
        opacity      = 0.3]
        (.05*\R,1.2*\R) circle (.3*\R);
  \draw[thick, blue] (.05*\R,1.2*\R) circle (.3*\R);
  \filldraw[thick] (.05*\R,1.2*\R) node[below=-.7mm]{\tiny $p$} circle (0.02*\R) node[above left=2.3mm and .3mm, red]{\tiny $V^\tm_{i-1}$} node[above right=1.1mm and .9mm, blue]{\tiny $V^\tm_i$};
  \filldraw[thick] (-.5*\R,1.24*\R) node[below=-.7mm]{\tiny $q$} circle (0.02*\R);
\end{tikzpicture}
}
\end{center}
We may then write
\begin{equation} \label{m A proof}
m_{q,Y}(A^\xi_{q}) = \Bigg[ \prod_{i=1}^r s\bigg( \a^i_{\xi_q} \otimes \Big( \lceil \xi_q^{-m_i} \rceil^{V^\tp_{i-1}}_{V^\tp_i} - \lceil \xi_q^{-m_i} \rceil^{V^\tm_i}_{V^\tm_{i-1}} \Big) \bigg) \Bigg]_Y.
\end{equation}
Expanding this out we get a sum of $2^r$ terms, each of which is the cohomology class of a $r$-fold $\Sym$-product of terms of the form
\begin{equation} \label{2 types of terms pre}
s \Big( \a^i_{\xi_q} \otimes \lceil \xi_q^{-m_i} \rceil^{V^\tp_{i-1}}_{V^\tp_i} \Big) \qquad \textup{or} \qquad
- s \Big( \a^i_{\xi_q} \otimes \lceil \xi_q^{-m_i} \rceil^{V^\tm_i}_{V^\tm_{i-1}} \Big),
\end{equation}
for each $i \in \{ 1, \ldots, r \}$. Due to the support properties of the different smooth bump functions present, in each case we can expand $\xi_q^{-m_i} = \frac{1}{(m_i-1)!} \partial_{\xi(q)}^{m_i-1} (\xi_p - \xi_p(q))^{-1}$ in the region where $|\xi_p(q)| < |\xi_p|$ or $|\xi_p(q)| > |\xi_p|$, respectively. Using also the fact that $\a^i_{\xi_q} = \a^i_{\xi_p}$ by Remark \ref{rem: a xi q no expand}, we may write \eqref{2 types of terms pre} as
\begin{align} \label{2 types of terms}
&\frac{1}{(m_i-1)!} \partial_{\xi(q)}^{m_i-1} \Bigg( \sum_{n < 0} \xi_p(q)^{-n-1} s \Big( \a^i_{\xi_p} \otimes \lceil \xi_p^n \rceil^{V^\tp_{i-1}}_{V^\tp_i} \Big) \Bigg) \notag\\
&\qquad\qquad \textup{or} \qquad
\frac{1}{(m_i-1)!} \partial_{\xi(q)}^{m_i-1} \Bigg( \sum_{n \geq 0} \xi_p(q)^{-n-1} s \Big( \a^i_{\xi_p} \otimes \lceil \xi_p^n \rceil^{V^\tm_i}_{V^\tm_{i-1}} \Big) \Bigg).
\end{align}
Define the annulus shaped open subsets $Y^+_i \coloneqq V^+_i \setminus \overline{V^+_{i-1}}$ and $Y^-_i \coloneqq V^-_{i-1} \setminus \overline{V^-_i}$ for $i \in \{ 1, \ldots, r \}$. Then taking the cohomology classes $[\cdot]_{Y_i^\tp}$ or $[\cdot]_{Y_i^\tm}$ of the expressions in \eqref{2 types of terms} gives
\begin{align} \label{2 types of terms 2}
\frac{1}{(m_i-1)!} \partial_{\xi(q)}^{m_i-1} \a[\xi_p(q)]_+ \in \U\L^\Sigma_\alpha(Y^+_i)
\quad \textup{or} \quad
\frac{1}{(m_i-1)!} \partial_{\xi(q)}^{m_i-1} \a[\xi_p(q)]_- \in \U\L^\Sigma_\alpha(Y^-_i).
\end{align}
In summary, we can rewrite \eqref{m A proof} explicitly as a sum of $2^r$ terms
\begin{equation} \label{m A proof rewrite}
m_{q,Y}(A^\xi_{q}) = \sum_{(\varepsilon_1, \ldots, \varepsilon_r) \in \{+, -\}^r} m_{(Y_1^{\varepsilon_1}, \ldots, Y_r^{\varepsilon_r}), Y} \bigg( \bigotimes_{i=1}^r \frac{1}{(m_i-1)!} \partial_{\xi(q)}^{m_i-1} \a[\xi_p(q)]_{\varepsilon_i} \bigg)
\end{equation}
where the support of the $i^{\rm th}$ term in each of the above factorisation products is $Y_i^{\varepsilon_i}$, namely it is determined by the sign $\varepsilon_i \in \{ +, - \}$ according to \eqref{2 types of terms 2}.
The desired result now follows from observing that the relative ordering of the $2 r$ annuli $Y_i^\pm$ for $i \in \{1, \ldots, r\}$ appearing in each of the $2^r$ terms of the sum \eqref{m A proof rewrite} coincides with the relative ordering of the $2r$ operators appearing in each of the $2^r$ terms of the desired normal ordered product.
\end{proof}
\end{proposition}

The coefficients in the expansion of the vertex operator from Proposition \ref{prop: vertex operator general}, namely
\begin{equation} \label{vertex algebra modes}
\mathcal Y\big( A^\xi_q, \underline{\xi_p(q)} \big) = \sum_{n, \bar n \in \ZZ} A^{\; \xi_p}_{(n, \bar n)} \, \xi_p(q)^{-n-1} \bar \xi_p(q)^{-\bar n-1},
\end{equation}
define the \emph{vertex modes} $A^{\; \xi_p}_{(n, \bar n)} \in \U\L^\Sigma_\alpha(Y)$ of the state $A \in \hVV^{\fl,\alpha}$ at $p \in \Sigma^\circ$ in the local coordinate $\xi$ for any $n, \bar n \in \ZZ$. Just as the elements \eqref{X bar X vertex modes} of $\U\L^\Sigma_\alpha(Y)$ did, each vertex mode of any given state $A \in \hVV^{\fl,\alpha}$ also gives rise to an endomorphism
\begin{equation*}
A^{\; \xi_p}_{(n, \bar n)} : \hV^{\fl,\alpha}_p \longrightarrow \hV^{\fl,\alpha}_p, \qquad
B^\xi_p \longmapsto A^{\; \xi_p}_{(n, \bar n)} B^\xi_p
\end{equation*}
of $\hV^{\fl,\alpha}_p$ for each $n, \bar n \in \ZZ$, defined by forming the factorisation product $m_{(Y, p), U} \big( A^{\; \xi_p}_{(n, \bar n)} \otimes B^\xi_p \big)$ with the input state $B \in \hVV^{\fl,\alpha}$ prepared at $p$ in the local coordinate $\xi$. The fact that this defines an element of $\hV^{\fl,\alpha}_p$ follows from the next lemma.

\begin{lemma} \label{lem: field vertex modes}
In any local coordinate $\xi$ in the neighbourhood of a point $p \in \Sigma^\circ$ and for any $A, B \in \hVV^{\fl,\alpha}$ we have
\begin{equation*}
A^{\; \xi_p}_{(n, \bar n)} B^\xi_p = \big( A_{(n, \bar n)} B \big)^\xi_p
\end{equation*}
for some unique $A_{(n, \bar n)} B \in \hVV^{\fl,\alpha}$. Moreover, we have $A_{(n, \bar n)} B = 0$ if either $n$ or $\bar n$ is sufficiently large. In particular, we have $A_{(n, \bar n)} \vac = 0$ if either $n \geq 0$ or $\bar n \geq 0$.
\begin{proof}
We use the same notation as in the proof of Proposition \ref{prop: vertex operator general}. And just as in the latter, it is sufficient to consider only the chiral part of the state $B \in \hVV^{\fl,\alpha}$, i.e. we can suppose that $B =\Y^s_{(-n_s)} \ldots \Y^1_{(-n_1)} \vac$, since the anti-chiral part can be treated completely analogously.

Now for every $n \in \ZZ$, the vertex mode $A^{\; \xi_p}_{(n, -1)} \in \U\L^\Sigma_\alpha(Y)$ is the coefficient of $\xi_p(q)^{-n-1}$ in the operator \eqref{m A proof rewrite} from the proof of Proposition \ref{prop: vertex operator general}, consisting of a sum of $2^r$ terms. Aside from the two extremal terms with $\varepsilon_i = \pm$ for all $i \in \{ 1, \ldots, r \}$, the remaining $2^r-2$ terms will all contribute an infinite sum to the coefficient of $\xi_p(q)^{-n-1}$. However, it can be seen by repeatedly applying the identities \eqref{commutator check} from Proposition \ref{prop: Y map def} that only finitely many terms in these infinite sums will contribute non-trivially to the factorisation product
\begin{equation*}
m_{(Y, p), U} \Big( A^{\; \xi_p}_{(n, -1)} \otimes B^\xi_p \Big) \in \U\L^\Sigma_\alpha(U)
\end{equation*}
into a disc shaped open subset $U \supset Y$. It follows that this factorisation product can be written as $(A_{(n, -1)} B)^{\xi_p}_U \in \U\L^\Sigma_\alpha(U)$ for some state $A_{(n, -1)} B \in \hVV^{\fl,\alpha}$ which is unique by the injectivity of \eqref{cal Y map defined}. If $n \in \ZZ_{\geq 0}$ is too large then $A_{(n, -1)} B=0$ since none of the $2^r$ terms in the sum \eqref{m A proof rewrite} will contribute to the above factorisation product. Finally, for the `in particular' part, terms from the sum \eqref{m A proof rewrite} which contribute to the coefficient of $\xi_p(q)^{-n-1}$ for $n \in \ZZ_{\geq 0}$ must have $\varepsilon_i = -$ for some $i \in \{1, \ldots, r\}$ and all such terms annihilate the vacuum by \eqref{commutator check d}.
\end{proof}
\end{lemma}

In light of the mode expansion \eqref{vertex algebra modes} and the definition of the states $A_{(n, \bar n)} B \in \hVV^{\fl,\alpha}$ given in Lemma \ref{lem: field vertex modes}, it will be convenient to define the formal vertex operator map which gathers all these states into a formal sum defined as
\begin{equation} \label{vertex operator Y}
Y( A, \underline{\zeta} ) B = \sum_{n, \bar n \in \ZZ} A_{(n, \bar n)} B \, \zeta^{-n-1} \bar \zeta^{-\bar n-1},
\end{equation}
where $\zeta$ and $\bar\zeta$ are formal variables. This is related to the vertex operator in \eqref{vertex algebra modes} by
\begin{equation} \label{cal Y vs Y}
\mathcal Y\big( A^\xi_q, \underline{\xi_p(q)} \big) B^\xi_p = \Big( Y\big( A, \underline{\xi_p(q)} \big) B \Big)^\xi_p.
\end{equation}
In other words, by preparing all the states $A_{(n, \bar n)} B \in \hVV^{\fl,\alpha}$ for $n, \bar n \in \ZZ$, appearing in \eqref{vertex operator Y}, at the point $p \in U$ in the local coordinate chart $(U, \xi)$ and replacing the formal variables $\zeta$, $\bar \zeta$ in \eqref{vertex operator Y} by the complex numbers $\xi_p(q)$, $\bar \xi_p(q)$ for some other point $q \in U$, we get back the expansion of the factorisation product $m_{(p,q), U}(A^\xi_q \otimes B^\xi_p)$ in $q$ near $p$.

\subsubsection{Homogeneous vertex operators} \label{sec: hom vert op}

Recall from \S\ref{sec: hom vertex modes} the notion of homogeneous (anti-)chiral vertex modes for an element $\a \in \fl$ of definite conformal dimension $\Delta_\a$.
Given a monomial state $A \in \hVV^{\fl,\alpha}$, it is convenient to similarly introduce a notion of homogeneous vertex modes of $A$ by expanding the vertex operator from Proposition \ref{prop: vertex operator general} as
\begin{equation} \label{hom vertex modes def}
\mathcal Y\big( A^\xi_q, \underline{\xi_p(q)} \big) = \sum_{n, \bar n \in \ZZ} A^{\; \xi_p}_{[n, \bar n]} \xi_p(q)^{-n-\Delta_A} \bar \xi_p(q)^{-\bar n- \bar\Delta_A},
\end{equation}
which is to be compared with \eqref{vertex algebra modes}. The coefficients $A^{\; \xi_p}_{[n, \bar n]} \in \U\L^\Sigma_\alpha(Y)$ for $n, \bar n \in \ZZ$ are called the \emph{homogenous vertex modes} of the monomial state $A \in \hVV^{\fl,\alpha}$ at $p \in \Sigma$ in the local coordinate $\xi$. They are related to the vertex modes by a simple shift, cf. \eqref{hom modes a def},
\begin{equation} \label{hom vertex modes relation}
A^{\; \xi_p}_{[n, \bar n]} = A^{\; \xi_p}_{(n + \Delta_A - 1, \bar n + \bar\Delta_A - 1)}.
\end{equation} 
We also introduce the notation $A_{[n, \bar n]} = A_{(n + \Delta_A - 1, \bar n + \bar\Delta_A - 1)} \in \End \hVV^{\fl, \alpha}$, see Lemma \ref{lem: field vertex modes}.

\begin{lemma} \label{lem: chiral weights}
Let $A \in \hVV^{\fl,\alpha}$ be any monomial state and $n, \bar n \in \ZZ$.
Its vertex modes at $p \in \Sigma$ in any local coordinate $\xi$ are homogeneous elements of $\End \hV^{\fl,\alpha}_p$ in the local coordinate $\xi$, of chiral and anti-chiral weights
\begin{equation*}
\wgt \big( A^{\; \xi_p}_{(n, \bar n)} \big) = \Delta_A  - n - 1, \qquad
\wgtb \big( A^{\; \xi_p}_{(n, \bar n)} \big) = \bar\Delta_A  - \bar n - 1.
\end{equation*}
Likewise, the homogeneous vertex modes of $A$ at $p$ in the local coordinate $\xi$ are homogeneous elements of $\End \hV^{\fl,\alpha}_p$ in the local coordinate $\xi$, of chiral and anti-chiral weights
\begin{equation*}
\wgt \big( A^{\; \xi_p}_{[n, \bar n]} \big) = - n, \qquad
\wgtb \big( A^{\; \xi_p}_{[n, \bar n]} \big) = - \bar n.
\end{equation*}
\begin{proof}
By definition of the $\ZZ^2$-grading on $\hV^{\fl,\alpha}_p$ relative to the local coordinate $\xi$, it is clear that the homogeneous chiral and anti-chiral vertex modes \eqref{endo Xn bar Xn} are homogeneous with weights
\begin{equation*}
\wgt \big( \a^{\; \xi_p}_{[n]} \big) = - n, \qquad \wgtb \big( \a^{\; \xi_p}_{[n]} \big) = 0,
\qquad
\wgt \big( \bar \a^{\; \xi_p}_{[\bar n]} \big) = 0, \qquad \wgtb \big( \bar \a^{\; \xi_p}_{[\bar n]} \big) = - \bar n
\end{equation*}
for every $n, \bar n \in \ZZ$.
The last statement now follows using the explicit expression for the vertex operator of the monomial state $A \in \hVV^{\fl,\alpha}$ from Proposition \ref{prop: vertex operator general} and the definition \eqref{hom vertex modes def} of the homogeneous vertex modes, recalling that $\Delta_A = \sum_{i=1}^r (\Delta_{\a^i} + m_i - 1)$ and noting using \eqref{vertex operator hom rewrite} that the chiral factor $\partial_{\xi_p(q)}^{m_i-1} \a^i[\xi_p(q)]$ for $i \in \{ 1, \ldots, r\}$ consists of operators of the form
\begin{equation*}
\a^i_{[k]} \xi_p(q)^{-k - \Delta_{\a^i} - m_i + 1}
\end{equation*}
with $k \in \ZZ$ and likewise for the anti-chiral factors. The first result then follows by \eqref{hom vertex modes relation}.
\end{proof}
\end{lemma}

Given any monomial state $A \in \hVV^{\fl,\alpha}$, we will refer to
\begin{equation} \label{homogeneous vertex op}
\mathcal X\big( A^\xi_q, \underline{\xi_p(q)} \big) \coloneqq \xi_p(q)^{\Delta_A} \bar \xi_p(q)^{\bar\Delta_A} \, \mathcal Y\big( A^\xi_q, \underline{\xi_p(q)} \big)
= \sum_{n, \bar n \in \ZZ} A^{\; \xi_p}_{[n, \bar n]} \xi_p(q)^{-n} \bar \xi_p(q)^{-\bar n}
\end{equation}
as the \emph{homogeneous vertex operator of $A^\xi_q \in \hV^{\fl,\alpha}_q$ at $q$ in the coordinate $\xi$}, see e.g. \cite[Remark 3.1.25]{LepowskyLi} in the chiral case. We extend this notion by linearity to all states in $\hV^{\fl,\alpha}_q$.

\subsubsection{Translation operators}

Recall the translation operators $D, \bar D : \hVV^{\fl,\alpha} \to \hVV^{\fl,\alpha}$ introduced in \eqref{D bar D def}. Their action on $\hV^{\fl,\alpha}_p$ was identified in Proposition \ref{prop: Y deriv} with that of the derivative operators $\partial_{\xi(p)}, \bar \partial_{\xi(p)} : \hV^{\fl,\alpha}_p \to \hV^{\fl,\alpha}_p$, respectively. In terms of the vertex modes \eqref{vertex algebra modes} of the vertex operator from Proposition \ref{prop: vertex operator general} their actions is given by the usual formula.

\begin{lemma} \label{lem: derivative vertex}
For any $A \in \hVV^{\fl,\alpha}$ and $n, \bar n \in \ZZ$ we have
\begin{equation*}
(D A)^{\, \xi_p}_{(n,\bar n)} = - n A^{\, \xi_p}_{(n-1,\bar n)}, \qquad
(\bar D A)^{\, \xi_p}_{(n,\bar n)} = - \bar n A^{\, \xi_p}_{(n,\bar n-1)}.
\end{equation*}
in $\End \hV^{\fl,\alpha}_p$, for any local coordinate $\xi$ in the neighbourhood of a point $p \in \Sigma^\circ$.
\begin{proof}
The first result follows from comparing the mode expansion of both sides of the relation
\begin{equation*}
\mathcal Y\big( (D A)^\xi_q, \underline{\xi_p(q)} \big) = m_{q,Y}\big( (D A)^\xi_{q} \big) = m_{q,Y}( \partial_{\xi(q)} A^\xi_{q} ) = \partial_{\xi(q)} \mathcal Y\big( A^\xi_q, \underline{\xi_p(q)} \big),
\end{equation*}
where the first and last equalities make use of Proposition \ref{prop: vertex operator general} and the second equality is by Proposition \ref{prop: Y deriv}. The proof of the second result is completely analogous.
\end{proof}
\end{lemma}

\begin{lemma} \label{lem: vacuum axiom}
We have $Y( \vac, \underline{\zeta} ) = \id_{\hVV^{\fl,\alpha}}$, and for any $A \in \hVV^{\fl,\alpha}$ we have the Taylor expansion $Y( A, \underline{\zeta} ) \vac = e^{\zeta D} e^{\bar \zeta \bar D} A$.
In particular, $D A = A_{(-2,-1)} \vac$ and $\bar D A = A_{(-1,-2)} \vac$.
\begin{proof}
Since $m_{q,Y}\big( \vac^\xi_q \big) = [1]_Y$ acts as the identity in $\End \hV^{\fl,\alpha}$ the first result is immediate from the definition of $\mathcal Y$ in Proposition \ref{prop: vertex operator general}.
In any chart $\xi : U \to \CC$ with $p, q \in U$ we have
\begin{align*}
\mathcal Y\big( A^\xi_q, \underline{\xi_p(q)} \big) \vac^{\xi_p}_U &= m_{(Y, p), U} \big( m_{q,Y}(A^\xi_q) \otimes \vac^\xi_p \big) = m_{q, U}(A^\xi_q) = A^{\xi_q}_U\\
&= e^{\xi_p(q) \partial_{\xi(p)}} e^{\bar \xi_p(q) \partial_{\bar \xi(p)}} A^{\xi_p}_U = \big( e^{\xi_p(q) D} e^{\bar \xi_p(q) \bar D} A \big)^{\xi_p}_U
\end{align*}
where $Y$ is an annulus shaped open subset containing $q$ and encircling the point $p$. In the second last equality we have Taylor expanded the expression $A^{\xi_q}_U$ in $q$ near $p$ and the final step is by Proposition \ref{prop: Y deriv}, see in particular its proof. The result now follows from \eqref{cal Y vs Y} and the injectivity of the linear map \eqref{cal Y map defined} by Proposition \ref{prop: Y map injective}.
\end{proof}
\end{lemma}

\begin{proposition}
For any $A, B \in \hVV^{\fl,\alpha}$ we have $Y( B, \underline{\zeta} ) A = e^{\zeta D} e^{\bar \zeta \bar D} Y( A, - \underline{\zeta} ) B$.
\begin{proof}
By associativity \eqref{PFA commutativity} of the factorisation product we have
\begin{equation*}
\mathcal Y\big( B^\xi_q, \underline{\xi_p(q)} \big) A^\xi_p = \mathcal Y\big( A^\xi_p, \underline{\xi_q(p)} \big) B^\xi_q
\end{equation*}
since both sides are given by $m_{(p, q), U}(A^\xi_p \otimes B^\xi_q)$ in any local chart $\xi : U \to \CC$ containing $p$ and $q$. In other words, in terms of the map $Y$ defined in \eqref{vertex operator Y} we have
\begin{equation*}
\Big( Y\big( B, \underline{\xi_p(q)} \big) A\Big)^\xi_p = \Big( Y\big( A, \underline{\xi_q(p)} \big) B \Big)^\xi_q
= \Big( e^{\xi_p(q) D} e^{\bar \xi_p(q) \bar D} Y\big( A, \underline{\xi_q(p)} \big) B \Big)^\xi_p
\end{equation*}
where the last step is by Lemma \ref{lem: vacuum axiom}, see in particular its proof.
\end{proof}
\end{proposition}

\subsubsection{(Anti-)chiral states} \label{sec: anti-chiral states}

The chiral vertex operator \eqref{X as X+ X-} and its anti-chiral analogue for any $\a \in \fl$ correspond to the states $\a_{(-1)}\vac^\xi_q$ and $\bar \a_{(-1)} \vac^\xi_q$, respectively, namely
\begin{equation} \label{X as vertex operator}
\a[\xi_p(q)] = \mathcal Y\big( \a_{(-1)} \vac^\xi_q, \xi_p(q) \big), \qquad
\bar\a[\bar\xi_p(q)] = \mathcal Y\big( \bar\a_{(-1)} \vac^\xi_q, \bar \xi_p(q) \big)
\end{equation}
and comparing the mode expansion of both sides it follows that for any $n, \bar n \in \ZZ$ we have
\begin{equation} \label{modes lvl 1 state}
\big( \a_{(-1)} \vac \big)^{\; \xi_p}_{(n, \bar n)} = \a^{\; \xi_p}_{(n)} \delta_{\bar n, -1}, \qquad
\big( \bar \a_{(-1)} \vac \big)^{\; \xi_p}_{(n, \bar n)} = \bar \a^{\; \xi_p}_{(\bar n)} \delta_{n, -1}.
\end{equation}
Notice $\Delta_{\a_{(-1)} \vac} = \bar\Delta_{\bar \a_{(-1)} \vac} = \Delta_\a$ and $\Delta_{\bar\a_{(-1)} \vac} = \bar\Delta_{\a_{(-1)} \vac} = 0$ for any $\a \in \fl$ so that also
\begin{align*}
\big( \a_{(-1)} \vac \big)^{\; \xi_p}_{[n,\bar n]} = \a^{\; \xi_p}_{[n]} \delta_{\bar n, 0}, \qquad
\big( \bar \a_{(-1)} \vac \big)^{\; \xi_p}_{[n, \bar n]} = \bar \a^{\; \xi_p}_{[\bar n]} \delta_{n, 0}.
\end{align*}

More generally, if $A \in \hVV^{\fl,\alpha}$ is chiral then by Proposition \ref{prop: Y deriv} the expansion of $m_{q, Y}(A^\xi_q)$ from Proposition \ref{prop: vertex operator general} is holomorphic, i.e. of the form $\mathcal Y\big( A^\xi_q, \xi_p(q) \big) = \sum_{n \in \ZZ} A^{\, \xi_p}_{(n)} \xi_p(q)^{-n-1}$, in which case its modes can be extracted as contour integrals
\begin{subequations} \label{(anti)chiral modes}
\begin{equation} \label{chiral mode}
A^{\, \xi_p}_{(n)} \coloneqq A^{\, \xi_p}_{(n,-1)} = \frac{1}{2 \pi \ii} \int_{c_p} m_{q,Y}(A^\xi_q) \xi_p(q)^n \d\xi(q)
\end{equation}
for every $n \in \ZZ$, using any counterclockwise oriented contour $c_p$ in $Y$ encircling the point $p$. If instead $A \in \hVV^{\fl,\alpha}$ is anti-chiral then we similarly have $\mathcal Y\big( A^\xi_q, \bar \xi_p(q) \big) = \sum_{n \in \ZZ} A^{\, \xi_p}_{(n)} \bar \xi_p(q)^{- n-1}$ whose modes we can extract as contour integrals
\begin{equation} \label{anti-chiral mode}
A^{\, \xi_p}_{(n)} \coloneqq A^{\, \xi_p}_{(-1, n)} = - \frac{1}{2 \pi \ii} \int_{c_p} m_{q,Y}(A^\xi_q) \bar \xi_p(q)^n \d\bar \xi(q).
\end{equation}
\end{subequations}
Although we have introduced the same abbreviated notation in both \eqref{chiral mode} and \eqref{anti-chiral mode}, this should not lead to confusion since the full correct notation $A^{\, \xi_p}_{(n, -1)}$ or $A^{\, \xi_p}_{(-1, n)}$ can be restored by noting that $A$ is chiral or anti-chiral, respectively.

Recall that the vacuum $\vac$ is both chiral and anti-chiral. It follows from Lemma \ref{lem: vacuum axiom} that its modes are $\vac^{\, \xi_p}_{(n,\bar n)} = [1]_Y \delta_{n, -1} \delta_{\bar n, -1}$ for all $n, \bar n \in \ZZ$ which is consistent with \eqref{(anti)chiral modes}.

\subsection{Borcherds type identities} \label{sec: Borcherds Vertex}

In an ordinary vertex algebra the modes of arbitrary states satisfy algebraic relations known as the Borcherds identities \cite{Borcherds:1983sq, FHL, LepowskyLi}. In the present setting, the vertex algebra modes of both chiral and anti-chiral states \eqref{(anti)chiral modes} also satisfy the same standard Borcherds identities. In fact, slightly more general `Borcherds type' identities also hold, given below in Proposition \ref{prop: Borcherds}, between the vertex modes of any (anti-)chiral state and those of a generic state.
However, there does not appear to be algebraic `Borcherds type' identities relating the vertex modes of two non-chiral states (see \cite[Remark 1.9]{Moriwaki:2020cxf}). This is the reason why the Borcherds identity is replaced by a different axiom in the various existing frameworks for full vertex operator algebras \cite{Huang:2005gz, Moriwaki:2020cxf, Singh:2023mom}. See \S\ref{sec: Full VOA axioms} for details.

\subsubsection{Moving points} \label{sec: moving points}

We can keep track of the dependence of the linear map \eqref{cal Y map defined} on the location of the point $p \in U$ where the state $A$ is inserted as follows. Given any open subset $U \subset \Sigma^\circ$ in $\Top(\Sigma)$, equipped with a local coordinate $\xi : U \to \CC$, we define the map
\begin{equation} \label{moving points 1}
\Phi^{1, \xi}_U : U \times \hVV^{\fl,\alpha} \longrightarrow \U\L^\Sigma_\alpha(U), \qquad
(p, A) \longmapsto m_{p,U}(A^\xi_p).
\end{equation}
The superscript $\xi$ refers to the coordinate used to prepare the state $A$ at the point $p$.
We can also generalise the linear map \eqref{cal Y map defined} to describe the insertion of $n \in \ZZ_{\geq 1}$ states from $\hVV^{\fl,\alpha}$ at a subset of $n$ points in $\Sigma^\circ$ as follows. Given any open subset $U \subset \Sigma^\circ$ in $\Top(\Sigma)$ and a collection of points $p_i \in U$ with $i \in \{ 1, \ldots, n \}$ for $n \in \ZZ_{\geq 1}$, we define morphisms
\begin{equation} \label{fac prod pi to U}
m_{(p_i), U} : \bigotimes_{i=1}^n \lim \, (\U\L^\Sigma_\alpha)_{p_i} \longrightarrow \U\L^\Sigma_\alpha(U)
\end{equation}
as the composition $m_{(p_i), U} \coloneqq m_{(U_i), U} \circ \big( \! \bigotimes_{i=1}^n m_{p_i, U_i} \big)$ for any inclusion of $n$ disjoint subsets $\sqcup_{i=1}^n U_i \subset U$. Combining the linear map \eqref{fac prod pi to U} with the local realisations \eqref{Y at p} at each $p_i$ in the local coordinate $\xi$, we then obtain a linear map
\begin{equation} \label{n Fg to Ug}
\bigotimes_{i=1}^n \hVV^{\fl,\alpha} \longrightarrow \U\L^\Sigma_\alpha(U), \qquad (A_i)_{i=1}^n \longmapsto m_{(p_i), U}\bigg( \bigotimes_{i=1}^n (A_i)^\xi_{p_i} \bigg).
\end{equation}
In fact, one could more generally consider different local coordinates $\xi_i$, $i \in \{ 1, \ldots, n \}$ around each of the points $p_i$ to obtain a linear map
\begin{equation} \label{n Fg to Ug xi i}
\bigotimes_{i=1}^n \hVV^{\fl,\alpha} \longrightarrow \U\L^\Sigma_\alpha(U), \qquad (A_i)_{i=1}^n \longmapsto m_{(p_i), U}\bigg( \bigotimes_{i=1}^n (A_i)^{\xi_i}_{p_i} \bigg).
\end{equation}
It is useful, as in \eqref{moving points 1}, to keep track of the location of the points $p_i$ at which the individual states $A_i$ are prepared. For any open subset $U \subset \Sigma^\circ$ in $\Top(\Sigma)$ and any $n \in \ZZ_{\geq 1}$ we let
\begin{equation*}
\Conf_n(U) \coloneqq \big\{ (p_i)_{i=1}^n \in U^{\times n} \,\big|\, p_i \neq p_j \; \text{for all}\; i \neq j \in \{ 1, \ldots, n \} \big\}
\end{equation*}
denote the configuration space of $n$ distinct points in $U$. In particular, for $n = 1$ this is just $\Conf_1(U) = U$.
We can now define the analogue of \eqref{moving points 1} for $n$ points as
\begin{align} \label{moving points n}
\Phi^{n, \xi}_U : \Conf_n(U) \times \bigotimes_{i=1}^n \hVV^{\fl,\alpha} \longrightarrow \U\L^\Sigma_\alpha(U), \qquad
\big( (p_i)_{i=1}^n, (A_i)_{i=1}^n \big) \longmapsto m_{(p_i), U}\bigg( \bigotimes_{i=1}^n (A_i)^\xi_{p_i} \bigg),
\end{align}
where again the superscript $\xi$ refers to the coordinate used to prepare each state $A_i$ at $p_i$.
Again, more generally, we could use different local coordinates $\xi_i$, $i \in \{ 1, \ldots, n \}$ around each of the points $p_i$ as in \eqref{n Fg to Ug xi i}, in which case we may define a map
\begin{align} \label{moving points n coord}
\Phi^{n, \bm \xi}_U : \Conf_n(U) \times \bigotimes_{i=1}^n \hVV^{\fl,\alpha} \longrightarrow \U\L^\Sigma_\alpha(U), \qquad
\big( (p_i)_{i=1}^n, (A_i)_{i=1}^n \big) \longmapsto m_{(p_i), U}\bigg( \bigotimes_{i=1}^n (A_i)^{\xi_i}_{p_i} \bigg),
\end{align}
which depends on the collection $\bm \xi = (\xi_i)_{i=1}^n$ of local coordinates around each point $p_i$.

\subsubsection{`Borcherds type' identities and consequences} \label{sec: Borcherds type id}

Recall the state $A_{(n, \bar n)} B \in \hVV^{\fl,\alpha}$ given by Lemma \ref{lem: field vertex modes} for any $A, B \in \hVV^{\fl,\alpha}$ and $n, \bar n \in \ZZ$.

\begin{proposition} \label{prop: Borcherds}
Let $A, B, C \in \hVV^{\fl,\alpha}$ and $k, \bar k, m, \bar m, n \in \ZZ$. We have the identities in $\hVV^{\fl,\alpha}$:
\begin{itemize}
  \item[$(i)$] If $A$ is chiral then
\begin{align*}
&\sum_{j \geq 0} \binom m j \big( A_{(n+j)} B \big)_{(m+k-j, \bar k)} C\\
&\qquad= \sum_{j \geq 0} (-1)^j \binom n j A_{(m+n-j)} B_{(k+j, \bar k)} C - \sum_{j \geq 0} (-1)^{n+j} \binom n j B_{(n+k-j, \bar k)} A_{(m+j)} C.
\end{align*}
  \item[$(ii)$] If $A$ is anti-chiral then
\begin{align*}
&\sum_{j \geq 0} \binom {\bar m} j \big( A_{(n+j)} B \big)_{(k, \bar m+\bar k-j)} C\\
&\quad= \sum_{j \geq 0} (-1)^j \binom n j A_{(\bar m+n-j)} B_{(k,\bar k+j)} C - \sum_{j \geq 0} (-1)^{n+j} \binom n j B_{(k, n+\bar k-j)} A_{(\bar m+j)} C.
\end{align*}
\end{itemize}
\begin{proof}
We prove only the chiral case $(i)$. The proof of the anti-chiral case $(ii)$ is completely analogous. So let $A, B \in \hVV^{\fl,\alpha}$ and suppose that $\bar D A = 0$.

Let $U \subset \Sigma^\circ$ be an open subset equipped with a local coordinate $\xi : U \to \CC$. Recall the map $\Phi^{3, \xi}_U$ in \eqref{moving points n} and define $F_{A, B, C} : \Conf_3(U) \to \U\L^\Sigma_\alpha(U)$ by
\begin{align} \label{FABC}
F_{A,B,C}(q, p, t) &\coloneqq \Phi^{3, \xi}_U\big( (q,p,t), (A, B,C) \big) \xi_p(q)^n \xi_t(q)^m \xi_t(p)^k \bar \xi_t(p)^{\bar k} \notag\\
&\, = m_{(q, p, t), U} \big( A^\xi_q \otimes B^\xi_p \otimes C^\xi_t \big) \xi_p(q)^n \xi_t(q)^m \xi_t(p)^k \bar \xi_t(p)^{\bar k}
\end{align}
for pairwise distinct points $q, p, t \in U$. Since we are assuming that $\bar D A = 0$, the above depends holomorphically on $q$ by Proposition \ref{prop: Y deriv}. By Cauchy's theorem we then have the identity
\begin{align} \label{contour identity Borcherds}
\frac{1}{2\pi \ii} \int_{c_p} F_{A,B,C}(q, p, t) \d\xi(q) = \frac{1}{2\pi \ii} \int_{c_t} F_{A,B,C}(q, p, t) \d\xi(q) - \frac{1}{2\pi \ii} \int_{c'_t} F_{A,B,C}(q, p, t) \d\xi(q)
\end{align}
using the standard deformation of contour argument
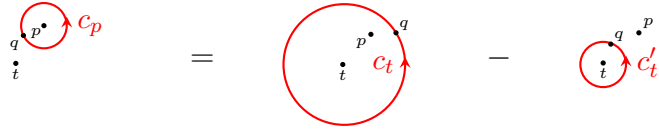
\begin{center}
\raisebox{5mm}{\begin{tikzpicture}
\def\R{1}
  \filldraw[thick] (-.4*\R,1.3*\R) node[below=-.7mm]{\tiny $t$}
circle (0.02*\R);
  \draw[thick, red] (-.03*\R,1.8*\R) ellipse (.3cm and .3cm);
  \draw[->,>=stealth, red, thick] (.27*\R,1.92*\R);
  \filldraw[thick] (-.3*\R,1.67*\R) node[below left=-1.2mm]{\tiny $q$} circle (0.02*\R);
  \filldraw[thick] (-.03*\R,1.8*\R) node[below left=-1.5mm]{\tiny $p$} node[right=3mm, red]{$c_p$} circle (0.02*\R);
\end{tikzpicture}
}
\qquad \raisebox{8mm}{\large $=$} \qquad
\begin{tikzpicture}
\def\R{1}
  \draw[thick, red] (-.28*\R,1.3*\R) circle (.8cm);
  \draw[->,>=stealth, red, thick] (.52*\R,1.4*\R);
  \filldraw[thick] (.4*\R,1.72*\R) node[above right=-1.4mm]{\tiny $q$} circle (0.02*\R);
  \filldraw[thick] (.07*\R,1.7*\R) node[below left=-1.2mm]{\tiny $p$} circle (0.02*\R);
  \filldraw[thick] (-.3*\R,1.3*\R) node[below=-.7mm]{\tiny $t$}
node[right=2.5mm, red]{$c_t$} circle (0.02*\R);
\end{tikzpicture}
\qquad \raisebox{8mm}{\large $-$} \qquad
\raisebox{5mm}{\begin{tikzpicture}
\def\R{1}

  \draw[thick, red] (-.4*\R,1.28*\R) circle (.3cm);
  \draw[->,>=stealth, red, thick] (-.1*\R,1.4*\R);
  \filldraw[thick] (-.3*\R,1.55*\R) node[above right=-1.4mm]{\tiny $q$} circle (0.02*\R);
  \filldraw[thick] (.07*\R,1.7*\R) node[above right=-1.2mm]{\tiny $p$} circle (0.02*\R);
  \filldraw[thick] (-.4*\R,1.3*\R) node[below=-.7mm]{\tiny $t$} node[below right= -3.5mm and 3mm, red]{$c'_t$} 
circle (0.02*\R);
\end{tikzpicture}
}
\end{center}
where $c_t$ and $c'_t$ are counterclockwise oriented contours around $t$, as depicted, and $c_p$ is a small counterclockwise oriented contour around the point $p$.

On the left hand side of \eqref{contour identity Borcherds}, since $q$ is closer to $p$ than $p$ is to $t$, we can expand the factor $\xi_t(q)^m = \big( \xi_t(p) + \xi_p(q) \big)^m$ of \eqref{FABC} in small $\xi_p(q)$.
Using also Proposition \ref{prop: vertex operator general} gives
\begin{equation*}
F_{A,B,C}(q, p, t) = \sum_{j \geq 0} \binom m j \mathcal Y \Big( \mathcal Y \big( A^\xi_q, \xi_p(q) \big) B^\xi_p, \ul{\xi_t(p)} \Big) C^\xi_t \, \xi_p(q)^{n+j} \xi_t(p)^{m+k-j} \bar \xi_t(p)^{\bar k}.
\end{equation*}
In the first integral on the right hand side of \eqref{contour identity Borcherds}, since $p$ is closer to $t$ than $q$ is, we can expand the factor $\xi_p(q)^n = \big( \xi_t(q) - \xi_t(p) \big)^n$ of \eqref{FABC} in small $\xi_t(p)$.
Using also Proposition \ref{prop: vertex operator general} we have the following expansion
\begin{equation*}
F_{A,B,C}(q, p, t) = \sum_{j \geq 0} (-1)^j \binom n j \mathcal Y\big( A^\xi_q, \xi_t(q) \big) \mathcal Y \big( B^\xi_p, \ul{\xi_t(p)} \big) C^\xi_t \, \xi_t(p)^{k+j} \xi_t(q)^{m+n-j} \bar \xi_t(p)^{\bar k}.
\end{equation*}
Likewise, in the second integral on the right hand side of \eqref{contour identity Borcherds} we can expand the same factor $\xi_p(q)^n = (-1)^n \big( \xi_t(p) - \xi_t(q) \big)^n$ of \eqref{FABC} in small $\xi_t(q)$.
Using also Proposition \ref{prop: vertex operator general} yields
\begin{equation*}
F_{A,B,C}(q, p, t) = \sum_{j \geq 0} (-1)^{n+j} \binom n j \mathcal Y\big( B^\xi_p, \ul{\xi_t(p))} \big) \mathcal Y\big( A^\xi_q, \xi_t(q) \big) C^\xi_t \, \xi_t(q)^{m+j} \xi_t(p)^{n+k-j} \bar \xi_t(p)^{\bar k}.
\end{equation*}
Upon integrating the above expansions in $\xi(q)$ along the contours $c_p$, $c_t$ and $c'_t$, respectively, and using the definition of the vertex modes in \eqref{vertex algebra modes} we obtain the three terms of the desired identity as the coefficient of $\xi_t(p)^0 \bar \xi_t(p)^0$.
\end{proof}
\end{proposition}

\begin{remark} \label{rem: no general Borcherds}
It is clear from the proof of Proposition \ref{prop: Borcherds} that `Borcherds type' identities can only be derived in the case when one of the two states involved is either chiral or anti-chiral. If this is not the case then we cannot use Cauchy's residue theorem as in \eqref{contour identity Borcherds} to relate the three terms in the Borcherds identity (see \cite[Remark 1.9]{Moriwaki:2020cxf}).
\end{remark}

The following consequences of the `Borcherds type' identities are equivalent to those in \cite[Lemma 3.11]{Moriwaki:2020cxf}, in the case of a chiral state $A$. We have also explicitly stated the version for an anti-chiral state $A$ for completeness.

\begin{corollary} \label{cor: com nord vertex}
Let $A, B, C \in \hVV^{\fl,\alpha}$ and $k, \bar k, m, \bar m \in \ZZ$. We have the identities in $\hVV^{\fl,\alpha}$:
\begin{itemize}
  \item[$(i)$] If $A$ is chiral then
\begin{align*}
\big[ A_{(m)}, B_{(k, \bar k)} \big] C &= \sum_{j \geq 0} \binom m j \big( A_{(j)} B \big)_{(m+k-j, \bar k)} C,\\
\big( A_{(-1)} B \big)_{(k, \bar k)} C &= \sum_{j \geq 0} A_{(-j-1)} B_{(k+j, \bar k)} C + \sum_{j < 0} B_{(k+j, \bar k)} A_{(-j-1)} C.
\end{align*}
  \item[$(ii)$] If $A$ is anti-chiral then
\begin{align*}
\big[ A_{(\bar m)}, B_{(k,\bar k)} \big] C &= \sum_{j \geq 0} \binom {\bar m} j \big( A_{(j)} B \big)_{(k, \bar m+\bar k-j)} C,\\
\big( A_{(-1)} B \big)_{(k, \bar k)} C
&= \sum_{j \geq 0} A_{(-j-1)} B_{(k, \bar k+j)} C + \sum_{j < 0} B_{(k, \bar k+j)} A_{(-j-1)} C.
\end{align*}
\end{itemize}
\begin{proof}
The commutator formulae both follow from taking $n=0$ in the Borcherds type identities of Proposition \ref{prop: Borcherds}. The second relation in $(i)$ (resp. $(ii)$) follows from Proposition \ref{prop: Borcherds}$(i)$ (resp. $(ii)$) in the case $n=-1$ and $m=0$ (resp. $\bar m =0$).
\end{proof}
\end{corollary}

It is instructive to see that we recover the Lie algebra relations of the generators of $\hg$ and $\hbg$ from \S\ref{sec: Vec structure} as a special case of the commutation relations in Corollary \ref{cor: com nord vertex}. We consider the three main examples from \S\ref{sec: main examples} separately.

\paragraph{Kac-Moody:}
It follows from the action of the vertex modes of any $\X \in \g$ in \eqref{endo Xn bar Xn} and using Proposition \ref{prop: Y map def} that for any $\Y \in \g$, $n \in \ZZ_{\geq 0}$ and $p \in \Sigma^\circ$ we have
\begin{subequations} \label{positive modes lvl 1}
\begin{align}
\X^{\; \xi_p}_{(n)} \Y_{(-1)} \vac^\xi_p &= [\X, \Y]_{(-1)} \vac^\xi_p \, \delta_{n,0} + \kappa(\X, \Y) \vac^\xi_p \delta_{n,1},\\
\bar \X^{\; \xi_p}_{(n)} \bar \Y_{(-1)} \vac^\xi_p &= \overline{[\X, \Y]}_{(-1)} \vac^\xi_p \, \delta_{n,0} + \kappa(\X, \Y) \vac^\xi_p \delta_{n,1}
\end{align}
\end{subequations}
in $\hV^{\fl,\alpha}_p$ for any local coordinate $\xi$ in the neighbourhood of $p$. Then as a simple application of Corollary \ref{cor: com nord vertex}, recalling also \eqref{modes lvl 1 state}, we recover \eqref{KM algebra vertex modes def} with $\ms k$ and $\bar{\ms k}$ both set to $1$.

\paragraph{Virasoro:} Again, by definition of the vertex modes in \eqref{endo Xn bar Xn} and using Proposition \ref{prop: Y map def}, for any $n \in \ZZ_{\geq 0}$, $p \in \Sigma^\circ$ and working in any local coordinate $\xi$ near $p$ we have
\begin{subequations} \label{Virasoro positive modes lvl 1}
\begin{align}
\Omega^{\; \xi_p}_{(n)} \Omega_{(-1)} \vac^\xi_p &= \Omega_{(-2)} \vac^\xi_p \, \delta_{n,0} + 2 \Omega_{(-1)} \vac^\xi_p \, \delta_{n,1} + \tfrac c2 \vac^\xi_p \delta_{n,3},\\
\bar\Omega^{\; \xi_p}_{(n)} \bar\Omega_{(-1)} \vac^\xi_p &= \bar\Omega_{(-2)} \vac^\xi_p \, \delta_{n,0} + 2 \bar\Omega_{(-1)} \vac^\xi_p \, \delta_{n,1} + \tfrac c2 \vac^\xi_p \delta_{n,3}.
\end{align}
\end{subequations}
As a consequence of Corollary \ref{cor: com nord vertex} we then recover the Virasoro algebra in the form \eqref{Virasoro algebra vertex modes def} with $\ms k$ and $\bar{\ms k}$ both set to $1$. To deal with the first term on the right hand sides of \eqref{Virasoro positive modes lvl 1} we use the fact that $\Omega_{(-2)} \vac = D(\Omega_{(-1)} \vac)$ and $\bar \Omega_{(-2)} \vac = \bar D(\bar \Omega_{(-1)} \vac)$ by definition of the translation operators in \eqref{D bar D def} and then apply Lemma \ref{lem: derivative vertex} and the identities \eqref{modes lvl 1 state}.

\paragraph{$\bm\beta \bm \gamma$ system:}
For any $n \in \ZZ_{\geq 0}$ and any local coordinate $\xi$ near $p \in \Sigma^\circ$ we find
\begin{subequations} \label{beta gamma positive modes lvl 1}
\begin{align}
\beta^{\; \xi_p}_{(n)} \gamma_{(-1)} \vac^\xi_p &= \vac^\xi_p \delta_{n,0},\\
\bar\beta^{\; \xi_p}_{(n)} \bar\gamma_{(-1)} \vac^\xi_p &= \vac^\xi_p \delta_{n,0}.
\end{align}
\end{subequations}
We then immediately recover the infinite-dimensional Weyl algebra relations \eqref{beta gamma algebra vertex modes def}, again with $\ms k$ and $\bar{\ms k}$ both set to $1$, by applying Corollary \ref{cor: com nord vertex}.

\subsection{Conformal and anti-conformal states} \label{sec: chiral states}

In this subsection we introduce the conformal and anti-conformal states $\Omega, \bar \Omega \in \hVV^{\fl,\alpha}$, which are respectively chiral and anti-chiral, and whose vertex $n^{\rm th}$-modes generate infinitesimal local coordinate transformations. An important use of the non-negative shifted vertex $n^{\rm th}$-modes of the conformal and anti-conformal states, namely $\Omega_{(n+1)}$ and $\bar \Omega_{(n+1)}$ for $n \geq 0$, is in establishing a generalisation of Huang's change of variable formula to the present case of full vertex operator algebras. This is the content of Corollary \ref{cor: Huang's formula} below. As an application we also show how to define an invariant bilinear form on $\hVV^{\fl,\alpha}$, in the sense of Proposition \ref{prop: inv bilinear form}.

\subsubsection{(Anti-)conformal states} \label{sec: conformal states}

We call a chiral state $\Omega \in \VV^{\fl, \alpha} \subset \hVV^{\fl, \alpha}$ \emph{conformal} if
\begin{subequations} \label{Omega positive prod}
\begin{equation} \label{Omega positive prod a}
\Omega_{(n)} \Omega = D \Omega \, \delta_{n,0} + 2 \Omega \, \delta_{n,1} + \frac{c}{2} \vac\, \delta_{n, 3}
\end{equation}
for every $n \geq 0$ and, moreover, for any $A \in \hVV^{\fl, \alpha}$ of chiral conformal dimension $\Delta_A$ we have
\begin{equation} \label{Omega positive prod b}
\Omega_{(0)} A = D A, \qquad 
\Omega_{(1)} A = \Delta_A A.
\end{equation}
\end{subequations}
In particular, it follows that $\Delta_\Omega = 2$.

Similarly, an \emph{anti-conformal state} is an anti-chiral state $\bar\Omega \in \overline{\VV}^{\fl, \alpha} \subset \hVV^{\fl, \alpha}$ such that
\begin{subequations} \label{bOmega positive prod}
\begin{equation} \label{bOmega positive prod a}
\bar\Omega_{(n)} \bar\Omega = D \bar\Omega \, \delta_{n,0} + 2 \bar\Omega \, \delta_{n,1} + \frac{c}{2} \vac\, \delta_{n, 3}
\end{equation}
for every $n \geq 0$, and for any $A \in \hVV^{\fl, \alpha}$ of anti-chiral conformal dimension $\bar\Delta_A$ we have
\begin{equation} \label{bOmega positive prod b}
\bar\Omega_{(0)} A = \bar D A, \qquad
\bar\Omega_{(1)} A = \bar\Delta_A A.
\end{equation}
\end{subequations}
In particular, $\bar\Delta_{\bar\Omega} = 2$. Because $\Omega$ is chiral and $\bar\Omega$ is anti-chiral we also have that $\Omega_{(n)} \bar \Omega = 0$ for all $n \geq 0$. The parameter $c \in \RR$ entering \eqref{bOmega positive prod a} and \eqref{bOmega positive prod b} is the \emph{central charge}.

Since the (anti-)chiral conformal dimensions of the states $\Omega, \bar\Omega \in \hVV^{\fl, \alpha}$ are different from $1$, it is convenient to work with their homogeneous vertex $n^{\rm th}$-modes introduced in \S\ref{sec: hom vert op} which we will denote by
\begin{equation} \label{Virasoro modes}
L_{(n)} \coloneqq \Omega_{[n]} = \Omega_{(n+1)}, \qquad
\bar L_{(n)} \coloneqq \bar \Omega_{[n]} = \bar \Omega_{(n+1)}
\end{equation}
for every $n \in \ZZ$. The notations $L_{(n)}$, $\bar L_{(n)} \in \End \hVV^{\fl ,\alpha}$ are slightly misleading since $L$ and $\bar L$ do not represent states in $\hVV^{\fl, \alpha}$: the states in question are $\Omega, \bar\Omega \in \hVV^{\fl, \alpha}$ and \eqref{Virasoro modes} represent their homogeneous vertex $n^{\rm th}$-modes. On the other hand, the more standard notation $L_n$, $\bar L_n$ for these endomorphisms will be reserved for the Fourier modes introduced in \S\ref{sec: Fourier modes def} below so we will keep using the notation \eqref{Virasoro modes}, hoping this will not cause any confusion. It will also be useful to introduce the notation
\begin{equation} \label{Virasoro modes xi}
L^{\xi_p}_{(n)} \coloneqq \Omega^{\xi_p}_{[n]} = \Omega^{\xi_p}_{(n+1)}, \qquad
\bar L^{\xi_p}_{(n)} \coloneqq \bar \Omega^{\xi_p}_{[n]} = \bar \Omega^{\xi_p}_{(n+1)}
\end{equation}
for any local coordinate $\xi : U \to \CC$ in a neighbourhood of a point $p \in U$.

It follows form the first equations in \eqref{Omega positive prod b} and \eqref{bOmega positive prod b} that $L_{(-1)} = D$ and $\bar L_{(-1)} = \bar D$ act as the translation operators \eqref{D bar D def}, and from the second equations in \eqref{Omega positive prod b} and \eqref{bOmega positive prod b} that $L_{(0)}$ and $\bar L_{(0)}$ measure the chiral and anti-chiral conformal dimensions of a state. It follows also from Lemma \ref{lem: field vertex modes} that for any given state $A \in \hVV^{\fl,\alpha}$ there exists a $N_A \in \ZZ_{\geq 0}$ such that $L_{(k)} A = 0$ and $\bar L_{(k)} A = 0$ for all $k \geq N_A$. Moreover, it follows from Lemma \ref{lem: chiral weights} that $L_{(k)}$ and $\bar L_{(k)}$ lower the conformal dimension when $k \geq 1$ and hence both they act locally nilpotently on $\hVV^{\fl, \alpha}$ in the sense that for every $A \in \hVV^{\fl,\alpha}$ there exists $r_A, \bar r_A \in \ZZ_{\geq 1}$ such that
\begin{equation} \label{locally nilpotent}
\big( L_{(k)} \big)^{r_A} A = 0, \qquad
\big( \bar L_{(k)} \big)^{\bar r_A} A = 0.
\end{equation}

As an application of Corollary \ref{cor: com nord vertex} using \eqref{Omega positive prod a}, \eqref{bOmega positive prod a} and Lemma \ref{lem: derivative vertex}, we deduce that $L_{(n)}, \bar L_{(n)} \in \End \hVV^{\fl, \alpha}$ generate a direct sum of two copies of the Virasoro algebra, i.e.
\begin{align*}
[ L_{(m)}, L_{(n)} ] &= (m-n) L_{(m+n)} + \frac{m^3-m}{12} \, c \,\delta_{m+n,0},\\
[ \bar L_{(m)}, \bar L_{(n)} ] &= (m-n) \bar L_{(m+n)} + \frac{m^3-m}{12} \,  c \, \delta_{m+n,0}
\end{align*}
and $[ L_{(m)}, \bar L_{(n)} ] = 0$ for any $m, n \in \ZZ$. More importantly for what follows, when $m, n \in \ZZ_{\geq -1}$ we have a direct sum of two copies of the simpler Witt algebra
\begin{equation} \label{Witt algebra}
[ L_{(m)}, L_{(n)} ] = (m-n) L_{(m+n)}, \qquad
[ \bar L_{(m)}, \bar L_{(n)} ] = (m-n) \bar L_{(m+n)}.
\end{equation}

The (anti-)conformal states in each of the three main examples from \S\ref{sec: main examples} are as follows.

\paragraph{Kac-Moody:}
For simplicity we assume here that the Lie algebra $\g$ is simple. The extension to the reductive case is straightforward, decomposing the Lie algebra $\g = \bigoplus_i \g_i \oplus \mathfrak z$ as a direct sum of its simple components $\g_i$ and its centre $\mathfrak z$. The construction of (anti-)conformal states in the more general case when $\g = \mathcal T^n \f \coloneqq \f[t]/ t^{n+1} \f[t]$ is a Takiff algebra \cite{Takiff} for any reductive Lie algebra $\f$ and any $n \in \ZZ_{\geq 0}$ can be treated similarly along the lines of \cite{Quella:2020uhk}.

Let $\{I_a\}_{a=1}^{\dim \g}$ and $\{I^a\}_{a=1}^{\dim \g}$ be dual bases of $\g$ with respect to a non-degenerate symmetric invariant bilinear form $\kappa_0 : \g \otimes \g \to \CC$. For any $\X \in \g$ we then have $\kappa_0(\X, I^b) I_b = \X$, where we use summation convention for Lie algebra indices $b = 1, \ldots, \dim \g$. Since every non-degenerate symmetric invariant bilinear form on the simple Lie algebra $\g$ is proportional to the Killing form $\kappa_\g : \g \otimes \g \to \CC$, we denote by $\frac{\kappa}{\kappa'}$ the relative coefficient of proportionality between any two non-degenerate symmetric invariant bilinear forms $\kappa, \kappa' : \g \otimes \g \to \CC$. The \emph{critical} level is defined as $\kappa_c \coloneqq - \frac 12 \kappa_\g : \g \otimes \g \to \CC$. We assume that the level $\kappa : \g \otimes \g \to \CC$ entering the definition of $\hVV^{\fl,\alpha}$ is non-critical, i.e. $\kappa \neq \kappa_c$. We then define the \emph{conformal} and \emph{anti-conformal states} of $\hVV^{\fl,\alpha}$, respectively, as
\begin{equation} \label{(anti)conformal states KM}
\conf \coloneqq \frac{\kappa_0}{2(\kappa - \kappa_c)} I_{a(-1)} I^a_{(-1)} \vac, \qquad
\bar\conf \coloneqq \frac{\kappa_0}{2(\kappa - \kappa_c)} \bar I_{a(-1)} \bar I^a_{(-1)} \vac.
\end{equation}
It is a standard fact that these satisfy the relations \eqref{Omega positive prod a} and \eqref{Omega positive prod b} with central charge 
$c \coloneqq \kappa \, \dim \g/ (\kappa - \kappa_c)$.

It is straightforward to check, for instance using the commutation relations from Corollary \ref{cor: com nord vertex} in the case where $A = \X_{(-1)} \vac$ and $B = \Omega$ along with the defining relations \eqref{KM algebra vertex modes def} in $\fl$, that the homogeneous vertex $n^{\rm th}$-modes \eqref{Virasoro modes} satisfy
\begin{align} \label{Lk comm with modes KM}
\big[ L_{(k)}, \X_{(-n)} \big] = n \,\X_{(k-n)}, \qquad
\big[ \bar L_{(k)}, \bar \X_{(-n)} \big] = n \,\bar \X_{(k-n)}.
\end{align}
for every $k, n \in \ZZ$. It then immediately follows that \eqref{Omega positive prod b} and \eqref{bOmega positive prod b} hold for any state $A \in \hVV^{\fl , \alpha}$ of definite (anti-)chiral conformal dimension.

Although the formulae \eqref{(anti)conformal states KM} for the (anti-)conformal states of $\hVV^{\fl, \alpha}$ take a different form in the case when $\g$ is reductive, the resulting relations \eqref{Lk comm with modes KM} remain true for arbirary reductive $\g$. We also expect \eqref{Lk comm with modes KM} to hold more generally in the case of a Takiff algebra $\g = \mathcal T^n \f$ for any reductive Lie algebra $\f$ and any $n \in \ZZ_{\geq 0}$.

\paragraph{Virasoro:}
We define the \emph{conformal} and \emph{anti-conformal states} of $\hVV^{\fl,\alpha}$ as
\begin{equation} \label{(anti)conformal states Vir}
\Omega \coloneqq \Omega_{(-1)} \vac, \qquad
\bar \Omega \coloneqq \bar\Omega_{(-1)} \vac.
\end{equation}
We are abusing notation slightly by using $\Omega$ to denote both the basis element of $\fl = \text{span}_\CC \{ \Omega \}$ and the above conformal state in $\hVV^{\fl,\alpha}$. Since these live in different spaces, it should be clear from the context which we mean. This abuse of notation is also justified by the identity \eqref{modes lvl 1 state} which implies that for any $n \in \ZZ$ we have $\Omega_{(n)} = \big(\Omega_{(-1)} \vac\big)_{(n)}$ with $\Omega \in \fl$, so that $\Omega_{(n)}$ denotes the same element of $\End \hVV^{\fl,\alpha}$ whether $\Omega \in \fl$ or $\Omega \in \hVV^{\fl, \alpha}$.

It follows from the defining relations \eqref{Virasoro algebra vertex modes def} in $\fl$ that \eqref{(anti)conformal states Vir} satisfy the relations \eqref{Omega positive prod a} and \eqref{Omega positive prod b} for any central charge $c$. The relations \eqref{Virasoro algebra vertex modes def}, with the central elements $\ms k$ and $\bar{\ms k}$ set to $1$ as in the definition of $\hVV^{\fl, \alpha}$ from \S\ref{sec: Vec structure}, can also be rewritten as
\begin{subequations} \label{Lk comm with modes Vir}
\begin{align}
\big[ L_{(k)}, \Omega_{(-n)} \big] &= (n+k+1) \Omega_{(k-n)} + \frac{k^3-k}{12} c \, \delta_{k-n, 1},\\
\big[ \bar L_{(k)}, \bar \Omega_{(-n)} \big] &= (n+k+1) \bar \Omega_{(k-n)} + \frac{k^3-k}{12} c \, \delta_{k-n, 1}.
\end{align}
\end{subequations}
for every $k, n \in \ZZ$. Since the terms proportional to $c$ on the right hand sides vanish when $k=-1$ or $k=0$, it follows that \eqref{Omega positive prod b} and \eqref{bOmega positive prod b} hold for any state $A \in \hVV^{\fl , \alpha}$ of definite (anti-)chiral conformal dimension.

\paragraph{$\bm \beta \bm \gamma$ system:}
Here we define the \emph{conformal} and \emph{anti-conformal states} of $\hVV^{\fl,\alpha}$ as
\begin{equation} \label{(anti)conformal states beta gamma}
\Omega = \beta_{(-1)} \gamma_{(-2)} \vac, \qquad
\bar\Omega = \bar\beta_{(-1)} \bar\gamma_{(-2)} \vac.
\end{equation}
It follows again from Corollary \ref{cor: com nord vertex} and the defining relations \eqref{beta gamma algebra vertex modes def} in $\fl$, that the homogeneous vertex $n^{\rm th}$-modes \eqref{Virasoro modes} satisfy
\begin{subequations} \label{Lk comm with modes beta gamma}
\begin{alignat}{2}
\big[ L_{(k)}, \beta_{(-n)} \big] &= n \beta_{(k-n)}, &\qquad
\big[ L_{(k)}, \gamma_{(-n)} \big] &= (n-k-1) \gamma_{(k-n)},\\
\big[ \bar L_{(k)}, \bar \beta_{(-n)} \big] &= n \bar \beta_{(k-n)}, &\qquad
\big[ \bar L_{(k)}, \bar \gamma_{(-n)} \big] &= (n-k-1) \bar \gamma_{(k-n)}
\end{alignat}
\end{subequations}
for every $k, n \in \ZZ$, and hence that \eqref{Omega positive prod b} and \eqref{bOmega positive prod b} hold for any state $A \in \hVV^{\fl , \alpha}$ of definite (anti-)chiral conformal dimension.

\subsubsection{Change of variable formula}

Given two local coordinates $\xi, \eta : U \to \CC$ on an open subset $U \subset \Sigma^\circ$, we denote the change of coordinate from $\eta$ to $\xi$ by $\varrho^{\eta \to \xi} : \xi \circ \eta^{-1} : \eta(U) \to \xi(U)$, or simply $\varrho$ for short when the coordinates $\eta$ and $\xi$ are clear from the context. Furthermore, given any point $p \in U$, we denote the corresponding change of shifted local coordinate from $\eta_p$ to $\xi_p$ by $\varrho^{\eta \to \xi}_p : \eta_p(U) \to \xi_p(U)$, or simply by $\varrho_p$ for short. Explicitly, we have
\begin{subequations} \label{b coeff definition}
\begin{equation} \label{varrho p Taylor}
\xi_p = \varrho_p(\eta_p) = \varrho\big( \eta_p + \eta(p) \big) - \xi(p) = \varrho'\big( \eta(p) \big) \Bigg(\eta_p + \sum_{n \geq 2} \frac{1}{n!} \frac{\varrho^{(n)}\big( \eta(p) \big)}{\varrho' \big( \eta(p) \big)} \eta_p^n \Bigg)
\end{equation}
where we note that $\varrho'\big( \eta(p) \big) \neq 0$ because the coordinate transformation $\varrho_p$ is invertible. Since the set of local (holomorphic) coordinate transformations fixing the origin is generated by holomorphic vector fields vanishing at $0$, i.e. $\ell^{\eta_p}_k \coloneqq - \eta_p^{k+1} \partial_{\eta_p}$ for $k \geq 0$, following \cite[\S 2.1]{Huang}, see also \cite[\S6.3.1]{FBbook}, it is natural to rewrite the above power series in the form
\begin{equation} \label{b coeff def}
\xi_p = \exp \bigg( \!\! - \sum_{k \geq 1} b_k \ell^{\eta_p}_k \bigg) b_0^{- \ell^{\eta_p}_0} \eta_p.
\end{equation}
By equating this power series term by term with the one in \eqref{varrho p Taylor} we find that the first few coefficients are given explicitly by
\begin{gather}
b_0 = \varrho' \big( \eta(p) \big), \qquad
b_1 = \frac 12 (P \varrho)\big( \eta(p) \big), \qquad
b_2 = \frac 1{3!} (S \varrho)\big( \eta(p) \big), \notag \\
\label{b coeff def c} b_3 = \frac 1{4!} \Big( (S \varrho)'\big( \eta(p) \big) - (S \varrho)\big( \eta(p) \big) (P \varrho)\big( \eta(p) \big) \Big), \qquad \ldots
\end{gather}
\end{subequations}
where $(P f)(t) \coloneqq f''(t) / f'(t)$ is the pre-Schwarzian derivative of $f$ at $t$, i.e. the logarithmic derivative of $f'$, and $(S f)(t) \coloneqq (P f)'(t) - \frac 12 (P f)(t)^2$ is the Schwarzian derivative of $f$ at $t$.

Recall from \S\ref{sec: general setup} that in the shifted local (anti-)holomorphic coordinates $\xi_p : U_+ \to \CC$ and $\bar\xi_p : U_- \to \CC$ around a point $p \in U$, see \eqref{shifted coords def}, we have the isomorphism of \dg{} vector spaces $\Omega^{0,\chain}_c\big( \pi^{-1}(U), L \big) \cong \fl_{\xi_p} \otimes \Omega^{0,\chain}_c(U_+) \oplus \fl_{\bar\xi_p} \otimes \Omega^{0,\chain}_c(U_-)$. Letting $L_X$ denote the Lie derivative of a vector field $X$ on $U$, we obtain linear maps
\begin{equation}
L_{\ell^{\xi_p}_k}, L_{\bar\ell^{\xi_p}_k} : \L^\Sigma_\alpha(U) \longrightarrow \L^\Sigma_\alpha(U)
\end{equation}
for $k \geq -1$.
These extend by the Leibniz rule to endomorphisms of $\U\L^\Sigma_\alpha(U)$ and in turn by taking limits over neighbourhoods $U \ni p$ to linear maps
\begin{equation}
L_{\ell^{\xi_p}_k}, L_{\bar\ell^{\xi_p}_k} : \hV^{\fl, \alpha}_p \longrightarrow \hV^{\fl, \alpha}_p.
\end{equation}

The following proposition is a generalisation of Proposition \ref{prop: Y deriv}, which concerned only the translation operators $L_{(-1)} = D$ and $\bar L_{(-1)} = \bar D$, to all operators $L_{(k)}$ and $\bar L_{(k)}$ for $k \geq -1$.

\begin{proposition} \label{prop: L coord change}
For any $A \in \hVV^{\fl, \alpha}$ and local coordinate $\xi$ in the neighbourhood of $p \in \Sigma^\circ$, we have the relations
\begin{equation} \label{L is coord change}
( L_{(k)} A )^\xi_p = L_{\ell^{\xi_p}_k} A^\xi_p, \qquad
( \bar L_{(k)} A )^\xi_p = L_{\bar\ell^{\xi_p}_k} A^\xi_p
\end{equation}
in $\hV^{\fl, \alpha}_p$ for every $k \in \ZZ_{\geq -1}$.
\begin{proof}
By linearity it is sufficient to consider a state $A \in \hVV^{\fl,\alpha}$ of the form \eqref{gen state Vkx}. And just as in the proof of Proposition \ref{prop: vertex operator general} we may focus on a chiral state $A = \a^r_{(-m_r)} \ldots \a^1_{(-m_1)} \vac$ since the treatment of the anti-chiral part is completely analogous. Since both sides of \eqref{L is coord change} are defined using the Leibniz rule it then suffices to show that
\begin{equation} \label{L l compare}
\big[ L_{(k)}^{\xi_p}, \a^{i\, \xi_p}_{(-m_i)} \big] = L_{\ell^{\xi_p}_k} \a^{i\, \xi_p}_{(-m_i)}
= L_{\ell^{\xi_p}_k} \Big[ s \Big( \a^i_{\xi_p} \otimes \Big(\! - \xi_p^{-m_i} \bar\partial \rho^{V^\tm_\tp}_{V^\tp_\tp} \Big) \Big) \Big]_Y,
\end{equation}
where the second step is by definition of the vertex modes in \eqref{X bar X vertex modes}.
We now consider separately the three main examples from \S\ref{sec: main examples}.

\paragraph{Kac-Moody:} Here $\a^i = \X^i \in \g$ for all $i = 1, \ldots, r$. And by \eqref{Lk comm with modes KM} we have
\begin{equation} \label{L l compare KM left}
\big[ L_{(k)}, \a^i_{(-m_i)} \big] = \big[ L_{(k)}, \X^i_{(-m_i)} \big] = m_i \,\X^i_{(k-m_i)}
\end{equation}
for the left hand side of \eqref{L l compare}. As for the right hand side, by definition of the Lie derivative we have $L_{\ell^{\xi_p}_k} \xi_p^{-m_i} = \ell^{\xi_p}_k (\xi_p^{-m_i}) = m_i \xi_p^{k-m_i}$ so this reads
\begin{align*}
L_{\ell^{\xi_p}_k} \X^{i\, \xi_p}_{(-m_i)} &= m_i \Big[ \X^i \otimes \Big(\! - \xi_p^{k-m_i} \bar\partial \rho^{V^\tm_\tp}_{V^\tp_\tp} \Big) \Big]_Y - \Big[ s \Big( \X^i \otimes \Big( \xi_p^{-m_i} \bar\partial \big( \ell^{\xi_p}_k \rho^{V^\tm_\tp}_{V^\tp_\tp} \big) \Big) \Big) \Big]_Y\\
&= m_i \,\X^{i \, \xi_p}_{(k-m_i)} - \Big[ \bar\partial \Big( s\Big( \X^i \otimes \Big( \xi_p^{-m_i} \ell^{\xi_p}_k \rho^{V^\tm_\tp}_{V^\tp_\tp} \Big) \Big) \Big) \Big]_Y.
\end{align*}
The first term on the right is exactly the vertex mode at $p$ in the coordinate $\xi$ associated with the loop generator in $\hg$ on the right hand side of \eqref{L l compare KM left}. The second term on the right hand side is the cohomology of a $\bar\partial$-exact element of $\fl_{\xi_p} \otimes \Omega^{0,1}_c(Y_+)$ which therefore vanishes.

\paragraph{Virasoro:} Now $\a^i = \Omega$ for all $i = 1, \ldots, r$ and by \eqref{Lk comm with modes Vir}, the left hand side of \eqref{L l compare} reads
\begin{align} \label{L l compare Virasoro left}
\big[ L_{(k)}, \a^i_{(-m_i)} \big] = \big[ L_{(k)}, \Omega_{(-m_i)} \big] &= (m_i+k+1) \Omega_{(k-m_i)} + \frac{k^3-k}{12} c \, \delta_{k-m_i, 1}.
\end{align}
For the right hand side of \eqref{L l compare} we recall that $\Omega_{\xi_p} = - \partial_{\xi_p}$ on which the Lie derivative of $\ell^{\xi_p}_k$ acts as $- L_{\ell^{\xi_p}_k} \partial_{\xi_p} = - (k+1) \xi_p^k \partial_{\xi_p}$. Using also $L_{\ell^{\xi_p}_k} \xi_p^{-m_i} = m_i \xi_p^{k-m_i}$ as above we find
\begin{align} \label{L l compare Virasoro right}
L_{\ell^{\xi_p}_k} \Omega^{\xi_p}_{(-m_i)}
&= (m_i+k+1) \Big[ s \Big( \Omega_{\xi_p} \otimes \Big(\! - \xi_p^{k-m_i} \bar\partial \rho^{V^\tm_\tp}_{V^\tp_\tp} \Big) \Big) \Big]_Y - \Big[ s \Big( \Omega_{\xi_p} \otimes \Big( \xi_p^{-m_i} \bar\partial \big( \ell^{\xi_p}_k \rho^{V^\tm_\tp}_{V^\tp_\tp} \big) \Big) \Big) \Big]_Y \notag\\
&\qquad\qquad\qquad\qquad\qquad\qquad\qquad + \frac{k^3-k}{24 \pi \ii} c \int_U \xi_p^{k-m_i-2} \d \xi_p \wedge \bar\partial \rho^{V^\tm_\tp}_{V^\tp_\tp}.
\end{align}
To explain the origin of the last term, recall from \S\ref{sec: main examples} that in the Virasoro case the $2$-cocycle \eqref{cocycle Virasoro} depends explicitly on the coordinate up to a $2$-coboundary \eqref{coboundary def}. And under a change of coordinate $\eta \mapsto \xi$ this induces an isomorphism of unital local Lie algebras \eqref{Phi iso change alpha} given by a shift along the central element. Here we are considering the Lie derivative action of the vector field $\ell^{\xi_p}_k$ which corresponds to the infinitesimal change of coordinate $\xi_p \mapsto \xi'_p \coloneqq \xi_p - \epsilon \xi_p^{k+1}$. But for an infinitesimal coordinate transformation $\eta \mapsto \xi = \eta + \epsilon v(\eta)$, the Schwarzian derivative to leading order in $\epsilon$ is $(S \xi)(\eta) = \epsilon v'''(\eta)$ which here takes the form $- (k^3 - k) \xi_p^{k-2}$. This explains the last term on the right hand side of \eqref{L l compare Virasoro right}. Using Lemma \ref{lem: residue integral} to evalute the integral we obtain exactly the last term on the right hand side of \eqref{L l compare Virasoro left}. As in the Kac-Moody case, the first term on the right hand side of \eqref{L l compare Virasoro right} corresponds to that of \eqref{L l compare Virasoro left} and the second term on the right hand side of \eqref{L l compare Virasoro left} vanishes.

\paragraph{$\bm \beta \bm \gamma$ system:} Here we have either $\a^i = \beta$ or $\a^i = \gamma$ for each $i = 1, \ldots, r$.

In the first case, since $\a^i_{\xi_p} = \beta_{\xi_p} = 1$ is constant, the computation works exactly as in the Kac-Moody case and we find that \eqref{L l compare} matches the first equation in \eqref{Lk comm with modes beta gamma}.

In the second case, the left hand side of \eqref{L l compare} reads
\begin{align} \label{L l compare beta gamma left}
\big[ L_{(k)}, \a^i_{(-m_i)} \big] = \big[ L_{(k)}, \gamma_{(-m_i)} \big] &= (m_i-k-1) \gamma_{(k-m_i)}
\end{align}
by virtue of the second equation in \eqref{Lk comm with modes beta gamma}. Recalling from \S\ref{sec: main examples} that we have $\a^i_{\xi_p} = \d \xi_p$ and using the fact that $L_{\ell^{\xi_p}_k} \d \xi_p = - (k+1) \xi_p^k \d \xi_p$ we deduce
\begin{align} \label{L l compare beta gamma right}
L_{\ell^{\xi_p}_k} \gamma^{\xi_p}_{(-m_i)}
&= (m_i-k-1) \Big[ s \Big( \gamma_{\xi_p} \otimes \Big(\! - \xi_p^{k-m_i} \bar\partial \rho^{V^\tm_\tp}_{V^\tp_\tp} \Big) \Big) \Big]_Y.
\end{align}
This agrees exactly with vertex mode at $p$ in the coordinate $\xi$ of the loop generator in $\hg$ on the right hand side of \eqref{L l compare beta gamma left}.
\end{proof}
\end{proposition}

We define a linear map $R^{\eta \to \xi}_p : \hVV^{\fl,\alpha} \to \hVV^{\fl,\alpha}$ by, cf. \cite[(6.3.3)]{FBbook},
\begin{equation} \label{R eta xi def}
R^{\eta \to \xi}_p \coloneqq \exp \bigg( \!\! - \sum_{k \geq 1} \big( b_k L_{(k)} + \bar b_k \bar L_{(k)} \big) \bigg) b_0^{- L_{(0)}} \bar b_0^{- \bar L_{(0)}}.
\end{equation}
This is well defined since the last two factors act diagonally as $b_0^{- \Delta_A} \bar b_0^{- \bar \Delta_A}$ on any monomial state $A$ as in \eqref{gen state Vkx}, of chiral and anti-chiral conformal dimensions $\Delta_A$ and $\bar \Delta_A$ respectively, and the exponential is well defined since $L_{(k)}, \bar L_{(k)} : \hVV^{\fl,\alpha} \to \hVV^{\fl,\alpha}$ are locally nilpotent for $k \geq 1$.

It follows immediately from the definitions \eqref{Virasoro modes}, \eqref{Virasoro modes xi} and from Lemma \ref{lem: field vertex modes} that for any state $A \in \hVV^{\fl,\alpha}$ and any $k \in \ZZ$ we have
\begin{equation} \label{L action A}
L^{\; \xi_p}_{(k)} A^\xi_p = (L_{(k)} A)^\xi_p, \qquad
\bar L^{\; \xi_p}_{(k)} A^\xi_p = (\bar L_{(k)} A)^\xi_p.
\end{equation}
It is then also convenient to introduce the linear map $\R^{\eta \to \xi}_p : \hV^{\fl,\alpha}_p \to \hV^{\fl,\alpha}_p$ defined by
\begin{equation} \label{R cal eta xi def}
\R^{\eta \to \xi}_p \coloneqq \exp \bigg( \!\! - \sum_{k \geq 1} \Big( b_k L_{(k)}^{\, \eta_p} + \bar b_k \bar L_{(k)}^{\, \eta_p} \Big) \bigg) b_0^{- L_{(0)}^{\, \eta_p}} \bar b_0^{- \bar L_{(0)}^{\, \eta_p}},
\end{equation}
with the property $\R^{\eta \to \xi}_p A^\eta_p = ( R^{\eta \to \xi}_p A )^\eta_p$ as a result of the identity \eqref{L action A}. In particular, the operator $\R^{\eta \to \xi}_p : \hV^{\fl,\alpha}_p \to \hV^{\fl,\alpha}_p$ naturally takes as input a state $A \in \hVV^{\fl,\alpha}$ prepared at $p$ in the local coordinate $\eta$ and returns another state $R^{\eta \to \xi}_p A \in \hVV^{\fl,\alpha}$ also prepared at $p$ in the same local coordinate $\eta$. This is to be contrasted with the linear map $(\cdot)^{\eta \to \xi}_p : \hV^{\fl,\alpha}_p \to \hV^{\fl,\alpha}_p$ defined in \eqref{change of preparation} which given a state $A \in \hVV^{\fl,\alpha}$ prepared at $p$ in the local coordinate $\eta$ returns the same state $A$ also prepared at $p$ but now in the new coordinate $\xi$. These two maps, in fact, coincide.

\begin{theorem} \label{thm: change of coordinate}
Let $p \in U \subset \Sigma^\circ$ be a point in a connected open subset equipped with two local coordinates $\xi , \eta : U \to \CC$. For any state $A \in \hVV^{\fl,\alpha}$ we have $A^\xi_p = \R^{\eta \to \xi}_p A^\eta_p$.
\begin{proof}
Consider the operator
\begin{equation} \label{hat varrho def}
\widehat{\varrho}^{\eta \to \xi}_p \coloneqq \exp \bigg( \!\! - \sum_{k \geq 1} \big( b_k L_{\ell^{\eta_p}_k} + \bar b_k L_{\bar \ell^{\eta_p}_k} \big) \bigg) b_0^{-L_{\ell^{\eta_p}_0}} \bar b_0^{-L_{\bar \ell^{\eta_p}_0}} : \hV^{\fl, \alpha}_p \longrightarrow \hV^{\fl, \alpha}_p.
\end{equation}
Since the Lie derivatives along the vector fields $\ell^{\eta_p}_k$ and $\bar\ell^{\eta_p}_k$ generate local (anti-)holomorphic coordinate transformations fixing the point $p$, it follows by definition of the coefficients $b_k$ and $\bar b_k$ in \eqref{b coeff definition} that \eqref{hat varrho def} implements the coordinate transformation $\eta_p \mapsto \xi_p$ in $\hV^{\fl, \alpha}_p$. Therefore
\begin{equation*}
A^\xi_p = \widehat{\varrho}^{\eta \to \xi}_p A^\eta_p = ( R^{\eta \to \xi} A)^\eta_p = \R^{\eta \to \xi} A^\eta_p
\end{equation*}
where the second equality is by definition \eqref{R eta xi def} and Proposition \ref{prop: L coord change}, while the last equality is by definition \eqref{R cal eta xi def}.
\end{proof}
\end{theorem}

The following gives an analogue of Huang's ``change of variable'' formula \cite[p.176-177]{Huang}, see also \cite[Lemma 6.5.6]{FBbook}, in the present full vertex algebra setting.
\begin{corollary} \label{cor: Huang's formula}
With the same setting as in Theorem \ref{thm: change of coordinate}, for any $A, B \in \hVV^{\fl,\alpha}$ we have
\begin{equation*}
\R^{\eta \to \xi}_p \bigg( \Big( \mathcal Y\big( A^\xi_q, \underline{\xi_p(q)} \big) B^\xi_p \Big)_p^{\xi \to \eta} \bigg)
= \mathcal Y\big( \R^{\eta \to \xi}_q A^\eta_q, \underline{\eta_p(q)} \big) \R^{\eta \to \xi}_p B^\eta_p.
\end{equation*}
\begin{remark}
The slightly awkward use of the linear map $(\cdot)^{\xi \to \eta}_p : \hV^{\fl,\alpha}_p \to \hV^{\fl,\alpha}_p$ on the left hand side ensures that this is an equality between states prepared at $p$ in the same local coordinate $\eta$. Indeed, the left hand side above could also be written simply as $\mathcal Y\big( A^\xi_q, \underline{\xi_p(q)} \big) B^\xi_p$ but then the equality in Corollary \ref{cor: Huang's formula} would be comparing a set of states prepared at $p$ in the local coordinate $\xi$ with a set of states prepared at $p$ in the local coordinate $\eta$. In other words, in terms of the formal vertex operator $Y(A, \underline{\zeta}) B$ introduced in \eqref{vertex operator Y} for any $A, B \in \hVV^{\fl,\alpha}$, the identity in Corollary \ref{cor: Huang's formula} takes the more recognisable form
\begin{equation*}
R^{\eta \to \xi}_p \Big( Y\big( A, \underline{\xi_p(q)} \big) B \Big) = Y\big( R^{\eta \to \xi}_q A, \underline{\eta_p(q)} \big) R^{\eta \to \xi}_p B. \qedhere
\end{equation*}
\end{remark}
\begin{proof}[Proof of Corollary \ref{cor: Huang's formula}.]
By Theorem \ref{thm: change of coordinate} and the definition of the linear map $(\cdot)^{\xi \to \eta}_p$, for any state $C \in \hVV^{\fl,\alpha}$ we have $C^\xi_p = \R^{\eta \to \xi}_p \big( ( C^\xi_p )_p^{\xi \to \eta} \big)$. In particular, it follows that
\begin{equation*}
\mathcal Y\big( A^\xi_q, \underline{\xi_p(q)} \big) B^\xi_p = \R^{\eta \to \xi}_p \bigg( \Big( \mathcal Y\big( A^\xi_q, \underline{\xi_p(q)} \big) B^\xi_p \Big)_p^{\xi \to \eta} \bigg).
\end{equation*}
On the other hand, we have
\begin{align*}
\mathcal Y\big( A^\xi_q, \underline{\xi_p(q)} \big) B^\xi_p &= m_{(Y, p), U} \big( m_{q, Y}(A^\xi_q) \otimes B^\xi_p \big)
= m_{(Y, p), U} \Big( m_{q, Y}\Big( \big( R^{\eta \to \xi}_q A \big)^\eta_q \Big) \otimes \big( R^{\eta \to \xi}_p B \big)^\eta_p \Big)\\
&= \mathcal Y\Big( \big( R^{\eta \to \xi}_q A \big)^\eta_q, \underline{\eta_p(q)} \Big) \big( R^{\eta \to \xi}_p B \big)^\eta_p
= \mathcal Y\big( \R^{\eta \to \xi}_q A^\eta_q, \underline{\eta_p(q)} \big) \R^{\eta \to \xi}_p B^\eta_p,
\end{align*}
where in the first and third equalities we used Proposition \ref{prop: vertex operator general} and the second equality is by Theorem \ref{thm: change of coordinate} again.
\end{proof}
\end{corollary}

\subsection{Invariant bilinear form} \label{sec: invariant bilinear}

From now on we will specialise to the case of the $2$-sphere
\begin{equation} \label{Sigma 2-sphere}
\Sigma = S^2.
\end{equation}
Since $\Sigma$ is orientable, in this case the double is a disjoint union $\Dz = \Sigma_+ \sqcup \Sigma_-$ of two copies of $\Sigma$ equipped with opposite orientations.
As a simple application of Theorem \ref{thm: change of coordinate}, in \S\ref{sec: inv bilinear def} below we will construct a canonical invariant bilinear form on the vector space $\hVV^{\fl,\alpha}$. In particular, this will be used in \S\ref{sec: Full VOA axioms} to show that $\hVV^{\fl, \alpha}$ satisfies the axioms of a full vertex algebra \cite{Moriwaki:2020cxf}.
To define this invariant bilinear form, we first need to a notion of vacuum state.

\subsubsection{Vacuum state}

Recall from \S\ref{sec: twisted PFE} that $\Sigma \in \Top(\Sigma)$. However, $\Sigma = S^2$ cannot be covered by a single coordinate patch, so Theorem \ref{thm: cohomology gSigma} does not apply to the open subset $U = \Sigma$ and in particular $\U\L^\Sigma_\alpha(\Sigma)$ is not given by the isomorphism in \eqref{Ug(U) isomorphism}. Instead we have the following result which depends on the choice of holomorphic vector bundle $L$ among the three examples in \S\ref{sec: main examples}.

\begin{lemma} \label{lem: iso correlation}
In the Kac-Moody and Virasoro cases we have an isomorphism
\begin{equation} \label{iso for sphere}
\langle \cdot \rangle : \U\L^\Sigma_\alpha(\Sigma) \overset{\cong}\longrightarrow \CC
\end{equation}
such that $\big\langle [1]_{\Sigma} \big\rangle = 1$. In the $\beta\gamma$ system case we have instead $\U\L^\Sigma_\alpha(\Sigma) \cong 0$.
\begin{proof}
We consider the three cases separately.

\paragraph{Kac-Moody:} We have $H^0\big( \L^\Sigma_\alpha(\Sigma) \big) = \g \oplus \g$ and $H^1\big( \L^\Sigma_\alpha(\Sigma) \big) = \CC$ since
\begin{equation*}
H^0\big( \g \otimes \Omega^{0,\chain}(\Dz) \big) = \g \oplus \g, \qquad
H^1\big( \g \otimes \Omega^{0,\chain}(\Dz) \big) = 0.
\end{equation*}
Consider the unital \dg{} Lie algebra $\g \oplus \g \oplus \CC{\bf 1}$ with ${\bf 1}$ of degree $1$ and equipped with the trivial Lie bracket. We then have a quasi-isomorphism of unital \dg{} Lie algebras
\begin{equation*}
\phi : \g \oplus \g \oplus \CC{\bf 1} \overset{\simeq}\longrightarrow \L^\Sigma_\alpha(\Sigma)
\end{equation*}
sending $(\X, \Y) \in \g \oplus \g$ to the constant function in $\g \otimes \Omega^{0,0}(\Dz) = \g \otimes \Omega^{0,0}(\Sigma_+) \oplus \g \otimes \Omega^{0,0}(\Sigma_-)$ equal to $\X$ on $\Sigma_+$ and $\Y$ on $\Sigma_-$, and ${\bf 1}$ to $\kent \in \L^\Sigma_\alpha(\Sigma)$.
Since the functor $\bCE_\chain$ introduced in \S\ref{sec: udgLie} preserved quasi-isomorphisms by Proposition \ref{prop: bCE functor} we obtain a quasi-isomorphism
\begin{equation*}
\bCE_\chain(\phi) : \Sym\big( (\g \oplus \g)[1] \big) \overset{\simeq}\longrightarrow \bCE_\chain\big( \L^\Sigma_\alpha(\Sigma) \big)  
\end{equation*}
where the differential $\d_{\CE}$ on the domain is trivial.
Taking the $0^{\rm th}$-cohomology then yields an isomorphism $\CC \SimTo \U\L^\Sigma_\alpha(\Sigma)$ which sends $1 \in \CC$ to $[1]_\Sigma \in \U\L^\Sigma_\alpha(\Sigma)$ and whose inverse then gives the required isomorphism \eqref{iso for sphere}.

\paragraph{Virasoro:}
We have $H^0\big( \L^\Sigma_\alpha(\Sigma) \big) = \CC^6$ and $H^1\big( \L^\Sigma_\alpha(\Sigma) \big) = \CC$ since
\begin{equation*}
H^0\big( \Omega^{0,\chain}(\Dz, T^{1,0} \Dz) \big) = \CC^6, \qquad
H^1\big( \Omega^{0,\chain}(\Dz, T^{1,0} \Dz) \big) = 0
\end{equation*}
where the $0^{\rm th}$ cohomology is given by global (anti-)holomorphic vector fields on $\CP$.
Consider therefore the unital \dg{} Lie algebra $\CC^6 \oplus \CC{\bf 1}$, with ${\bf 1}$ of degree 1, again equipped with the trivial Lie bracket.
We then have a quasi-isomorphism of unital \dg{} Lie algebras
\begin{equation*}
\phi : \CC^6 \oplus \CC{\bf 1} \overset{\simeq}\longrightarrow \L^\Sigma_\alpha(\Sigma)
\end{equation*}
which sends $(a,b,c, a',b',c') \in \CC^6$ to the pair of global (anti-)holomorphic vector fields given by $\partial_\xi \otimes (a \xi^2 + b \xi + c) \in \fl_\xi \otimes \Omega^{0,0}(\Sigma_+)$ and $\partial_{\bar\xi} \otimes (a' \bar\xi^2 + b' \bar\xi + c') \in \fl_{\bar\xi} \otimes \Omega^{0,0}(\Sigma_-)$ in any choice of holomorphic coordinate $\xi : \CC = \CP \setminus \{\infty\} \to \CC$, and ${\bf 1}$ to $\kent \in \L^\Sigma_\alpha(\Sigma)$. Applying the functor $\bCE_\chain$ we obtain a quasi-isomorphism
\begin{equation*}
\bCE_\chain(\phi) : \Sym\big( \CC^6[1] \big) \overset{\simeq}\longrightarrow \bCE_\chain\big( \L^\Sigma_\alpha(\Sigma) \big)  
\end{equation*}
where the differential $\d_{\CE}$ on the domain is again trivial so that taking the $0^{\rm th}$ cohomology we again obtain the inverse of the required isomorphism \eqref{iso for sphere}.

\paragraph{$\bm \beta \bm \gamma$ system:}
We have $H^0\big( \L^\Sigma_\alpha(\Sigma) \big) = \CC^2$ and $H^1\big( \L^\Sigma_\alpha(\Sigma) \big) = \CC^3$ coming from the fact that
\begin{equation*}
H^0\big( \Omega^{0,\chain}(\Dz, L) \big) = \CC^2, \qquad
H^1\big( \Omega^{0,\chain}(\Dz, L) \big) = \CC^2.
\end{equation*}
We then form the unital \dg{} Lie algebra $\fl \oplus \bar\fl \oplus \CC {\bf 1}$ with $\fl = \text{span}_\CC \{\beta, \gamma\}$ and where we set $\bar\fl \coloneqq \text{span}_\CC \{ \bar\beta, \bar\gamma\}$. Here $\beta$ and $\bar\beta$ of degree $0$ and $\gamma$, $\bar\gamma$ and ${\bf 1}$ all of degree $1$. The Lie bracket is given by $[\beta, \gamma] = {\bf 1} = [\bar\beta,\bar\gamma]$ and $[\beta, \bar\gamma] = 0 = [\bar\beta, \gamma]$. We have a quasi-isomorphism of unital \dg{} Lie algebras
\begin{equation*}
\phi : \fl \oplus \bar\fl \oplus \CC{\bf 1} \overset{\simeq}\longrightarrow \L^\Sigma_\alpha(\Sigma)
\end{equation*}
which sends $\beta$ and $\bar\beta$ to the pair of constant functions $1 \in \Omega^{0,0}(\Sigma_+)$ and $1 \in \Omega^{0,0}(\Sigma_-)$, and sends $\gamma$ and $\bar\gamma$ to the pair of $(1,1)$-forms $\d \xi \wedge \xi^{-1} \bar\partial\chi \in \Omega^{1,1}(\Sigma_+)$ and $\d \bar\xi \wedge \bar\xi^{-1} \bar\partial\chi' \in \Omega^{1,1}(\Sigma_-)$ where $\chi \in \Omega^{0,0}(\Sigma_+)$ is constant equal to $1$ in a neighbourhood of $\xi = 0$ and constant equal to $0$ in a neighbourhood of $\xi = \infty$, and likewise for $\chi' \in \Omega^{0,0}(\Sigma_-)$ in the holomorphic coordinate $\bar\xi$ on $\Sigma_-$.
Applying the functor $\bCE_\chain$ we obtain a quasi-isomorphism
\begin{equation*}
\bCE_\chain(\phi) : \Sym\big( (\fl \oplus \bar\fl)[1] \big) \overset{\simeq}\longrightarrow \bCE_\chain\big( \L^\Sigma_\alpha(\Sigma) \big)  
\end{equation*}
where the differential on the domain is $\d_{\CE} = \d_{[\cdot,\cdot]}$. A degree $0$ element in the domain is a linear combination of elements of the form $(s\gamma)^n (s\bar\gamma)^{\bar n}$ with $n, \bar n \in \ZZ_{\geq 0}$. But the latter is always $\d_{\CE}$-exact since we can write it, for instance, as $- \d_{\CE}\big( \frac{1}{n+1} (s\beta)(s\gamma)^{n+1} (s\bar\gamma)^{\bar n} \big)$ using the fact that $- \d_{\CE}\big( (s\beta)(s\gamma) \big) = s[\beta,\gamma] = s{\bf 1} = 1$ where the last step is from working in the quotient in the definition of the functor $\bCE_\chain$ in \eqref{bar CE functor}. The result for the $\beta\gamma$ system now follows from taking the $0^{\rm th}$ cohomology of the above quasi-isomorphism.
\end{proof}
\end{lemma}

A linear map $\U\L^\Sigma_\alpha(\Sigma) \to \CC$, as in \eqref{iso for sphere} but not necessarily an isomorphism, defines the notion of a \emph{state} for the prefactorisation algebra $\U\L^\Sigma_\alpha$ in the sense of \cite[Definition 4.9.1]{CGBook1}. We will use this below to define an invariant bilinear form on $\hVV^{\fl, \alpha}$.

Note that while Lemma \ref{lem: iso correlation} provides us with a state in both the Kac-Moody and Virasoro cases, it fails to provide a suitable notion of state for the $\beta\gamma$ system. It is clear form the proof of Lemma \ref{lem: iso correlation} that the culprit for the vanishing of $\U\L^\Sigma_\alpha(\Sigma)$ are the global holomorphic sections of $\Omega^{0,0}(\Sigma_\pm)$ and $\Omega^{1,1}(\Sigma_\pm)$ which pair non-trivially to $1$ under the differential $\d_{\CE}$. Removing these zero modes can be achieved by adding a mass term which will have the effect of modifying the global observables of the factorisation algebra $\U\L^\Sigma_\alpha$ over $\Sigma = S^2$ while preserving the observables over small open sets \cite[\S 6.1.2]{GwilliamThesis}, see also \cite[Lemma 2.5.1]{CGBook1}. We will not pursue this idea further here and from now on, when discussing the invariant bilinear form on $\hVV^{\fl, \alpha}$, we will therefore only focus on the Kac-Moody and Virasoro cases.

\subsubsection{Invariant bilinear form on \texorpdfstring{$\hVV^{\fl,\alpha}$}{F}} \label{sec: inv bilinear def}

Let us fix antipodal points $o, o'$ on $\Sigma$ and a coordinate $u : \Sigma \setminus \{ o' \} \to \CC$ with $u(o) = 0$. We also let $u^{-1} : \Sigma \setminus \{ o \} \to \CC$ be its inverse defined by $u^{-1}(p) \coloneqq u(p)^{-1}$ so that, in particular, $u^{-1}(o') = 0$.
Consider the linear map \eqref{n Fg to Ug xi i} in the case $n=2$ and $U = \Sigma$, associated with the points $o, o' \in \Sigma$ and with the local coordinates $u$ and $u^{-1}$ near these points, respectively. We obtain a linear map
\begin{equation} \label{bilinear form Fg}
\begin{array}{l}
(\cdot, \cdot) : \hVV^{\fl,\alpha} \otimes \hVV^{\fl,\alpha} \longrightarrow \CC, \\[3mm]
(B, C) \coloneqq \Big\langle m_{(o', o), \Sigma}\big( B^{u^{\tmo}}_{o'} \otimes C^u_o \big) \Big\rangle.
\end{array}
\qquad\qquad
\raisebox{-12mm}
{\begin{tikzpicture}
  \shade[ball color = lightgray, opacity = 0.5] (0,0,0) circle (1);
  \tdplotsetrotatedcoords{0}{0}{120};
 
  \fill[tdplot_rotated_coords, blue!90,
        opacity      = 0.3]
        ({1/sqrt(5)},0,{-2/sqrt(5)}) arc (0:360:{1/sqrt(5)});
  \draw[thick, blue, dashed, tdplot_rotated_coords] ({1/sqrt(5)},0,{-2/sqrt(5)}) arc (0:360:{1/sqrt(5)}) node[below right=0mm and -1mm]{\tiny $V$} node[black, below left=-2mm and 12mm]{\tiny $\Sigma$};
  
  \fill[tdplot_rotated_coords, red!90,
        opacity      = 0.4]
        ({1/sqrt(5)},0,{2/sqrt(5)}) arc (0:360:{1/sqrt(5)});
  \draw[thick, red, tdplot_rotated_coords] ({1/sqrt(5)},0,{2/sqrt(5)}) arc (0:360:{1/sqrt(5)}) node[above right=0mm and -1mm]{\tiny $V'$};

  \draw (0,0) circle (1);
  \draw[tdplot_rotated_coords, fill] (0,0,1) node[above=.7mm]{\tiny $o'$} circle (0.5pt);
  \draw[tdplot_rotated_coords, fill, black!90] (0,0,-1) node[black, below=.7mm]{\tiny $o$} circle (0.5pt);
\end{tikzpicture}
}
\end{equation}
Specifically, the state $B \in \hVV^{\fl,\alpha}$ is prepared at $o' \in V'$ in the local chart $(V', u^{-1})$ while the state $C \in \hVV^{\fl,\alpha}$ is prepared at $o \in V$ in the local chart $(V, u)$ with $V \cap V' = \emptyset$. We then apply the factorisation product $m_{(V', V), \Sigma} : \U\L^\Sigma_\alpha(V') \otimes \U\L^\Sigma_\alpha(V) \to \U\L^\Sigma_\alpha(\Sigma)$ to the element $B^{u^{\tmo}}_{V'} \otimes C^u_V$ followed by the isomorphism $\langle \cdot \rangle$ from Lemma \ref{lem: iso correlation} to obtain a complex number.

\begin{proposition} \label{prop: inv bilinear form}
In the Kac-Moody and Virasoro cases, the bilinear form \eqref{bilinear form Fg} is invariant in the sense that
\begin{equation*}
\big( B, Y(A, \underline{\zeta}) C \big) = \Big( Y\big( e^{\zeta L_{(1)} + \bar \zeta \bar L_{(1)}} (-1)^{L_{(0)} + \bar L_{(0)}} \zeta^{- 2 L_{(0)}} \bar \zeta^{- 2 \bar L_{(0)}} A, \underline{\zeta}^{-1} \big) B, C \Big)
\end{equation*}
for any states $A, B, C \in \hVV^{\fl,\alpha}$.
\begin{remark}
This identity coincides exactly with \cite[\S3.1]{Moriwaki:2020dlj} or \cite[eq. (4)]{Adamo:2024etu}, except for the fact that the operator $(-1)^{L_{(0)} + \bar L_{(0)}}$ here is replaced by $(-1)^{L_{(0)} - \bar L_{(0)}}$ there. However, since for us $\bar L_{(0)}$ acts as an integer it follows that these two operators coincide in our case. 
\end{remark}
\begin{remark}
In the $\beta\gamma$ system case the identity is vacuous as both sides vanish identically.
\end{remark}
\begin{proof}[Proof of Proposition \ref{prop: inv bilinear form}]
The change of coordinate $u^{-1} \to u$ on $U = \Sigma \setminus \{ o, o' \}$ is given simply by $\varrho(t) = t^{-1}$. Letting $p \in U$ and denoting its coordinate by $w \coloneqq u(p) \in \CC^\times$ we find that $b_0 = - w^2$ and $b_1 = - w$ while $b_n = 0$ for all $n \geq 2$. In particular, using Theorem \ref{thm: change of coordinate} and definition \eqref{R eta xi def}, for any state $A \in \hVV^{\fl,\alpha}$ we have
\begin{equation} \label{change of coord z inverse}
A^u_p = \exp \Big( w\, L_{(1)}^{\, u^{\tmo}_p} + \bar w\, \bar L_{(1)}^{\, u^{\tmo}_p} \Big) \big( \! - w^2 \big)^{- L_{(0)}^{\, u^{\tmo}_p}} \big( \! - \bar w^2 \big)^{- \bar L_{(0)}^{\, u^{\tmo}_p}} A^{u^{\tmo}}_p = \widetilde{A}^{u^\tmo}_p,
\end{equation}
where in the last step we have used \eqref{L action A} and introduced the state
\begin{equation*}
\widetilde{A} \coloneqq \exp\big( w L_{(1)} + \bar w \bar L_{(1)} \big) (-1)^{L_{(0)} + \bar L_{(0)}} w^{- 2 L_{(0)}} \bar w^{- 2 \bar L_{(0)}} A \in \hVV^{\fl,\alpha}.
\end{equation*}

Consider now the linear map \eqref{n Fg to Ug xi i} for $n=3$ distinct points $o, o', p \in \Sigma$ with $U = \Sigma$. At the points $o$ and $o'$ we still use the local coordinates $u$ and $u^{-1}$, respectively, and at the point $p$ we will also use $u$. The linear map \eqref{n Fg to Ug xi i} for $U = \Sigma$ combined with the isomorphism from Lemma \ref{lem: iso correlation} then gives
\begin{equation*}
\hVV^{\fl,\alpha} \otimes \hVV^{\fl,\alpha} \otimes \hVV^{\fl,\alpha} \longrightarrow \CC, \qquad (B, A, C) \longmapsto \Big\langle m_{(o', p, o), \Sigma}\big( B^{u^\tmo}_{o'} \otimes A_p^u \otimes C^u_o \big) \Big\rangle.
\end{equation*}
Using the associativity \eqref{PFA commutativity} of the factorisation product, we can compute the factorisation product on the right hand side by first computing the factorisation product
\begin{equation*}
m_{(p, o), V}(A^u_p \otimes C^u_o) = \mathcal Y\big( A^u_p, \underline{u_o(p)} \big) C^u_o = \big( Y(A, \underline{w}) C \big)^u_V
\end{equation*}
for some open $V \subset \Sigma \setminus \{ o' \}$ containing $p$ and $o$, where we have used the fact that $u(o) = 0$ so that $u_o(p) = w$, followed by the factorisation product $m_{(o', V), \Sigma} \big( B^{u^\tmo}_{o'} \otimes m_{(p, o), V}(A^u_p \otimes C^u_o) \big)$. By definition \eqref{bilinear form Fg} of the bilinear form on $\hVV^{\fl,\alpha}$, the result of this computation can be written simply as $\big( B, Y(A, \underline{w}) C \big)$. On the other hand, we can compute the same factorisation product by first using \eqref{change of coord z inverse} to rewrite the state $A$ prepared at $p$ in the coordinate $u$ in terms of states prepared in the coordinate $u^{-1}$. This allows us to first compute the factorisation product
\begin{equation*}
m_{(o', p), V'}(B^{u^{-1}}_{o'} \otimes A^u_p) = \mathcal Y\big( \widetilde{A}^{u^\tmo}_p, \underline{u^{-1}_{o'}(p)} \big) B^{u^\tmo}_{o'} = \big( Y\big( \widetilde{A}, \underline{w}^{-1} \big) B \big)^{u^\tmo}_{V'}
\end{equation*}
for some choice of open subset $V' \subset \Sigma \setminus \{ o \}$ containing the points $p$ and $o'$, where we have used the fact that $u^{-1}(o') = 0$ and $u^{-1}(p) = w^{-1}$. By subsequently performing the remaining factorisation product with $C^u_o$, namely $m_{(V', o), \Sigma} \big( m_{(o', p), V'}(B^{u^\tmo}_{o'} \otimes A^u_p) \otimes C^u_o \big)$, and comparing the result with the other computation described above we deduce the claim.
\end{proof}
\end{proposition}

\subsubsection{Full vertex algebra axioms} \label{sec: Full VOA axioms}

As recalled in \S\ref{sec: full VOA intro}, there are a number of mathematical formulations of the notion of full, or non-chiral, vertex operator algebras capable of describing both chiral and anti-chiral states in a full conformal field theory. Since the emphasis in each formulation is slightly different they also go by distinct names: \emph{OPE-algebras} in \cite{KapustinOrlov, Rosellen}, \emph{full field algebras} in \cite{Huang:2005gz, Kong:2006wa}, \emph{full vertex algebras} in \cite{Moriwaki:2020dlj, Moriwaki:2020cxf} and \emph{non-chiral vertex operator algebras} in \cite{Singh:2023mom}.
We now show that our geometric realisation \eqref{cal Y map defined} of the vector space $\hVV^{\fl, \alpha}$ using the prefactorisation algebra $\U\L^\Sigma_\alpha$ naturally endows it with the structure of a full vertex algebra in the sense of \cite{Moriwaki:2020dlj, Moriwaki:2020cxf}, justifying the use of the name `full vertex algebra' for the vector space $\hV^{\fl, \alpha}_p$.

Recall from \cite{Moriwaki:2020dlj, Moriwaki:2020cxf}, see also \cite{Moriwaki:2021epl}, that a full vertex algebra is an $\RR^2$-graded vector space $F = \bigoplus_{h, \bar h \in \RR^2} F_{h, \bar h}$ over $\CC$ equipped with a linear map
\begin{equation*}
Y(\--, \ul{z}) : F \longrightarrow \End(F)[[z^{\pm 1}, \bar z^{\pm 1}, |z|^\RR]], \qquad
A \longmapsto Y(A, \ul{z}) = \sum_{\substack{r, s \in \RR\\ r-s \in \ZZ}} A_{(r,s)} z^{-r-1} \bar z^{-s-1},
\end{equation*}
where the space $\End(F)[[z^{\pm 1}, \bar z^{\pm 1}, |z|^\RR]]$ consists of formal sums as above with $A_{(r,s)} \in \End(F)$,
and a non-zero element $\vac \in F_{0,0}$ satisfying the following axioms:
\begin{itemize}
    \item[FV1)] For any $A, B \in F$,
\begin{itemize}
    \item[(a)] There exists $N \in \RR$ such that $A_{(r,s)} B = 0$ for any $r \geq N$ or $s \geq N$, and
    \item[(b)] For any $H \in \RR$, the set $\big\{ (r,s) \in \RR^2 \,|\, A_{(r,s)} B \neq 0 \;\text{and}\; r+s \geq H \big\}$ is finite,
\end{itemize}
    \item[FV2)] For any $h, \bar h \in \RR$, $F_{h, \bar h} \neq 0$ implies $h - \bar h \in \ZZ$,
    \item[FV3)] For any $A \in F$, $Y(A,\ul{z}) \vac \in F[[z,\bar z]]$ and $\lim_{z \to 0} Y(A, \ul{z}) \vac = A_{(-1,-1)} \vac = A$,
    \item[FV4)] $Y(\vac, \ul{z}) = \id_F \in \End(F)$,
    \item[FV5)] For any $A, B, C \in F$ and $u \in F^\vee \coloneqq \bigoplus_{h, \bar h \in \RR} F_{h, \bar h}^\ast$ with $F_{h, \bar h}^\ast$ the dual of $F_{h, \bar h}$, the formal power series
\begin{equation*}
u\big( Y(A, \ul{z_1}) Y(B, \ul{z_2}) C \big), \qquad
u\big( Y(B, \ul{z_2}) Y(A, \ul{z_1}) C \big), \qquad
u\big( Y\big( Y(A, \ul{z_1 - z_2})B, \ul{z_2} \big) C \big)
\end{equation*}
are the expansions of the same single-valued real analytic function on
\begin{equation*}
Y_2 \coloneqq \{ (z_1, z_2) \in \CC^2 \,|\, z_1 \neq 0, z_2 \neq 0, z_1 \neq z_2 \}
\end{equation*}
in the respective regions $\{ (z_1, z_2) \in Y_2 \,|\, |z_1| > |z_2| \}$, $\{ (z_1, z_2) \in Y_2 \,|\, |z_2| > |z_1| \}$ and $\{ (z_1, z_2) \in Y_2 \,|\, |z_2| > |z_1 - z_2| \}$ of $Y_2$,
    \item[FV6)] For any $h,h',\bar h, \bar h', r, s \in \RR$, $(F_{h, \bar h})_{(r,s)} F_{h', \bar h'} \subset F_{h + h' - r - 1, \bar h + \bar h' - s - 1}$.
\end{itemize}

Recall the formal vertex operator $Y(\--, \ul{\zeta})$ defined in \eqref{vertex operator Y} using the geometric realisation \eqref{cal Y map defined} of the vector space $\hVV^{\fl, \alpha}$ in terms of the prefactorisation algebra $\U\L^\Sigma_\alpha$.

\begin{theorem} \label{thm: F is a full VOA}
In the Kac-Moody and Virasoro cases, the tuple $\big( \hVV^{\fl, \alpha}, Y(\--, \ul{\zeta}), \vac \big)$ is a full vertex algebra.
\begin{remark}
In the $\beta\gamma$ system case our proof of axiom (FV5) does not apply since we make explicit use of the isomorphism \eqref{iso for sphere} from Lemma \ref{lem: iso correlation}.
\end{remark}
\begin{proof}[Proof of Theorem \ref{thm: F is a full VOA}]
Recall from \S\ref{sec: hom vertex modes} that we consider there is a natural $\ZZ^2$-grading on the vector space $\hVV^{\fl, \alpha}$ such that $A \in (\hVV^{\fl, \alpha})_{h, \bar h}$ if $A \in \hVV^{\fl, \alpha}$ is a homogeneous state with chiral and anti-chiral conformal dimensions $\Delta_A = h$ and $\bar\Delta_A = \bar h$. Axiom (FV2) is then immediate.

Part (a) of axiom (FV1) follows from the `Moreover' part of Lemma \ref{lem: field vertex modes} and part (b) then also follows since in our case we have $A_{(n, \bar n)} B = 0$ unless $(n, \bar n) \in \ZZ^2$.

Axioms (FV3) and (FV4) follow from the second and first statements in Lemma \ref{lem: vacuum axiom}.

To see axiom (FV5), let $p, q \in \Sigma \setminus \{ o', o \}$ be distinct points. Let $u_1 \coloneqq u(q)$ and $u_2 \coloneqq u(p)$. Let $W, A, B, C \in \hVV^{\fl, \alpha}$ and consider the expression
\begin{equation*}
\mu(u_1, u_2) \coloneqq \Big\langle \Phi^{4, (u^{-1}, u, u, u)}_U\big( (o',q,p,o), (W, A, B, C) \big) \Big\rangle \in \CC
\end{equation*}
defined by using the map \eqref{moving points n coord} for $n=4$ and the tuple of coordinates $\bm \xi = (u^{-1}, u, u, u)$, and also the isomorphism \eqref{iso for sphere} from Lemma \ref{lem: iso correlation}.
By Proposition \ref{prop: vertex operator general} and the definition of the bilinear form in \eqref{bilinear form Fg}, this has the three expansions
\begin{equation*}
\big( W, Y( A, \ul{u_1}) Y( B, \ul{u_2}) C \big), \qquad
\big( W,Y( B, \ul{u_2}) Y( A, \ul{u_1}) C \big), \qquad
\big( W, Y \big( Y( A, \ul{u_1 - u_2}) B, \ul{u_2} \big) C \big)
\end{equation*}
in the respective regions of $\Conf_3(U)$ determined by $|u_1| > |u_2|$, $|u_2| > |u_1|$ and $|u_2| > |u_1 - u_2|$.

Finally, Lemma \ref{lem: chiral weights} implies that for any homogeneous states $A, B \in \hVV^{\fl, \alpha}$ and any $n, \bar n \in \ZZ$, the state $A_{(n,\bar n)} B \in \hVV^{\fl, \alpha}$ is homogeneous with chiral conformal dimension $\Delta_A + \Delta_B - n - 1$ and anti-chiral conformal dimension $\bar\Delta_A + \bar\Delta_B - \bar n - 1$, from which axiom (FV6) follows.
\end{proof}
\end{theorem}

\subsection{Reality conditions} \label{sec: unitarity}

In this subsection we construct an anti-linear involution on $\hVV^{\fl,\alpha}$ and describe its action on the vertex operator products. We then combine it with the invariant bilinear form \eqref{bilinear form Fg} to define a Hermitian sesquilinear form $\langle \cdot, \cdot \rangle$ on $\hVV^{\fl,\alpha}$.

We let $\tau : \Dz \to \Dz$ be an orientation reversing involution of $\Dz$ with respect to which the coordinate $u$ used at the end of \S\ref{sec: chiral states} has the property that
\begin{equation} \label{u reality condition}
\overline{u \circ \tau} = u^{-1}.
\end{equation}
A concrete example will be described at the start of \S\ref{sec: Fourier modes} below, where $\tau$ will correspond to the Lorentzian involution $\tau_L : \Dz \to \Dz$ from \S\ref{sec: intro}.

\subsubsection{Anti-linear involution on \texorpdfstring{$\hVV^{\fl,\alpha}$}{F}}

Recall from \S\ref{sec: g Sigma reality} that we have an anti-linear involution $\tau : \fl \SimTo \fl$. We extend it to an anti-linear involution of the untwisted Kac-Moody algebras $\hg$ and $\hbg$ defined in \S\ref{sec: Vec structure} by letting it act by complex conjugation on the coefficients of the Laurent polynomials in the second tensor factor, i.e. $x \, t^n \mapsto \bar x \, t^n$ and $x \, \bar t^n \mapsto \bar x \, \bar t^n$ for any $x \in \CC$ and $n \in \ZZ_{>0}$, where $\bar x$ denotes the complex conjugate of $x$, and fixing the central extensions ${\ms k} \mapsto {\ms k}$ and $\bar {\ms k} \mapsto \bar{\ms k}$. In turn, this induces an anti-linear involution $\hat \tau : \hVV^{\fl,\alpha} \SimTo \hVV^{\fl,\alpha}$ on the full affine vertex algebra defined by the Leibniz rule, namely it is given on a monomial state $A \in \hVV^{\fl,\alpha}$ as in \eqref{gen state Vkx} by
\begin{equation} \label{hat tau on Fg}
\hat \tau A \coloneqq (\tau \a^r)_{(-m_r)} \ldots (\tau \a^1)_{(-m_1)} \overline{(\tau \b^{\bar r})}_{(-n_{\bar r})} \ldots \overline{(\tau \b^1)}_{(-n_1)} \vac
\end{equation}
and extended by anti-linearity to all of $\hVV^{\fl,\alpha}$.

Recall also the anti-linear isomorphism of prefactorisation algebras $\hat\tau : \U\L^\Sigma_\alpha \SimTo \tau^\ast \U\L^\Sigma_\alpha$ from Proposition \ref{prop: equiv PFac}. It induces an anti-linear isomorphism $\hat\tau : (\U\L^\Sigma_\alpha)_p \SimTo (\tau^\ast \U\L^\Sigma_\alpha)_p$, i.e. a natural isomorphism of diagrams $\Top(\Sigma)_p \to \Vec_\CC$ of shape the category $\Top(\Sigma)_p$ defined in \S\ref{sec: Vec structure}, whose components are all anti-linear maps. By the universal property of limits we obtain a unique anti-linear isomorphism
\begin{equation} \label{hat tau limit at p}
\hat \tau : \lim (\U\L^\Sigma_\alpha)_p \overset{\cong}\longrightarrow \lim (\tau^\ast \U\L^\Sigma_\alpha)_p
\end{equation}
such that for every neighbourhood $U \in \Top(\Sigma)_p$ of $p \in \Sigma$ we have the commutative diagram
\begin{equation} \label{hat tau limit at p diag}
\begin{tikzcd}[column sep=18mm]
\lim (\U\L^\Sigma_\alpha)_p \arrow[r, "m_{p, U}"] \arrow[d, "\hat\tau"', dashed] & \U\L^\Sigma_\alpha(U) \arrow[d, "\hat\tau"] \\
\lim (\tau^\ast \U\L^\Sigma_\alpha)_p \arrow[r, "m_{\tau(p), \tau(U)}"'] & \U\L^\Sigma_\alpha \big(\tau(U) \big)
\end{tikzcd}
\end{equation}
We denote by $m_{\tau(p), \tau(U)}$ the canonical linear map $\lim (\tau^\ast \U\L^\Sigma_\alpha)_p \to \U\L^\Sigma_\alpha(\tau(U))$ in the bottom of this diagram since, for any inclusion of open subsets $U \subset V \subset \Sigma$, the factorisation product $\U\L^\Sigma_\alpha(\tau(U)) \to \U\L^\Sigma_\alpha(\tau(V))$ of the prefactorisation algebra $\tau^\ast \U\L^\Sigma_\alpha$ is given by $m_{\tau(U), \tau(V)}$ in terms of the factorisation products of the prefactorisation algebra $\U\L^\Sigma_\alpha$, in other words
\begin{equation*}
m_{U, V}^{\tau^\ast \U\L^\Sigma_\alpha} = m_{\tau(U), \tau(V)}^{\U\L^\Sigma_\alpha}.
\end{equation*}

Since $\tau : \Sigma_\pm \to \Sigma_\pm$ is an anti-holomorphic map, given any local holomorphic coordinate $\xi : U \to \CC$ on an open subset $U \subset \Sigma$, the coordinate $\hat\tau \xi = \overline{\xi \circ \tau} : \tau(U) \to \CC$ is also holomorphic.

\begin{lemma} \label{lem: tau Au}
For any $A \in \hVV^{\fl,\alpha}$ and a local coordinate $\xi$ in a neighbourood of $p\in \Sigma$ we have $\hat \tau \big( A^\xi_p \big) = (\hat \tau A)^{\hat\tau \xi}_{\tau(p)}$. In particular, \eqref{hat tau limit at p} induces an anti-linear isomorphism $\hat \tau : \hV^{\fl,\alpha}_p \SimTo \hV^{\fl,\alpha}_{\tau(p)}$.
\begin{proof}
By linearity we can assume that $A \in \hVV^{\fl,\alpha}$ is a monomial state and we will only consider the chiral case $A = \a^r_{(-m_r)} \ldots \a^1_{(-m_1)} \vac$ for some $\a^i \in \fl$ and $m_i \in \ZZ_{\geq 1}$ with $i \in \{ 1, \ldots, r \}$ and $r \in \ZZ_{\geq 0}$, since the treatment of the anti-chiral part is completely analogous.
For any neighbourhood $U$ of the point $p$ equipped with the local coordinate $\xi : U \to \CC$ we then have
\begin{align*}
&m_{\tau(p), \tau(U)} \big( \hat \tau \big( A^\xi_p \big) \big) = \hat \tau ( A^\xi_U )
= \Bigg[ \prod_{i=1}^r s\Big( \tau \big( \a^i_{\xi_p} \big) \otimes \hat\tau \Big( \lceil \xi_p^{-m_i} \rceil^{U_{i-1}}_{U_i} \Big) \Big) \Bigg]_{\tau(U)}\\
&\qquad\qquad = \Bigg[ \prod_{i=1}^r s\Big( (\tau \a)^i_{(\hat\tau \xi)_{\tau(p)}} \otimes \Big\lceil (\hat\tau \xi)_{\tau(p)}^{-m_i} \Big\rceil^{\tau(U_{i-1})}_{\tau(U_i)} \Big) \Bigg]_{\tau(U)}
= (\hat\tau A)^{\hat \tau \xi}_{\tau(U)} = m_{\tau(p), \tau(U)} \big( (\hat \tau A)^{\hat\tau \xi}_{\tau(p)} \big)
\end{align*}
where in the first step we have used the commutativity of the diagram \eqref{hat tau limit at p diag}. The second step is by definition of the anti-linear isomorphism $\hat \tau : \U\L^\Sigma_\alpha \SimTo \tau^\ast \U\L^\Sigma_\alpha$ in Proposition \ref{prop: equiv PFac}. In the third step we have used the fact that $\hat\tau : \Omega^{0, \chain}_c \SimTo \tau^\ast \Omega^{0, \chain}_c$ is an anti-linear isomorphism of cosheaves of commutative \dg{} algebras, which in particular commutes with $\bar\partial$, and
\begin{equation*}
\hat\tau(\xi_p^{-m_i}) = \big( \hat\tau \xi - (\hat\tau \xi)(\tau(p)) \big)^{-m_i} = \big( (\hat\tau \xi)_{\tau(p)} \big)^{-m_i}.
\end{equation*}
We have also used the fact that $\hat\tau \big( \rho^{U_{i-1}}_{U_i} \big) \in \Omega^{0,0}_c(\tau(U_i))^1_{\tau(U_{i-1})}$. In the second last step we have used the definition of $\hat\tau : \hVV^{\fl,\alpha} \SimTo \hVV^{\fl,\alpha}$ in \eqref{hat tau on Fg}.
\end{proof}
\end{lemma}

\begin{proposition}
For any $A, B \in \hVV^{\fl,\alpha}$ and $n, \bar n \in \ZZ$, we have
\begin{equation*}
\hat\tau \big( A_{(n, \bar n)} B \big) = (\hat\tau A)_{(n, \bar n)} \hat\tau B.
\end{equation*}
In other words, $\hat\tau : \hVV^{\fl,\alpha} \SimTo \hVV^{\fl,\alpha}$ is an anti-linear \emph{automorphism} of the full vertex algebra $\hVV^{\fl,\alpha}$.
\begin{proof}
Let $p, q \in U$ be two points in an open subset $U \subset \Sigma$ equipped with a local coordinate $\xi : U \to \CC$. Using Proposition \ref{prop: vertex operator general} and the definition \eqref{vertex algebra modes} of the vertex modes, we have
\begin{align*}
\hat\tau \big( m_{(q, p), U} (A^\xi_q \otimes B^\xi_p) \big)
&= \hat\tau \Big( \mathcal Y\big( A^\xi_q, \underline{\xi_p(q)} \big) B^\xi_p \Big)
= \hat\tau \bigg(\sum_{n, \bar n \in \ZZ} A^{\, \xi_p}_{(n, \bar n)} B^\xi_p \, \xi_p(q)^{-n-1} \bar \xi_p(q)^{-\bar n-1} \bigg)\\
&= \sum_{n, \bar n \in \ZZ} \hat\tau \big( A^{\, \xi_p}_{(n, \bar n)} B^\xi_p \big) \, (\hat\tau \xi)_{\tau(p)}(\tau(q))^{-n-1} (\overline{\hat\tau \xi})_{\tau(p)}(\tau(q))^{-\bar n-1}\\
&= \sum_{n, \bar n \in \ZZ} \hat\tau \Big( \big( A_{(n, \bar n)} B \big)^\xi_p \Big) \, (\hat\tau \xi)_{\tau(p)}(\tau(q))^{-n-1} (\overline{\hat\tau \xi})_{\tau(p)}(\tau(q))^{-\bar n-1}\\
&= \sum_{n, \bar n \in \ZZ} \Big( \hat\tau \big( A_{(n, \bar n)} B \big) \Big)^{\hat\tau \xi}_{\tau(p)} \, (\hat\tau \xi)_{\tau(p)}(\tau(q))^{-n-1} (\overline{\hat\tau \xi})_{\tau(p)}(\tau(q))^{-\bar n-1}
\end{align*}
where in the third step we used $\overline{\xi_p(q)} = \overline{\xi(q)} - \overline{\xi(p)} = (\hat\tau \xi)(\tau(q)) - (\hat\tau \xi)(\tau(p)) = (\hat\tau \xi)_{\tau(p)}(\tau(q))$, in the fourth step we used Lemma \ref{lem: field vertex modes} and in the last step we used Lemma \ref{lem: tau Au}.
On the other hand, we can also compute this as follows
\begin{align*}
\hat\tau \big( m_{(q, p), U} (A^\xi_q \otimes B^\xi_p) \big)
&= m_{(\tau(q), \tau(p)), \tau(U)} \Big((\hat\tau A)^{\hat\tau \xi}_{\tau(q)} \otimes (\hat\tau B)^{\hat\tau \xi}_{\tau(p)}\Big)\\
&= \mathcal Y \Big( (\hat\tau A)^{\hat\tau \xi}_{\tau(q)}, \underline{(\hat\tau \xi)_{\tau(p)} (\tau(q))} \Big) (\hat\tau B)^{\hat\tau \xi}_{\tau(p)}\\
&= \sum_{n, \bar n \in \ZZ} (\hat\tau A)^{\, (\hat\tau \xi)_{\tau(p)}}_{(n, \bar n)} (\hat\tau B)^{\hat\tau \xi}_{\tau(p)} \, (\hat\tau \xi)_{\tau(p)}(\tau(q))^{-n-1} (\overline{\hat\tau \xi})_{\tau(p)}(\tau(q))^{-\bar n-1}\\
&= \sum_{n, \bar n \in \ZZ} \Big( (\hat\tau A)_{(n, \bar n)} \hat\tau B \Big)^{\hat\tau \xi}_{\tau(p)} \, (\hat\tau \xi)_{\tau(p)}(\tau(q))^{-n-1} (\overline{\hat\tau \xi})_{\tau(p)}(\tau(q))^{-\bar n-1}
\end{align*}
where in the first step we used Proposition \ref{prop: equiv PFac} and Lemma \ref{lem: tau Au}, in the second step Proposition \ref{prop: vertex operator general}, in the third step the definition \eqref{vertex algebra modes} and in the final step we used Lemma \ref{lem: field vertex modes}.
\end{proof}
\end{proposition}

\subsubsection{Hermitian sesquilinear form on \texorpdfstring{$\hVV^{\fl,\alpha}$}{F}}

Using the bilinear map \eqref{bilinear form Fg} and the anti-linear involution $\hat\tau : \hVV^{\fl,\alpha} \SimTo \hVV^{\fl,\alpha}$ we introduce the sesquilinear form
\begin{equation} \label{Hermitian form}
\langle \cdot, \cdot \rangle \coloneqq \big( \hat\tau(\cdot), \cdot \big) : \hVV^{\fl,\alpha} \otimes \hVV^{\fl,\alpha} \longrightarrow \CC, \qquad
\langle B, C \rangle \coloneqq \Big\langle m_{(o', o), \Sigma}\big( ( \hat\tau B)^{u^\tmo}_{o'} \otimes C^u_o \big) \Big\rangle
\end{equation}
which is anti-linear in the first argument and linear in the second. The vertex algebra $\hVV^{\fl,\alpha}$ equipped with the anti-linear involution $\hat\tau : \hVV^{\fl,\alpha} \SimTo \hVV^{\fl,\alpha}$ is said to be \emph{unitary} if the sesquilinear form \eqref{Hermitian form} is positive definite \cite{Moriwaki:2020dlj, Adamo:2024etu}.

\begin{lemma} \label{lem: iso correlation reality}
The isomorphism $\langle \cdot \rangle : \U\L^\Sigma_\alpha(\Sigma) \SimTo \CC$ from Lemma \ref{lem: iso correlation}, in the Kac-Moody and Virasoro cases, is such that for any $\mathcal A \in \U\L^\Sigma_\alpha(\Sigma)$ we have $\overline{\langle \mathcal A \rangle} = \langle \hat\tau \mathcal A \rangle$.
\begin{proof}
Recall that the isomorphism from Lemma \ref{lem: iso correlation} is the inverse of $\CC \SimTo \U\L^\Sigma_\alpha(\Sigma)$, $a \mapsto [a]_\Sigma$ which is clearly $\ZZ_2$-equivariant, with the action of $\t \in \ZZ_2$ given by complex conjugation on $\CC$ and by the anti-linear isomorphism $\hat\tau : \U\L^\Sigma_\alpha(\Sigma) \SimTo \U\L^\Sigma_\alpha(\Sigma)$ from Proposition \ref{prop: equiv PFac}.
\end{proof}
\end{lemma}

\begin{proposition}
The sesquilinear form \eqref{Hermitian form} is Hermitian, in the Kac-Moody and Virasoro cases.
\begin{proof}
Let $B, C \in \hVV^{\fl,\alpha}$. Using Lemma \ref{lem: tau Au} with $p = o$ and $\xi = u$, recalling that $\tau(o) = o'$ and $\hat\tau u = u^{-1}$ by \eqref{u reality condition}, we can rewrite \eqref{Hermitian form} as
$\langle B, C \rangle = \big\langle m_{(o', o), \Sigma}\big( \hat\tau( B^u_o ) \otimes C^u_o \big) \big\rangle$.
Then
\begin{align*}
\overline{\langle B, C \rangle} = \Big\langle \hat\tau \Big( m_{(o', o), \Sigma}\big( \hat\tau (B^u_o) \otimes C^u_o \big) \Big) \Big\rangle = \Big\langle m_{(o, o'), \Sigma}\big( B^u_o \otimes \hat\tau(C^u_o) \big) \Big\rangle = \langle C, B \rangle
\end{align*}
where the first step is by definition \eqref{Hermitian form} and Lemma \ref{lem: iso correlation reality}, the second step is by Proposition \ref{prop: equiv PFac} and the last step is again by definition \eqref{Hermitian form}.
\end{proof}
\end{proposition}

In the remainder of this section we focus only on the Kac-Moody case since the statement of Proposition \eqref{prop: pairing Kac property} relates to a similar formula, in the chiral setting, for a Hermitian form on a Kac-Moody algebra module, see for instance \cite[\S9.4]{KacBook} or \cite[(11.5.1)]{KacBook}.

Recall from \S\ref{sec: Vec structure} that any state $A \in \hVV^{\fl,\alpha}$ can be written as $A = g \vac$ for a unique element $g \in U(\hg_- \oplus \hbg_-)$. The negative of the anti-linear involution $\tau : \g \SimTo \g$ defines an anti-linear \emph{anti}-involution $-\tau : \g \SimTo \g$. We extend this to an anti-linear anti-involution $(\cdot)^\dag : \hg \oplus \hbg \SimTo \hg \oplus \hbg$ of $\hg \oplus \hbg$ defined by (cf. \cite[\S\S2.7, 7.6]{KacBook} and \cite[\S4.2]{DongLin} in the chiral case, where $\omega_0$ denotes the compact anti-linear involution)
\begin{equation} \label{check tau def}
\big( \X_{(-m)} + \bar{\Y}_{(-n)} + a \, {\ms k} + b \, \bar{\ms k} \big)^\dag \coloneqq (- \tau \X)_{(m)} + (\overline{- \tau \Y})_{(n)} + \bar a \, {\ms k} + \bar b \, \bar{\ms k}
\end{equation}
for any $\X, \Y \in \g$, $m, n \in \ZZ$ and $a, b \in \CC$.
We then extend this to an anti-linear anti-involution of  $U(\hg_- \oplus \hbg_-)$ by letting $(gg')^\dag = g'^\dag g^\dag$ for any $g, g' \in U(\hg_- \oplus \hbg_-)$, cf. \cite[\S11.5]{KacBook}.

\begin{proposition} \label{prop: pairing Kac property}
For any $g, g' \in U(\hg \oplus \hbg)$ we have $\big\langle g \vac, g' \vac \big\rangle = \big\langle m_{o, \Sigma}\big( g^\dag g' \vac^u_o \big) \big\rangle$.
\begin{proof}
By anti-linearity in $g \in U(\hg \oplus \hbg)$, it is enough to consider monomial $g$.
Let $A = g \vac \in \hVV^{\fl,\alpha}$ be the monomial state in \eqref{gen state Vkx} so that $g = \X^r_{(-m_r)} \ldots \X^1_{(-m_1)} \bar \Y^r_{(-n_{\bar r})} \ldots \bar \Y^1_{(-n_1)} \in U(\hg \oplus \hbg)$, and let $C = g' \vac \in \hVV^{\fl,\alpha}$ be aribtrary.

Now by definition of \eqref{Hermitian form} we have $\big\langle g \vac, g' \vac \big\rangle = ( \hat\tau A, C )$. To compute the right hand side, note that for any open subset $o' \in V' \subset \Sigma$ we have
\begin{equation} \label{hat A computing}
(\hat\tau A)^{u^\tmo}_{V'} = \Bigg[ \prod_{i=1}^r s\Big( \tau \X^i \otimes \lceil u^{m_i} \rceil^{U_{i-1}}_{U_i} \Big) \prod_{j=1}^{\bar r} s\Big( \tau \Y^j \otimes \lceil \bar u^{n_j} \rceil^{V_{j-1}}_{V_j} \Big) \Bigg]_{V'},
\end{equation}
where recall that $o'_+ \in U_0 \Subset \ldots \Subset U_r \subset V'_+ \subset \Sigma_+$ and $o'_- \in V_0 \Subset \ldots \Subset V_{\bar r} \subset V'_- \subset \Sigma_-$ is a choice of nested sequence of open subsets. In simplifying the above expression we have used the fact that $u^{-1}(o') = 0$ so that the shifted coordinate $u^{-1}_{o'}$ is the same as the unshifted one, namely $u^{-1}_{o'} = u^{-1}$. Note that each factor in the above cohomology now has a positive power of $u$ since $(u^{-1})^{-m_i} = u^{m_i}$ and similarly for the anti-chiral part.

For $i \in \{ 0, \ldots, r \}$ we define the open subset $U'_i \coloneqq (\overline{U_i})^c$ as the complement in $\Sigma$ of the closure of $U_i$ and similarly for each $j \in \{ 0, \ldots, \bar r \}$ we define $V'_j \coloneqq (\overline{V_j})^c$. These are open subsets of some open neighbourhoods $W_\pm \subset \Sigma_\pm \setminus \{ o'_\pm \}$ of $o_\pm \in \Sigma_\pm$. In fact, they form nested sequences of open subsets $o_+ \in U'_r \Subset \ldots \Subset U'_0 \subset W_+ \subset \Sigma_+$ and $o_- \in V'_{\bar r} \Subset \ldots \Subset V'_0 \subset W_- \subset \Sigma_-$. We can then introduce smooth bump functions
\begin{equation*}
\rho_{U'_{i-1}}^{U'_i} \coloneqq 1 - \rho^{U_{i-1}}_{U_i} \in \Omega^{0,0}_c(U'_{i-1})^1_{U'_i}, \qquad
\rho_{V'_{j-1}}^{V'_j} \coloneqq 1 - \rho^{V_{j-1}}_{V_j} \in \Omega^{0,0}_c(V'_{j-1})^1_{V'_j}
\end{equation*}
for each $i \in \{ 1, \ldots, r \}$ and $j \in \{ 1, \ldots, \bar r \}$.
We may now rewrite \eqref{hat A computing} as
\begin{equation} \label{hat A computing 2}
(\hat\tau A)^{u^\tmo}_{V'} = \Bigg[ \prod_{i=1}^r s\Big( \! - \tau \X^i \otimes \lceil u^{m_i} \rceil_{U'_{i-1}}^{U'_i} \Big) \prod_{j=1}^{\bar r} s\Big( \! - \tau \Y^j \otimes \lceil \bar u^{n_j} \rceil_{V'_{j-1}}^{V'_j} \Big) \Bigg]_{V'}
\end{equation}
where we note, crucially, that an extra minus sign appeared in each term through the replacement of smooth bump functions in \eqref{up arrow def U V} since
\begin{equation*}
\bar\partial \rho^{U_{i-1}}_{U_i} = - \bar\partial \rho_{U'_{i-1}}^{U'_i}, \qquad
\bar\partial \rho^{V_{j-1}}_{V_j} = - \bar\partial \rho_{V'_{j-1}}^{V'_j}.
\end{equation*}
Recalling the definition of the vertex modes in \eqref{X bar X vertex modes}, the cohomology class \eqref{hat A computing 2} represents the product of vertex modes $(- \tau \X^i)^u_{(m_i)}$ and $(\overline{- \tau \Y^j})^u_{(n_j)}$ but in the reverse order compared to \eqref{hat A computing} due to the reversal of the nested sequences of open subsets $U'_i$ and $V'_j$ compared to $U_i$ and $V_j$. By associativity \eqref{PFA commutativity} of the factorisation product and definition \eqref{check tau def} we now have
\begin{align*}
m_{(o', o), \Sigma} \big( (\hat\tau A)^{u^\tmo}_{o'} \otimes C^u_o \big) &= m_{(V', V), \Sigma} \big( (\hat\tau A)^{u^\tmo}_{V'} \otimes C^u_V \big)\\
&= m_{(Y, V), \Sigma} \big( (g^\dag)^u_Y \otimes C^u_V \big) = \big( g^\dag C \big)^u_{\Sigma}
\end{align*}
where in the second step $Y \subset W \setminus \overline{V}$ is an annulus shaped open subset which encircles the open subset $V \ni o$. The result follows by applying the isomorphism $\langle \cdot \rangle$ from Lemma \ref{lem: iso correlation}.
\end{proof}
\end{proposition}

\section{Operator formalism for \texorpdfstring{$\hVV^{\fl,\alpha}$}{Vg}} \label{sec: Fourier modes}

In this section we use the prefactorisation algebra $\U\L^\Sigma_\alpha\in \PFac(\Sigma, \Vec_\CC^{\otimes})$ from \S\ref{sec: twisted PFE} to describe the operator formalism \cite[\S6]{CFTbook} for $\hVV^{\fl,\alpha}$. That is, we describe quantum operators associated with states in $\hVV^{\fl,\alpha}$ on the cylinder $\Sigma^\prime \coloneqq \RR \times S^1 \cong \CC / 2 \pi \ZZ$, where $\RR$ corresponds to the time direction and $S^1 = \RR / 2 \pi \ZZ$ to the compactified space direction.
It will be convenient to extend $\Sigma^\prime$ to a $2$-sphere $\Sigma = S^2$ by adding the points at `infinity' on both ends of the cylinder.
Recall that in \eqref{Sigma 2-sphere} we had specialised to the real $2$-dimensional manifold $\Sigma = S^2$ and we will thus keep doing so throughout this section.
We make the identification $\Sigma^\prime = \Sigma \setminus \{ o, o' \}$, where $o$ and $o'$ are the points added to the `bottom' and `top' of the infinite cylinder, respectively.

\subsection{Fourier modes and quantum operators} \label{sec: Fourier modes def}

\subsubsection{From the cylinder to the plane}

On a simply connected open subset $U \subset \Sigma^\prime$ we define a local holomorphic coordinate $\vartheta : U \to \CC$ by choosing a representative of the coset $\CC/ 2 \pi \ZZ$ continuously over $U$. For later convenience, we choose the orientation on $\Sigma^\prime_+$ such that $\vartheta \to \ii \infty$ as we approach the bottom of the cylinder, i.e. the point $o \in \Sigma$.
This coordinate satisfies
\begin{equation} \label{vartheta property}
\bar\vartheta = \vartheta \circ \tau,
\end{equation}
where $\tau : \Sigma^\prime \to \Sigma^\prime$ denotes the orientation reversing involution given by complex conjugation on $\Sigma^\prime = \CC / 2 \pi \ZZ$, extended to an orientation reversing involution $\tau : \Sigma \to \Sigma$ as $\tau(o) = o'$. We also define $\tau : \Dz \to \Dz$ by letting it act as $\tau : \Sigma \to \Sigma$ on each copy $\Sigma_\pm$ of $\Sigma$. This orientation reversing involution corresponds to the Lorentzian involution $\tau_L : \Dz \to \Dz$ from \S\ref{sec: intro}.

It is crucial to note that while the coordinate $\vartheta : U \to \CC$ on any simply connected open subset $U \subset \Sigma$ is defined only up to an additive multiple of $2 \pi$, the shifted local coordinate $\vartheta_p \coloneqq \vartheta - \vartheta(p) : U \to \CC$, introduced in \eqref{shifted coords def}, is unambiguously defined. This will mean that we can use the geometric realisation linear maps \eqref{cal Y map defined} or \eqref{Y at p} to prepare states in $\hVV^{\fl,\alpha}$ at any $p \in \Sigma^\prime$ using the local coordinate $\vartheta$.

However, on any open subset $W \subset \Sigma^\prime$ wrapping around the cylinder, a continuous representative $\vartheta : W \to \CC$ of the coset $\CC/2 \pi \ZZ$ over $W$ is necessarily multi-valued, and hence does not provide a well defined coordinate on $W$. Moreover, $\vartheta$ is not defined at $o, o' \in \Sigma$. We thus introduce two other holomorphic coordinates
\begin{equation} \label{u uInv charts}
u : \Sigma \setminus \{ o' \} \longrightarrow \CC, \qquad
u^{-1} : \Sigma \setminus \{ o \} \longrightarrow \CC
\end{equation}
related to the above local coordinate $\vartheta : U \to \CC$ on a simply connected open subset $U \subset \Sigma^\prime$ by the holomorphic changes of coordinate $u = e^{\ii \vartheta}$ and $u^{-1} = e^{- \ii \vartheta}$. Note, in particular, that the multi-valuedness of the local coordinate $\vartheta$ drops out from the local coordinates \eqref{u uInv charts}. These coordinates satisfy all the properties of the charts with the same name from \S\ref{sec: chiral states} and \S\ref{sec: unitarity}. For instance, $u(o) = 0$ and $u^{-1}(o') = 0$, and the property \eqref{u reality condition} is a consequence of \eqref{vartheta property}.

We summarise the above setup in the following picture
\begin{equation*}
\raisebox{-17mm}
{\begin{tikzpicture}
\def\R{1}
\def\x{5.5}
\def\y{-1.5}
  \shade[ball color = lightgray, opacity = 0.5] (0,0,0) circle (1);
  \tdplotsetrotatedcoords{0}{0}{120};
 
  \draw[thick, dashed, tdplot_rotated_coords] (1,0,0) arc (0:180:1);
  \draw[thick, tdplot_rotated_coords] (1,0,0) arc (0:-180:1) node[right=1mm]{\tiny $S^1$} node[black, below left=6mm and -3mm]{\tiny $\Sigma$};

  \draw (0,0) circle (1);
  \draw[tdplot_rotated_coords, fill] (0,0,1) node[above=.7mm]{\tiny $o'$} circle (0.5pt);
  \draw[tdplot_rotated_coords, fill, black!90] (0,0,-1) node[black, below=.7mm]{\tiny $o$} circle (0.5pt);

  \draw[blue, thick, -stealth, yshift=2mm] (1.5,0) -- node[blue, above=-.5mm]{\scriptsize $\ii \vartheta = \log u$} (4,0);
  \draw[blue, thick, stealth-, yshift=-2mm] (1.5,0) -- node[blue, below=-.5mm]{\scriptsize $u = e^{\ii \vartheta}$} (4,0);

  \fill[left color   = gray!10!black,
        right color  = gray!10!black,
        middle color = gray!10,
        shading      = axis,
        opacity      = 0.15]
    (\x+\R,\y+.5*\R) -- (\x+\R,\y+2.5*\R)  arc (360:180:\R cm and 0.3cm)
          -- (\x-\R,\y+.5*\R) arc (180:360:\R cm and 0.3cm);
  \fill[top color    = gray!90!,
        bottom color = gray!2,
        middle color = gray!30,
        shading      = axis,
        opacity      = 0.03]
    (\x,\y+2.5*\R) circle (\R cm and 0.3cm);
  \draw (\x-\R,\y+2.5*\R) -- (\x-\R,\y+.5*\R) arc (180:360:\R cm and 0.3cm)
               -- (\x+\R,\y+2.5*\R) ++ (-\R,0) circle (\R cm and 0.3cm);
  \draw[dashed] (\x-\R,\y+.5*\R) arc (180:0:\R cm and 0.3cm);

  \draw[thick, dashed] (\x-\R,\y+1.5*\R) arc (180:0:\R cm and 0.3cm);
  \draw[thick] (\x-\R,\y+1.5*\R) arc (180:360:\R cm and 0.3cm) node[below left=1.5mm and .5mm]{\tiny $S^1$};

  \draw (\x,\y+.2*\R) node[below=1mm]{\tiny $\Sigma^\prime$};
  
  \draw[thick, -stealth] (\x+\R+.2,-.1*\R) -- (\x+\R+.2, -.9*\R) node[right]{\tiny $\vartheta \to \ii \infty$};
\end{tikzpicture}
}
\end{equation*}

\begin{proposition} \label{prop: Fourier series}
Let $A \in \hVV^{\fl,\alpha}$, $q \in \Sigma^\prime$ and let $W \ni q$ be an annulus shaped open subset around $o \in \Sigma$ in the coordinate $u$. The expansion of $m_{q, W}(A^{\ii \vartheta}_q)$ in the coordinate $u(q)$, namely
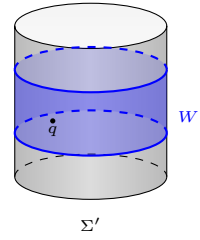
\begin{center}
\raisebox{12mm}{
\begin{minipage}[b]{0.65\linewidth}
\begin{align*}
A\big( \underline{\vartheta(q)} \big) \coloneqq \mathcal Y\big( A^{\ii \vartheta}_q, \underline{u(q)} \big) = \sum_{n, \bar n \in \ZZ} A_{n, \bar n} \, e^{-\ii n \vartheta(q)} e^{\ii \bar n \bar \vartheta(q)}, \qquad\qquad
\end{align*}
\end{minipage}
}
\begin{minipage}[b]{0.25\linewidth}
\begin{tikzpicture}
\def\R{1}
  \fill[left color   = gray!10!black,
        right color  = gray!10!black,
        middle color = gray!10,
        shading      = axis,
        opacity      = 0.15]
    (\R,.5*\R) -- (\R,2.5*\R)  arc (360:180:\R cm and 0.3cm)
          -- (-\R,.5*\R) arc (180:360:\R cm and 0.3cm);
  \fill[top color    = gray!90!,
        bottom color = gray!2,
        middle color = gray!30,
        shading      = axis,
        opacity      = 0.03]
    (0,2.5*\R) circle (\R cm and 0.3cm);
  \draw (-\R,2.5*\R) -- (-\R,.5*\R) arc (180:360:\R cm and 0.3cm)
               -- (\R,2.5*\R) ++ (-\R,0) circle (\R cm and 0.3cm);
  \draw[dashed] (-\R,.5*\R) arc (180:0:\R cm and 0.3cm);
  \fill[blue!90,
        opacity      = 0.15]
    (\R,1.08*\R) -- (\R,1.92*\R)  arc (0:180:\R cm and 0.3cm)
          -- (-\R,1.08*\R) arc (180:0:\R cm and 0.3cm);
  \fill[blue!90,
        opacity      = 0.3]
    (\R,1.08*\R) -- (\R,1.92*\R)  arc (360:180:\R cm and 0.3cm)
          -- (-\R,1.08*\R) arc (180:360:\R cm and 0.3cm);
  \draw[thick, dashed, blue] (-\R,1.08*\R) arc (180:0:\R cm and 0.3cm) node[above right, blue]{\tiny $W$};
  \draw[thick, blue] (-\R,1.08*\R) arc (180:360:\R cm and 0.3cm);
  \draw[thick, dashed, blue] (-\R,1.92*\R) arc (180:0:\R cm and 0.3cm);
  \draw[thick, blue] (-\R,1.92*\R) arc (180:360:\R cm and 0.3cm);

  \filldraw[thick] (-.5*\R,1.24*\R) node[below=-.7mm]{\tiny $q$}
circle (0.02*\R);
  \draw (0, .2*\R) node[below=1mm]{\tiny $\Sigma^\prime$};
\end{tikzpicture}
\end{minipage}
\end{center}
is a homogeneous vertex operator at $o \in \Sigma$ in the coordinate $u$. That is, its Fourier coefficients $A_{n, \bar n}$ are the homogeneous vertex modes $\wt A^{\; u}_{[n, \bar n]}$ of some state $\wt A \in \hVV^{\fl,\alpha}$ at $o \in \Sigma$. We will refer to this as the \emph{quantum operator} on $\Sigma^\prime$ associated with the state $A \in \hVV^{\fl,\alpha}$.
\begin{proof}
Let $A \in \hVV^{\fl, \alpha}$ be any state of definite (anti-)chiral conformal dimensions $\Delta_A$ and $\bar\Delta_A$. Using Theorem \ref{thm: change of coordinate} to change from the local coordinate $u$ to $\ii \vartheta = \log u$ we have
\begin{equation} \label{A theta to u}
A^{\ii \vartheta}_q = \R^{u \to \ii \vartheta}_q A^u_q = \big( R^{u \to \ii \vartheta} A \big)^u_q = u(p)^{\Delta_A} \bar u(p)^{\bar\Delta_A} \Bigg( \exp \bigg( \!\! - \sum_{k \geq 1} \big( b_k L_{(k)} + \bar b_k \bar L_{(k)} \big) \bigg) A \Bigg)^u_q
\end{equation}
where in the last step we used the explicit expression for the operator $R^{u \to \ii \vartheta}$ from \eqref{R eta xi def} and the facts that $b_0 = u(p)^{-1}$, $L_{(0)} A = \Delta_A A$ and $\bar L_{(0)} A = \bar \Delta_A A$.

Now notice that $L_{(k)}$ and $\bar L_{(k)}$ lower the (anti-)chiral conformal dimension of a state by $k$ since $[L_{(0)}, L_{(k)}] = - k L_{(k)}$ by \eqref{Witt algebra} and similarly for $\bar L_{(k)}$. Also observe that $b_k$ is proportional to $u(p)^{-k}$ for all $k \geq 1$ since
\begin{align*}
u(p)^{-1} \exp \bigg( \sum_{k \geq 1} b_k u_p^{k+1} \partial_{u_p} \bigg) u_p &= \ii \vartheta_p = \log\big( u_p + u(p) \big) - \log\big( u(p) \big)\\
&= u(p)^{-1} \bigg( u_p + \sum_{k \geq 1} \frac{(-1)^k u_p^{k+1}}{(k+1) \, u(p)^k} \bigg),
\end{align*}
where the first equality is by definition of the $b_k$ for $k \geq 1$ in \eqref{b coeff def} and using the fact that $b_0 = u(p)^{-1}$.
We then immediately deduce from \eqref{A theta to u} that we can write $A^{\ii \vartheta}_q$ as a finite sum
\begin{equation*}
A^{\ii \vartheta}_q = \sum_i u(p)^{\Delta_{B_i}} \bar u(p)^{\bar \Delta_{B_i}} (B_i)^u_q
\end{equation*}
for some states $B_i \in \hVV^{\fl, \alpha}$ of definite (anti-)chiral conformal dimensions $\Delta_{B_i}$ and $\bar\Delta_{B_i}$. Then
\begin{equation*}
\mathcal Y\big( A^{\ii \vartheta}_q, \underline{u(q)} \big) = \sum_i u(p)^{\Delta_{B_i}} \bar u(p)^{\bar \Delta_{B_i}} \mathcal Y\big( (B_i)^u_q, \underline{u(q)} \big) = \mathcal X\big( \wt A^u_q, \underline{u(q)} \big)
\end{equation*}
with $\wt A \coloneqq \sum_i B_i$, where in the last equality we used the definition of the homogeneous vertex operator in \eqref{homogeneous vertex op}. The result now follows.
\end{proof}
\end{proposition}

\subsubsection{(Anti-)chiral states}

Recall from \S\ref{sec: anti-chiral states} that the chiral vertex operators \eqref{X as X+ X-} and its anti-chiral analogue for any $\a \in \fl$ correspond to the states $\a_{(-1)}\vac^\xi_q$ and $\bar \a_{(-1)} \vac^\xi_q$, respectively, see \eqref{X as vertex operator}.
Similarly, we introduce the quantum operators associated to any $\a \in \fl$ as
\begin{align*}
\a\big( \vartheta(q) \big) &= \mathcal Y\big( \a_{(-1)} \vac^{\ii \vartheta}_q, u(q) \big) = \sum_{n \in \ZZ} \a_n e^{- \ii n \vartheta(q)},\\
\bar\a \big( \bar\vartheta(q) \big) &= \mathcal Y\big( \bar\a_{(-1)} \vac^{\ii \vartheta}_q, \bar u(q) \big) = \sum_{n \in \ZZ} \bar\a_n e^{\ii n \bar\vartheta(q)}.
\end{align*}
In other words, comparing with the general mode expansion from \eqref{prop: Fourier series}, the Fourier modes $\a_n$ and $\bar \a_n$ are defined as, cf. \eqref{modes lvl 1 state},
\begin{equation*}
\a_n \coloneqq \big( \a_{(-1)} \vac \big)_{n, 0}, \qquad
\bar \a_n \coloneqq \big( \bar \a_{(-1)} \vac \big)_{0, n}.
\end{equation*}
We can describe these Fourier modes more explicitly in each of the three examples from \S\ref{sec: main examples}.

\paragraph{Kac-Moody:}
We note using \eqref{Lk comm with modes KM} that for any $\X \in \g$ we have
\begin{equation} \label{Ln on X state}
L_{(n)} \X_{(-1)} \vac = \X_{(-1)} \vac \delta_{n,0}, \qquad \bar L_{(n)} \bar \X_{(-1)} \vac = \bar \X_{(-1)} \vac \delta_{n,0}    
\end{equation}
for all $n \geq 0$. Using this we may apply the formula \eqref{A theta to u} for the change coordinates $u \mapsto \ii \vartheta$ to the states $\X_{(-1)} \vac$ and $\bar \X_{(-1)} \vac$ to obtain
\begin{equation*}
\X_{(-1)} \vac^{\ii \vartheta}_p = u(p) \, \X_{(-1)} \vac^u_p, \qquad
\bar \X_{(-1)} \vac^{\ii \vartheta}_p = \bar u(p) \, \bar \X_{(-1)} \vac^u_p.
\end{equation*}
It then follows from \eqref{hom vertex modes def} that the quantum operators from Proposition \ref{prop: Fourier series}, associated to the (anti-)chiral states $\X_{(-1)}\vac$ and $\bar \X_{(-1)} \vac$, are respectively given by
\begin{equation} \label{quantum fields X bar X}
\X\big( \vartheta(p) \big) = \sum_{n \in \ZZ} \X_n e^{-\ii n \vartheta(p)}, \qquad
\bar \X\big( \bar \vartheta(p) \big) = \sum_{n \in \ZZ} \bar \X_n e^{\ii n \bar\vartheta(p)},
\end{equation}
where the Fourier modes are related to the homogeneous vertex modes in the coordinate $u$ as
\begin{equation*}
\X_n = \X^{\; u}_{[n]}, \qquad
\bar \X_n = \bar \X^{\; u}_{[n]}
\end{equation*}
for every $n \in \ZZ$.

\paragraph{Virasoro:}
Using the formula \eqref{A theta to u} for changing coordinates $u \mapsto \ii \vartheta$, we find
\begin{equation} \label{Omega theta to u}
\Omega^{\ii \vartheta}_p = u(p)^2 \Bigg( \exp \bigg( \!\! - \sum_{k \geq 1} \big( b_k L_{(k)} + \bar b_k \bar L_{(k)} \big) \bigg) \Omega \Bigg)^u_p
= u(p)^2 \Omega^u_p - \frac{c}{24} \vac^u_p,
\end{equation}
and similarly for the anti-chiral state $\bar\Omega$.
To see the last step note that by the defining property \eqref{Omega positive prod a} of a conformal state $\Omega$ we have
$L_{(k)} \Omega = \frac{c}{2} \vac \delta_{k,2}$ for all $k \geq 1$ and also $b_2 = \frac{1}{12} u(p)^{-2}$. Expanding both sides of \eqref{Omega theta to u} in $u(p)$ we obtain
\begin{equation*}
\sum_{n \in \ZZ} L_n e^{- \ii n \vartheta(p)} = \sum_{n \in \ZZ} \Omega^u_{[n]} e^{- \ii n \vartheta(p)} - \frac{c}{24}
\end{equation*}
where on the left hand side we used the definition of the Fourier coefficients in Proposition \ref{prop: Fourier series}, which are conventionally denoted as
$L_n \coloneqq \Omega_n$ and $\bar L_n \coloneqq \bar \Omega_n$,
and on the right hand side, for the first term we used the definition of the homogeneous vertex $n^{\rm th}$-modes, see \eqref{hom vertex modes def}, in the coordinate $u$ on the plane and for the second term the fact that $\mathcal Y\big( \vac^u_p, u(p) \big)$ is the identity operator. The above is the well known relation between the stress energy tensor on the cylinder and on the plane, see e.g. \cite[(5.138)]{CFTbook}. In particular, we have
\begin{equation*}
L_n = \Omega^u_{[n]} - \frac{c}{24} \delta_{n,0}, \qquad
\bar L_n = \bar \Omega^u_{[n]} - \frac{c}{24} \delta_{n,0}
\end{equation*}
for every $n \in \ZZ$.

\paragraph{$\bm \beta \bm \gamma$ system:}
Using \eqref{Lk comm with modes beta gamma} we find
\begin{alignat*}{2}
L_{(n)} \beta_{(-1)} \vac &= \beta_{(-1)} \vac \delta_{n,0}, &\qquad
\bar L_{(n)} \bar \beta_{(-1)} \vac &= \bar \beta_{(-1)} \vac \delta_{n,0},\\
L_{(n)} \gamma_{(-1)} \vac &= 0, &\qquad
\bar L_{(n)} \bar \gamma_{(-1)} \vac &= 0
\end{alignat*}
and therefore apply the formula \eqref{A theta to u} for the change coordinates $u \mapsto \ii \vartheta$ to the states $\beta_{(-1)} \vac$, $\bar \beta_{(-1)} \vac$, $\gamma_{(-1)} \vac$ and $\bar \gamma_{(-1)} \vac$ we find
\begin{alignat*}{2}
\beta_{(-1)} \vac^{\ii \vartheta}_p &= u(p) \, \beta_{(-1)} \vac^u_p, &\qquad
\bar \beta_{(-1)} \vac^{\ii \vartheta}_p &= \bar u(p) \, \bar \beta_{(-1)} \vac^u_p\\
\gamma_{(-1)} \vac^{\ii \vartheta}_p &= \gamma_{(-1)} \vac^u_p, &\qquad
\bar \gamma_{(-1)} \vac^{\ii \vartheta}_p &= \bar \gamma_{(-1)} \vac^u_p.
\end{alignat*}
It then follows from \eqref{hom vertex modes def} that the quantum operators from Proposition \ref{prop: Fourier series} associated to these (anti-)chiral states are given by
\begin{subequations} \label{quantum fields beta gamma}
\begin{alignat}{2}
\beta\big( \vartheta(p) \big) &= \sum_{n \in \ZZ} \beta_n e^{-\ii n \vartheta(p)}, &\qquad
\bar \beta\big( \bar \vartheta(p) \big) &= \sum_{n \in \ZZ} \bar \beta_n e^{\ii n \bar\vartheta(p)},\\
\gamma\big( \vartheta(p) \big) &= \sum_{n \in \ZZ} \gamma_n e^{-\ii n \vartheta(p)}, &\qquad
\bar \gamma\big( \bar \vartheta(p) \big) &= \sum_{n \in \ZZ} \bar \gamma_n e^{\ii n \bar\vartheta(p)}
\end{alignat}
\end{subequations}
where the Fourier modes are related to the homogeneous vertex modes in the coordinate $u$ as
\begin{equation*}
\beta_n = \beta^{\; u}_{[n]}, \qquad
\bar \beta_n = \bar \beta^{\; u}_{[n]}, \qquad
\gamma_n = \gamma^{\; u}_{[n]}, \qquad
\bar \gamma_n = \bar \gamma^{\; u}_{[n]}
\end{equation*}
for every $n \in \ZZ$.

\medskip

More generally, if $A \in \hVV^{\fl,\alpha}$ is a chiral monomial state then its associated quantum operator from Proposition \ref{prop: Fourier series} is holomorphic in this case so that
\begin{equation*}
A\big( \vartheta(q) \big) \coloneqq \mathcal Y\big( A^{\ii \vartheta}_q, u(q) \big) = \sum_{n \in \ZZ} A_n \, e^{-\ii n \vartheta(q)},
\end{equation*}
i.e. $A$ has Fourier modes $A_{n, \bar n} = A_n \delta_{\bar n, 0}$. Moreover, since $m_{q, W} (A^{\ii \vartheta}_q)$ depends holomorphically on $\vartheta(q)$ by Proposition \ref{prop: Y deriv}, we can extract the Fourier modes using a contour integral
\begin{equation} \label{Fourier modes for chiral b}
A_n = \frac{1}{2 \pi} \int_{\gamma} m_{q, W} (A^{\ii \vartheta}_q) e^{\ii n \vartheta(q)} \d\vartheta(q)
\end{equation}
along any contour $\gamma$ wrapping once around the cylinder $\Sigma^\prime$.
Likewise, if $A \in \hVV^{\fl,\alpha}$ is anti-chiral then the quantum operator from Proposition \ref{prop: Fourier series} is anti-holomorphic and we have
\begin{equation*}
A\big( \bar \vartheta(q) \big) \coloneqq \mathcal Y\big( A^{\ii \vartheta}_q, \bar u(q) \big) = \sum_{n \in \ZZ} A_n \, e^{\ii n \bar \vartheta(q)},
\end{equation*}
i.e. $A$ has Fourier modes $A_{n, \bar n} = A_{\bar n} \delta_{n, 0}$. Since in this case $m_{q, W} (A^{\ii \vartheta}_q)$ depends holomorphically on $\bar \vartheta(q)$ by Proposition \ref{prop: Y deriv}, we can extract these Fourier modes using a similar contour integral
\begin{equation*}
A_n = \frac{1}{2 \pi} \int_{\gamma} m_{q, W} (A^{\ii \vartheta}_q) e^{- \ii n \bar \vartheta(q)} \d\bar \vartheta(q).
\end{equation*}
As in the case of vertex modes discussed in \S\ref{sec: VA products}, we use the same notation for Fourier modes of chiral and anti-chiral states. This should not lead to confusion for the same reason as before, namely the type of Fourier mode in question is determined by the chirality of the state.

The vacuum state $\vac \in \hVV^{\fl,\alpha}$ is both chiral and anti-chiral, with Fourier modes given by $\vac_{n, \bar n} = \delta_{n, 0} \delta_{\bar n, 0}$ for $n, \bar n \in \ZZ$. Here we are omitting the factor of $[1]_W \in \U\L^\Sigma_\alpha(W)$ for brevity since it corresponds to the identity operator.

\subsubsection{Translation operators}

Recall the endomorphisms $D, \bar D : \hVV^{\fl,\alpha} \to \hVV^{\fl,\alpha}$ in \eqref{D bar D def} and Proposition \ref{prop: Y deriv}. The following is the analogue of Lemma \ref{lem: derivative vertex} for Fourier modes.

\begin{lemma} \label{lem: derivative Fourier}
For any $A \in \hVV^{\fl,\alpha}$ and $n, \bar n \in \ZZ$ we have
\begin{equation*}
(D A)_{n,\bar n} = - n A_{n,\bar n}, \qquad
(\bar D A)_{n,\bar n} = - \bar n A_{n,\bar n}.
\end{equation*}
\begin{proof}
This is very similar to the proof of Lemma \ref{lem: derivative vertex}. We have
\begin{align*}
\mathcal Y\big( (D A)^{\ii \vartheta}_q, \underline{u(q)} \big) &= m_{q,W}\big( (D A)^{\ii \vartheta}_q \big) = m_{q,W}( -\ii \partial_{\vartheta(q)} A^{\ii \vartheta}_q ) = -\ii \partial_{\vartheta(q)} \mathcal Y\big( A^{\ii \vartheta}_q, \underline{u(q)} \big),\\
\mathcal Y\big( (\bar D A)^{\ii \vartheta}_q, \underline{u(q)} \big) &= m_{q,W}\big( (\bar D A)^{\ii \vartheta}_q \big) = m_{q,W}( \ii \partial_{\bar \vartheta(q)} A^{\ii \vartheta}_q ) = \ii \partial_{\bar \vartheta(q)} \mathcal Y\big( A^{\ii \vartheta}_q, \underline{u(q)} \big)
\end{align*}
where the first and last steps in both cases use Proposition \ref{prop: vertex operator general} and the second steps are by Proposition \ref{prop: Y deriv}. The result follows from the definition of Fourier modes in Proposition \ref{prop: Fourier series}.
\end{proof}
\end{lemma}

\subsection{Borcherds type identities} \label{sec: Borcherds Fourier}

In Proposition \ref{prop: Borcherds} we showed that the vertex modes of states in $\hVV^{\fl,\alpha}$ satisfy the Borcherds type identities when one of the states involved is (anti-)chiral. In Theorem \ref{thm: Fourier Borcherds} below we derive similar identities for Fourier modes.

Recall from \S\ref{sec: Vec structure} that the relative ordering of the (anti-)chiral vertex modes in a state in $\hVV^{\fl,\alpha}$ prepared at a point $p \in \Sigma^\circ$ in a local coordinate $\xi$ is encoded in terms of the prefactorisation algebra $\U\L^\Sigma_\alpha$ by associating each factor with disjoint concentric annuli shaped open subsets around $p$. In the case of Fourier modes, the annuli shaped open subsets around $o \in \Sigma$ in the local coordinate $u$ correspond to open strips around the cylinder $\Sigma^\prime$, as depicted below. In particular, the centre of these annuli is $o \in \Sigma$ so the product of Fourier modes is from left to right if their corresponding open strips are ordered from top to bottom along $\Sigma'$. In pictures,
\begin{center}
\begin{tikzpicture}
\def\R{1}
  \fill[left color   = gray!10!black,
        right color  = gray!10!black,
        middle color = gray!10,
        shading      = axis,
        opacity      = 0.15]
    (\R,.5*\R) -- (\R,2.5*\R)  arc (360:180:\R cm and 0.3cm)
          -- (-\R,.5*\R) arc (180:360:\R cm and 0.3cm);
  \fill[top color    = gray!90!,
        bottom color = gray!2,
        middle color = gray!30,
        shading      = axis,
        opacity      = 0.03]
    (0,2.5*\R) circle (\R cm and 0.3cm);
  \draw (-\R,2.5*\R) -- (-\R,.5*\R) arc (180:360:\R cm and 0.3cm)
               -- (\R,2.5*\R) ++ (-\R,0) circle (\R cm and 0.3cm);
  \draw[dashed] (-\R,.5*\R) arc (180:0:\R cm and 0.3cm);

  \fill[red!90,
        opacity      = 0.15]
    (\R,\R) -- (\R,1.3*\R)  arc (0:180:\R cm and 0.3cm)
          -- (-\R,\R) arc (180:0:\R cm and 0.3cm);
  \fill[red!90,
        opacity      = 0.3]
    (\R,\R) -- (\R,1.3*\R)  arc (360:180:\R cm and 0.3cm)
          -- (-\R,\R) arc (180:360:\R cm and 0.3cm);
  \draw[thick, dashed, red] (-\R,\R) node[above left=-1.5mm and 0mm, red]{$A_{m, \bar m}$} arc (180:0:\R cm and 0.3cm);
  \draw[thick, red] (-\R,\R) arc (180:360:\R cm and 0.3cm);
  \draw[thick, dashed, red] (-\R,1.3*\R) arc (180:0:\R cm and 0.3cm);
  \draw[thick, red] (-\R,1.3*\R) arc (180:360:\R cm and 0.3cm);

  \fill[blue!90,
        opacity      = 0.15]
    (\R,1.6*\R) -- (\R,1.9*\R)  arc (0:180:\R cm and 0.3cm)
          -- (-\R,1.6*\R) arc (180:0:\R cm and 0.3cm);
  \fill[blue!90,
        opacity      = 0.3]
    (\R,1.6*\R) -- (\R,1.9*\R)  arc (360:180:\R cm and 0.3cm)
          -- (-\R,1.6*\R) arc (180:360:\R cm and 0.3cm);
  \draw[thick, dashed, blue] (-\R,1.6*\R) node[above left=-1.5mm and 0mm, blue]{$B_{n, \bar n}$} arc (180:0:\R cm and 0.3cm);
  \draw[thick, blue] (-\R,1.6*\R) arc (180:360:\R cm and 0.3cm);
  \draw[thick, dashed, blue] (-\R,1.9*\R) arc (180:0:\R cm and 0.3cm);
  \draw[thick, blue] (-\R,1.9*\R) arc (180:360:\R cm and 0.3cm);
\end{tikzpicture}
\raisebox{12mm}{$\quad\longrightarrow \quad$}
\begin{tikzpicture}
\def\R{1}
  \fill[left color   = gray!10!black,
        right color  = gray!10!black,
        middle color = gray!10,
        shading      = axis,
        opacity      = 0.15]
    (\R,.5*\R) -- (\R,2.5*\R)  arc (360:180:\R cm and 0.3cm)
          -- (-\R,.5*\R) arc (180:360:\R cm and 0.3cm);
  \fill[top color    = gray!90!,
        bottom color = gray!2,
        middle color = gray!30,
        shading      = axis,
        opacity      = 0.03]
    (0,2.5*\R) circle (\R cm and 0.3cm);
  \draw (-\R,2.5*\R) -- (-\R,.5*\R) arc (180:360:\R cm and 0.3cm)
               -- (\R,2.5*\R) ++ (-\R,0) circle (\R cm and 0.3cm);
  \draw[dashed] (-\R,.5*\R) arc (180:0:\R cm and 0.3cm);

  \fill[green!70!black,
        opacity      = 0.15]
    (\R,.8*\R) -- (\R,2.1*\R)  arc (0:180:\R cm and 0.3cm)
          -- (-\R,.8*\R) arc (180:0:\R cm and 0.3cm);
  \fill[green!70!black,
        opacity      = 0.3]
    (\R,.8*\R) -- (\R,2.1*\R)  arc (360:180:\R cm and 0.3cm)
          -- (-\R,.8*\R) arc (180:360:\R cm and 0.3cm);
  \draw[thick, dashed, green!70!black] (-\R,.8*\R) arc (180:0:\R cm and 0.3cm) node[above right=4mm and 0mm, green!70!black]{$B_{n, \bar n} A_{m, \bar m}$};
  \draw[thick, green!70!black] (-\R,.8*\R) arc (180:360:\R cm and 0.3cm);
  \draw[thick, dashed, green!70!black] (-\R,2.1*\R) arc (180:0:\R cm and 0.3cm);
  \draw[thick, green!70!black] (-\R,2.1*\R) arc (180:360:\R cm and 0.3cm);
\end{tikzpicture}
\end{center}

The N{\o}rlund polynomials $\Ber^{(a)}_n(x)$, for any $n \in \ZZ_{\geq 0}$ and $a \in \ZZ$, are defined by
\begin{equation} \label{Norlund def}
\frac{e^{x t}}{(e^t - 1)^a} = \sum_{n \geq 0} \frac{\Ber^{(a)}_n(x)}{n!} t^{n-a}.
\end{equation}
They generalise the Bernoulli polynomials $\Ber_n(x)$ for $n \in \ZZ_{\geq 0}$ which correspond to the special case $a = 1$, i.e. $\Ber_n(x) = \Ber^{(1)}_n(x)$. The values $\Ber_n(1)$ are the Bernoulli numbers.
An alternative definition of these numbers that we shall also use is $\Ber_n(0)$, which differ from the first definition only in its first term: $\Ber_1(0) = - \frac 12$ and $\Ber_1(1) = \frac 12$ while $\Ber_n(0) = \Ber_n(1)$ for all $n \geq 2$.

In Theorem \ref{thm: Fourier Borcherds} and Proposition \ref{prop: Fourier com} below we use the convention that $0^0 = 1$.

\begin{theorem} \label{thm: Fourier Borcherds}
Let $A, B \in \hVV^{\fl,\alpha}$ and $k, m, n, \bar n \in \ZZ$. We have the following identities:
\begin{itemize}
  \item[(i)] If $A$ is chiral then
\begin{align*}
&\sum_{j \geq 0} (-1)^j \binom{k}{j} A_{m-j} B_{n+j, \bar n} - \sum_{j \geq 0} (-1)^{k+j} \binom{k}{j} B_{n+k-j, \bar n} A_{m-k+j}\\
&\qquad\qquad\qquad = \sum_{j, r \geq 0} \frac{\Ber^{(-k)}_j(-k) m^r}{j! r!} \big( A_{(j+k+r)} B \big)_{m+n, \bar n}.
\end{align*}
  \item[(ii)] If $A$ is anti-chiral then
\begin{align*}
&\sum_{j \geq 0} (-1)^j \binom{k}{j} A_{m-j} B_{n, \bar n+j} - \sum_{j \geq 0} (-1)^{k+j} \binom{k}{j} B_{n, \bar n+k-j} A_{m-k+j}\\
&\qquad\qquad\qquad = \sum_{j, r \geq 0} \frac{\Ber^{(-k)}_j(-k) m^r}{j! r!} \big( A_{(j+k+r)} B \big)_{n, m+\bar n}.
\end{align*}
\end{itemize}
\begin{proof}
Suppose first that $\bar D A = 0$. Let $W \subset \Sigma^\prime$ be an open subset in $\Top(\Sigma)$ wrapping around $\Sigma^\prime$.
Recall the map in \eqref{moving points n} and define $F_{A, B} : \Conf_2(W) \to \U\L^\Sigma_\alpha(W)$ by
\begin{align*}
F_{A,B}(q,p) &\coloneqq \Phi^{2, \ii \vartheta}_W\big( (q, p), (A, B) \big) u_p(q)^k u(p)^n u(q)^{m-k} \bar u(p)^{\bar n}\\
&= m_{(q,p), W} \big( A^{\ii \vartheta}_q \otimes B^{\ii \vartheta}_p \big) u_p(q)^k u(p)^n u(q)^{m-k} \bar u(p)^{\bar n}.
\end{align*}
This depends holomorphically on $q \in W$ by Proposition \ref{prop: Y deriv}. By the standard deformation of contour argument
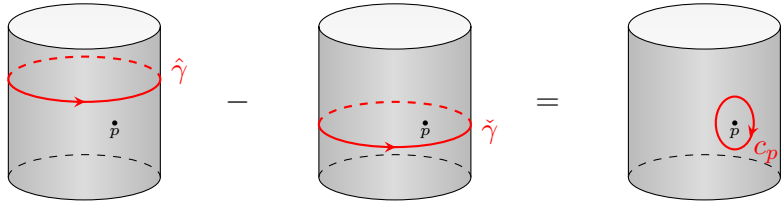
\begin{center}
\begin{tikzpicture}
\def\R{1}
  \fill[left color   = gray!10!black,
        right color  = gray!10!black,
        middle color = gray!10,
        shading      = axis,
        opacity      = 0.15]
    (\R,.5*\R) -- (\R,2.5*\R)  arc (360:180:\R cm and 0.3cm)
          -- (-\R,.5*\R) arc (180:360:\R cm and 0.3cm);
  \fill[top color    = gray!90!,
        bottom color = gray!2,
        middle color = gray!30,
        shading      = axis,
        opacity      = 0.03]
    (0,2.5*\R) circle (\R cm and 0.3cm);
  \draw (-\R,2.5*\R) -- (-\R,.5*\R) arc (180:360:\R cm and 0.3cm)
               -- (\R,2.5*\R) ++ (-\R,0) circle (\R cm and 0.3cm);
  \draw[dashed] (-\R,.5*\R) arc (180:0:\R cm and 0.3cm);
  \draw[thick, red, dashed] (-\R,1.8*\R) arc (180:0:\R cm and 0.3cm);
  \draw[thick, red] (-\R,1.8*\R) arc (180:360:\R cm and 0.3cm) node[above right=-2mm and 0mm, red]{$\hat \gamma$};
  \draw[->,>=stealth, red, thick] (0,1.5*\R) -- (0.01,1.5*\R);
  \filldraw[thick] (.4*\R,1.22*\R) node[below=-.7mm]{\tiny $p$} circle (0.02*\R);
\end{tikzpicture}
\; \raisebox{12mm}{\large $-$} \qquad
\begin{tikzpicture}
\def\R{1}
  \fill[left color   = gray!10!black,
        right color  = gray!10!black,
        middle color = gray!10,
        shading      = axis,
        opacity      = 0.15]
    (\R,.5*\R) -- (\R,2.5*\R)  arc (360:180:\R cm and 0.3cm)
          -- (-\R,.5*\R) arc (180:360:\R cm and 0.3cm);
  \fill[top color    = gray!90!,
        bottom color = gray!2,
        middle color = gray!30,
        shading      = axis,
        opacity      = 0.03]
    (0,2.5*\R) circle (\R cm and 0.3cm);
  \draw (-\R,2.5*\R) -- (-\R,.5*\R) arc (180:360:\R cm and 0.3cm)
               -- (\R,2.5*\R) ++ (-\R,0) circle (\R cm and 0.3cm);
  \draw[dashed] (-\R,.5*\R) arc (180:0:\R cm and 0.3cm);
  \draw[thick, red, dashed] (-\R,1.2*\R) arc (180:0:\R cm and 0.3cm);
  \draw[thick, red] (-\R,1.2*\R) arc (180:360:\R cm and 0.3cm) node[below right=-2mm and 0mm, red]{$\check \gamma$};
  \draw[->,>=stealth, red, thick] (0,.9*\R) -- (0.01,.9*\R);
  \filldraw[thick] (.4*\R,1.22*\R) node[below=-.7mm]{\tiny $p$} circle (0.02*\R);
\end{tikzpicture}
\; \raisebox{12mm}{\large $=$} \qquad
\begin{tikzpicture}
\def\R{1}
  \fill[left color   = gray!10!black,
        right color  = gray!10!black,
        middle color = gray!10,
        shading      = axis,
        opacity      = 0.15]
    (\R,.5*\R) -- (\R,2.5*\R)  arc (360:180:\R cm and 0.3cm)
          -- (-\R,.5*\R) arc (180:360:\R cm and 0.3cm);
  \fill[top color    = gray!90!,
        bottom color = gray!2,
        middle color = gray!30,
        shading      = axis,
        opacity      = 0.03]
    (0,2.5*\R) circle (\R cm and 0.3cm);
  \draw (-\R,2.5*\R) -- (-\R,.5*\R) arc (180:360:\R cm and 0.3cm)
               -- (\R,2.5*\R) ++ (-\R,0) circle (\R cm and 0.3cm);
  \draw[dashed] (-\R,.5*\R) arc (180:0:\R cm and 0.3cm);
  \draw[thick, red] (.4*\R,1.22*\R) ellipse (.25cm and .35cm);
  \draw[stealth-, red, thick] (.6,1.04) -- (.61,1.07);
  \filldraw[thick] (.4*\R,1.22*\R) node[below=-.7mm]{\tiny $p$} node[below right=1.5mm, red]{$c_p$} circle (0.02*\R);
\end{tikzpicture}
\end{center}
we then have the identity
\begin{equation} \label{Fourier contour arg}
\frac{1}{2 \pi} \bigg( \int_{\hat \gamma} - \int_{\check \gamma} \bigg) F_{A,B}(q,p) \d\vartheta(q) = \frac{1}{2 \pi \ii} \int_{c_p} F_{A,B}(q,p) \d\big( \ii \vartheta(q) \big),
\end{equation}
where $\hat \gamma$ and $\check \gamma$ are contours wrapping the cylinder which pass, respectively, above and below the point $p$ and $c_p$ is a small \emph{clockwise} oriented contour around $p$.

On the right hand side of \eqref{Fourier contour arg}, since $q \in c_p$ is close to $p$ we have the expansion
\begin{equation*}
F_{A,B}(q, p) = \sum_{j, r \geq 0} \frac{\Ber^{(-k)}_j(-k) m^r}{j! r!} \mathcal Y\Big( \mathcal Y\big( A^{\ii \vartheta}_q, \ii \vartheta_p(q) \big) B^{\ii \vartheta}_p, \underline{u(p)} \Big) \big( \ii \vartheta_p(q) \big)^{j+r+k} u(p)^{m+n} \bar u(p)^{\bar n},
\end{equation*}
where we used the definition \eqref{Norlund def} of the N{\o}rlund polynomials. We can perform the integral on the right hand side of \eqref{Fourier contour arg} using \eqref{chiral mode}. For this, note that the orientation of the integration variable $\ii \vartheta(q)$, see the choice of orientation on $\Sigma^\prime_+$ we made at the start of \S\ref{sec: Fourier modes def}, matches the clockwise orientation of the contour $c_p$. By definition of Fourier modes in Proposition \ref{prop: Fourier series}, the right hand side of the desired identity then corresponds to the coefficient of $u(p)^0 \bar u(p)^0$.

In the first integral on the left hand side of \eqref{Fourier contour arg}, by the relative positioning of $q \in \hat\gamma$ and the point $p$ as depicted in the above picture, we have $|u(p)| < |u(q)|$ and hence the expansion
\begin{equation*}
F_{A,B}(q, p) = \sum_{j \geq 0} (-1)^j \binom{k}{j} \mathcal Y\big( A^{\ii \vartheta}_q, u(q) \big) \mathcal Y\big( B^{\ii \vartheta}_p, \underline{u(p)} \big) u(q)^{m-j} u(p)^{n+j} \bar u(p)^{\bar n}.
\end{equation*}
Performing the integral over $\hat\gamma$ using the property \eqref{Fourier modes for chiral b} and using the definition of the Fourier modes in Proposition \ref{prop: Fourier series}, the first term on the left hand side of the desired identity is given again by the coefficient of $u(p)^0 \bar u(p)^0$.

Likewise, in the second integral on the left hand side of \eqref{Fourier contour arg} we have $q \in \check\gamma$ which lies below $p$ so that $|u(p)| > |u(q)|$ and hence we have the expansion
\begin{align*}
F_{A,B}(q, p) &= \sum_{j \geq 0} (-1)^{k+j} \binom{k}{j} \mathcal Y\big( B^{\ii \vartheta}_p, \underline{u(p)} \big) \mathcal Y\big( A^{\ii \vartheta}_q, u(q) \big) u(q)^{m-k+j} u(p)^{n+k-j} \bar u(p)^{\bar n}.
\end{align*}
Upon integrating over $\check\gamma$ and extracting the coefficient of $u(p)^0 \bar u(p)^0$ this then gives the second term on the left hand side of the desired identity, using again \eqref{Fourier modes for chiral b} and Proposition \ref{prop: Fourier series}.

The proof in the anti-chiral case $D A = 0$ is very similar.
\end{proof}
\end{theorem}

As with vertex modes, cf. Corollary \ref{cor: com nord vertex}, the `Borcherds type' identities for Fourier modes lead to commutator and normal ordering formulae for Fourier modes of composite states.

\begin{proposition} \label{prop: Fourier com}
Let $A, B \in \hVV^{\fl,\alpha}$ and $m, n, \bar n \in \ZZ$. We have the following identities:
\begin{itemize}
  \item[(i)] If $A$ is chiral then
\begin{align*}
[ A_m, B_{n, \bar n}] &= \sum_{r \geq 0} \frac{m^r}{r!} \big( A_{(r)} B \big)_{m+n, \bar n},\\
\big( A_{(-1)} B \big)_{n, \bar n} &= \sum_{j \leq 0} A_j B_{n-j, \bar n}
+ \sum_{j > 0} B_{n-j, \bar n} A_j + \sum_{r \geq 0} \frac{\zeta(-r)}{r!} \big( A_{(r)} B \big)_{n, \bar n}.
\end{align*}
  \item[(ii)] If $A$ is anti-chiral then
\begin{align*}
[ A_m, B_{n, \bar n} ] &= \sum_{r \geq 0} \frac{m^r}{r!} \big( A_{(r)} B \big)_{n, m+\bar n},\\
\big( A_{(-1)} B \big)_{n, \bar n} &= \sum_{j \leq 0} A_j B_{n, \bar n-j} + \sum_{j > 0} B_{n, \bar n-j} A_j + \sum_{r \geq 0} \frac{\zeta(-r)}{r!} \big( A_{(r)} B \big)_{n, \bar n}.
\end{align*}
\end{itemize}
\begin{proof}
Taking $k=0$ in the identities from Theorem \ref{thm: Fourier Borcherds} $(i)$ and $(ii)$ we obtain the first desired relations in $(i)$ and $(ii)$, respectively. Taking instead $m=0$ and $k=-1$ in the identities from Theorem \ref{thm: Fourier Borcherds}$(i)$ and $(ii)$ gives the second desired relations in $(i)$ and $(ii)$ upon noting the fact that $\zeta(-r) = - \frac{\Ber_{r+1}(1)}{r+1}$ for every $r \in \ZZ_{\geq 0}$.
\end{proof}
\end{proposition}

\begin{remark} \label{rem: normal ordering zeta}
The combination of the commutator and normal ordering formulae for Fourier modes in $\hVV^{\fl,\alpha}$ from Proposition \ref{prop: Fourier com} have the following interesting heuristic interpretation.

Replacing $n$ in the first relation of Proposition \ref{prop: Fourier com}$(i)$ by $n - m$, we can rewrite it as
\begin{equation*}
A_m B_{n-m, \bar n} = B_{n-m, \bar n} A_m + \sum_{r \geq 0} \frac{m^r}{r!} \big( A_{(r)} B \big)_{n, \bar n}.
\end{equation*}
Taking the formal sum of the latter over $m \in \ZZ_{> 0}$ and formally ``replacing'' the infinite sum $\sum_{m > 0} m^r$ by the $\zeta$-value $\zeta(-r)$, we see that the right hand side coincides exactly with the last two terms on the right hand side of the second relation in Proposition \ref{prop: Fourier com}$(i)$. Therefore, by virtue of the first relation in Proposition \ref{prop: Fourier com}$(i)$, the right hand side of the second relation in Proposition \ref{prop: Fourier com}$(i)$ can be interpreted as a $\zeta$-function regularisation of the formal infinite sum
$\sum_{m \in \ZZ} A_m B_{n-m, \bar n}$. In other words, we have the formal equality
\begin{equation*}
\big( A_{(-1)} B \big)_{n, \bar n} \;\; \text{``}\!=\!\text{''}\;\; \sum_{m \in \ZZ} A_m B_{n-m, \bar n}
\end{equation*}
where the non-sensical infinite sum on the right hand side, which corresponds to the Fourier $(n, \bar n)^{\rm th}$-mode of the naive product $A\big( \vartheta(q) \big) B\big( \ul{\vartheta(q)} \big)$, is given meaning by bringing all its terms into normal ordered form and using $\zeta$-function regularisation on the resulting divergent sum. The same reasoning applies to the relations in Proposition \ref{prop: Fourier com}$(ii)$.
\end{remark}

We have the following immediate applications of Proposition \ref{prop: Fourier com}.

\begin{corollary} \label{cor: Fourier com L and X}
The Fourier modes of the (anti-)conformal states satisfy
\begin{alignat*}{2}
\big[ L_m, L_n \big] &= (m-n) L_{m+n} + \frac{m^3}{12} c \, \delta_{m+n,0}, &\qquad
\big[ \bar L_m, \bar L_n \big] &= (m-n) \bar L_{m+n} + \frac{m^3}{12} c \, \delta_{m+n,0}
\end{alignat*}
and $[L_m, \bar L_n] = 0$ for $m, n \in \ZZ$.
In the Kac-Moody case we have the non-trivial commutators
\begin{alignat*}{2}
[L_m, \X_n] &= - n \, \X_{m+n}, &\qquad
\big[ \X_m, \Y_n \big] &= [\X, \Y]_{m+n} + m \, \kappa(\X, \Y) \delta_{m+n, 0},\\
[\bar L_m, \bar \X_n] &= - n \, \bar \X_{m+n}, &\qquad
\big[ \bar \X_m, \bar \Y_n \big] &= \overline{[\X, \Y]}_{m+n} + m \, \kappa(\X, \Y) \delta_{m+n, 0}
\end{alignat*}
for any $\X, \Y \in \g$ and $m, n \in \ZZ$. In the $\beta\gamma$ system case we have
\begin{alignat*}{3}
[L_m, \beta_n] &= - n \, \beta_{m+n}, &\qquad
[L_m, \gamma_n] &= - (m+n) \, \gamma_{m+n}, &\qquad
\big[ \beta_m, \gamma_n \big] &= \delta_{m+n, 0},\\
[\bar L_m, \bar \beta_n] &= - n \, \bar \beta_{m+n}, &\qquad
[\bar L_m, \bar \gamma_n] &= - (m+n) \, \bar \gamma_{m+n}, &\qquad
\big[ \bar \beta_m, \bar \gamma_n \big] &= \delta_{m+n, 0}
\end{alignat*}
\begin{proof}
These all follow immediately from Proposition \ref{prop: Fourier com}. For the first set of equations we use in particular the defining property \eqref{Omega positive prod a} of the (anti-)conformal states and Lemma \ref{lem: derivative Fourier}. For the second set of equations we use the identities \eqref{positive modes lvl 1} and the fact that
\begin{equation*}
\Omega_{(n)} \X_{(-1)} \vac = D(\X_{(-1)} \vac) \delta_{n,0} + \X_{(-1)} \vac \delta_{n,1}
\end{equation*}
for all $n \geq 0$, and similarly for the anti-chiral version, cf. \eqref{Ln on X state}. For the third set of equations we use the identities \eqref{beta gamma positive modes lvl 1} and the fact that
\begin{align*}
\Omega_{(n)} \beta_{(-1)} \vac &= D(\beta_{(-1)} \vac) \delta_{n,0} + \beta_{(-1)} \vac \delta_{n,1},\\
\Omega_{(n)} \gamma_{(-1)} \vac &= D(\gamma_{(-1)} \vac) \delta_{n,0}
\end{align*}
for all $n \geq 0$, and similarly for the anti-chiral versions.
\end{proof}
\end{corollary}

\subsection{Reality conditions} \label{sec: reality conditions}

In this section we discuss reality conditions on quantum operators defined in Proposition \ref{prop: Fourier series}. Specifically, we introduce a natural notion of adjoint operator, defined by the action of the anti-linear isomorphism $\hat\tau$ of $\U\L^\Sigma_\alpha$ from Proposition \ref{prop: equiv PFac} and show in Proposition \ref{prop: adjoint operator} that it matches the usual adjoint with respect to the Hermitian sesquilinear form introduced in \S\ref{sec: unitarity}.

\begin{lemma} \label{lem: Y equiv}
For any $A \in \hVV^{\fl,\alpha}$ and $p \in \Sigma^\prime$ we have $\hat \tau \big( A^{\ii \vartheta}_p \big) = (-1)^{\Delta_A + \bar \Delta_A} (\hat \tau A)^{\ii \vartheta}_{\tau(p)}$.
\begin{proof}
Using Lemma \ref{lem: tau Au} and the fact that $\hat\tau(\ii \vartheta) = - \ii \vartheta$ which follows from \eqref{vartheta property}, we deduce $\hat \tau \big( A^{\ii \vartheta}_p \big) = (\hat \tau A)^{- \ii \vartheta}_{\tau(p)}$. By the explicit expression \eqref{cal Y map defined b} for $A^{- \ii \vartheta}_{\tau(p)}$, using the fact that
\begin{align*}
(\tau \a)^i_{-\ii \vartheta_{\tau(p)}} \otimes \big\lceil (-\ii \vartheta_{\tau(p)})^{-m_i} \big\rceil^{U_{i-1}}_{U_i} &= (-1)^{m_i - \Delta_\a + 1} (\tau \a)^i_{\ii \vartheta_{\tau(p)}} \otimes \big\lceil (\ii \vartheta_{\tau(p)})^{-m_i} \big\rceil^{U_{i-1}}_{U_i}, \\
(\tau \b)^j_{\ii \bar\vartheta_{\tau(p)}} \otimes \big\lceil (\ii \bar \vartheta_{\tau(p)})^{-n_j} \big\rceil^{V_{j-1}}_{V_j} &= (-1)^{n_j - \Delta_\b + 1} (\tau \b)^j_{-\ii \bar\vartheta_{\tau(p)}} \otimes \big\lceil (- \ii \bar \vartheta_{\tau(p)})^{-n_j} \big\rceil^{V_{j-1}}_{V_j}
\end{align*}
for $i \in \{1, \ldots, r\}$ and $j \in \{1, \ldots, \bar r \}$, the result now follows.
\end{proof}
\end{lemma}

Given any $\mathcal O \in \U\L^\Sigma_\alpha(W)$ on an annulus shaped open subset $W \subset \Sigma^\prime$ around $o \in \Sigma$ in the coordinate $u$,
we say that it acts on $\hV^{\fl,\alpha}_o$ if it induces an endomorphism
\begin{equation*}
\mathcal O : \hV^{\fl,\alpha}_o \longrightarrow \hV^{\fl,\alpha}_o, \qquad
B_o^u \longmapsto (\mathcal O \, B)^u_o
\end{equation*}
where the state $\mathcal O \, B \in \hVV^{\fl,\alpha}$ is defined by the factorisation product $m_{(W, o), U} ( \mathcal O \otimes B^u_o) = ( \mathcal O \, B )^u_U$ into a larger open subset $U \subset \Sigma$ of $o \in \Sigma$. We then obtain an endomorphism $\mathcal O : \hVV^{\fl,\alpha} \to \hVV^{\fl,\alpha}$. For instance, the vertex modes $A^u_{(n,\bar n)}$ for $n ,\bar n \in \ZZ$ of any state $A \in \hVV^{\fl,\alpha}$ act on $\hV^{\fl,\alpha}_o$ by Lemma \ref{lem: field vertex modes}.
We define the adjoint of $\mathcal O \in \U\L^\Sigma_\alpha(W)$ as
\begin{equation} \label{adjoint operator}
\mathcal O^\dag \coloneqq \hat\tau(\mathcal O) \in \U\L^\Sigma_\alpha\big( \tau(W) \big),
\end{equation}
where we can also view $\tau(W) \subset \Sigma^\prime$ as an annulus shaped open subset around $o \in \Sigma$ in the coordinate $u$. In particular, $\mathcal O^\dag$ then also acts on $\hV^{\fl,\alpha}_o$.

Recall the quantum operator $A\big( \underline{\vartheta(p)} \big) \in \U\L^\Sigma_\alpha(W)$ associated to a state $A \in \hVV^{\fl,\alpha}$, as defined in Proposition \ref{prop: Fourier series}. Its adjoint is given by the following.

\begin{proposition} \label{prop: Fourier equiv}
For a monomial state $A \in \hVV^{\fl,\alpha}$ we have
\begin{equation} \label{adjoint quantum operator}
A\big( \underline{\vartheta(p)} \big)^\dag = (-1)^{\Delta_A + \bar\Delta_A} (\hat\tau A)\big( \underline{\bar \vartheta(p)} \big) \in \U\L^\Sigma_\alpha\big( \tau(W) \big),
\end{equation}
or equivalently, for any $n, \bar n \in \ZZ$ we have $A_{n, \bar n}^\dag = (-1)^{\Delta_A + \bar\Delta_A} (\hat \tau A)_{- n, - \bar n}$.
\begin{proof}
By definition \eqref{adjoint operator} of the adjoint we have
\begin{align*}
A\big( \underline{\vartheta(p)} \big)^\dag &= \hat\tau \Big( A\big( \underline{\vartheta(p)} \big) \Big) = \hat\tau \Big( \mathcal Y\big( A^{\ii \vartheta}_p, \underline{u(p)} \big) \Big) = \hat\tau \big( m_{p, W}( A^{\ii \vartheta}_p ) \big)\\
&= m_{\tau(p), \tau(W)} \big( \hat\tau A^{\ii \vartheta}_p \big) = (-1)^{\Delta_A + \bar \Delta_A} m_{\tau(p), \tau(W)} \big( (\hat \tau A)^{\ii \vartheta}_{\tau(p)} \big)\\
&= (-1)^{\Delta_A + \bar \Delta_A} \mathcal Y\big( (\hat\tau A)^{\ii \vartheta}_{\tau(p)}, \underline{u(\tau(p))} \big)
= (-1)^{\Delta_A + \bar\Delta_A} (\hat \tau A)\big( \underline{\vartheta(\tau(q))} \big)
\end{align*}
where the fourth step we uses Proposition \ref{prop: equiv PFac} and the fifth step is by Lemma \ref{lem: Y equiv}. The result \eqref{adjoint quantum operator} now follows by using the property \eqref{vartheta property}. Using the definition of the Fourier modes in Proposition \ref{prop: Fourier series} to rewrite both sides of the last expression we find
\begin{equation*}
\sum_{n, \bar n \in \ZZ} A^\dag_{n, \bar n} \, e^{\ii n \bar \vartheta(q)} e^{-\ii \bar n \vartheta(q)} = (-1)^{\Delta_A + \bar\Delta_A} \sum_{n, \bar n \in \ZZ} (\hat\tau A)_{n, \bar n} \, e^{-\ii n \vartheta(\tau(q))} e^{\ii \bar n \bar \vartheta(\tau(q))},
\end{equation*}
where on the left hand side we applied the anti-linear map $(\cdot)^\dag$ defined in \eqref{adjoint operator} to each term of the Fourier series from Proposition \ref{prop: Fourier series}. The final result about the Fourier modes now follows from comparing both sides of the above using the property \eqref{vartheta property}.
\end{proof}
\end{proposition}

The above notation $(\cdot)^\dag$ for the adjoint of Fourier modes is consistent with the anti-linear anti-involution $(\cdot)^\dag$ on the modes of $\hg \oplus \hbg$ introduced in the Kac-Moody case in \eqref{check tau def}. 

Indeed, recall from after \eqref{quantum fields X bar X} that the Fourier modes $\X_n$ and $\bar \X_n$ coincide with the vertex modes \eqref{X bar X vertex modes} at $o \in \Sigma$ in the $u$ coordinate for all $n \in \ZZ$, which in turn have the same action as modes in $\hg$ and $\hbg$, respectively, when acting on states prepared at $o \in \Sigma$ in the $u$ coordinate by \eqref{endo Xn bar Xn}. Also, the commutation relations of the Fourier modes in Corollary \ref{cor: Fourier com L and X} coincide with the Lie algebra relations \eqref{KM algebra vertex modes def} in $\hg \oplus \hbg$, with the central elements ${\ms k}$ and $\bar{\ms k}$ both replaced by $1$.
Now the definition \eqref{check tau def} implies that $( \X_n )^\dag = - (\tau \X)_{-n}$ and $( \bar{\Y}_{\bar n} )^\dag = - (\overline{\tau \Y})_{-\bar n}$, which are equivalent to the statement of Proposition \ref{prop: Fourier equiv} for the states $\X_{(-1)} \vac$ and $\bar Y_{(-1)} \vac$.
Furthermore, the property from Proposition \ref{prop: pairing Kac property} of the Hermitian sesquilinear form on $\hVV^{\fl,\alpha} \cong U(\hg \oplus \hbg) \vac$ defined in \eqref{Hermitian form} follows from the next two propositions.

\begin{proposition} \label{prop: adjoint operator}
Suppose $\mathcal O \in \U\L^\Sigma_\alpha(W)$, for an annulus shaped open subset $W \subset \Sigma^\prime$ around $o \in \Sigma$ in the coordinate $u$, acts on $\hV^{\fl,\alpha}_o$. Then for any $B, C \in \hVV^{\fl,\alpha}$ we have
\begin{equation*}
\langle B, \mathcal O \, C \rangle = \langle \mathcal O^\dag B, C \rangle.
\end{equation*}
In particular, for all states $A, B, C \in \hVV^{\fl,\alpha}$ we have
$\big\langle B, A(\underline{\zeta}) C \big\rangle = \big\langle A(\underline{\zeta})^\dag B, C \big\rangle$ or, equivalently, we have $\langle B, A_{n, \bar n} C \rangle = \langle A^\dag_{n, \bar n} B, C \rangle$ for any $n, \bar n \in \ZZ$.
\begin{proof}
This is very similar to the proof of Proposition \ref{prop: inv bilinear form}. For any $B, C \in \hVV^{\fl,\alpha}$, we compute
\begin{equation*}
\Big\langle m_{(o', W, o), \Sigma}\big( (\hat\tau B)^{u^\tmo}_{o'} \otimes \mathcal O \otimes C^u_o \big) \Big\rangle \in \CC
\end{equation*}
in two different ways using the associativity \eqref{PFA commutativity} of the factorisation product.

On the one hand, by definition of the action of $\mathcal O$ on $\hV^{\fl,\alpha}_o$ we have
\begin{equation*}
\Big\langle m_{(o', U), \Sigma}\Big( (\hat\tau B)^{u^\tmo}_{o'} \otimes m_{(W, o), U}\big( \mathcal O \otimes C^u_o \big) \Big) \Big\rangle = \Big\langle m_{(o', U), \Sigma}\Big( (\hat\tau B)^{u^\tmo}_{o'} \otimes ( \mathcal O \, C )^u_U \Big) \Big\rangle
\end{equation*}
for some open subset $o \in U \subset \Sigma \setminus \{ o' \}$. By definition \eqref{Hermitian form} of the Hermitian sesquilinear form on $\hVV^{\fl,\alpha}$, the above is simply $\langle B, \mathcal O \, C \rangle$. 

On the other hand, since $\mathcal O = \hat\tau (\mathcal O^\dag)$ we can compute the same factorisation product as
\begin{equation*}
\Big\langle m_{(o', W, o), \Sigma}\big( \hat\tau (B^u_o) \otimes \hat\tau ( \mathcal O^\dag ) \otimes C^u_o \big) \Big\rangle = \Big\langle m_{(U', o), \Sigma}\Big( \big( \hat\tau ( \mathcal O^\dag B) \big)^{u^\tmo}_{U'} \otimes C^u_o \Big) \Big\rangle
\end{equation*}
where in the second expression we introduced an open $o' \in U' \subset \Sigma \setminus \{ o \}$. By definition \eqref{Hermitian form} of the Hermitian sesquilinear form on $\hVV^{\fl,\alpha}$, the above is just $\langle \mathcal O^\dag B, C \rangle$, hence the result.
\end{proof}
\end{proposition}

\begin{proposition}
For any $A, B \in \hVV^{\fl,\alpha}$ and $m, \bar m, n, \bar n \in \ZZ$ we have
\begin{equation*}
(A_{m, \bar m} B_{n, \bar n})^\dag = B_{n, \bar n}^\dag A_{m, \bar m}^\dag.
\end{equation*}
In other words, taking the adjoint \eqref{adjoint operator} is an anti-linear \emph{anti}-involution on Fourier modes.
\begin{proof}
Recall that the product of Fourier modes, described explicitly at the start of \S\ref{sec: Borcherds Fourier}, is induced by the factorisation product of $\U\L^\Sigma_\alpha$. By Proposition \ref{prop: equiv PFac}, for any inclusion $U \sqcup V \subset W$ of disjoint annuli shaped open subsets around $o \in \Sigma$ in the $u$ coordinate into another larger such annuli shaped open subset $W \subset \Sigma^\prime$, we have the commutative diagram
\begin{equation*}
\begin{tikzcd}[column sep=30mm]
\U\L^\Sigma_\alpha(U) \otimes \U\L^\Sigma_\alpha(V) \arrow[r, "m_{(U,V), W}"] \arrow[d, "\hat\tau \,\otimes\, \hat\tau"'] & \U\L^\Sigma_\alpha(W) \arrow[d, "\hat\tau"] \\
\U\L^\Sigma_\alpha \big(\tau(U) \big) \otimes \U\L^\Sigma_\alpha \big(\tau(V) \big) \arrow[r, "m_{(\tau(U), \tau(V)), \tau(W)}"'] & \U\L^\Sigma_\alpha \big(\tau(W) \big)
\end{tikzcd}
\end{equation*}
The result now follows since $\tau$ reverses the ordering along the cylinder $\Sigma^\prime$ of annuli shaped open subsets around $o \in \Sigma$ in the $u$ coordinate.
\end{proof}
\end{proposition}

\appendix

\section{Twisted prefactorisation envelopes} \label{sec: twisted PF env}

\subsection{Unital \dg{} Lie algebras} \label{sec: udgLie}

Following \cite[Example 2.9]{Bruinsma:2018knq}, let $\uLie$ be the \dg{} operad generated by an element $b \in \uLie(2)$ in degree $0$ and an element $u \in \uLie(0)$ in degree $1$, which we represent graphically as trees
\begin{equation*}
b \; \longleftrightarrow \; \raisebox{-3mm}{\begin{tikzpicture}[scale=.5pt]
\draw (.293,.293) -- (1,1);
\draw (1.707,.293) -- (1,1);
\draw (1,1) -- (1,1.9);
\filldraw[black] (1,1) circle (3pt);
\end{tikzpicture}}, \qquad
u \; \longleftrightarrow \; \raisebox{0mm}{\begin{tikzpicture}[scale=.5pt]
\draw (0,1) -- (0,1.9);
\draw[circle,draw=black, fill=white] (0,1) circle (3pt);
\end{tikzpicture}} \;,
\end{equation*}
such that $\d b = 0$, $\d u=0$ and subject to three relations which we can represent most easily in graphical form as
\begin{equation} \label{uLie relations}
\raisebox{-7mm}{\begin{tikzpicture}[scale=.5pt]
\draw (.293,.293) node[below]{\mbox{\tiny $1$}} -- (1,1);
\draw (1.707,.293) node[below]{\mbox{\tiny $2$}} -- (1,1);
\draw (1,1) -- (1,1.9);
\filldraw[black] (1,1) circle (3pt);
\end{tikzpicture}}
\; + \; \raisebox{-7mm}{\begin{tikzpicture}[scale=.5pt]
\draw (.293,.293) node[below]{\mbox{\tiny $2$}} -- (1,1);
\draw (1.707,.293) node[below]{\mbox{\tiny $1$}} -- (1,1);
\draw (1,1) -- (1,1.9);
\filldraw[black] (1,1) circle (3pt);
\end{tikzpicture}}
\; = 0,  \qquad
\raisebox{-7mm}{\begin{tikzpicture}[scale=.5pt]
\draw (.293,.293) node[below]{\mbox{\tiny $1$}} -- (1,1);
\draw (1.707,.293) node[below]{\mbox{\tiny $3$}} -- (1,1);
\draw (1.354,.646) -- (1,.293) node[below]{\mbox{\tiny $2$}};
\draw (1,1) -- (1,1.9);
\filldraw[black] (1,1) circle (3pt);
\filldraw[black] (1.354,.646) circle (3pt);
\end{tikzpicture}}
\;+\;
\raisebox{-7mm}{\begin{tikzpicture}[scale=.5pt]
\draw (.293,.293) node[below]{\mbox{\tiny $2$}} -- (1,1);
\draw (1.707,.293) node[below]{\mbox{\tiny $1$}} -- (1,1);
\draw (1.354,.646) -- (1,.293) node[below]{\mbox{\tiny $3$}};
\draw (1,1) -- (1,1.9);
\filldraw[black] (1,1) circle (3pt);
\filldraw[black] (1.354,.646) circle (3pt);
\end{tikzpicture}}
\;+\;
\raisebox{-7mm}{\begin{tikzpicture}[scale=.5pt]
\draw (.293,.293) node[below]{\mbox{\tiny $3$}} -- (1,1);
\draw (1.707,.293) node[below]{\mbox{\tiny $2$}} -- (1,1);
\draw (1.354,.646) -- (1,.293) node[below]{\mbox{\tiny $1$}};
\draw (1,1) -- (1,1.9);
\filldraw[black] (1,1) circle (3pt);
\filldraw[black] (1.354,.646) circle (3pt);
\end{tikzpicture}}
\; = 0, \qquad
\raisebox{-3.5mm}{\begin{tikzpicture}[scale=.5pt]
\draw (.293,.293) -- (1,1);
\draw (1.707,.293) -- (1,1);
\draw (1,1) -- (1,1.9);
\filldraw[black] (1,1) circle (3pt);
\draw[circle,draw=black, fill=white] (1.707,.293) circle (3pt);
\end{tikzpicture}}
\; = \; 0,
\end{equation}
where the numbers below the trees indicate input permutations. We shall refer to an algebra in $\dgVec_\CC$ over the \dg{} operad $\uLie$ as a \emph{unital \dg{} Lie algebra}. We let $\udgLie_\CC \coloneqq \Alg_{\uLie}(\dgVec_\CC)$ denote the category of unital \dg{} Lie algebras. Explicitly, a unital \dg{} Lie algebra $L$ is described by a multifunctor $\uLie \to \dgVec_\CC^{\otimes}$. More explicitly, such a multifunctor singles out an object $L \in \dgVec_\CC$ as the image of the single object of $\uLie$ and closed linear maps $[\cdot, \cdot] : L \otimes L \to L$ and $\eta : \CC \to L$, of degrees $0$ and $1$ respectively, as the images of $b \in \uLie(2)$ and $u \in \uLie(0)$. In particular, the first two relations in \eqref{uLie relations} make $L$ into a \dg{} Lie algebra and the last relation says that the image of the \emph{unit} $\eta : \CC \to L$, which is closed in $L$, is also central in $L$.

\subsubsection{Monoidal structure}

For any unital \dg{} Lie algebras $L, L' \in \udgLie_\CC$, with respective units $\eta : \CC \to L$ and $\eta' : \CC \to L'$, we define the unital \dg{} Lie algebra
\begin{equation*}
L \boplus L' \coloneqq (L \oplus L') \big/ \im(\eta - \eta')
\end{equation*}
where the quotient is by the image of the linear map $\eta - \eta' : \CC \to L \oplus L'$, which is a \dg{} Lie ideal of $L \oplus L'$. The unit in $L \boplus L'$ is the map induced by $\eta : \CC \to L \oplus L'$, or equivalently by $\eta' : \CC \to L \oplus L'$. The identity object for the monoidal product $\boplus$ on $\udgLie_\CC$ is the trivial unital \dg{} Lie algebra $\CC[-1]$.

We will use the following universal property of the direct sum of unital \dg{} Lie algebras.

\begin{lemma} \label{lem: direct sum Lie}
Let $f_i : L_i \to L$ for $i \in I$ be any collection of morphisms of unital \dg{} Lie algebras, with indexing set $I$, such that $[\im f_i, \im f_j] = 0$ for every $i \neq j \in I$. There exists a unique morphism of unital \dg{} Lie algebras $\overline{\bigoplus}_{j \in I} L_j \to L$ such that the diagram
\begin{equation*}
\begin{tikzcd}
& L & \\
L_i \arrow[ur, "f_i"] \arrow[r, "\iota_i"'] &\overline{\bigoplus}_{j \in I} L_j \arrow[u, "\exists!"', dashed]
\end{tikzcd}
\end{equation*}
is commutative for each $i \in I$, where $\iota_i : L_i \to \overline{\bigoplus}_{j \in I} L_j$ is the canonical embedding.

Moreover, suppose $f'_i : L'_i \to L'$ for $i \in I$ is another collection of morphisms of unital \dg{} Lie algebras such that $[\im f'_i, \im f'_j] = 0$ for every $i \neq j \in I$ and suppose we are given morphisms of unital \dg{} Lie algebras $\phi : L \to L'$ and $\phi_i : L_i \to L'_i$ for each $i \in I$ such that $\phi \circ f_i = f'_i \circ \phi_i$. Then there exists a unique morphism of unital \dg{} Lie algebras $\overline{\bigoplus}_{j \in I} \phi_j : \overline{\bigoplus}_{j \in I} L_j \to \overline{\bigoplus}_{j \in I} L'_j$ making the following diagram
\begin{equation*}
\begin{tikzcd}
& L \arrow[rd, "\phi"] &\\
L_i \arrow[rd, "\phi_i"'] \arrow[r, "\iota_i"'] \arrow[ru, "f_i"] & \overline{\bigoplus}_{j=1}^n L_j \arrow[u] \arrow[rd, dashed, "\exists!" very near end] & L'\\
& L'_i \arrow[r, "\iota'_i"'] \arrow[ru, "f'_i", very near start, crossing over] & \overline{\bigoplus}_{j=1}^n L'_j \arrow[u]
\end{tikzcd}
\end{equation*}
commute for each $i \in I$, where the vertical morphisms are defined as above.
\begin{proof}
We first define $h : \bigoplus_{j \in I} L_j \to L$ by $(\ms x_i)_{i \in I} \mapsto \sum_{i \in I} f_i(\ms x_i)$. This is well defined since the sum over $i \in I$ is finite by virtue of $\ms x_i$ being zero for all but finitely many $i \in I$. And since $f_i$ are morphisms of unital \dg{} Lie algebras we have $f_i \circ \eta_i = \eta$ where $\eta_i$ denotes the unit in each $L_i$ and $\eta$ the unity in $L$. It follows that $h$ factors through a map $\bar h : \overline{\bigoplus}_{j \in I} L_j \to L$, defined by $[(\ms x_i)_{i \in I}] \mapsto \sum_{i \in I} f_i(\ms x_i)$ where $[(\ms x_i)_{i \in I}] \in \overline{\bigoplus}_{j \in I} L_j$ is the class of $(\ms x_i)_{i \in I} \in \bigoplus_{j \in I} L_j$. In other words, $\bar h\big( [(\ms x_i)_{i \in I}] \big) = h\big( (\ms x_i)_{i \in I} \big)$.
But we also have
\begin{align*}
\big[ h ((\ms x_i)_{i \in I}), h((\ms y_i)_{i \in I}) \big] &= \sum_{i, j \in I} [f_i(\ms x_i), f_j(\ms y_j)] = \sum_{i \in I} [f_i(\ms x_i), f_i(\ms y_i)]\\
&= \sum_{i \in I} f_i([\ms x_i, \ms y_i]) = h\big( ([\ms x_i, \ms y_i])_{i \in I} \big) = h\big( [(\ms x_i)_{i\in I}, (\ms y_i)_{i\in I}] \big)
\end{align*}
where the second step is by the assumption that $[\im f_i, \im f_j] = 0$ for $i \neq j \in I$. The third step follows since each $f_i$ is a morphism of unital \dg{} Lie algebras and the last step uses the \dg{} Lie algebra structure on $\overline{\bigoplus}_{i \in I} L_i$. So $h$ is a morphism of \dg{} Lie algebras and hence $\bar h$ is a morphism of unital \dg{} Lie algebras. By construction, the latter is unique such that $h \circ \iota_i = f_i$.

\medskip

Let us consider now the second claim. Since $\im (\iota'_i \circ \phi_i) \subset \im \iota'_i$ for each $i \in I$ it follows that $[\im (\iota'_i \circ \phi_i), \im (\iota'_j \circ \phi_j)] = 0$ for each $i \neq j \in I$. By the first part of the lemma applied to the morphisms of unital \dg{} Lie algebras $\iota'_i \circ \phi_i : L_i \to \overline{\bigoplus}_{j=1}^n L'_j$ we thus have a unique morphism of unital \dg{} Lie algebras $\overline{\bigoplus}_{j \in I} \phi_j : \overline{\bigoplus}_{j \in I} L_j \to \overline{\bigoplus}_{j \in I} L'_j$ which makes the bottom square of the second diagram in the statement commute. It remains to show that it also makes the square on the right of the diagram commute, i.e. that $\bar h' \circ \big( \overline{\bigoplus}_{j \in I} \phi_j \big) = \phi \circ \bar h$ where $\bar h$ and $\bar h'$ are the two vertical morphisms defined as above.

Now $\phi \circ f_i = f'_i \circ \phi_i$ and $[\im f'_i, \im f'_j] = 0$ for every $i \neq j \in I$ from which it follows that $[\im (\phi \circ f_i), \im (\phi \circ f_j)] = 0$ for every $i \neq j \in I$. Therefore, applying the first part of the lemma to the collection of morphisms of unital \dg{} Lie algebras $\phi \circ f_i : L_i \to L'$ we obtain a unique mophism of unital \dg{} Lie algebras $g : \overline{\bigoplus}_{j=1}^n L_j \to L'$ such that $g \circ \iota_i = \phi \circ f_i$ for each $i \in I$. Yet we have
\begin{equation*}
\bar h' \circ \bigg( \overline{\bigoplus}_{j \in I} \phi_j \bigg) \circ \iota_i = \bar h' \circ \iota'_i \circ \phi_i = f'_i \circ \phi_i = \phi \circ f_i = \phi \circ \bar h \circ \iota_i
\end{equation*}
from which we deduce that $\bar h' \circ \big( \overline{\bigoplus}_{j \in I} \phi_j \big) = g = \phi \circ \bar h$ by uniqueness of $g$, as required.
\end{proof}
\end{lemma}

\subsubsection{Chevalley-Eilenberg functor for \texorpdfstring{$\udgLie_\CC$}{udgLie}} \label{sec: CE for udgLie}

The homological Chevalley-Eilenberg functor is the symmetric monoidal functor
\begin{align} \label{CE functor}
\CE_\chain : (\dgLie_\CC, \oplus) &\longrightarrow (\dgVec_\CC, \otimes) \notag\\
L &\longmapsto \CE_\chain(L) \coloneqq \Big( \big( \Sym (L[1]) \big)^\chain, \d_{\CE} \Big),
\end{align}
where $(\Sym V)^\chain \coloneqq \bigoplus_{n \geq 0} \Sym^n V$ is the free commutative graded algebra on a graded vector space $V$ and $\Sym^n V$ is its component of word length $n$. Here $L[1] \coloneqq \CC[1] \otimes L \in \dgVec_\CC$ denotes the suspension of the \dg{} vector space $L$ where $\CC[1]$ is the \dg{} vector space with $\CC$ placed in degree $-1$. For any element $v \in V$ of a \dg{} vector space $V$ we let $sv \coloneqq 1 \otimes v \in V[1]$ denote its suspension, and we use the notation $s^{-1} v \coloneqq 1 \otimes v \in V[-1] \coloneqq \CC[-1] \otimes V$ to denote its inverse suspension.
The product in $(\Sym V)^\chain$ is denoted by concatenation.
The differential on $\CE_\chain(L)$ is $\d_\CE \coloneqq \d_{L[1]} + \d_{[\cdot,\cdot]}$ where $\d_{L[1]} : L[1] \to L[1]$ is induced by the differential $\d_L : L \to L$ of the \dg{} Lie algebra $L$ and $\d_{[\cdot,\cdot]} : ( \Sym (L[1]) )^\chain \to ( \Sym (L[1]) )^\chain$ is the unique coderivation of the graded coalgebra $( \Sym (L[1]) )^\chain$ extending the degree $1$ map $\Sym^2 (L[1]) \to L[1]$, where $\Sym^2 (L[1])$ is the graded symmetric tensor square, induced by the Lie bracket $[\cdot, \cdot] : L \otimes L \to L$, see for instance \cite[Lemma 22.2]{RationalHomotopy}.
Explicitly, $\Sym^2 (L[1]) \to L[1]$ is given by $s\ms x \, s\ms y \mapsto (-1)^{|\ms x|+1} s[\ms x, \ms y]$ for any homogeneous $\ms x, \ms y \in L$. This is well defined since $s\ms x \, s\ms y = (-1)^{|\ms x| |\ms y| + |\ms x| + |\ms y| + 1} s\ms y \, s\ms x$ is sent to $(-1)^{|\ms x| |\ms y| + |\ms x| + |\ms y| + 1} (-1)^{|\ms y|+1} s[\ms y, \ms x] = (-1)^{|\ms x|+1} s[\ms x, \ms y]$ using the graded \emph{skew}-symmetry of the Lie bracket, namely $[\ms y, \ms x] = - (-1)^{|\ms x| |\ms y|} [\ms x, \ms y]$. Note that $\d_{L[1]} \d_{[\cdot,\cdot]} = - \d_{[\cdot,\cdot]} \d_{L[1]}$ and $\d_{[\cdot,\cdot]}^2 = 0$.

We shall need a variant of the functor \eqref{CE functor} for unital \dg{} Lie algebras defined in the next proposition. Let $I^i_L \coloneqq \{ \mathcal A \, s(\eta(1)) - \mathcal A \,|\, \mathcal A \in \CE_i(L) \}$ for every $i \in \ZZ$. This is a subspace of $\CE_i(L)$ since $\eta(1)$ is of degree $1$. Moreover, since $\eta(1)$ is a cocycle, i.e. $\d_L \eta(1) = 0$, and is central in $L$, it follows that $I^\chain_L$ is a \dg{} vector subspace of $\CE_\chain(L)$.

\begin{proposition} \label{prop: bCE functor}
We have a symmetric monoidal functor
\begin{align} \label{bar CE functor}
\bCE_\chain : (\udgLie_\CC, \boplus) &\longrightarrow (\dgVec_\CC, \otimes) \notag\\
L &\longmapsto \bCE_\chain(L) \coloneqq \CE_\chain(L) \big/ I^\chain_L
\end{align}
which preserves quasi-isomorphisms.
\begin{proof}
For any morphism of unital \dg{} Lie algebras $f : L \to L'$, the morphism of \dg{} vector spaces $\CE_\chain(f) : \CE_\chain(L) \to \CE_\chain(L')$ is given in degree $i \in \ZZ$ by $\CE_i(f) = \big( \Sym f[1] \big)^i$, where the morphism of \dg{} vector spaces $f[1] : L[1] \to L'[1]$ is the suspension of $f : L \to L'$. Since $f(\eta(1)) = \eta'(1)$ it follows that $\CE_\chain(f)(I^\chain_L) \subset I^\chain_{L'}$ and therefore $\CE_\chain(f)$ induces a morphism of \dg{} vector spaces $\bCE_\chain(f) : \bCE_\chain(L) \to \bCE_\chain(L')$.

Since the functor \eqref{CE functor} preserves quasi-isomorphisms, if $f : L \qSimTo L'$ is a quasi-isomorphism then so is $\CE_\chain(f) : \CE_\chain(L) \qSimTo \CE_\chain(L')$. Now the morphism $\bCE_\chain(f) : \bCE_\chain(L) \to \bCE_\chain(L')$ is a retract of the latter since we have a commutative diagram
\begin{equation*}
\begin{tikzcd}
\bCE_\chain(L) \arrow[r, "i_L"] \arrow[d, "\bCE_\chain(f)"'] & \CE_\chain(L) \arrow[r, "q_L"] \arrow[d, "\CE_\chain(f)"] & \bCE_\chain(L) \arrow[d, "\bCE_\chain(f)"]\\
\bCE_\chain(L') \arrow[r, "i_{L'}"'] & \CE_\chain(L') \arrow[r, "q_{L'}"'] & \bCE_\chain(L')
\end{tikzcd}
\end{equation*}
where $q_L : \CE_\chain(L) \to \bCE_\chain(L)$ is the canonical map and $i_L : \bCE_\chain(L) \to \CE_\chain(L)$ is given by taking the representative with no factors of $s(\eta(1))$ so that we clearly have $q_L \circ i_L = \id_{\bCE_\chain(L)}$ and similarly $q_{L'} \circ i_{L'} = \id_{\bCE_\chain(L')}$.
Hence $\bCE_\chain(f) : \bCE_\chain(L) \qSimTo \bCE_\chain(L')$ is a quasi-isomorphism.

It remains to show that the functor $\bCE_\chain$ is symmetric monoidal. In particular, we must show that given any unital \dg{} Lie algebras $L, L' \in \udgLie_\CC$ we have a canonical isomorphism of \dg{} vector spaces
\begin{equation*}
\bCE_\chain(L \boplus L') \cong \bCE_\chain(L) \otimes \bCE_\chain(L').
\end{equation*}
To see this, let $J^i \coloneqq \{ \mathcal A \, s(\eta(1)) - \mathcal A \, s(\eta'(1)) \,|\, \mathcal A \in \CE_i(L \oplus L') \}$ for every $i \in \ZZ$. Since $\eta(1)$ and $\eta'(1)$, thought of as elements in $L \oplus L'$, are both central cocycles of degree $1$, it follows that $J^\chain$ is a \dg{} vector subspace of $\CE_\chain(L \oplus L')$. Moreover, we have a canonical isomorphism of \dg{} vector spaces $\CE_\chain(L \boplus L') \cong \CE_\chain(L \oplus L') / J^\chain$. Also, introducing the \dg{} vector subspace $K^\chain$ of $\CE_\chain(L \oplus L')$ with components $K^i \coloneqq \{ \mathcal A \, s(\eta(1)) - \mathcal A + \mathcal A' \, s(\eta'(1)) - \mathcal A' \,|\, \mathcal A, \mathcal A' \in \CE_i(L \oplus L') \}$ for every $i \in \ZZ$, of which $J^\chain$ is an obvious \dg{} vector subspace, we have a canonical isomorphism $I^\chain_{L \boplus L'} \cong K^\chain / J^\chain$. We therefore have
\begin{align*}
&\bCE_\chain(L \boplus L') = \CE_\chain(L \boplus L') / I^\chain_{L \boplus L'} \cong \big( \CE_\chain(L \oplus L') / J^\chain \big) \big/ \big( K^\chain / J^\chain \big) \cong \CE_\chain(L \oplus L') / K^\chain\\
&\quad \cong \big( \CE_\chain(L) \otimes \CE_\chain(L') \big) / \big( I^\chain_L \otimes \CE_\chain(L') + \CE_\chain(L) \otimes I^\chain_{L'} \big) = \bCE_\chain(L) \otimes \bCE_\chain(L'),
\end{align*}
where in the first isomorphism we made use of the two isomorphisms stated above. The next isomorphism is by the third isomorphism theorem and the last isomorphism uses the fact that the homological Chevalley-Eilenberg functor $\CE_\chain$ is symmetric monoidal.

Also, the result of applying the functor $\bCE_\chain$ to the identity object $\CC[-1]$ of the symmetric monoidal product $\boplus$ on $\udgLie_\CC$ is isomorphic to $\CC[0]$, i.e. the identity object of the symmetric monoidal product $\otimes$ on $\dgVec_\CC$. Hence $\bCE_\chain$ is a symmetric monoidal functor.
\end{proof}
\end{proposition}

It will be useful to also introduce, mainly for the purpose of proving part of Proposition \ref{prop: Y map injective}, the category of \emph{pointed} \dg{} vector spaces, denoted $\udgVec_\CC$, whose objects are vector spaces $V$ equipped with a degree $1$ map $\eta : \CC \to V$ and whose morphisms preserve this structure. We also introduce the symmetric monoidal functor
\begin{align} \label{bSym functor def}
\bSym_\chain : (\udgVec_\CC, \boplus) &\longrightarrow (\dgVec_\CC, \otimes) \notag\\
V &\longmapsto \bSym_\chain(V) \coloneqq \big( \Sym (V[1]) \big)^\chain \big/ J^\chain_V
\end{align}
where the differential on $\big( \Sym (V[1]) \big)^\chain$ is the one induced from $\d_{V[1]} : V[1] \to V[1]$ and we set $J^i_V \coloneqq \{ \A \, s(\eta(1)) - \A \,|\, \A \in (\Sym (V[1]))^i \}$ for each $i \in \ZZ$, which is a subpsace of $(\Sym (V[1]))^i$ since $\eta(1)$ has degree $1$. Moreover, since $\eta(1)$ is a cocycle, namely $\d_V \eta(1) = 0$, these form a \dg{} vector subspace $J^\chain_V$ of $\big( \Sym (V[1]) \big)^\chain$. The proof that \eqref{bSym functor def} is symmetric monoidal is completely analogous to that of Proposition \ref{prop: bCE functor}.

\subsection{Unital local Lie algebras} \label{sec: uLocLie}

Let $\L$ be a precosheaf of unital \dg{} Lie algebras on a manifold $D$, with extension morphisms denoted by $\ext_{U, V} : \L(U) \to \L(V)$ for any inclusion of subsets $U \subset V$ in $D$. We say that $\L$ is a \emph{unital local Lie algebra} on $D$ if for any finite collection $\{ U_i \}_{i=1}^n$ of disjoint open subsets $U_i \subset V$ of an open subset $V \subset D$ we have
\begin{equation} \label{LocLie condition}
\big[ \im \big( \! \ext_{U_i, V} \! \big), \im \big(\! \ext_{U_j, V} \! \big) \big] = 0,
\end{equation}
for every $i \neq j$. We denote by $\uLocLie_\CC(D)$ the category of unital local Lie algebras on $D$, where a morphism of unital local Lie algebras is defined as a morphism of the underlying precosheaves of unital \dg{} Lie algebras.

\begin{proposition} \label{prop: constructing prefac 1}
We have a canonical functor
\begin{equation*}
\uLocLie_\CC(D) \longrightarrow \PFac(D, \udgLie_\CC^{\boplus}).
\end{equation*}
More explicitly, any $\L \in \uLocLie_\CC(D)$ defines an object in $\PFac(D, \udgLie_\CC^{\boplus})$ which by a slight abuse of notation we also denote by $\L$. Moreover, every morphism $\L \to \L'$ of $\uLocLie_\CC(D)$ induces a morphism of $\PFac(D, \udgLie_\CC^{\boplus})$ which we also denote $\L \to \L'$.
\begin{proof}
Let $\L \in \uLocLie_\CC(D)$. Given any open subsets $\sqcup_{i=1}^n U_i \subset V$, by Lemma \ref{lem: direct sum Lie} we obtain a unique morphism of \dg{} Lie algebras $m^\L_{(U_i), V} : \overline{\bigoplus}_{j=1}^n \L(U_j) \to \L(V)$ such that the diagram
\begin{equation*}
\begin{tikzcd}
& \L(V) & \\
\L(U_i) \arrow[ur, "\ext^\L_{U_i, V}"] \arrow[r] & \overline{\bigoplus}_{j=1}^n \L(U_j) \arrow[u, "\exists! m^\L_{(U_i), V}"', dashed]
\end{tikzcd}
\end{equation*}
commutes. The composition property of these morphisms, required in order for $\L$ to define a prefactorisation algebra, then follows from their uniqueness property.
Explicitly, we have the following commutative diagram
\begin{equation*}
\begin{tikzcd}
& & \L(W) & \\
& \L(V_i) \arrow[ur, "\ext^\L_{V_i, W}"] \arrow[r] & \overline{\bigoplus}_{i=1}^n \L(V_i) \arrow[u, "\exists! m^\L_{(V_i), W}"', dashed]\\
\L(U_{ij}) \arrow[ur, "\ext^\L_{U_{ij}, V_i}"] \arrow[r] & \overline{\bigoplus}_{j=1}^{m_i} \L(U_{ij}) \arrow[u, "\exists! m^\L_{(U_{ij}), V_i}"', dashed] \arrow[r] & \overline{\bigoplus}_{i=1}^n \overline{\bigoplus}_{j=1}^{m_i} \L(U_{ij}) \arrow[u, "\overline{\bigoplus}_{i=1}^n m^\L_{(U_{ij}), V_i}"', dashed] \arrow[uu, "\exists! m^\L_{(U_{ij}), W}"', bend right=90, dashed]
\end{tikzcd}
\end{equation*}
where the two small commutative triangles are given by the universal property in Lemma \ref{lem: direct sum Lie}. The square on the bottom right of the diagram is commutative by definition of the morphism of unital \dg{} Lie algebras $\overline{\bigoplus}_{i=1}^n m^\L_{(U_{ij}), V_i}$, in the second part of Lemma \ref{lem: direct sum Lie}. The extra morphism from the bottom right to the top right is given by the universal property from Lemma \ref{lem: direct sum Lie} applied to the big outside triangle. Its uniqueness implies the required composition property. It follows that $\L$ defines an element of $\PFac(D, \udgLie_\CC^{\boplus})$, as required.

\medskip

Let $\phi : \L \to \L'$ be a morphism of $\uLocLie_\CC(D)$. By using the second part of Lemma \ref{lem: direct sum Lie} we have a morphism $\overline{\bigoplus}_{j=1}^n \phi_{U_j} : \overline{\bigoplus}_{j=1}^n \L(U_j) \to \overline{\bigoplus}_{j=1}^n \L'(U_j)$ making the following diagram
\begin{equation*}
\begin{tikzcd}
& \L(V) \arrow[rd, "\phi_V"] &\\
\L(U_i) \arrow[rd, "\phi_{U_i}"'] \arrow[r] \arrow[ru] & \overline{\bigoplus}_{j=1}^n \L(U_j) \arrow[u, "m^\L_{(U_i), V}"'] \arrow[rd, dashed, "\exists!" very near end] & \L'(V)\\
& \L'(U_i) \arrow[r] \arrow[ru, crossing over] & \overline{\bigoplus}_{j=1}^n \L'(U_j) \arrow[u, "m^{\L'}_{(U_i), V}"']
\end{tikzcd}
\end{equation*}
commutative. In particular, the commutativity of the right hand square is equivalent to the statement that $\phi : \L \to \L'$ is a morphism of $\PFac(D, \udgLie_\CC^{\boplus})$, as required.
\end{proof}
\end{proposition}

Viewing any $\L \in \uLocLie_\CC(D)$ as a multifunctor $\Top(D)^\sqcup \to \udgLie_\CC^{\boplus}$, by Proposition \ref{prop: constructing prefac 1}, we can consider its post-composition with the symmetric monoidal functor \eqref{bar CE functor} from Proposition \ref{prop: bCE functor} to obtain a prefactorisation algebra $\bCE_\chain \, \L \in \PFac(D, \dgVec_\CC^{\otimes})$. Post-composing the latter with the lax monoidal $0^{\rm th}$ cohomology functor $H^0 : \dgVec_\CC \to \Vec_\CC$ then yields the \emph{twisted prefactorisation envelope} of the unital local Lie algebra $\L$, denoted
\begin{equation} \label{PFA envelope}
\U \L \coloneqq H^0 \, \bCE_\chain \, \L \in \PFac(D, \Vec_\CC^{\otimes}).
\end{equation}

\end{document}